\documentclass[11pt]{article}

\usepackage[left=1in, top=1in, right=1in, bottom=1in]{geometry}
\usepackage{amsmath}
\usepackage{multirow}
\usepackage{bbm}
\usepackage{extarrows}
\usepackage{bm}
\usepackage{array}
\usepackage{xfrac}
\usepackage{amsthm}
\usepackage[dvipsnames]{xcolor}
\usepackage{mathrsfs}
\usepackage{amsfonts}
\usepackage{amssymb}
\usepackage{enumerate}
\usepackage{subfigure}
\usepackage{graphicx}
\usepackage{tikz}
\usetikzlibrary{arrows}
\usepackage{multicol}
\usepackage{wrapfig}
\usepackage{algpseudocode}
\algtext*{EndIf}
\usepackage[Algorithmus]{algorithm}
\usepackage{xcolor}
\usepackage[linktoc=page]{hyperref}
\hypersetup{
    colorlinks,
    linkcolor={red},
    citecolor={blue},
    urlcolor={lightgray}
}

\newtheorem{lemma}{Lemma}
\newtheorem{mainlemma}{Main Lemma}

\newtheorem{theorem}{Theorem}
\newtheorem*{theorem*}{Theorem}
\newtheorem{fact}{Fact}

\newtheorem{definition}{Definition}
\newtheorem{corollary}{Corollary}
\newtheorem{example}{Example}

\title{Tight Approximation Ratio of Anonymous Pricing}
\author{
Yaonan Jin\thanks{Department of IEDA, Hong Kong University of Science and Technology. {\tt yjinan@ust.hk}}
\and Pinyan Lu\thanks{ITCS, Shanghai University of Finance and Economics. {\tt lu.pinyan@mail.shufe.edu.cn}}
\and
Qi Qi\thanks{Department of IEDA, Hong Kong University of Science and Technology. {\tt kaylaqi@ust.hk}}
\and
Zhihao Gavin Tang\thanks{Department of Computer Science, the University of Hong Kong. {\tt zhtang@cs.hku.hk}}
\and
Tao Xiao\thanks{Department of Computer Science, Shanghai Jiao Tong University. {\tt xt\_1992@sjtu.edu.cn}}
}

\date{}

\newcommand{\opt}{\textsf{OPT}}

\newcommand{\upm}{\textsf{UP}}

\newcommand{\ap}{\textsf{AP}}

\newcommand{\R}{\mathcal{R}}
\newcommand{\Q}{\mathcal{Q}}
\newcommand{\C}{\mathcal{C}^*}
\newcommand{\bupp}{\textsf{BUPP}}
\newcommand{\bmumd}{\textsf{BMUMD}}

\newcommand{\reg}{\textsc{Reg}}
\newcommand{\tri}{\textsc{Tri}}
\newcommand{\cont}{\textsc{Cont}}

\newcommand{\FF}{\overline{F}}
\newcommand{\ff}{\overline{f}}
\newcommand{\DD}{\overline{D}}
\newcommand{\rr}{\overline{r}}
\newcommand{\uu}{\overline{u}}
\newcommand{\vv}{\overline{v}}
\newcommand{\qq}{\overline{q}}
\newcommand{\hp}{\widehat{p}}

\newcommand{\ggamma}{\overline{\gamma}}
\newcommand{\hgamma}{\widehat{\gamma}}
\newcommand{\PPhi}{\overline{\Phi}}

\newcommand{\eqdef}{\stackrel{\textrm{def}}{=}}
\newcommand{\E}{\mathbb{E}}

\begin{document}
\maketitle

\begin{abstract}
We consider two canonical Bayesian mechanism design settings. In the \emph{single-item} setting, we prove \emph{tight} approximation ratio for anonymous pricing: compared with Myerson Auction, it extracts at least $\frac{1}{2.62}$-fraction of revenue; there is a matching lower-bound example.

In the \emph{unit-demand single-buyer} setting, we prove \emph{tight} approximation ratio between the simplest and optimal deterministic mechanisms: in terms of revenue, uniform pricing admits a $2.62$-approximation of item pricing; we further validate the tightness of this ratio.

These results settle two open problems asked in~\cite{H13,CD15,AHNPY15,L17,JLTX18}. As an implication, in the \emph{single-item} setting: we improve the approximation ratio of the second-price auction with anonymous reserve to $2.62$, which breaks the state-of-the-art upper bound of $e \approx 2.72$.
\end{abstract}


\setcounter{page}{0}
\thispagestyle{empty}

\newpage

\section{Introduction}
\label{sec:intro}
Consider two simple revenue-maximization scenarios: (1)~one seller has $n$ items to sell to a unit-demand buyer; (2)~one seller has a single item to sell to $n$ potential buyers. In both cases, the seller only knows the buyers' valuation distributions $\{F_i\}_{i=1}^n$ (rather than their exact valuations for the items), and thus would like to maximize his expected revenue.

The simplest mechanisms are to post a fixed price for the item(s). In the former case, the seller can post a uniform price for all items, after which the unit-demand buyer will take his favorite item (given that its valuation exceeds the price). In the second case, the seller can post an anonymous price for the single item, and then the first buyer (in arbitrary orders), who values the item no less than the price, will take the item.

In practice, the above simple pricing schemes are widely used. In terms of revenue, however, they are not necessarily optimal. In the first scenario, the seller can increase his expected revenue by posting item-specific prices\footnote{It is noteworthy, in the unit-demand single-buyer setting, that finding optimal item pricing is NP-hard~\cite{CDPSY18}.}, hence the \emph{item pricing} mechanisms in the literature~\cite{CHK07,CHMS10,CD15,CMS15,CDW16,CDPSY18}. In the second scenario, the seller can even organize an auction, and therefore increase his revenue by leveraging the buyers' competition. To this end, it is well-known that the optimal mechanism is the remarkable \emph{Myerson Auction}~\cite{M81}.

Compared with the simple pricing schemes, how much extra revenues can the complicated mechanisms derive? This is a central question in the theory of Bayesian mechanism design, compiled as the ``\emph{simple versus optimal}'' program. As we quote from the survey by Lucier~\cite{L17}: ``\emph{an interesting question is how well one can approximate the optimal revenue using an anonymous price, rather than personalized prices.}''
 In this paper, we settle the approximation ratios of both simple pricing schemes mentioned above.
\vspace{5.5pt} \\
\fbox{\begin{minipage}{\textwidth}
\begin{theorem}
\label{thm:opt_ap}
To sell a single item among multiple buyers with independent regular distributions, the supremum of the ratio of Myerson Auction to anonymous pricing equals to $\C \approx 2.6202$.
\end{theorem}
\end{minipage}}
\vspace{5.5pt} \\
\fbox{\begin{minipage}{\textwidth}
\begin{theorem}
\label{thm:bupp_upm}
To sell heterogeneous items to a unit-demand buyer with independent regular distributions for different items, the supremum of the ratio of optimal item pricing to uniform pricing equals to $\C \approx 2.6202$.
\end{theorem}
\end{minipage}}
\vspace{5.5pt}

In statements of both theorems, we introduce constant
\[
\C \eqdef 2 + \displaystyle{\int_1^{\infty}} \left(1 - e^{-\Q(x)}\right) dx \approx 2.6202,
\]
where $\Q(p) \eqdef \ln\left(\frac{p^2}{p^2 - 1}\right) - \frac{1}{2} \cdot \sum\limits_{k = 1}^{\infty} \frac{1}{k^2} \cdot p^{-2k}$ for all $p \in (1, \infty)$. Noticeably, the imposed \emph{regularity} assumption on distributions (whose formal definition is deferred to Section~\ref{sec:prelim}) is very standard in economics, and is used in previous works as well. When we allow arbitrarily weird distributions, by contrast, both gaps can be as large as $n$ (see~\cite{AHNPY15} for more details). Since such weird distributions are uncommon in practice, these results might not be that informative.

The main body of this paper is devoted to proving Theorem~\ref{thm:opt_ap}. Given that, the upper-bound part of Theorem~\ref{thm:bupp_upm} comes immediately, after we apply the so-called \emph{copying} technique developed in~\cite{CHK07,CHMS10,CMS15}. Further, the lower bound example for Theorem~\ref{thm:bupp_upm} follows from that for Theorem~\ref{thm:opt_ap}, after re-interpretation and fine-turning. All analyses of Theorem~\ref{thm:bupp_upm} can be found in Section~\ref{sec:extension}. As another implication of Theorem~\ref{thm:opt_ap}, the ratio of Myerson Auction to the \emph{second-price auction with anonymous reserve}~\cite{M81,HR09,AFHH13,CGM15,JLTX18} is improved to $\C \approx 2.62$. This is a main open problem left by Hartline and Roughgarden in~\cite{HR09}, who proved an upper bound of $4$, and a lower bound of $2$. Prior to our work, the known best upper bound is $e \approx 2.72$~\cite{AHNPY15}, and the known best lower bound is $2.15$~\cite{JLTX18}. Whether our upper bound of $\C \approx 2.62$ is tight or not remains unknown.

\subsection{Our Techniques}

To conquer Theorem~\ref{thm:opt_ap}, like \cite{AHNPY15,JLTX18}, we represent the ratio by mathematical programming, and then manually solve the optimal objective value. The variables of this programming are an instance of the underlying mechanism design problem, which consists of $n$ different distributions $\{F_i\}_{i=1}^n$ and the number of $n$ itself. Given a certain instance, (1)~the constraint is, for any posted-price, that the revenue from anonymous pricing is at most $1$; and (2)~the objective is to maximize the revenue from Myerson Auction.

The revenue from anonymous pricing is relatively apparent (in terms of distributions $\{F_i\}_{i=1}^n$), but the revenue from Myerson Auction is far more complicated. To escape from this dilemma, Alaei et al.~\cite{AHNPY15} relaxed Myerson Auction to the so-called \emph{ex-ante relaxation}, and thus obtain a ``handy'' upper bound of objective function. However, this relaxation incurs an inevitable cost: although the exact ratio of $e$ was caught (between ex-ante relaxation and anonymous pricing), this bound is no longer tight for the original programming.

One can infer from the above, we must investigate Myerson Auction directly, and acquire its revenue (in terms of distributions $\{F_i\}_{i=1}^n$). To do so, a main obstacle is that Myerson Auction is stated in terms of \emph{virtual valuations}~\cite{M81}. We introduce the virtual value distributions $\{D_i\}_{i = 1}^n$ in Section~\ref{sec:prelim}, hence the desired objective function.
Throughout our analysis, three different formats for a certain distribution are involved: value CDF $F_i$, virtual value CDF $D_i$ and \emph{revenue-quantile curve} $r_i(q)$. There are one-to-one correspondences among these formats. In different parts of the proof, we use the format that is best for our analyses. To this end, another challenge occur -- how to relate these three formats -- which is quite non-trivial. For this, we observe several interesting identities (among the formats), which enables our proof.

In both~\cite{AHNPY15,JLTX18}, a very first step (towards their programming) is to show, that a worst-case instance is reached by the so-called \emph{triangular distributions} -- a subset of regular distributions. Regularity of a distribution requires its revenue-quantile curve to be a \emph{concave} function. By contrast, that curve of a triangular distribution is basically a triangle, which lies in the boundary between concavity and convexity. In other words, triangular distributions are in the boundary of regular distributions.
As a salient feature, describing a triangular distribution requires merely two real numbers. In this, the programming in both~\cite{AHNPY15,JLTX18} is greatly simplified, after the reductions (to triangular distributions) mentioned above.

Unfortunately, our programming seems not to admit such a reduction. In fact, what enables reductions in~\cite{AHNPY15,JLTX18} is that both objective functions\footnote{In specific, the objective function in~\cite{AHNPY15} is the aforementioned \emph{ex-ante relaxation}, and the objective function in~\cite{JLTX18} is the so-called \emph{sequential posted-pricing} mechanism.} therein contain a fantastic structure: given an instance $\{F_i\}_{i = 1}^n$, they only depend on a specific set of value-probability pairs $\big\{\big(v_i, F_i(v_i)\big)\big\}_{i = 1}^n$, rather than any other details. Therefore, we can push the instance to ``extreme''\footnote{This ``extreme'' is indeed a triangular instance: as mentioned before, triangular distributions are in the boundary of regular distributions; besides, each pair $\big(v_i, F_i(v_i)\big)$ refers to the two real number needed to specify the distribution.}, while keeping the pairs $\big\{\big(v_i, F_i(v_i)\big)\big\}_{i = 1}^n$ fixed, and keeping the distributions regular. This is exactly the reductions in~\cite{AHNPY15,JLTX18}. Back to our objective function, the revenue from Myerson Auction depends on all details about instance $\{F_i\}_{i = 1}^n$; this is why the desired reduction seems hopeless.

Despite of annoyance from regularity, we can still calibrate a worst-case instance from different angles. We notice that a (worst-case) revenue-quantile curve must be a triangle/quadrangle with one \emph{curved-edge}. It is this curved-edge that makes the revenue-quantile curve difficult to handle. With a concave revenue-quantile curve (namely a regular distribution), the curved-edge is upper-bounded by each tangent line of it, and is lower-bounded by the line connecting its two ends. In our analysis, we often exploit these two lines (rather than the original curved-edge) to measure a revenue-quantile curve. These two lines respectively serve as upper and lower bounds, which results in certain properties of a worst-case distribution.


For all programming in~\cite{AHNPY15,JLTX18} and in this paper, a worst-case instance is reached when there are infinite buyers (namely when $n \rightarrow \infty$), and each buyer's buying-probability is infinitesimal. Such instances refer to what we call \emph{continuous instance}. In such a instance, conceivably, the revenue from Myerson Auction is given by an integration. This explains why constant $\C \approx 2.62$ (in our main theorems) comes from an improper integration.
To settle our programming, it remains to show that a continuous instance does give the best revenue. We adopt a direct proof plan: given a fixed number $n \in \mathbb{N}$ and a feasible instance $\{F_i\}_{i = 1}^n$, \emph{we gradually transform it to a targeted continuous instance, without hurting the object value.}

In both~\cite{AHNPY15,JLTX18}, a similar philosophy is adopted. Another salient feature of triangular instance simplifies those proofs once again: there is a natural total-order\footnote{As mentioned before, a triangular distribution (1)~has a triangular revenue-quantile curve; and (2)~can be entirely defined by using two parameters. One of the two parameter is exactly the \emph{slope} of a specific edge (namely the \emph{tangent} of a specific angle) of the triangle. In both~\cite{AHNPY15,JLTX18}, the total-order refers to the decreasing-order of the slopes.} in a triangular instance. In other words, a set of triangular distributions can be transformed one by one in that order (to a target continuous instance), without introducing interface in each step.

However, generic regular distributions do not follow such total-order. Therefore, (locally) each distribution has to be modified piece by piece; (globally) all distributions have to be modified simultaneously; which incur many technical challenges. To conquer these issues, \emph{potential function} comes to the rescue: we find a natural potential to indicate status of the current instance; in each step of our transformation, the potential declines by a pre-fixed amount. Hence, the transformation will terminate after finite steps, left with a continuous instance (as desired). In some sense, this potential function is the fourth representation of a distribution; in this term, the constraint of the programming becomes simple enough to deal with.

\subsection{Related Works}

Both of anonymous pricing and uniform pricing mechanisms are widely studied in literature~\cite{GHKKKM05,B08,BH08,H13,AHNPY15,CD15,HH15,DFK16,JLTX18}. In the single-item setting, another important family of pricing schemes is the \emph{sequential posted-pricing} mechanisms~\cite{CHK07,HKS07,CEGMM10a,CHMS10,Y11,H13,A14,DFK16,AEEHK17,CFHOV17,L17,ACK18,BGLPS18,JLTX18}, which allows buyer-specific pricing strategies, and thus dominates anonymous pricing in revenue. Tight ratio between
sequential posted-pricing and anonymous pricing is known to be $2.62$~\cite{JLTX18}.

In the single-buyer unit-demand setting, item pricing schemes cover all deterministic mechanisms, among which finding the optimum is NP-hard~\cite{CDPSY18}. When randomization is allowed, the seller can further improve revenue by employing \emph{lottery}~\cite{HN13,CDOPSY15,CMS15,CDW16}. Chen et al.~\cite{CDOPSY15} settled the complexity of finding/describing optimal randomized mechanism.

In broader multi-item settings (with single or multiple buyers having unit-demand and other utility functions), optimal mechanisms are even far more complicated. In these setting, the last two decades have seen a number of works on proving intractability of optimal mechanisms, and even more abundant works on proving that simple mechanisms approximate the optima within constant factors. In this amount of space, evaluating so extensive literature would be impossible.
As a guideline, readers can turn to hardness results in~\cite{HN12,DDT14,CDOPSY15,R16:b,BGN17,CDPSY18,CMPY18}, and approximation results in~\cite{CHK07,CHMS10,CD15,BILW14,CMS15,Y15,CDW16,CZ17}, and the references therein.


\section{Notations and Preliminaries}
\label{sec:prelim}
In most part of the paper, we mainly focus on the simplest single-item environment. One seller has a single item to sell, among $n$ potential buyers. Each buyer $i \in [n]$ has valuation drawn from distribution $F_i$. Therefore, an instance is described by a $n$ distributions $\{F_i\}_{i = 1}^{n}$. We would study these distributions in great detail. In particular, we make use of three mathematically equivalent formats to describe the distributions. In the following three subsections, we will introduce them and some related notations one-by-one.

\subsection{CDF and PDF}
The most natural way to describe distributions is their cumulative density function (CDF) and  probability density function (PDF).
\begin{itemize}
\item $F_i$: cumulative density function (CDF) of a distribution. For convenience, we would assume that cumulative density function $F_i$ is left-continuous\footnote{That is, with random variable $x$ drawn from the underlying distribution, $F_i(p) = \Pr\big\{x < p\big\}$.} (rather than right-continuous, as usual). When without ambiguity, we also use $F_i$ to denote the corresponding distribution.
\item $f_i$: probability density function (PDF) of a distribution. By assuming that distribution $F_i$ is regular, soon after we will see (1)~the support is a \emph{single} closed interval; and (2)~there exists at most one probability-mass (which must be the support-supremum, if exists). Hence, we can safely assume that probability density function $f_i$ is well-defined and left-continuous.
\end{itemize}
We introduce some more useful notions for distribution $F_i$ as follows. (see Figure~\ref{fig:prelim:cdf} for demonstration).
\begin{itemize}
\item $u_i \in (0, \infty]$: support-supremum of distribution $F_i$. Note that $u_i$ can be infinity.
\item $v_i \in \arg\max \big\{p \cdot \big(1 - F_i(p)\big):\,\, p \in (0, \infty]\big\}\in (0, \infty]$. If there are multiple alternative $v_i$'s, we break ties by choosing the greatest one. Clearly, $v_i$ can also be infinity, yet cannot exceed $u_i$.
\item $q_i \eqdef 1 - F_i(v_i) \in [0, 1]$.
\end{itemize}
To interpret these notions intuitively, consider a single-item auction: the seller posts price $p$ for the item, and a single buyer with value drawn from distribution $F_i$ decides whether to buy. Since then, $u_i$ is the greatest possible value (of the buyer); $v_i$ is the revenue-optimal price (posted by the seller), termed the \emph{monopoly-price}; $q_i$ is the selling probability corresponding to monopoly-price $v_i$, termed \emph{monopoly-quantile}. We can easily check that $q_i = 0$ iff $v_i = u_i = \infty$.

\subsection{Regularity and Virtual Value}
\label{subsec:prelim-reg-vv}

Conventionally, regularity is elaborated by the notion of \emph{virtual value}: given CDF $F_i$, suppose that its PDF $f_i$ is well-defined on $p \in (0, \infty)$, and further define virtual value function\footnote{When $F_i(p) = 1$ and $f_i(p) = 0$, we let $\frac{1 - F_i(p)}{f_i(p)} = 0$. When $F_i(p) < 1$ and $f_i(p) = 0$, we let $\frac{1 - F_i(p)}{f_i(p)} = \infty$.}
\[
\Phi_i(p) \eqdef p - \frac{1 - F_i(p)}{f_i(p)} \quad\quad\quad\quad \forall p \in (0, \infty].
\]
Denote by $\reg$ the set of regular distributions. By standard notion~\cite{M81}, distribution $F_i$ is regular (namely $F_i \in \reg$) iff its virtual value function $\Phi_i$ is \emph{non-decreasing} on $p \in (0, \infty]$. As a result, for a regular distribution $F_i \in \reg$, we know\footnote{For Fact~\textbf{(a)}, assume to the contrary that the support contains two closed intervals $[p_1, p_2]$ and $[p_3, p_4]$, where $0 < p_1 < p_2 < p_3 < p_4 < \infty$. When $p \in (p_1, p_2)$, definitely $\Phi_i(p)$ is finite. When $p \in (p_2, p_3)$, we have $F_i(p) < 1$ and $f_i(p) = 0$, which means $\Phi_i(p)$ is negative infinity. This violates the monotonicity of $\Phi_i$. Similarly, Fact~\textbf{(b)} can be verified by contradiction (respect the monotonicity of $\Phi_i$).}
\begin{description}
\item [(a):] The support must be a \emph{single} closed interval;
\item [(b):] There exists at most one probability-mass (which must be the support-supremum, if exists).
\end{description}

Given the monotonicity of virtual value function $\Phi_i$, we define its inverse function as follows:
\[
\Phi_i^{-1}(p) \eqdef \arg\max \big\{x \in (0, \infty]:\,\, \Phi_i(x) \leq p\big\} \quad\quad\quad\quad \forall p \in (0, \infty].
\]
One may notice that this function is supported on positive virtual values. Indeed, to compute the revenue from Myerson Auction, only positive virtual values are involved (see Fact~\ref{fact:opt_rev} in Section~\ref{subsec:prelim-mechanism}), due to which we adopt this way of definition.

Further, we introduce virtual value CDF $D_i$ of this distribution:
\[
D_i(p) = F_i\big(\Phi_i^{-1}(p)\big) \quad\quad\quad\quad \forall p \in (0, \infty].
\]
By definition, we know that\footnote{Throughout the paper, for any function $g$, we use $g(p^+)$ to denote the right limit of $g$ at $p$, i.e., $\lim\limits_{x \to p^+} g(x)$.} $\Phi_i(v_i^+) \geq 0$ and $D_i(p) = 1 - q_i$ for all $p \in \big(0, \Phi_i(v_i^+)\big]$, which are illustrated in Figure~\ref{fig:prelim:virtaul_cdf}. Clearly, an equivalent condition for regularity is that the virtual value CDF is well-defined.

Essentially, value CDF $F_i$ can be transformed to virtual value CDF $D_i$, and vice versa. The above arguments capture one-side transformation (since virtual value function $\Phi_i$ is given by $F_i$). The other-side transformation is relatively unintuitive, and is deferred to Fact~\ref{fact2}.2 in Appendix~\ref{app:prelim}.

\begin{figure}[H]
\centering
\subfigure[value CDF]{
\begin{tikzpicture}[thick, smooth, scale = 2.5]
\node[above] at (0, 1.1) {\small $F_i(p)$};
\draw[->] (0, 0) -- (0, 1.1);
\node[right] at (1.6, 0) {\small $p$};
\draw[->] (0, 0) -- (1.6, 0);
\draw[thick, color = blue, domain = 0.8333: 1.6] plot (\x, {\x - sqrt((\x - 1)^2 + 0.16)});
\draw[thick, color = blue, domain = 0: 0.8333] plot (\x, {(2 * \x + 5 - sqrt(4 * \x^2 - 30 * \x + 25)) / 12.5});
\node[below] at (0.8333, 0) {\small $v_i$};
\node[left] at (0, 1) {\small $1$};
\node[below] at (0, 0) {\small $0$};
\node[left] at (0, 0.4) {\small $1 - q_i$};
\draw[style = dashed] (0, 1) -- (1.6, 1);
\draw[style = dashed] (0.8333, 0) -- (0.8333, 0.4);
\draw[style = dashed] (0, 0.4) -- (0.8333, 0.4);
\draw[very thick] (0pt, 1) -- (1.25pt, 1);
\draw[very thick] (0pt, 0.4) -- (1.25pt, 0.4);
\draw[very thick] (0.8333, 0pt) -- (0.8333, 1.25pt);
\label{fig:prelim:cdf}
\end{tikzpicture}
}
\quad
\subfigure[virtual value CDF]{
\begin{tikzpicture}[thick, smooth, scale = 2.5]
\node[above] at (0, 1.1) {\small $D_i(p)$};
\draw[->] (0, 0) -- (1.1, 0);
\draw[->] (0, 0) -- (0, 1.1);
\node[right] at (1.1, 0) {\small $p$};
\node[below] at (0.4, 1.15pt) {\scriptsize $\Phi_i(v_i^+)$};
\node[below] at (1, 0.5pt) {\scriptsize $\Phi_i(u_i)$};
\node[left] at (0, 1) {\small $1$};
\node[below] at (0, 0) {\small $0$};
\node[left] at (0, 0.4) {\small $1 - q_i$};
\draw[style = dashed] (0, 1) -- (1, 1);
\draw[style = dashed] (0.4, 0) -- (0.4, 0.4);
\draw[style = dashed] (1, 0) -- (1, 1);
\draw[color = blue] (0, 0.4) -- (0.4, 0.4) -- (1, 1) -- (1.1, 1);

\draw[very thick] (0pt, 1) -- (1.25pt, 1);
\draw[very thick] (0pt, 0.4) -- (1.25pt, 0.4);
\draw[very thick] (0.4, 0pt) -- (0.4, 1.25pt);
\draw[very thick] (1, 0pt) -- (1, 1.25pt);
\label{fig:prelim:virtaul_cdf}
\end{tikzpicture}
}
\quad
\subfigure[revenue-quantile curve]{
\begin{tikzpicture}[thick, smooth, scale = 2.5]
\draw[->] (0, 0) -- (1.1, 0);
\draw[->] (0, 0) -- (0, 1.1);
\node[above] at (0, 1.1) {\small $r_i(q)$};
\node[right] at (1.1, 0) {\small $q$};
\node[below] at (0, 0) {\small $0$};
\draw[thick, color = blue, domain = 0: 0.6] plot (\x, {1.16 - (\x - 1)^2});
\draw[thick, color = blue, domain = 0.6: 1] plot (\x, {-6.25 * (\x - 0.2) * (\x - 1)});
\draw[dashed] (0.6, 0) -- (0.6, 1);
\draw[dashed] (0, 1) -- (0.6, 1);
\draw[very thick] (1, 0) -- (1, 1.25pt);
\node[below] at (1, 0) {\small $1$};
\draw[very thick] (0.6, 0) -- (0.6, 1.25pt);
\node[below] at (0.6, 0) {\small $q_i$};
\draw[very thick] (0, 1) -- (1.25pt, 1);
\draw[very thick] (0, 0.16) -- (1.25pt, 0.16);
\node[left] at (0, 1) {\small $v_i q_i$};
\node[left] at (0, 0.16) {\small $r_i(0)$};
\draw[color = blue, fill = red] (0, 0.16) circle(0.625pt);
\draw[color = blue, fill = red] (0.6, 1) circle(0.625pt);
\end{tikzpicture}
\label{fig:prelim:rq-curve}
}
\caption{Demonstrations for value CDF, virtual value CDF and revenue-quantile curve.}
\end{figure}
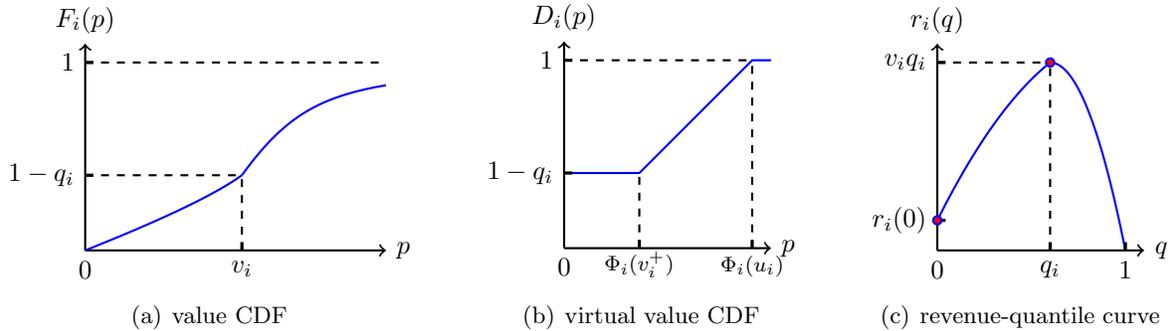

\subsection{Revenue-Quantile Curve}
\label{subsec:prelim-rq-curve}

Equivalently, to depict the whole distribution, one can use revenue-quantile curve $r_i(q)$ that takes a quantile $q \in [0, 1]$ as input, and outputs corresponding revenue (see Figure~\ref{fig:prelim:rq-curve} for illustration). Formally, we have $r_i(0) = \lim\limits_{p \rightarrow \infty} p \cdot \big(1 - F_i(p)\big)$ and
\[
r_i(q) = q \cdot \max\big\{p \in (0, \infty]:\,\, F_i(p) \leq 1 - q\big\} \quad\quad\quad\quad  \forall q \in (0, 1].
\]
For brevity, we might interchange notations $F_i$, $D_i$ and $r_i(q)$ to denote an underlying distribution.

By standard notion, an equivalent condition for regularity is that revenue-quantile curve $r_i(q)$ is \emph{continuous} and \emph{concave} on $q \in [0, 1]$. Clearly, both properties together imply that
\begin{itemize}
\item $r_i(q)$ is left-/right-differentiable at any $q \in (0, 1)$, and is right-differentiable at $q = 0$, and is left-differentiable at $q = 1$;
\item By standard notion, the virtual value at $p = \big.r_i(q)\big/q$ is precisely the right-derivative of $r_i(q)$. Formally, $\Phi_i\big(\big.r_i(q)\big/q\big) = \partial_+ r_i(q)$, for all $q \in [0, 1)$;
\item Respecting the semi-differentiability of $r_i(q)$, we can safely assume that PDF $f_i$ is well-defined (recall that $\Phi_i(p) = p - \frac{1 - F_i(p)}{f_i(p)}$).
\end{itemize}
Later, we will interchange all equivalent definitions of regularity, whenever one is more convenient for our use. In terms of revenue-quantile curve $r_i(q)$, support-supremum $u_i$, monopoly-price $v_i$ and monopoly-quantile $q_i$ can be re-defined as follows:
\begin{itemize}
\item $u_i \eqdef \lim\limits_{q \rightarrow 0^+} \frac{r_i(q)}{q}$: clearly, support-supremum $u_i = \infty$ when $r_i(0) > 0$.
\item $q_i \in \arg\max \big\{r_i(q):\,\, q \in [0, 1]\big\}$: recall that $q_i$ is the quantile at monopoly-price $v_i$. If there are multiple alternative $q_i$'s, we break ties by choosing the smallest one.
\item $v_i \eqdef \frac{r_i(q_i)}{q_i}$: we know monopoly-price $v_i$ is the greatest one among the alternatives, because $r_i(q)$ is concave, and $q_i$ is the smallest alternative one.
\end{itemize}
Clearly, a concave revenue-quantile curve $r_i(q)$ can be converted into value CDF $F_i$ and virtual value CDF $D_i$, and vice versa. In Appendix~\ref{app:prelim}, we formalize this statement as Fact~\ref{fact1} and Fact~\ref{fact2}, and then obtain a structural result (i.e., Main Lemma~\ref{lem:virtual_value}) that will be useful in Section~\ref{subsec:upper3-opt}.


\subsection{Mechanisms}
\label{subsec:prelim-mechanism}

Until Section~\ref{sec:extension}, we always concentrate on single-item environment, where $n$ buyers draw bids\footnote{The posted-price strategy (concerned here) and all other mechanisms studied in this paper are truthful, and therefore each buyer always bids his true value.} $\{b_i\}_{i = 1}^{n}$ independently from distributions $\{F_i\}_{i = 1}^{n} \in \reg^n$. Furthermore, we explore revenue gap between the following two families of mechanisms.

\paragraph{Anonymous Posted-Pricing ($\ap$).}
In such a mechanism, the seller posts the same price of $p \in [0, \infty]$ for all buyers. The item is allocated if there exists a buyer $i$ bidding $b_i \geq p$ (and remains unsold if no such buyers exist), and the buyer who gets the item pays the posted-price $p$. It is easy to see that the expected revenue from such mechanism is
\[
\ap\big(p, \{F_i\}_{i = 1}^n\big) \eqdef p \cdot \left(1- \prod\limits_{i = 1}^{n} F_i(p)\right).
\]
Among all such mechanisms, we abuse $\ap\big(\{F_i\}_{i = 1}^n\big) \eqdef \max \big\{\ap\big(p, \{F_i\}_{i = 1}^n\big): \,\, p \in [0, \infty]\big\}$ to denote the optimal revenue. When term $\{F_i\}_{i = 1}^{n}$ is clear from the context, we would drop that for notational simplicity (the same below for $\opt$).

\paragraph{Myerson Auction ($\opt$).}
Given a regular instance $\{F_i\}_{i = 1}^n \in \reg^n$, Myerson Auction runs as follows\footnote{For a general instance, Myerson Auction can be more complicated, and randomization can be necessary.}: After receiving bids $\{b_i\}_{i = 1}^n$, the seller first calculates virtual values $\{\Phi_i(b_i)\}_{i = 1}^n$, and then allocates the item to the buyer with highest virtual value that is at least $0$, and charges this buyer with a threshold bid (for the buyer to keep winning). If no buyer has non-negative virtual values, the seller would withhold the item, and charge nothing.

Indeed, an allocation rule and a payment rule compose a mechanism. For a truthful single-item mechanism (e.g., Myerson Auction), it is well-known~\cite{M81} that the allocation rule is monotone\footnote{That is, in the specific mechanism, the winner keeps to win after unilaterally increasing his bid.}, and the payment rule satisfies the property from Lemma~\ref{lem:myerson} in the below.

\begin{lemma}[Myerson Lemma~\cite{M81}]
\label{lem:myerson}
Given a monotone allocation rule $\bm{x}(\cdot)$, there exists a unique payment rule $\bm{\pi}(\cdot)$ such that the mechanism is truthful:
\[
\pi_i(\bm{b}) = \displaystyle{\int}_0^{b_i} z \cdot dx_i(z, \bm{b}_{-i}) \quad\quad \forall i \in [n] \text{, } \forall \bm{b} \in \mathbb{R}_+^n.
\]
\end{lemma}
Equipped with this structural result, we capture the revenue from Myerson Auction in Fact~\ref{fact:opt_rev} (whose proof is deferred to Appendix~\ref{subapp:prelim:opt_rev}), in the concerned setting with regular distributions\footnote{For general distributions, all $D_i$'s in Fact~\ref{fact:opt_rev} should be replaced by the virtual value CDF's after ironing~\cite{M81}.}.

\begin{fact}[Characterization]
\label{fact:opt_rev}
Given a regular instance $\{F_i\}_{i = 1}^n \in \reg^n$,
\[
\opt\big(\{F_i\}_{i = 1}^n\big) = \sum\limits_{i = 1}^n r_i(0) + \displaystyle{\int_0^{\infty}} \left(1 - \prod\limits_{i = 1}^n D_i(x)\right) dx.
\]
\end{fact}

Thus far, we can easily formulate the revenue gap between $\opt$ and $\ap$ as Program~(\ref{prog:0}) in the below: the constraint on $p \in [0, 1]$ always holds, and thus is dropped.
\vspace{5.5pt} \\
\fcolorbox{black}{lightgray}{\begin{minipage}{\textwidth}
\begin{align}
\label{prog:0}\tag{P0}
\max\limits_{\{F_i\}_{i = 1}^n \in \reg^n} &\quad\quad \opt = \sum\limits_{i = 1}^n r_i(0) + \displaystyle{\int_0^{\infty}} \left(1 - \prod\limits_{i = 1}^n D_i(x)\right) dx \\
\label{cstr:ap0}\tag{C0}
\text{subject to:} &\quad\quad \ap(p) = p \cdot \left(1 - \prod\limits_{i = 1}^n F_i(p)\right) \leq 1 \quad\quad \forall p \in (1, \infty]
\end{align}
\end{minipage}}

\subsection{Lower-Bound Instance}
\label{subsec:extension:opt_ap_lower}

As a subset of regular distributions $\reg$, the family of \emph{triangular distributions} $\tri$ (introduced by Alaei et al.~\cite{AHNPY15}) is also used in this paper, serving to construct worst-case instances. Previously, with a triangular instance (i.e., Example~\ref{exp:opt-ap-lower}), Jin et al.~\cite{JLTX18} gave a lower bound of Program~(\ref{prog:0}) that matches to our Theorem~\ref{thm:opt_ap}. In the following, we present these notion and example.

\begin{figure}[H]
\centering
\subfigure[Revenue-quantile curve]{
\begin{tikzpicture}[thick, smooth, domain = 0: 1, scale = 2.5]
\draw[->, thick] (0, 0) -- (1.1, 0);
\draw[->, thick] (0, 0) -- (0, 1.1);
\node[above] at (0, 1.1) {\small $r_i(q)$};
\node[right] at (1.1, 0) {\small $q$};
\draw[dashed] (0.3333, 0) -- (0.3333, 1);
\draw[dashed] (0, 1) -- (0.3333, 1);
\draw[very thick] (1, 0pt) -- (1, 1.25pt);
\node[below] at (1, 0) {\small $1$};
\draw[very thick] (0.3333, 0pt) -- (0.3333, 1.25pt);
\node[below] at (0.3333, 0) {\small $q_i$};
\draw[very thick] (0pt, 1) -- (1.25pt, 1);
\node[left] at (0, 1) {\small $v_i q_i$};
\node[below] at (0, 0) {\small $0$};

\draw[color = blue, very thick] (0, 0) -- (0.3333, 1) -- (1, 0);
\draw [color = blue, fill = red] (0.3333, 1) circle(0.625pt);
\end{tikzpicture}
\label{fig:tri_revenue_quntile}
}
\quad\quad\quad\quad\quad\quad
\subfigure[Value CDF]{
\begin{tikzpicture}[thick, smooth, domain = 0: 1.5, scale = 2.5]
\draw[->] (0, 0) -- (2, 0);
\draw[->] (0, 0) -- (0, 1.1);
\node[above] at (0, 1.1) {\small $F_i(p)$};
\node[right] at (2, 0) {\small $p$};
\node[below] at (1.5, 0) {\small $v_i$};
\node[left] at (0, 1) {\small $1$};
\node[below] at (0, 0) {\small $0$};
\node[left] at (0, 0.6667) {\small $1 - q_i$};
\draw[style = dashed] (0, 1) -- (1.5, 1);
\draw[style = dashed] (1.5, 0) -- (1.5, 1);

\draw[color = blue] plot (\x, {\x / (\x + 0.75)});
\draw[style = dashed] (0, 0.6667) -- (1.5, 0.6667);
\draw[blue, fill = white] (1.5, 0.6667) circle (0.625pt);
\draw[color = blue] (1.5, 1) -- (2, 1);
\draw [color = blue, fill = red] (1.5, 1) circle(0.625pt);

\draw[very thick] (0pt, 1) -- (1.25pt, 1);
\draw[very thick] (0pt, 0.6667) -- (1.25pt, 0.6667);
\draw[very thick] (1.5, 0pt) -- (1.5, 1.25pt);
\end{tikzpicture}
\label{fig:tri_CDF}
}
\caption{Demonstration of triangular distributions}
\label{fig:tri}
\end{figure}
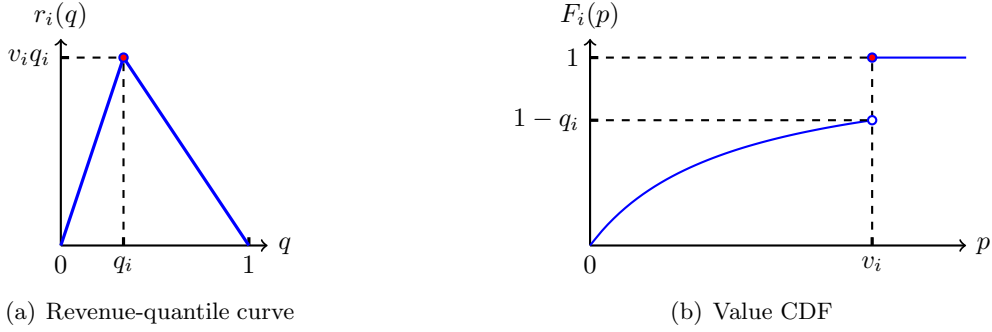

\paragraph{Triangular Distributions.}
Parameterized by $u_i = v_i \in (0, \infty]$ and $q_i \in [0, 1]$, a triangular distribution $\tri(v_i, q_i)$'s CDF is defined as:
\[
\label{eq:cdf_tri}
F_i(p) =
\begin{cases}
\frac{p \cdot (1 - q_i)}{p \cdot (1 - q_i) + v_i q_i} & \forall p \in [0, v_i] \\
1 & \forall p \in (v_i, \infty)
\end{cases}
\]
Intuitively, as Figure~\ref{fig:tri_revenue_quntile} conveys, the revenue-quantile curve of $\tri(v_i, q_i)$ looks like a triangle. Notice that the lower-bound instance by Jin et al.~\cite{JLTX18} involves a special triangular distribution\footnote{CDF $F_0(p) = \frac{p}{p + 1}$ corresponds to revenue-quantile curve $r_0(q) = 1 - q$, for all $q \in [0, 1]$.} $F_0(p) = \frac{p}{p + 1}$, for all $p \in (0, \infty)$. For brevity, we denote this distribution by $\tri(\infty)$.

\paragraph{Lower Bound.}
For triangular instances, Jin et al.~\cite{JLTX18} obtained an alternative formula for computing the revenue from Myerson Auction (see Fact~\ref{fact:revenue_spm_opt}). With a sophisticated triangular instance (i.e., Example~\ref{exp:opt-ap-lower}), they established a lower bound of $C^* \approx 2.62$ for Program~(\ref{prog:0}) (see Theorem~\ref{thm:opt_ap_lower}).

\begin{fact}[\cite{JLTX18}]
\label{fact:revenue_spm_opt}
For any triangular instance $\{\tri(v_i, q_i)\}_{i = 1}^n$ that $v_1 \geq v_2 \geq \cdots \geq v_n$,
\[
\opt\big(\{\tri(v_i, q_i)\}_{i = 1}^n\big) = \sum\limits_{i = 1}^n v_i q_i \cdot \prod\limits_{j = 1}^{i - 1} (1 - q_i).
\]
\end{fact}

\begin{example}[\cite{JLTX18}]
\label{exp:opt-ap-lower}
Given $\epsilon \in (0, 1)$, let $a \eqdef \Q^{-1}\left(\ln\frac{8}{\epsilon}\right) > 1$, $b \eqdef \frac{8}{\epsilon}$ and $\delta \eqdef \frac{b - a}{n}$. Define triangular instance $\{\tri(\infty)\} \cup \{\tri(v_i, q_i)\}_{i = 1}^{n + 1}$ as, for each $i \in [n + 1]$,
\[
v_i = b - (i - 1) \cdot \delta \quad\quad q_i = \frac{\R(v_i) - \R(v_{i - 1})}{v_i + \R(v_i) - \R(v_{i - 1})}.
\]
\end{example}

\begin{theorem}[\cite{JLTX18}]
\label{thm:opt_ap_lower}
Consider the triangular instance in Example~\ref{exp:opt-ap-lower}. When $n \in \mathbb{N}_+$ is sufficiently large, $\ap\big(\{\tri(\infty)\} \cup \{\tri(v_i, q_i)\}_{i = 1}^{n + 1}\big) \leq 1$ and further,
\[
\C \geq \opt\big(\{\tri(\infty)\} \cup \{\tri(v_i, q_i)\}_{i = 1}^{n + 1}\big) = 1 + \sum\limits_{i = 1}^{n + 1} v_i q_i \cdot \prod\limits_{j = 1}^{i - 1} (1 - q_i) \geq \C - \epsilon.
\]
\end{theorem}






\section{Proof Overview: Reductions and Optimizations}
\label{sec:overview}
Our main result is obtained by solving Program~(\ref{prog:0}). In this section, we would outline the proof of Theorem~\ref{thm:opt_ap}, and more details are deferred to Section~\ref{sec:upper1}, Section~\ref{sec:upper3} and appendices.

\subsection{Transforming Program~(\ref{prog:0}) into Program~(\ref{prog:1})}
\label{subsec:upper1-program}

We first propose several reductions among feasible instances of Program~(\ref{prog:0}). As a consequence, \emph{the optimal objective value of Program~(\ref{prog:0}) is upper-bounded by that of Program~(\ref{prog:1})}.
\vspace{5.5pt} \\
\fcolorbox{black}{lightgray}{\begin{minipage}{\textwidth}
\begin{align}
\label{prog:1}\tag{P1}
\underset{\{F_i\}_{i = 1}^n \in \reg^n}{\max} &\quad\quad \opt = 2 + \displaystyle{\int_1^{\infty}} \left(1 - \prod\limits_{i = 1}^n D_i(x)\right) dx \\
\label{cstr:value}\tag{C1}
\text{subject to:} &\quad\quad \text{$v_i > 1$ and $\Phi_i(v_i^+) \geq 1$ \quad\quad\quad\quad\quad\quad\quad\quad\quad\quad\quad\quad\quad\quad\, $\forall i \in [n]$} \\
\label{cstr:ap2}\tag{C2}
&\quad\quad \sum\limits_{i = 1}^n \ln\left(1 + p \cdot \frac{1 - F_i(p)}{F_i(p)}\right) \leq \R(p) \eqdef p\ln\left(\frac{p^2}{p^2 - 1}\right) \quad\quad \forall p \in (1, \infty)
\end{align}
\end{minipage}}

%

\paragraph{Special Buyer $\{\tri(1, 1)\}$.}
Given a feasible instance $\{F_i\}_{i = 1}^n$ of Program~(\ref{prog:0}), we would enroll a special buyer $\tri(1, 1)$, that is,
$F_{n + 1}(p) = D_{n + 1}(p) =
\begin{cases}
0 & \forall p \in (0, 1] \\
1 & \forall p \in (1, \infty)
\end{cases}$.
Afterwards,
\vspace{5.5pt} \\
\fcolorbox{white}{lightgray}{\begin{minipage}{\textwidth}
\begin{itemize}
\item The objective value of Program~(\ref{prog:0}) increases to $1 + \sum\limits_{i = 1}^n r_i(0) + \displaystyle{\int_1^{\infty}} \left(1 - \prod\limits_{i = 1}^n D_i(x)\right) dx$;
\item Constraint~(\ref{cstr:ap0}) still holds, since adding $F_{n + 1}$ has no effect on constraint~(\ref{cstr:ap0}).
\end{itemize}
\end{minipage}}
\vspace{5.5pt}

With this special buyer being recruited, each other buyer can contribute to the objective of Program~(\ref{prog:0}) only if his virtual value strictly exceeds $1$. Using the idea developed in~\cite{AHNPY15} to ``tailor'' distributions, we formalize this statement as Lemma~\ref{lem:shape} (whose proof is deferred to Section~\ref{subapp:upper1-shape}). Clearly, Lemma~\ref{lem:shape}.1 corresponds to constraint~(\ref{cstr:value}). For notational brevity, we will not explicitly mention this special buyer $\tri(1, 1)$ from now on.

\begin{lemma}
\label{lem:shape}
In a feasible instance $\{\tri(1, 1)\} \bigcup \{F_i\}_{i = 1}^n$ of Program~(\ref{prog:0}), w.l.o.g. for all $i \in [n]$,
\begin{enumerate}
\item If $q_i > 0$, then $\partial_+ r_i(q) > 1$ for all $q \in [0, q_i)$, $v_i = \frac{r_i(q_i)}{q_i} > 1$ and $\Phi_i(v_i^+) = \lim\limits_{q \rightarrow q_i^-} \partial_+ r_i(q) \geq 1$;
\item $r_i(q)$ is increasing on $q \in [0, q_i]$, and $r_i(q) = \frac{r_i(q_i)}{1 - q_i} \cdot (1 - q)$ when $q \in [q_i, 1]$.
\end{enumerate}
\end{lemma}

\paragraph{Special Buyer $\{\tri(\infty)\}$.}
We continue to explore worst-case instances of Program~(\ref{prog:0}), and thus obtain the following lemma (for which we will give a sketched proof in Section~\ref{subsec:upper1-overview}). From now on, we will not explicitly mention this special buyer $\tri(\infty)$ anywhere except for Section~\ref{sec:upper1}.
\vspace{5.5pt} \\
\fbox{\begin{minipage}{\textwidth}
\begin{mainlemma}
\label{lem:p/(p+1)}
In a worst-case instance $\{F_i\}_{i = 0}^n$ of Program~(\ref{prog:0}), w.l.o.g.
\begin{enumerate}
\item $F_0(p) = \frac{p}{p + 1}$ for all $p \in (0, \infty)$, that is, $r_0(q) = 1 - q$ for all $q \in [0, 1]$;
\item $r_i(0) = 0$ for all $i \in [n]$.
\end{enumerate}
\end{mainlemma}
\end{minipage}}

\paragraph{Getting Program~(\ref{prog:1}).}
Equipped with Main Lemma~\ref{lem:p/(p+1)}, to investigate a worst-case instance $\{F_i\}_{i = 0}^n$, we can safely rewrite Program~(\ref{prog:0}) as follows:
\vspace{5.5pt} \\
\fcolorbox{white}{lightgray}{\begin{minipage}{\textwidth}
\begin{itemize}
\item The objective value of Program~(\ref{prog:0}) becomes $2 + \displaystyle{\int_1^{\infty}} \left(1 - \prod\limits_{i = 1}^n D_i(x)\right) dx$, which is exactly same to that of Program~(\ref{prog:1});
\item Constraint~(\ref{cstr:ap0}) becomes $p \cdot \left(1 - \frac{p}{p + 1} \cdot \prod\limits_{i = 1}^n F_i(p)\right) \leq 1$, or equivalently,
    \begin{equation}
    \label{cstr:ap1}\tag{C3}
    -\sum\limits_{i = 1}^n \ln F_i(p) = \sum\limits_{i = 1}^n \ln\left(1 + \frac{1 - F_i(p)}{F_i(p)}\right) \leq \ln\left(\frac{p^2}{p^2 - 1}\right) \quad\quad \forall p \in (1, \infty).
    \end{equation}
\end{itemize}
\end{minipage}}
\vspace{5.5pt}

We would use a trick exploited in~\cite{AHNPY15} to relax constraint~(\ref{cstr:ap1}). Clearly, $\ln\left(1 + \frac{x}{p}\right) \geq \frac{1}{p} \cdot \ln(1 + x)$ for all $p > 1$ and $x \geq 0$. Applying this to constraint~(\ref{cstr:ap1}), we acquire constraint~(\ref{cstr:ap2}). Hitherto, it suffices to conquer Program~(\ref{prog:1}) instead of Program~(\ref{prog:0}).

\subsection{Continuous Instance}
\label{subsec:upper2-cont}

Our next step is to solve Program~(\ref{prog:1}). In fact, we have an optimal solution for that. We first introduce this solution and then prove its optimality. Besides, a guiding notion here refers to what we call \emph{continuous instance}: the worst-case instance of Program~(\ref{prog:1}) also falls into this family.

Before defining continuous instance, let us recall two functions (appeared in Program~(\ref{prog:1}) and Theorem~\ref{thm:opt_ap}, respectively): $\R(p) = p\ln\left(\frac{p^2}{p^2 - 1}\right)$ and $\Q(p) = \ln\left(\frac{p^2}{p^2 - 1}\right) - \frac{1}{2}Li_2\left(\frac{1}{p^2}\right)$. It turns out that these two functions are naturally correlated, based on the following mathematical facts (the proofs are deferred in Appendix~\ref{subapp:math_facts:RQ}).

\begin{lemma}[partially proved in~\cite{JLTX18}]
\label{lem:RQ}
For $\R(p) = p\ln\left(\frac{p^2}{p^2 - 1}\right)$ and $\Q(p) = \ln\left(\frac{p^2}{p^2 - 1}\right) - \frac{1}{2}Li_2\left(\frac{1}{p^2}\right)$,
\begin{enumerate}
\item $\R'(p) = p \cdot \Q'(p)$, $\R'(p) < \Q'(p) < 0$ and $\R''(p) > 0$, for all $p \in (1, \infty)$;
\item $\R(\infty) = \Q(\infty) = 0$ and $\R(1^+) = \Q(1^+) = \infty$;
\item $\frac{1}{p} \leq \R(p) \leq \frac{1}{p - 1}$ and $\frac{1}{p^2} \leq \left|\R'(p)\right| \leq \frac{1}{p^2 - p}$, for all $p \in (1, \infty)$.
\end{enumerate}
\end{lemma}

Next, let us define a continuous instance $\cont(\gamma)$, with parameter $\gamma \in [1, \infty)$. Intuitively, it is defined to be an instance with infinitely many triangular distributions where the constraint is tight for all choices of $p\geq \gamma$. A formal definition is shown in Definition~\ref{def:cont}, from which we know $\cont(\gamma)$ satisfies both constraints in Program~(\ref{prog:1}). In particular, constraint~(\ref{cstr:ap2}) is tight in the range of $p \in (\gamma, \infty)$.
\vspace{5.5pt} \\
\fbox{\begin{minipage}{\textwidth}
\begin{definition}
\label{def:cont}
Given $\gamma \in [1, \infty)$ and $m \in \mathbb{N}_+$, triangular instance $\{\tri(v_i, q_i)\}_{i = 1}^{n^2}$ satisfies that
\begin{itemize}
\item $v_i = \gamma + n - \frac{i - 1}{n}$, for all $i \in [n^2]$. For notational simplicity, let $v_0 \eqdef \infty$;
\item $\sum\limits_{j = 1}^i \ln\left(1 + \frac{v_i q_i}{1 - q_i}\right) = \R(v_i)$, that is, $q_i = \frac{e^{\R(v_i) - \R(v_{i - 1})} - 1}{v_i + e^{\R(v_i) - \R(v_{i - 1})} - 1}$, for all $i \in [n^2]$.
\end{itemize}
When $n$ approaches to infinity, $\{\tri(v_i, q_i)\}_{i = 1}^{n^2}$ converges to $\cont(\gamma)$.
\end{definition}
\end{minipage}}
\vspace{5.5pt}

The following lemma calibrates the objective function (i.e., the revenue from Myerson Auction) for continuous instances, and the proof is deferred in Appendix~\ref{subapp:upper2-cont}.

\begin{lemma}
\label{lem:opt_cont}
$\opt(\cont(\gamma)) = 2 + \displaystyle{\int_1^{\infty}} \left(1 - e^{-\Q\big(\max\{x, \gamma\}\big)}\right) dx$ is decreasing for all $\gamma \in [1, \infty)$.
\end{lemma}

Notably, by assigning $\gamma \leftarrow 1$ in Lemma~\ref{lem:opt_cont}, we capture the formula in Theorem~\ref{thm:opt_ap}.

\subsection{Catching Optimal Solution}
\label{subsec:overview:opt}

A high level idea to prove the optimality of $\cont(1)$ is, that we can covert any a feasible instance of Program~(\ref{prog:1}) into some continuous instance $\cont(\hgamma)$, without hurting the objective function.

\paragraph{Potential Function.}
A crucial ingredient used in deriving optimal solution is the potential function defined as follows. For a regular distribution $F_i \in \reg$, define its potential function
\[
\Psi\big(p, \{F_i\}\big) \eqdef \ln\left(1 + p \cdot \frac{1 - F_i(p)}{F_i(p)}\right) \quad\quad\quad\quad \forall p \in (0, \infty).
\]
For simplicity, we abuse notations as $\Psi\big(\{F_i\}\big) \eqdef \ln\left(1 + \frac{v_i q_i}{1 - q_i}\right)$ and $\Psi\big(p, \{F_i\}_{i = 1}^n\big) \eqdef \sum\limits_{i = 1}^n \Psi\big(p, \{F_i\}\big)$.

As shown in Lemma~\ref{lem:potential}, potential function $\Psi\big(p, \{F_i\}\big)$ is strictly decreasing on $p \in [v_i, u_i]$, and remains maximum $\Psi\big(\{F_i\}\big)$ when $p \in (0, v_i]$, and inherently remains $0$ when $p \in (u_i, \infty)$.

\begin{lemma}
\label{lem:potential}
Potential function $\Psi\big(p, \{F_i\}\big) \eqdef \ln\left(1 + p \cdot \frac{1 - F_i(p)}{F_i(p)}\right)$ is strictly decreasing on $p \in [v_i, u_i]$, and remains $\ln\left(1 + \frac{v_i q_i}{1 - q_i}\right)$ when $p \in (0, v_i]$.
\end{lemma}


\begin{figure}[H]
\centering
\subfigure[Potential function of the existing $F_i$]{
\begin{tikzpicture}[thick, smooth, scale = 2]
\draw[->] (0, 0) -- (2, 0);
\draw[->] (0, 0) -- (0, 1.1);
\node[above] at (0, 1.1) {\small $\Psi\big(p, \{F_i\}\big)$};
\node[right] at (2, 0) {\small $p$};
\node[below] at (1.8, 0) {\small $u_i$};
\node[below] at (1.4727, 0.8pt) {\small $\uu_i$};
\node[below] at (0.6, 0) {\small $v_i$};
\node[below] at (0, 0) {\small $0$};
\node[left] at (0, 1) {\small $\Psi\big(\{F_i\}\big)$};
\node[left] at (0, 0.2) {\small $\Delta_i$};

\draw[thin, dashed] (1.8, 0.1) -- (1.8, 0);
\draw[thin, dashed] (0.6, 1) -- (0.6, 0);
\draw[thin, dashed] (1.4727, 0) -- (1.4727, 0.2);
\draw[thin, dashed] (0, 0.2) -- (1.4727, 0.2);
\draw[color = blue] (0, 1) -- (0.6, 1);
\draw[color = blue, domain = 0.6: 1.8] plot (\x, {-0.35 + 0.81 / \x});

\draw[very thick] (0pt, 1) -- (1.25pt, 1);
\draw[very thick] (0pt, 0.2) -- (1.25pt, 0.2);
\draw[very thick] (1.8, 0pt) -- (1.8, 1.25pt);
\draw[very thick] (0.6, 0pt) -- (0.6, 1.25pt);
\draw[very thick] (1.4727, 0pt) -- (1.4727, 1.25pt);
\draw[thin, blue, fill = red] (1.8, 0.1) circle (0.5pt);
\draw[thin, blue, fill = red] (0.6, 1) circle (0.5pt);
\draw[thin, blue, fill = green] (1.4727, 0.2) circle (0.5pt);
\end{tikzpicture}
\label{fig:potential1}
}
\quad\quad\quad\quad\quad\quad
\subfigure[Potential function of the new $\FF_i$]{
\begin{tikzpicture}[thick, smooth, scale = 2]
\draw[->] (0, 0) -- (1.7, 0);
\draw[->] (0, 0) -- (0, 1.1);
\node[above] at (0, 1.1) {\small $\Psi\big(p, \{\FF_i\}\big)$};
\node[right] at (1.7, 0) {\small $p$};
\node[below] at (1.4727, 0) {\small $\uu_i \leq u_i$};
\node[below] at (0.6, 0) {\small $\vv_i = v_i$};
\node[below] at (0, 0) {\small $0$};
\node[left] at (0, 0.8) {\small $\Psi\big(\{\FF_i\}\big)$};

\draw[thin, dashed] (0.6, 0.8) -- (0.6, 0);
\draw[color = blue] (0, 0.8) -- (0.6, 0.8);
\draw[color = blue, domain = 0.6: 1.4727] plot (\x, {-0.55 + 0.81 / \x});

\draw[very thick] (0pt, 0.8) -- (1.25pt, 0.8);
\draw[very thick] (1.4727, 0pt) -- (1.4727, 1.25pt);
\draw[very thick] (0.6, 0pt) -- (0.6, 1.25pt);
\draw[thin, blue, fill = red] (1.4727, 0) circle (0.5pt);
\draw[thin, blue, fill = red] (0.6, 0.8) circle (0.5pt);
\end{tikzpicture}
\label{fig:potential2}
}
\caption{Demonstration for potential-based construction of $\{\FF_i\}_{i = 1}^n$ from $\{F_i\}_{i = 1}^n$}
\label{fig:potential}
\end{figure}
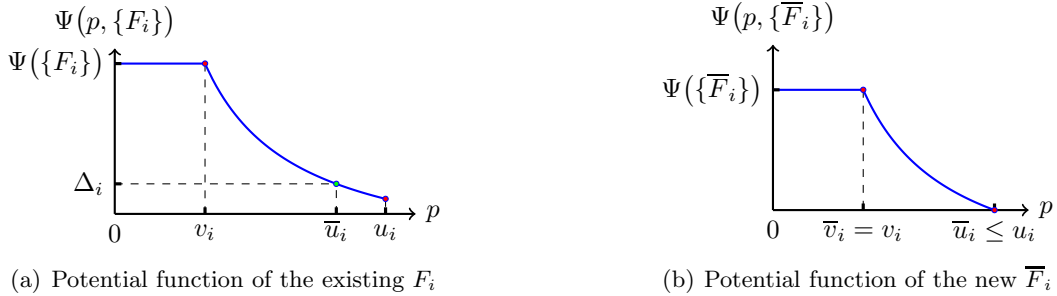
We note that the potential function for a triangular distribution is simply degenerated step function, and thus the potential argument was not needed in previous papers~\cite{AHNPY15,JLTX18}.

\paragraph{Recursive Construction.}
We adopt the most natural proof plan: \emph{construct a \textbf{list} (finite sequence) of instances, where the \textbf{head} is the given feasible instance $\{F_i\}_{i = 1}^n$, and the \textbf{tail} is a continuous instance $\cont(\hgamma)$}.
In the middle of the list, we have ``hybrid'' instances with both continuous and discrete components. For brevity, in an iteration of the recursive construction, re-denote by $\cont(\gamma) \cup \{F_i\}_{i = 1}^n$ and $\cont(\ggamma) \cup \{\FF_i\}_{i = 1}^n$ the existing and new instances, respectively. The new instance $\cont(\ggamma) \cup \{\FF_i\}_{i = 1}^n$ is acquired as follows:
\begin{description}
\item \framebox{\emph{Discrete-component}}: The new discrete-component $\{\FF_i\}_{i = 1}^n$ is derived from the existing discrete-component $\{F_i\}_{i = 1}^n$ in a ``\emph{uniform}'' fashion. \\
    \framebox{\emph{Local-uniformity}}: As Figure~\ref{fig:potential} demonstrates, for each $i \in [n]$, the new distribution $\FF_i$ is constructed by uniformly ``\emph{diminishing}'' the existing distribution $F_i$ as follows:
    \[
    \Psi\big(p, \{\FF_i\}\big) = \Psi\big(p, \{F_i\}\big) - \Delta_i \quad\quad \forall p \in (0, \uu_i],
    \]
    where $\Delta_i \geq 0$ is \emph{potential-decrease} contributed by this distribution, and $\uu_i$ is support-supremum of the new distribution $\FF_i$; \\
    \framebox{\emph{Global-uniformity}}: Let $u$ and $\uu$ be support-supremum of $\{F_i\}_{i = 1}^n$ and that of $\{\FF_i\}_{i = 1}^n$, respectively. In fact, the above $\Delta_i$'s and $\uu_i$'s are carefully chosen:
    \begin{itemize}
    \item Each distribution $F_i$ is ``\emph{strictly diminish}'' (that is, $\Delta_i > 0$) when $u_i = u$; and if so, each such new distribution $\FF_i$ conformably satisfies that $\uu_i = \uu$;
    \item Support-supremum $\uu$ of $\{\FF_i\}_{i = 1}^n$ is carefully chosen, so that total potential-decrease $\Delta \eqdef \sum\limits_{i = 1}^n \Delta_i$ ``\emph{usually}''\footnote{When potential-decrease $\Delta$ (in a specific iteration) is less than step-size $\Delta^*$, at least one $F_i$ will become vanishing (that is, $\Psi\big(\{\FF_i\}\big) = 0$). Clearly, this special case happens at most $n$ times. More details will be provided soon after.} equals to pre-fixed step-size $\Delta^* > 0$.
    \end{itemize}
    The \emph{local-uniformity} w.r.t. support-supremum will be formalized as Lemma~\ref{lem:alg_facts}.1 in Section~\ref{subsec:upper3-subroutine}. Moreover, the \emph{global-uniformity} w.r.t. potential-decrease means that the iteration terminates in finite iterations. Soon after, we will see these uniformities are crucial for proving the feasibility and optimality.
\item \framebox{\emph{Continuous-component}}: The new continuous-component $\cont(\ggamma)$ is obtained by ``\emph{augmenting}'' the continuous-component existing $\cont(\gamma)$, so as to ``\emph{offset}'' influences (on the constraints and objective value of Program~(\ref{prog:1})) caused by ``\emph{diminishing}'' the discrete-component.
\end{description}

\paragraph{Main Remarks.}
The benefit from such recursive and potential-based construction (refer to the involved ``\emph{uniformities}'', in particular) is twofold. To see so, let us first rewrite constraint~(\ref{cstr:ap2}) for the existing instance $\cont(\gamma) \cup \{F_i\}_{i = 1}^n$:
\vspace{5.5pt} \\
\fcolorbox{white}{lightgray}{\begin{minipage}{\textwidth}
\begin{equation}
\tag{\ref{cstr:ap2}}
\R\big(\max\{p, \gamma\}\big) + \Psi\big(p, \{F_i\}_{i = 1}^n\big) \leq \R(p) \quad\quad \forall p \in (1, \infty)
\end{equation}
\end{minipage}}
\vspace{0pt} \\
Observe that ``\emph{augmenting}'' continuous-component is just to ``\emph{offset}'' the slack of constraint~(\ref{cstr:ap2}) caused by ``\emph{diminishing}'' discrete-component. Based on the \emph{local-uniformity} and \emph{global-uniformity}, it is not hard to see that the new instance $\cont(\ggamma) \cup \{\FF_i\}_{i = 1}^n$ also satisfies constraint~(\ref{cstr:ap2}).

Additionally, the \emph{global-uniformity} w.r.t. support-supremum (cooperated with other details) allows us to decompose an iteration  into $n$ \emph{simpler processes}: In the $k$-th process,
\vspace{5.5pt} \\
\fcolorbox{white}{lightgray}{\begin{minipage}{\textwidth}
We only transform one existing distribution $F_k$ into another new distribution $\FF_k$, and then ``\emph{augment}'' the continuous-component correspondingly.
\end{minipage}}
\vspace{0pt} \\
Afterwards, we can exploit standard tools from mathematical analysis to verify, in each of these $n$ processes, that the objective value of Program~(\ref{prog:1}) monotonically increases. By induction, we conclude that $\opt\big(\cont(\ggamma) \cup \{\FF_i\}_{i = 1}^n\big) \geq \opt\big(\cont(\gamma) \cup \{F_i\}_{i = 1}^n\big)$.

\section{Main Proof: Reduction Part}
\label{sec:upper1}
In this and the next section, we will implement the aforementioned proof plan. We will describe most proof ideas and proof steps with sufficient details. To improve the readability, however, verifications of a few lemmas (those extremely technically involved ones) are postponed to appendices. This is why we call these two sections ``Main Proof''.

\subsection{Proof of Lemma~\ref{lem:shape}}
\label{subapp:upper1-shape}

\paragraph{[Lemma~\ref{lem:shape}].}
\emph{In a feasible instance $\{F_i\}_{i = 1}^{n + 1}$ of Program~(\ref{prog:0}), w.l.o.g. for all $i \in [n]$,
\begin{enumerate}
\item If $q_i > 0$, then $\partial_+ r_i(q) > 1$ for all $q \in [0, q_i)$, $v_i = \frac{r_i(q_i)}{q_i} > 1$ and $\Phi_i(v_i^+) = \lim\limits_{q \rightarrow q_i^-} \partial_+ r_i(q) \geq 1$;
\item $r_i(q)$ is increasing on $q \in [0, q_i]$, and $r_i(q) = \frac{r_i(q_i)}{1 - q_i} \cdot (1 - q)$ when $q \in [q_i, 1]$.
\end{enumerate}}

\begin{proof}
For each $i \in [n]$, the regularity of $F_i$ implies that revenue-quantile curve $r_i(q)$ is concave on $q \in [0, 1]$ (see Section~\ref{subsec:prelim-reg-vv}). As Figure~\ref{fig:lem:shape} demonstrates, we would ``tailor'' $F_i$ as follows: In term of revenue-quantile curve, let\footnote{We know $\partial_+ r_i(q_i) \leq 0 < 1$, in that $q_i = \mathop{\arg\max} \{r_i(q): \,\, q \in [0, 1]\}$. Therefore, $\qq_i = \min \{q \in [0, q_i]: \,\, \partial_+ r_i(q) \leq 1\}$ must be well-defined in $[0, q_i]$.} $\qq_i \eqdef \min \{q \in [0, q_i]: \, \partial_+ r_i(q) \leq 1\}$, and define
\[
\rr_i(q) \eqdef
\begin{cases}
r_i(q) & \forall q \in [0, \qq_i] \\
\frac{r_i(\qq_i)}{1 - \qq_i} \cdot (1 - q) & \forall q \in (\qq_i, 1]
\end{cases}.
\]
Clearly, $\rr_i(q)$ is also concave on $q \in [0, 1]$, and satisfies both claims in the lemma.

\begin{figure}[H]
\centering
\subfigure[Original $r_i(q)$]{
\begin{tikzpicture}[thick, smooth, scale = 2.5]
\draw[->] (0, 0) -- (1.1, 0);
\draw[->] (0, 0) -- (0, 1.1);
\node[above] at (0, 1.1) {\small $r_i(q)$};
\node[right] at (1.1, 0) {\small $q$};
\node[below] at (0, 0) {\small $0$};
\draw[very thick, color = blue, domain = 0: 0.6] plot (\x, {1.16 - (\x - 1)^2});
\draw[very thick, color = blue, domain = 0.6: 1] plot (\x, {-6.25 * (\x - 0.2) * (\x - 1)});
\draw[dashed] (0, 1) -- (0.6, 1);
\draw[dashed] (0.4, 0) -- (0.4, 0.8);
\draw[dashed] (0, 0.8) -- (0.4, 0.8);
\draw[dashed] (0.4, 0.8) -- (1, 0);
\draw[color = blue, fill = red] (0.4, 0.8) circle(0.625pt);
\draw[very thick] (1, 0) -- (1, 1.25pt);
\node[below] at (1, 0) {\small $1$};
\node[below] at (0.4, 0.025) {\small $\qq_i$};
\draw[very thick] (0, 1) -- (1.25pt, 1);
\draw[very thick] (0, 0.16) -- (1.25pt, 0.16);
\node[left] at (0, 1) {\small $r_i(q_i) = v_i q_i$};
\node[left] at (0, 0.8) {\small $r_i(\qq_i)$};
\node[left] at (0, 0.16) {\small $r_i(0)$};
\draw[very thick] (0.4, 0) -- (0.4, 1.25pt);
\draw[very thick] (0, 0.8) -- (1.25pt, 0.8);
\end{tikzpicture}
\label{fig:lem:shape.1}
}
\quad\quad\quad\quad\quad\quad
\subfigure[$\rr_i(q)$ after tailoring]{
\begin{tikzpicture}[thick, smooth, scale = 2.5]
\draw[->] (0, 0) -- (1.1, 0);
\draw[->] (0, 0) -- (0, 1.1);
\node[above] at (0, 1.1) {\small $\rr_i(q)$};
\node[right] at (1.1, 0) {\small $q$};
\node[below] at (0, 0) {\small $0$};
\draw[very thick, color = blue, domain = 0: 0.4] plot (\x, {1.16 - (\x - 1)^2});
\draw[very thick, color = blue] (0.4, 0.8) -- (1, 0);
\draw[dashed] (0.4, 0) -- (0.4, 0.8);
\draw[dashed] (0, 0.8) -- (0.4, 0.8);

\node[left] at (0, 0.16) {\small $\rr_i(0)$};

\draw[color = blue, fill = red] (0.4, 0.8) circle(0.625pt);
\draw[very thick] (1, 0) -- (1, 1.25pt);
\node[below] at (1, 0) {\small $1$};
\node[below] at (0.4, 0.025) {\small $\qq_i$};
\draw[very thick] (0, 0.16) -- (1.25pt, 0.16);
\node[left] at (0, 0.8) {\small $\rr_i(\qq_i) = \vv_i \qq_i$};
\draw[very thick] (0.4, 0) -- (0.4, 1.25pt);
\draw[very thick] (0, 0.8) -- (1.25pt, 0.8);
\end{tikzpicture}
\label{fig:lem:shape.2}
}
\caption{Demonstration for reduction in Lemma~\ref{lem:shape}.}
\label{fig:lem:shape}
\end{figure}
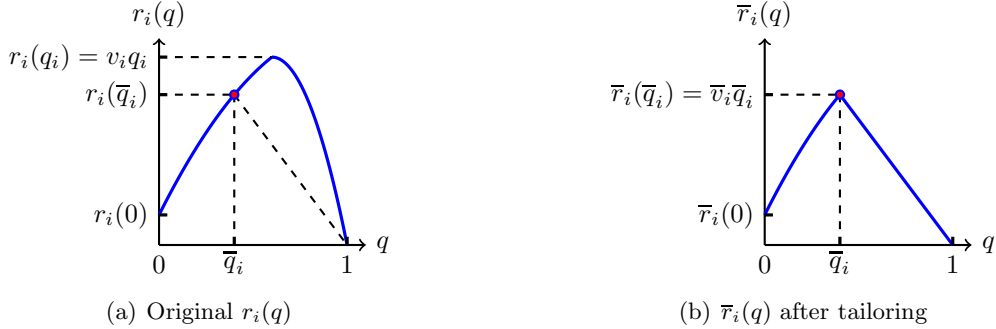

It remains to confirm the feasibility and optimality of $\{\FF_i\}_{i = 1}^n \bigcup \{F_{n + 1}\}$. Recall the reductions between a CDF and its revenue-quantile curve (see Fact~\ref{fact1}). For each $i \in [n]$, the concavity of $r_i(q)$ implies that $\rr_i(q) \leq r_i(q)$ for all $q \in [0, 1]$ and thus, $\FF_i(p) \geq F_i(p)$ for all $p \in (0, \infty)$. Therefore,
\[
\ap(p, \{\FF_i\}_{i = 1}^n \bigcup \{F_{n + 1}\}) \leq \ap(p, \{F_i\}_{i = 1}^{n + 1}) \leq 1,
\]
for all $p \in (0, \infty)$. On the other hand, we surely have $\rr_i(0) = r_i(0)$, and can infer from Fact~\ref{fact1}.2 that $\DD_i(p) = D_i(p)$ for all $p \in (1, \infty)$. Hence,
\[
\begin{aligned}
\opt\big(\{\FF_i\}_{i = 1}^n \bigcup \{F_{n + 1}\}\big)
= & 1 + \sum\limits_{i = 1}^n \rr_i(0) + \displaystyle{\int_1^{\infty}} \left(1 - \prod\limits_{i = 1}^n \DD_i(x)\right) dx \\
= & 1 + \sum\limits_{i = 1}^n r_i(0) + \displaystyle{\int_1^{\infty}} \left(1 - \prod\limits_{i = 1}^n D_i(x)\right) dx
= \opt\big(\{F_i\}_{i = 1}^{n + 1}\big).
\end{aligned}
\]
This completes the proof of the lemma.
\end{proof}

\subsection{Proof Plan of Main Lemma~\ref{lem:p/(p+1)}}
\label{subsec:upper1-overview}

The proof of Main Lemma~\ref{lem:p/(p+1)} consists of the following three parts.
\vspace{5.5pt} \\
\fbox{\begin{minipage}{\textwidth}
\paragraph{[Main Lemma~\ref{lem:p/(p+1)}].}
\emph{In a worst-case instance $\{F_i\}_{i = 0}^n$ of Program~(\ref{prog:0}), w.l.o.g.
\begin{enumerate}
\item $F_0(p) = \frac{p}{p + 1}$ for all $p \in (0, \infty)$, that is, $r_0(q) = 1 - q$ for all $q \in [0, 1]$;
\item $r_i(0) = 0$ for all $i \in [n]$.
\end{enumerate}}
\end{minipage}}

\paragraph{Step I: Special Buyer $\{F_0\}$.}
Our first step is to show, in the worst-case instance, that w.l.o.g. there is at most one buyer with monopoly-quantile of $0$.

Provided with a feasible instance $\{F_i\}_{i = 1}^n$, let $\Omega \eqdef \{i \in [n]: q_i = 0\}$. Based on Lemma~\ref{lem:shape}, for each $i \in \Omega$, we can assume w.l.o.g. that
\[
r_i(q) = r_i(0) \cdot (1 - q) \quad\quad \forall q \in [0, 1] \quad\quad\quad\quad\quad\quad\quad\quad F_i(p) = \frac{p}{p + r_i(0)} \quad\quad \forall p \in (0, \infty).
\]
Define $r_0(0) \eqdef \sum\limits_{i \in \Omega} r_i(0)$, and $F_0(p) \eqdef \frac{p}{p + r_0(0)}$ for all $p \in (0, \infty)$. Now consider $\{F_0\} \cup \{F_i\}_{i \in [n] \setminus \Omega}$:
\begin{itemize}
\item \emph{Optimality}: $\opt\big(\{F_0\} \cup \{F_i\}_{i \in [n] \setminus \Omega}\big) = \opt\big(\{F_i\}_{i = 1}^n\big)$, since $r_0(0) = \sum\limits_{i \in \Omega} r_i(0)$, and
    \[
    D_i(p) = 1 \quad\quad\quad\quad \forall p \in (0, \infty), \,\, \forall i \in \{0\} \cup \Omega;
    \]
\item \emph{Feasibility}: $\ap\big(p, \{F_0\} \cup \{F_i\}_{i \in [n] \setminus \Omega}\big) \leq \ap\big(p, \{F_i\}_{i = 1}^n\big)$ for all $p \in (0, \infty)$, in that
    \[
    F_0(p) = \left(1 + \frac{r_0(0)}{p}\right)^{-1} = \left(1 + \frac{1}{p} \cdot \sum\limits_{i \in \Omega} r_i(0)\right)^{-1} \geq \prod\limits_{i \in \Omega} \left(1 + \frac{1}{p} \cdot r_i(0)\right)^{-1} = \prod\limits_{i \in \Omega} F_i(p).
    \]
\end{itemize}
By Lemma~\ref{lem:shape} and the above reduction, w.l.o.g. in a worst-case instance $\{F_i\}_{i = 0}^n$,
\begin{enumerate}
\item With some $r_0(0) \in [0, 1]$, we have $r_0(q) = r_0(0) \cdot (1 - q)$ for all $q \in [0, 1]$;
\item For each $i \in [n]$, we have $q_i > 0$, $\partial_+ r_i(q) > 1$ for all $q \in [0, q_i)$, $v_i = \frac{r_i(q_i)}{q_i} > 1$.
\end{enumerate}

\paragraph{Step II: Auxiliary Property.}
The second and third steps are devoted to showing that $r_0(0) = 1$ (i.e., $F_0 = \tri(\infty)$) in a worst-case instance. Towards that, we prove certain technical properties of a worst-case instance (in a straightforward way): Assuming a given instance violates a certain property, we explicitly construct a new instance that (1)~satisfies the property; and (2)~keeps the revenue gap between $\opt$ and $\ap$.

Alaei et al.~\cite{AHNPY15} studied the revenue gap between \emph{ex-ante relaxation} and \emph{anonymous pricing}, and the special buyer $F_0 = \tri(\infty)$ also appears in a worst-case instance of that problem. Compared with their proof, however, ours is significantly more involved and technical\footnote{Alaei et al.~\cite{AHNPY15} also modeled the ``ex-ante relaxation vs. anonymous pricing'' problem as a program. Observing special structures of that program (i.e., special structures of ex-ante relaxation), they proved that the worst case must be a triangular instance. As mentioned before, provided with two parameters $v_i$ and $q_i$, it suffices to present a triangular distribution $\tri(v_i, q_i)$. By contrast, to define a regular distribution $F_i$ pointwisely (interpreted as an abstract function), infinity parameters are required.}.

To offer intuitions, we illustrate several most important constructions here, and leave rigorous proofs to Appendix~\ref{app:upper1}. The following property is a major ingredient of our proof.
\vspace{5.5pt} \\
\fbox{\begin{minipage}{\textwidth}
\begin{lemma}
\label{lem:p/(p+1)-3}
There exists a worst-case instance $\{F_i\}_{i = 0}^n$ of Program~(\ref{prog:0}), such that for all $i \in [n]$ with $r_i(0) > 0$,
\[
\partial_+ r_i(0) > 1 + r_0(0).
\]
\end{lemma}
\end{minipage}}

\begin{proof}[Proof (Sketch)]
Assume to the contrary: $\exists k \in [n]$ such that $r_k(0) > 0$ and $\partial_+ r_k(0) \leq 1 + r_0(0)$. We would construct two new distributions $\FF_0$ and $\FF_k$ (in terms of revenue-quantile curves $\rr_0$ and $\rr_k$), to replace the original distributions $F_0$ and $F_k$.
\begin{itemize}
\item Define $\rr_0(q) \eqdef \big(r_0(0) + r_k(0)\big) \cdot (1 - q)$ for all $q \in [0, 1]$;
\item Define $\qq_k \eqdef q_k$ and $\rr_k(q) \eqdef
    \begin{cases}
    r_k(q) - r_k(0) & \forall q \in [0, \qq_k] \\
    \frac{r_k(q_k) - r_k(0)}{1 - q_k} \cdot (1 - q) & \forall q \in (\qq_k, 1]
    \end{cases}$, as Figure~\ref{fig:lem:p/(p+1)-5} illustrates.
\end{itemize}

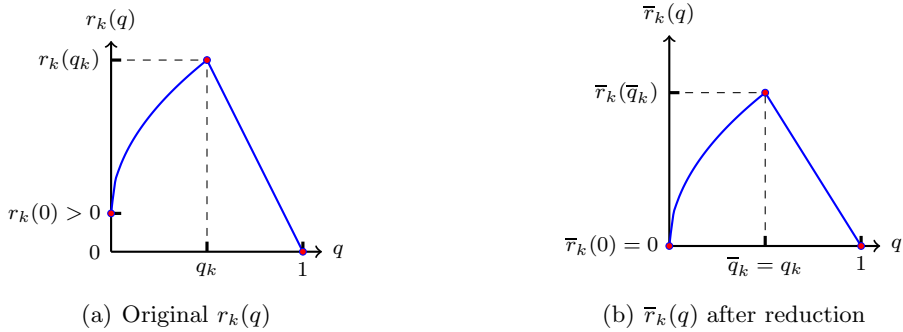
\begin{figure}[H]
\centering
\subfigure[Original $r_k(q)$]{
\begin{tikzpicture}[thick, smooth, domain = 0: 0.5, scale = 2.55]
\draw[->] (0, 0) -- (1.1, 0);
\draw[->] (0, 0) -- (0, 1.1);
\node[above] at (0, 1.1) {\scriptsize $r_k(q)$};
\node[right] at (1.1, 0) {\scriptsize $q$};
\node[left] at (0, 0) {\scriptsize $0$};
\draw[color = blue] plot (\x, {0.2 + 1.6 * sqrt(\x / 2)});
\draw[thin, dashed] (0, 1) -- (0.5, 1);
\draw[very thick] (0, 1) -- (1.5pt, 1);
\node[left] at (0, 1) {\scriptsize $r_k(q_k)$};
\draw[thin, dashed] (0.5, 0) -- (0.5, 1);
\draw[very thick] (0.5, 0) -- (0.5, 1.5pt);
\node[below] at (0.5, 0) {\scriptsize $q_k$};

\draw[very thick] (0, 0.2) -- (1.5pt, 0.2);
\node[left] at (0, 0.2) {\scriptsize $r_k(0) > 0$};

\draw[color = blue] (0.5, 1) -- (1, 0);
\draw[very thick] (1, 0) -- (1, 1.5pt);
\node[below] at (1, 0) {\scriptsize $1$};
\draw[thin, color = blue, fill = red] (1, 0) circle(0.5pt);
\draw[thin, color = blue, fill = red] (0.5, 1) circle(0.5pt);
\draw[thin, color = blue, fill = red] (0, 0.2) circle(0.5pt);
\end{tikzpicture}
}
\quad\quad\quad\quad\quad\quad
\subfigure[$\rr_k(q)$ after reduction]{
\begin{tikzpicture}[thick, smooth, domain = 0: 0.5, scale = 2.55]
\draw[->] (0, 0) -- (1.1, 0);
\draw[->] (0, 0) -- (0, 1.1);
\node[above] at (0, 1.1) {\scriptsize $\rr_k(q)$};
\node[right] at (1.1, 0) {\scriptsize $q$};
\node[left] at (0, 0) {\scriptsize $\rr_k(0) = 0$};
\draw[color = blue] plot (\x, {1.6 * sqrt(\x / 2)});
\draw[thin, dashed] (0, 0.8) -- (0.5, 0.8);
\draw[very thick] (0, 0.8) -- (1.5pt, 0.8);
\node[anchor = 0] at (0, 0.8) {\scriptsize $\rr_k(\qq_k)$};
\draw[thin, dashed] (0.5, 0) -- (0.5, 0.8);
\draw[very thick] (0.5, 0) -- (0.5, 1.5pt);
\node[below] at (0.5, 0) {\scriptsize $\qq_k = q_k$};

\draw[color = blue] (0.5, 0.8) -- (1, 0);
\draw[very thick] (1, 0) -- (1, 1.5pt);
\node[below] at (1, 0) {\scriptsize $1$};
\draw[thin, color = blue, fill = red] (1, 0) circle(0.5pt);
\draw[thin, color = blue, fill = red] (0.5, 0.8) circle(0.5pt);
\draw[thin, color = blue, fill = red] (0, 0) circle(0.5pt);
\end{tikzpicture}
}
\caption{Demonstrations for the construction in Lemma~\ref{lem:p/(p+1)-3}}
\label{fig:lem:p/(p+1)-5}
\end{figure}

The feasibility of the new instance, i.e., $\ap\big(p, \{\FF_0\} \bigcup \{\FF_k\} \bigcup \{F_i\}_{i \in [n] \setminus \{k\}}\big) \leq 1$ when $p \in (1, \infty]$, is provided in Appendix~\ref{subapp:upper1-p/(p+1)-3}. Besides, we have $\opt\left(\{\FF_0\} \bigcup \{\FF_k\} \bigcup \{F_i\}_{i \in [n] \setminus \{k\}}\right) = \opt\big(\{F_i\}_{i = 0}^n\big)$.
\begin{itemize}
\item The integration part of $\opt$ remains the same: ``Chop'' $r_k(0)$ from the $k$-th revenue-quantile curve -- this has no affect on the $k$-th virtual value CDF, in whole integrating range of $(1, \infty)$ (refer to all positive virtual values, as Figure~\ref{fig:lem:p/(p+1)-5} suggests);
\item The summation part of $\opt$ also remains the same: The loss of $r_k(0)$ incurred by ``chopping'' the $k$-th distribution is precisely compensated by ``filling up'' the special $0$-th distribution.
\end{itemize}
We shall emphasize, after such reduction, that $\sum\limits_{i = 0}^n r_i(0)$ remains the same, whereas $r_0(0)$ increases strictly. These facts will be useful for later proofs.
\end{proof}

\paragraph{Step III: Case Analyses.}
In the rest of Section~\ref{sec:upper1}, we always assume the property from Lemma~\ref{lem:p/(p+1)-3}. Towards Main Lemma~\ref{lem:p/(p+1)}, our final step refers to case analyses, based on the value of $r_0(0)$. Specifically, we provide different construction schemes when $r_0(0) \leq \frac{1}{\sqrt{3}}$ and when $r_0(0) \ge \frac{1}{2}$.
\vspace{5.5pt} \\
\fbox{\begin{minipage}{\textwidth}
\begin{lemma}
\label{lem:p/(p+1)-1}
There exists a worst-case instance $\{F_i\}_{i = 0}^n$ of Program~(\ref{prog:0}) such that $r_0(0) > \frac{1}{\sqrt{3}}$.
\end{lemma}
\end{minipage}}
\vspace{5.5pt} \\
\fbox{\begin{minipage}{\textwidth}
\begin{lemma}
\label{lem:p/(p+1)-4}
Given a feasible instance $\{F_i\}_{i = 0}^n$ of Program~(\ref{prog:0}) that $r_0(0) \geq \frac{1}{2}$, there exists another feasible instance $\{\FF_i\}_{i = 0}^n$ such that $\opt\big(\{\FF_i\}_{i = 0}^n\big) \geq \opt\big(\{F_i\}_{i = 0}^n\big)$, and
\[
\rr_0(q) = 1 - q \quad\quad \forall q \in [0, 1] \quad\quad\quad\quad\quad\quad\quad\quad \rr_i(0) = 0 \quad\quad \forall i \in [n].
\]
\end{lemma}
\end{minipage}}
\vspace{5.5pt}

Note that $\frac{1}{\sqrt{3}} > \frac{1}{2}$, combining both lemmas together leads to Main Lemma~\ref{lem:p/(p+1)} immediately. Assuming the property from Lemma~\ref{lem:p/(p+1)-3}, we will outline how to settle the first case (i.e., Lemma~\ref{lem:p/(p+1)-1}) in Section~\ref{subsec:upper1-lem:p/(p+1)-1}. Furthermore, the proof of the second case (i.e., Lemma~\ref{lem:p/(p+1)-4}) is of the same spirit as that of Lemma~\ref{lem:p/(p+1)-3}, and is deferred to Appendix~\ref{subapp:upper1-p/(p+1)-4}.

\subsection{Proof Plan of Lemma~\ref{lem:p/(p+1)-1}}
\label{subsec:upper1-lem:p/(p+1)-1}

To conquer Lemma~\ref{lem:p/(p+1)-1}, we shall define a useful parameter\footnote{Each regular distribution $F_i$ has at most one probability-mass (which must be support-supremum $u_i$ if exists). Therefore, function $\ap(p)$ is continuous when $p$ is sufficiently large. We define $\beta \eqdef \infty$ when $\lim\limits_{p \to \infty} \ap(p) = 1$.}:
\[
\beta\big(\{F_i\}_{i = 0}^n\big) \eqdef \max \big\{p \in (1, \infty]: \,\, \ap\big(p, \{F_i\}_{i = 0}^n\big) = 1\big\}.
\]
For brevity, we often drop the term $\{F_i\}_{i = 0}^{n}$, when there is no ambiguity from the context. Recall Lemma~\ref{lem:shape}.1 that $v_i > 1$ for all $i \in [n]$. The following lemmas are proved in Appendix~\ref{subapp:upper1-p/(p+1)-1}.

\begin{lemma}
\label{lem:ap:increase}
$\ap\big(p, \{F_i\}_{i = 0}^n\big)$ is strictly increasing, when $1 < p < \min \big\{v_i: \,\, i \in [n]\big\}$.
\end{lemma}

\begin{lemma}
\label{lem:ap_infinity}
$\lim\limits_{p \to \infty} \ap\big(p, \{F_i\}_{i = 0}^n\big) = \sum\limits_{i = 0}^n r_i(0)$.
\end{lemma}

In a worst-case instance $\{F_i\}_{i = 0}^n$ of Program~(\ref{prog:0}), definitely $\max \big\{\ap(p):\,\, p \in (1, \infty]\big\} = 1$, as otherwise we can always scale up all distributions. Based on Lemma~\ref{lem:ap:increase}, $\max \big\{\ap(p):\,\, p \in (1, \infty]\big\}$ is reached when $p = \beta \geq \min \big\{v_i: \,\, i \in [n]\big\} > 1$ (notice that the maximizer $\beta$ can be infinity). Hence, $\beta$ must be well-defined (can be infinity). By definition we have $\ap(\beta) = 1$ and further,
\begin{itemize}
\item When $\beta \in (1, \infty)$, the feasibility to constraint~(\ref{cstr:ap0}) ensures $\partial_+ \ap(\beta) \leq 0$ and $\partial_- \ap(\beta) \geq 0$;
\item When $\beta = \infty$, it follows from Lemma~\ref{lem:ap_infinity} that $\sum\limits_{i = 0}^n r_i(0) = \ap(\infty) = \ap(\beta) = 1$.
\end{itemize}

\paragraph{Case I (when $\beta < \infty$):}
This case is a bit strange since the constraint of $\ap(p)\leq 1$ is not tight for $p>\beta$. Intuitively, one can slightly increase the distribution for $p>\beta$ so that the constraint of $\ap(p)\leq 1$ is still satisfied while the objective function is strictly increased. This seems always possible. The only reason that we cannot do this is due to the  hidden constraint that the distributions are regular. This motivates the following notion of \emph{critical instance} in Definition~\ref{def:critical}. The \emph{first} condition for criticality means, for each $i \in [n]$ with $F_i(\beta) < 1$, that revenue-quantile curve $r_i(q)$ is a line segment\footnote{In terms of notations in Definition~\ref{def:critical}, we have $r_i(q) = a_i \cdot q + b_i$ for all $q \in [0, 1 - F_i(\beta)]$.} on $q \in [0, 1 - F_i(\beta)]$ (this range refers to all values $p \in [\beta,\infty]$).
\vspace{5.5pt} \\
\fbox{\begin{minipage}{\textwidth}
\begin{definition}
\label{def:critical}
A critical instance $\{F_i\}_{i = 0}^n$ is feasible, has well-defined and finite $\beta$, and satisfies:
\begin{enumerate}
\item $\forall \big(i \in [n] \text{: } F_i(\beta) < 1\big)$, $\exists a_i > 1$, $\exists b_i \geq 0$: $F_i(p) = 1 - \frac{b_i}{p - a_i}$ for all $p \in (\beta, \infty)$;
\item $\partial_+ \ap\big(\beta\big) = 0$.
\end{enumerate}
\end{definition}
\end{minipage}}

\paragraph{Case I.1 (when $\{F_i\}_{i = 0}^n$ is non-critical):}
Given that $\beta \in (1, \infty)$, indeed we can transform any a \emph{non-critical} instance into another \emph{critical} instance, or an instance with $\beta = \infty$, while preserve the property from Lemma~\ref{lem:p/(p+1)-3}.

\begin{itemize}
\item Suppose the \emph{first} condition for criticality fails: There exists $k \in [n]$ with $q_{\beta} \eqdef 1 - F_k(\beta) > 0$, so that revenue-quantile curve $r_k(q)$ is not a line segment (yet still concave) on $q \in \left[0, q_{\beta}\right]$.

    Given $\Delta_r > 0$, as Figure~\ref{fig:non-critical:1} demonstrates, there is a \emph{unique} line segment $\rr_k(q)$ such that (1)~$\rr_k(0) = r_k(0) + \Delta_r$; and (2)~$\rr_k(q)$ is \emph{tangent} to $r_k(q)$ at $\big(\qq_{\beta}, r_k(\qq_{\beta})\big)$, for some $\qq_{\beta} > 0$. In the range of $q \in (\qq_{\beta}, 1]$, we further define $\rr_k(q) \eqdef r_k(q)$. Interpreted as a revenue-quantile curve, $\rr_k(q)$ corresponds to some distribution $\FF_k$.

    We shall find the \emph{maximum} $\Delta_r > 0$ such that $\ap\big(p, \{F_0\} \cup \{\FF_k\} \cup \{F_i\}_{i \in [n] \setminus \{k\}}\big) \leq 1$ for all $p \in (1, \infty]$. Such $\Delta_r$ and $\rr_k(q)$ always exist, which can be inferred from the following facts\footnote{Combining facts~\textbf{(a,b)} together implies the existence of candidates $\Delta_r > 0$. Also, combine facts~\textbf{(a,c)} together, the maximum $\Delta_r = \rr_k(0) - r_k(0)$ is no more than $r_k(q_{\beta}) - \partial_+ r_k(q_{\beta}) \cdot q_{\beta} - r_k(0)$. Note that line segment $\rr_k(q)$ is tangent to $r_k(q)$ at $\big(q_{\beta}, r_k(q_{\beta})\big)$, when $\rr_k(0) = r_k(q_{\beta}) - \partial_+ r_k(q_{\beta}) \cdot q_{\beta}$.}.
    \begin{description}
    \item [(a):] $r_k(q)$ is not a line segment (yet still concave) on $q \in \left[0, q_{\beta}\right]$.
    \item [(b):] Given $p \in (\beta, \infty]$, $\ap\big(p, \{F_i\}_{i = 0}^n\big)$ is \emph{strictly} smaller than $1$, due to the definition of $\beta$.
    \item [(c):] $\ap\big(\beta, \{F_i\}_{i = 0}^n\big) = 1$.
    \end{description}
    Since $\rr_k(q) \geq r_k(q)$ for all $q \in [0, 1]$, distribution $\FF_k$ stochastically dominates distribution $F_k$. Hence, after replacing $F_k$ by $\FF_k$, we know (1)~the revenue from $\opt$ can never decrease; and (2)~$\sum\limits_{i = 0}^n r_i(0)$ \emph{strictly} increases by $\Delta_r > 0$.
\item Otherwise, the \emph{second} condition for criticality fails, yet the \emph{first} holds: $\partial_+ \ap\big(\beta\big) \neq 0$, and for all $i \in [n]$ with $F_k(\beta) < 1$, revenue-quantile curve $r_k(q)$ is a line segment on $q \in [0, 1 - F_k(\beta)]$. Once again, we can find a distribution $\FF_k$ such that, after replacing $F_k$ by $\FF_k$, (1)~the revenue from $\opt$ can never decrease; and (2)~$\sum\limits_{i = 0}^n r_i(0)$ \emph{strictly} increases by $\Delta_r > 0$.

    To see so, we shall notice that\footnote{As mentioned before, $\partial_+ \ap(\beta) \leq 0$ and $\partial_+ \ap(\beta) \geq 0$, based on the feasibility to constraint~(\ref{cstr:ap0}). Moreover, since the \emph{second} condition for criticality fails, we have $\partial_+ \ap(\beta) \neq 0$.} $\partial_+ \ap(\beta) < 0$ and $\partial_+ \ap(\beta) \geq 0$. In other words, there exists $k \in [n]$ such that $q_{\beta} \eqdef 1 - F_k(\beta) > 0$ and $\partial_+ F_k(\beta) > \partial_- F_k(\beta)$.

    In terms of revenue-quantile curve, we have $\partial_+ r_k(q_{\beta}) < \partial_- r_k(q_{\beta})$, as Figure~\ref{fig:non-critical:2} suggests. As per this, the existence of desired distribution $\FF_k$ comes from same arguments as before.
\end{itemize}

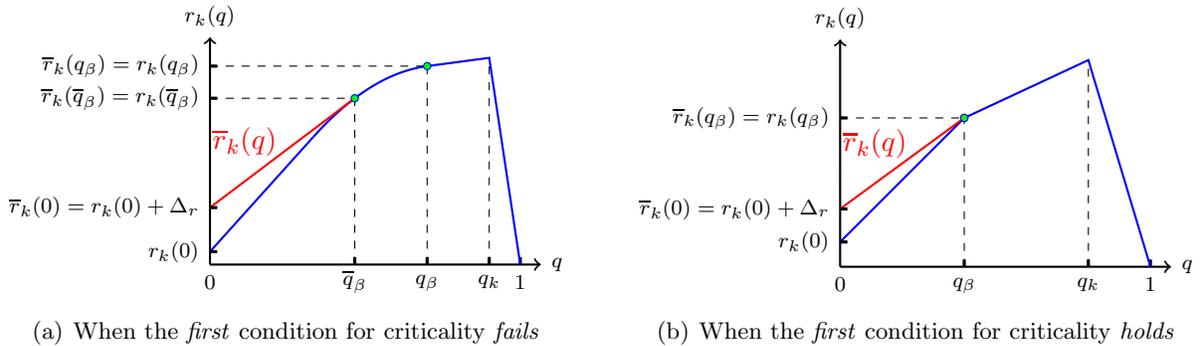
\begin{figure}[H]
\centering
\subfigure[When the \emph{first} condition for criticality \emph{fails}]{
\begin{tikzpicture}[thick, smooth, scale = 2.75]
\draw[->] (0, 0) -- (1.6, 0);
\draw[->] (0, 0) -- (0, 1.1);
\node[above] at (0, 1.1) {\scriptsize $r_k(q)$};
\node[right] at (1.6, 0) {\scriptsize $q$};
\node[below] at (0, 0) {\scriptsize $0$};
\draw[very thick] (1.5, 0) -- (1.5, 1pt);
\node[below] at (1.5, 0) {\scriptsize $1$};

\draw[color = blue] plot (0, 0.0622) -- (0.5, 0.6178) (1.05, 0.96) -- (1.35, 1) -- (1.5, 0);
\draw[color = blue, domain = 0.5: 1.05] plot (\x, {0.965 - 8 / 9 * (\x - 1.125)^2});
\draw[color = red] (0, 0.2756) -- (0.7, 0.8044);
\node[color = red, anchor = -15] at (0.375, 0.54) {$\rr_k(q)$};
\draw[thin, dashed] (1.35, 0) -- (1.35, 1);
\draw[very thick] (1.35, 0) -- (1.35, 1pt);
\node[below] at (1.35, 0) {\scriptsize $q_k$};

\draw[thin, dashed] (0, 0.8044) -- (0.7, 0.8044);
\draw[very thick] (0, 0.8044) -- (1pt, 0.8044);
\draw[thin, dashed] (0.7, 0) -- (0.7, 0.8044);
\draw[very thick] (0.7, 0) -- (0.7, 1pt);
\draw[thin, color = blue, fill = green] (0.7, 0.8044) circle(0.5pt);
\node[left] at (0, 0.8044) {\scriptsize $\rr_k(\qq_{\beta}) = r_k(\qq_{\beta})$};
\node[below] at (0.7, 0.6pt) {\scriptsize $\qq_{\beta}$};

\draw[very thick] (0, 0.0622) -- (1pt, 0.0622);
\node[left] at (0, 0.0622) {\scriptsize $r_k(0)$};

\draw[thin, dashed] (0, 0.96) -- (1.05, 0.96);
\draw[very thick] (0, 0.96) -- (1pt, 0.96);
\node[left] at (0, 0.96) {\scriptsize $\rr_k(q_{\beta}) = r_k(q_{\beta})$};

\draw[thin, dashed] (1.05, 0) -- (1.05, 0.96);
\draw[very thick] (1.05, 0) -- (1.05, 1pt);
\node[below] at (1.05, 0) {\scriptsize $q_{\beta}$};
\draw[thin, color = blue, fill = green] (1.05, 0.96) circle(0.5pt);

\draw[very thick] (0, 0.2756) -- (1pt, 0.2756);
\node[left] at (0, 0.2756) {\scriptsize $\rr_k(0) = r_k(0) + \Delta_r$};
\end{tikzpicture}
\label{fig:non-critical:1}
}
\quad
\subfigure[When the \emph{first} condition for criticality \emph{holds}]{
\begin{tikzpicture}[thick, smooth, scale = 2.75]
\draw[->] (0, 0) -- (1.6, 0);
\draw[->] (0, 0) -- (0, 1.1);
\node[above] at (0, 1.1) {\scriptsize $r_k(q)$};
\node[right] at (1.6, 0) {\scriptsize $q$};
\node[below] at (0, 0) {\scriptsize $0$};
\draw[very thick] (1.5, 0) -- (1.5, 1pt);
\node[below] at (1.5, 0) {\scriptsize $1$};
\draw[color = blue] plot (0, 0.12) -- (0.6, 0.72) -- (1.2, 1) -- (1.5, 0);
\draw[color = red] (0, 0.28) -- (0.6, 0.72);
\node[color = red, anchor = -25] at (0.375, 0.5) {$\rr_k(q)$};
\draw[thin, dashed] (1.2, 0) -- (1.2, 1);
\draw[very thick] (1.2, 0) -- (1.2, 1pt);
\node[below] at (1.2, 0) {\scriptsize $q_k$};

\draw[very thick] (0, 0.12) -- (1pt, 0.12);
\node[left] at (0, 0.12) {\scriptsize $r_k(0)$};

\draw[thin, dashed] (0, 0.72) -- (0.6, 0.72);
\draw[very thick] (0, 0.72) -- (1pt, 0.72);
\draw[thin, color = blue, fill = green] (0.6, 0.72) circle(0.5pt);
\node[left] at (0, 0.72) {\scriptsize $\rr_k(q_{\beta}) = r_k(q_{\beta})$};

\draw[thin, dashed] (0.6, 0) -- (0.6, 0.72);
\draw[very thick] (0.6, 0) -- (0.6, 1pt);
\node[below] at (0.6, 0) {\scriptsize $q_{\beta}$};
\draw[thin, color = blue, fill = green] (0.6, 0.72) circle(0.5pt);

\draw[very thick] (0, 0.28) -- (1pt, 0.28);
\node[left] at (0, 0.28) {\scriptsize $\rr_k(0) = r_k(0) + \Delta_r$};
\end{tikzpicture}
\label{fig:non-critical:2}
}
\caption{Demonstration for the reduction from a \emph{non-critical} instance to a \emph{critical} instance}
\label{fig:non-critical}
\end{figure}

After replacing $F_k$ by the above $\FF_k$, possible $\partial_+ \rr_k(0) \leq 1 + r_0(0)$, which violates the property from Lemma~\ref{lem:p/(p+1)-3}. In this case, we further apply the reduction involved in Lemma~\ref{lem:p/(p+1)-3} (to the $k$-th distribution). Afterwards, as mentioned before, $\sum\limits_{i = 0}^n r_i(0)$ remains same, yet $r_0(0)$ increases \emph{strictly}.

To sum up, each of the above reductions brings greater $\sum\limits_{i = 0}^n r_i(0)$ and greater $r_0(0)$, where one quantity increases \emph{strictly}. On the other hand, $\sum\limits_{i = 0}^n r_i(0) = \ap\big(\infty, \{F_i\}_{i = 0}^n\big) \leq 1$ and $r_0(0) \leq \frac{1}{\sqrt{3}}$\footnote{$\sum\limits_{i = 0}^n r_i(0) = \ap\big(\infty, \{F_i\}_{i = 0}^n\big) \leq 1$ is due to Lemma~\ref{lem:ap_infinity} and feasibility to constraint~(\ref{cstr:ap0}). $r_0(0) \leq \frac{1}{\sqrt{3}}$ respects the statement of Lemma~\ref{lem:p/(p+1)-1}.}. Thus, the procedure (invoking the above reductions) will converge, and bring either (1)~a critical instance $\{F_i\}_{i = 0}^n$ (see \textbf{Case I.2}); or (2)~an instance $\{F_i\}_{i = 0}^n$ with $\beta = \infty$ (see \textbf{Case II}).

\paragraph{Case I.2 (when $\{F_i\}_{i = 0}^n$ is critical):}
Given that $\beta \in (1, \infty)$, we are left to investigate all \emph{critical} instances. When a critical instance $\{F_i\}_{i = 0}^n$ satisfies the property from Lemma~\ref{lem:p/(p+1)-3}, the following lemma (whose proof is deferred to Appendix~\ref{subapp:upper1-p/(p+1)-1.1}) implies that $r_0(0) > \frac{1}{\sqrt{3}}$ in the worst case.
\vspace{5.5pt} \\
\fbox{\begin{minipage}{\textwidth}
\begin{lemma}
\label{lem:p/(p+1)-1.1}
Given a critical instance $\{F_i\}_{i = 0}^n$ of Program~(\ref{prog:0}) that satisfies the property from Lemma~\ref{lem:p/(p+1)-3} (i.e., $\partial_+ r_i(0) > 1 + r_0(0)$ for all $i \in [n]$ with $r_i(0) > 0$), then $r_0(0) > \frac{1}{\sqrt{3}}$.
\end{lemma}
\end{minipage}}

\paragraph{Case II (when $\beta = \infty$):}
In this case, as mentioned before, constraint~(\ref{cstr:ap0}) is tight when $p$ goes to infinity: $\sum\limits_{i = 0}^n r_i(0) = \ap(\infty) = \ap(\beta) = 1$. Given this, we testify Lemma~\ref{lem:p/(p+1)-2} in Appendix~\ref{subapp:upper1-p/(p+1)-2}, and thus narrow down the feasible space of Program~(\ref{prog:0}). Apparently, respect the property from Lemma~\ref{lem:p/(p+1)-3}, an instance $\{F_i\}_{i = 0}^n$ with $r_0(0) \leq \frac{1}{\sqrt{3}}$ and $\sum\limits_{i = 0}^n r_i(0) = 1$ can \emph{never} be feasible.
\vspace{5.5pt} \\
\fbox{\begin{minipage}{\textwidth}
\begin{lemma}
\label{lem:p/(p+1)-2}
Given a feasible instance $\{F_i\}_{i = 0}^n$ of Program~(\ref{prog:0}) that $\sum\limits_{i = 0}^n r_i(0) = 1$ and satisfies the property from Lemma~\ref{lem:p/(p+1)-3} (i.e., $\partial_+ r_i(0) > 1 + r_0(0)$ for all $i \in [n]$ with $r_i(0) > 0$), then $r_0(0) > \frac{1}{\sqrt{3}}$.
\end{lemma}
\end{minipage}}
\vspace{5.5pt}

To sum up, putting the above \textbf{Case I} and \textbf{Case II} (Lemma~\ref{lem:p/(p+1)-1.1} and Lemma~\ref{lem:p/(p+1)-2} essentially) together completes the proof of Lemma~\ref{lem:p/(p+1)-1}.

\section{Main Proof: Optimization Part}
\label{sec:upper3}
In the optimization part of the main proof, we are provided with Program~(\ref{prog:1}), which is rewritten as fellow. For convenience, constraint~(\ref{cstr:ap2}) is reformulated in terms of potential functions. Recall that $\Psi\big(p, \{F_i\}\big) = \ln\left(1 + p \cdot \frac{1 - F_i(p)}{F_i(p)}\right)$ and $\R(p) = p\ln\left(\frac{p^2}{p^2 - 1}\right)$, for all $p \in (0, \infty)$.
\vspace{5.5pt} \\
\fcolorbox{black}{lightgray}{\begin{minipage}{\textwidth}
\begin{align}
\tag{\ref{prog:1}}
\underset{\{F_i\}_{i = 1}^n \in \reg^n}{\max} &\quad\quad \opt = 2 + \displaystyle{\int_1^{\infty}} \left(1 - \prod\limits_{i = 1}^n D_i(x)\right) dx \\
\tag{\ref{cstr:value}}
\text{subject to:} &\quad\quad \text{$v_i > 1$ and $\Phi_i(v_i^+) \geq 1$ \quad\quad\quad\quad\quad\quad\quad\quad\,\,\, $\forall i \in [n]$} \\
\tag{\ref{cstr:ap2}}
&\quad\quad \Psi\big(p, \{F_i\}_{i = 1}^n\big) = \sum\limits_{i = 1}^n \Psi\big(p, \{F_i\}\big) \leq \R(p) \quad\quad \forall p \in (1, \infty)
\end{align}
\end{minipage}}
\vspace{5.5pt}


As mentioned in Section~\ref{subsec:overview:opt}, we adopt the most natural proof plan for solving Program~(\ref{prog:1}): \emph{construct a \textbf{list} (finite sequence) of instances, where the \textbf{head} is the given feasible instance $\{F_i\}_{i = 1}^n$, and the \textbf{tail} is a continuous instance $\cont(\hgamma)$}. This proof plan will be implemented as \textsf{MAIN} (namely Algorithm~\ref{alg1}) in Section~\ref{subsec:upper2-main}. To convey basic ideas, we would outline the main steps before introducing \textsf{MAIN}: with given precision $\epsilon \in (0, 1)$,
\begin{description}
\item [(a):] For the \textbf{head} instance $\{F_i\}_{i = 1}^n$, we would capture its \textbf{next} instance $\cont(\gamma^*) \cup \{F_i^*\}_{i = 1}^n$ (see line~(\ref{alg1-preprocess}) to line~(\ref{alg1:gamma}) of \textsf{MAIN}). While being feasible to Program~(\ref{prog:1}) (see Section~\ref{subsec:upper2-cstr-opt}), the \textbf{next} instance $\cont(\gamma^*) \cup \{F_i^*\}_{i = 1}^n$ further possesses three properties:
    \begin{enumerate}
    \item \emph{Hybridization}: Different from the (discrete instance) \textbf{head} $\{F_i\}_{i = 1}^n$, the \textbf{next} is a mix of continuous-component $\cont(\gamma^*)$ and discrete-component $\{F_i^*\}_{i = 1}^n$, and thus is ``closer'' to the (continuous instance) \textbf{tail} $\cont(\hgamma)$. Notably, $\gamma^* \leq \frac{33}{\epsilon^2}$ is finite (see Lemma~\ref{lem:constant});
    \item \emph{Constraint-slack} (see inequality~(\ref{eq:cstr:ap.0.2}) in Section~\ref{subsec:upper2-cstr-opt}): For discrete-component $\{F_i^*\}_{i = 1}^n$ of the \textbf{next} instance, let $u^*$ denote its support-supremum. W.r.t. the \textbf{next} instance, there is a fixed parameter $\delta^* \geq \frac{1}{32}\epsilon^2$ (see Lemma~\ref{lem:constant}) such that, for all $p \in (1, u^*]$,
        \[
        \text{LHS of~constraint~(\ref{cstr:ap2})} \leq \text{RHS of~constraint~(\ref{cstr:ap2})} - \delta^*;
        \]
    \item \emph{Near-optimality} (see Section~\ref{subsec:upper2-cstr-opt}): $\opt\big(\textbf{next}\big) \geq \opt\big(\textbf{head}\big) - 5\epsilon$.
    \end{enumerate}
\item [(b):] We acquire the remaining \textbf{list} iteratively, by invoking \textsf{SUBROUTINE} (namely Algorithm~\ref{alg2} in Section~\ref{subsec:upper3-subroutine}) in line~(\ref{alg1:construct}) of \textsf{MAIN}. Each remaining instance \textbf{current} is constructed from its \textbf{previous} instance in a tailor-made way (again, see Section~\ref{subsec:upper3-subroutine}). Here, we shall stress that a specific \textbf{current} is feasible to Program~(\ref{prog:1}), and preserves analogous properties:
    \begin{enumerate}
    \item \emph{Hybridization} (see Section~\ref{subsec:upper3-subroutine}): The \textbf{current} is constructed from its \textbf{previous} through ``augmenting'' the continuous-component, and ``diminishing'' the discrete-component. As a consequence, the \textbf{current} is ``closer'' to the (continuous instance) \textbf{tail} $\cont(\hgamma)$, compared with its \textbf{previous};
    \item \emph{Constraint-slack} (see Section~\ref{subsec:upper3-cstr}): For discrete-component of the \textbf{current}, denote by $\uu$ its support-supremum. Our tailor-made construction (see Section~\ref{subsec:upper3-subroutine}) automatically respects this property: given the \textbf{current}, in the range of $p \in (1, \uu]$, we have
        \[
        \text{LHS of~constraint~(\ref{cstr:ap2})} \leq \text{RHS of~constraint~(\ref{cstr:ap2})} - \delta^*;
        \]
    \item \emph{Optimality} (see Section~\ref{subsec:upper3-opt}): $\opt\big(\textbf{current}\big) \geq \opt\big(\textbf{previous}\big)$.
    \end{enumerate}
\end{description}
\fbox{\begin{minipage}{\textwidth}
Suppose that \emph{the \textbf{list} of instances is indeed finite} (to be affirmed soon after), then \textsf{SUBROUTINE} returns a (continuous instance) \textbf{tail} $\cont(\hgamma)$. Accordingly, applying proof plan~\textbf{(b.3)} inductively gives $\opt\big(\textbf{tail}\big) \geq \opt\big(\textbf{next}\big) \overset{(\textbf{a.3})}{\geq} \opt\big(\textbf{head}\big) - 5\epsilon$. In other words,
\[
\opt\big(\cont(\hgamma)\big) \geq \opt\big(\{F_i\}_{i = 1}^n\big) - 5\epsilon.
\]
\end{minipage}}
\begin{description}
\item [(c):] Now, we would briefly explain \emph{why the \textbf{list} is finite}, and defer more details to Section~\ref{subsec:upper3-subroutine}.
    \begin{itemize}
    \item In each iteration of \textsf{SUBROUTINE}, how we construct a \textbf{current} instance (from its \textbf{previous} instance) is \emph{potential-based} (recall Section~\ref{subsec:upper2-cont} for more details);
    \item As input of \textsf{SUBROUTINE}, the \textbf{next} instance $\cont(\gamma^*) \cup \{F_i^*\}_{i = 1}^n$ has a finite potential. On the other hand, the \textbf{tail} instance (if attainable) is of $0$-potential. To see the \textbf{list} is a \emph{finite} sequence of instances, it suffices to find a suitable and fixed\footnote{Here, ``fixed'' means the step-size $\Delta^* > 0$ only depends on the \textbf{head} instance $\{F_i\}_{i = 1}^n$ and precision $\epsilon > 0$.} \emph{step-size} $\Delta^* > 0$, which prescribes the \emph{potential-decrease}\footnote{That is, how ``closer'' (to the \textbf{tail} instance) a \textbf{current} instance is, compared with its \textbf{previous} instance.} per iteration;
    \item A workable step-size $\Delta^* > 0$ is defined in line~(\ref{alg1:Delta}) of \textsf{MAIN}. Such $\Delta^*$ is carefully chosen, to enable our proof plan~\textbf{(b.3)} that $\opt\big(\textbf{current}\big) \geq \opt\big(\textbf{previous}\big)$.
    \item Notably, such $\Delta^*$ is proportional to $\delta^*$, and is inversely proportional to $\gamma^{*2}$. Hence, proof plan~\textbf{(a.1)} and~\textbf{(a.2)} are both necessary: to make $\Delta^*$ bounded away from $0^+$, introduce continuous-component $\cont(\gamma^*)$, where $\gamma^* \leq \frac{33}{\epsilon^2}$, and constraint-slack $\delta^* \geq \frac{1}{32}\epsilon^2$.
    \end{itemize}
    To sum up, the desired \textbf{list} is finite, ending up with the (continuous instance) \textbf{tail} $\cont(\hgamma)$.
\end{description}
The above ideas will be implemented as algorithm \textsf{MAIN} (namely Algorithm~\ref{alg1}). Basically, algorithm \textsf{MAIN} consists of two parts: \textsf{PREPROCESSING} and \textsf{SUBROUTINE}.
\begin{itemize}
\item \emph{Input}: \textsf{MAIN} takes a feasible instance $\{F_i\}_{i = 1}^n$ of Program~(\ref{prog:1}) and a fixed $\epsilon > 0$ as input;
\item \emph{Output}: \textsf{MAIN} outputs a continuous instance $\cont(\hgamma)$ that is provably better than the input instance $\{F_i\}_{i = 1}^n$ (in terms of the objective value, within $5\epsilon$-loss);
\item \textsf{PREPROCESSING}: with precision $\epsilon > 0$ and input instance $\{F_i\}_{i = 1}^n$, we construct a ``\emph{hybrid}'' instance $\cont(\gamma^*) \cup \{F_i^*\}_{i = 1}^n$;
\item \textsf{SUBROUTINE}: we construct a finite sequence of instances, and therefore gradually transform $\cont(\gamma^*) \cup \{F_i^*\}_{i = 1}^n$ into $\cont(\hgamma)$. \textsf{SUBROUTINE} is an iterative algorithm: in each iteration, by invoking algorithm \textsf{DIMINISH} (see Section~\ref{subsec:upper3-subroutine}), we get a new regular instance that is feasible to Program~(\ref{prog:1}), and generates a greater objective value than previous instances;
\end{itemize}
In Section~\ref{subsec:upper2-cstr-opt}, we will prove that $\cont(\gamma^*) \cup \{F_i^*\}_{i = 1}^n$ is regular, feasible to Program~(\ref{prog:1}), and gives an objective value at least $\opt\big(\{F_i\}_{i = 1}^n\big) - 5\epsilon$. Further, all proofs about \textsf{SUBROUTINE} are given in Section~\ref{subsec:upper3-subroutine}, Section~\ref{subsec:upper3-reg-cstr} and Section~\ref{subsec:upper3-opt}. Particularly, we will show that \textsf{SUBROUTINE} terminates after finite iterations, and then outputs our final continuous instance $\cont(\hgamma)$.

\subsection{\textsf{PREPROCESSING}}
\label{subsec:upper2-main}

We first present algorithm \textsf{MAIN}, which invokes algorithm \textsf{SUBROUTINE} at the end. Noticeably, benefits from the \textsf{PREPROCESSING} part (namely line~(\ref{alg1:preprocessing}) to line~(\ref{alg1:Delta}) in the below) are threefold: (1)~it introduces a continuous instance $\cont(\gamma^*)$, incurring only a small loss of objective value; (2)~it gives a \emph{uniform} upper bound $u^*$ (depending on $\epsilon$) for supports of all distributions $F_i^*$'s; and (3)~it specifies a \emph{fixed} step-size $\Delta^*$ for algorithm \textsf{SUBROUTINE}, which is helpful in proving that algorithm \textsf{SUBROUTINE} (and thus algorithm \textsf{MAIN}) terminates in finite time.

\begin{algorithm}[H]
\caption{$\textsf{MAIN}(\epsilon, \{F_i\}_{i = 1}^n)$}
\begin{algorithmic}[1]
\Require fixed-precision $\epsilon \in (0, 1)$; feasible instance $\{F_i\}_{i = 1}^n$.
\Ensure continuous instance $\cont(\hgamma)$.
\State Let $u^* \eqdef \R^{-1}(\epsilon)$.
\label{alg1:preprocessing}
\Comment{\emph{Support-supremum of $\{F_i^*\}_{i = 1}^n$.}}
\State Define $D_i^*(p) \eqdef
\begin{cases}
D_i\big(1 + \epsilon\big) & \forall p \in (0, 1] \\
D_i\big((1 + \epsilon) \cdot p\big) & \forall p \in (1, u^*] \\
1 & \forall p \in (u^*, \infty)
\end{cases}$  for all $i \in [n]$.
\label{alg1-preprocess}
\Comment{\emph{Discrete-component $\{F_i^*\}_{i = 1}^n$.}}
\State Let $v^* \eqdef \min \big\{v_i^*: \,\, i \in [n]\big\}$ and $\kappa^* \eqdef \prod\limits_{i = 1}^n (1 - q_i^*)$.
\Comment{\emph{Fixed parameters.}}
\State Let $\delta^* \eqdef \frac{1}{2} \cdot \min \left\{\R(p) - \Psi\big(p, \{F_i^*\}_{i = 1}^n\big): \,\, p \in (1, u^*]\right\}$.
\Comment{\emph{Constraint-slack parameter (fixed).}}
\State Define $\gamma^* \eqdef \R^{-1}(\delta^*)$.
\label{alg1:gamma}
\Comment{\emph{Continuous-component $\cont(\gamma^*)$.}}
\State Let $\Delta^* \eqdef \frac{v^* - 1}{3\gamma^{*2}} \cdot \kappa^* \cdot \delta^*$.
\label{alg1:Delta}
\Comment{\emph{Step-size parameter (fixed).}}
\State \Return $\cont(\hgamma) \eqdef \textsf{SUBROUTINE}\big(\Delta^*, u^*, \cont(\gamma^*) \cup \{F_i^*\}_{i = 1}^n\big)$.
\label{alg1:construct}
\end{algorithmic}
\label{alg1}
\end{algorithm}

From Lemma~\ref{lem:terminate} in Section~\ref{subsec:upper3-subroutine}, we know $\hgamma \in (1, \infty)$ is well-defined. Here is our overall method:
\begin{flalign*}
\opt\big(\{F_i\}_{i = 1}^n\big) - 5\epsilon
\leq & \opt\big(\cont(\gamma^*) \cup \{F_i^*\}_{i = 1}^n\big) & \text{(See Lemma~\ref{lem:obj_approx1} and~\ref{lem:obj_approx2} in Section~\ref{subsec:upper2-cstr-opt})} \\
\leq & \opt\big(\cont(\hgamma)\big) & \text{(See Lemma~\ref{lem:contopt} in Section~\ref{subsec:upper3-opt})} \\
\leq & \opt\big(\cont(1)\big). & \text{(By Lemma~\ref{lem:opt_cont}, and since $\hgamma \in (1, \infty)$)}
\end{flalign*}
Since $\epsilon \in (0, 1)$ is fixed before launching algorithm \textsf{MAIN}, we can choose an arbitrary small $\epsilon \in (0, 1)$, hence the upper-bound part of Theorem~\ref{thm:opt_ap}: for all feasible instance $\{F_i\}_{i = 1}^n$ of Program~(\ref{prog:1}),
\[
\opt\big(\{F_i\}_{i = 1}^n\big) \leq \opt\big(\cont(1)\big) = 2 + \displaystyle{\int_1^{\infty}} \left(1 - e^{-\Q(x)}\right) dx.
\]
Prior to characterizing the near-optimal hybrid instance $\cont(\gamma^*) \cup \{F_i^*\}_{i = 1}^n$, we introduce the following lemma (whose proof is deferred to Appendix~\ref{subapp:upper2-alg}) to measure the quantities defined in algorithm \textsf{MAIN}, which turns out to be crucial for proving the near-optimality.
\begin{lemma}
\label{lem:constant}
$\frac{1}{32}\epsilon^2 \leq \delta^* \leq \frac{1}{2}\epsilon$, and $1 < v^* \leq u^* \leq \gamma^* \leq \frac{33}{\epsilon^2}$, and $\kappa^* \in (0, 1]$, and $\Delta^* \in (0, \delta^*)$.
\end{lemma}

\subsection{Regularity, Feasibility and Near-Optimality of $\cont(\gamma^*) \cup \{F_i^*\}_{i = 1}^n$}
\label{subsec:upper2-cstr-opt}

We turn to investigate the regularity, feasibility and $\epsilon$-optimality of instance $\cont(\gamma^*) \cup \{F_i^*\}_{i = 1}^n$.

\paragraph{Regularity.}
Recall Definition~\ref{def:cont}, continuous instance $\cont(\gamma^*)$ consists of infinite ``small'' triangular distributions, hence the regularity of $\cont(\gamma^*)$. We also have $\{F_i^*\}_{i = 1}^n \in \reg^n$, in that each $D_i^*$ is well-defined on $p \in (0, \infty)$ (see the equivalent conditions for regularity in Section~\ref{subsec:prelim-reg-vv}).

\paragraph{Feasibility.}
The feasibility to constraint~(\ref{cstr:value}) is trivial. For continuous instance $\cont(\gamma^*)$, by definition it satisfies constraint~(\ref{cstr:value}). For each $i \in [n]$, $D_i^*(p)$ remains $D_i(1 + \epsilon)$ when $p \in (0, 1]$. This indicates that $\Phi_i(v_i^+) \geq 1$ and $v_i > 1$, due to the definition of virtual value.

The feasibility to constraint~(\ref{cstr:ap2}) is given by $\R\big(\max\{p, \gamma^*\}\big) + \Psi\big(p, \{F_i^*\}_{i = 1}^n\big) \leq \R(p)$, for all $p \in (1, \infty)$. We would handle this in the following two cases:
\begin{description}
\item [Case I (when $u^* < p < \infty$):] In this case, $D_i^*(p) = F_i^*(p) = 1$ for each $i \in [n]$, due to line~(\ref{alg1-preprocess}) of algorithm \textsf{MAIN}. Thus, the feasibility of $\cont(\gamma^*) \cup \{F_i^*\}_{i = 1}^n$ comes from that of $\cont(\gamma^*)$;
\item [Case II (when $1 < p \leq u^*$):] In this case, $\R\big(\max\{p, \gamma^*\}\big) = \R(\gamma^*) = \delta^*$, due to Lemma~\ref{lem:constant} that $\gamma^* \geq u^*$. Besides, note that $\delta^* = \frac{1}{2} \cdot \min \left\{\R(p) - \Psi\big(p, \{F_i^*\}_{i = 1}^n\big): \,\, p \in (1, u^*]\right\} \geq \frac{1}{32}\epsilon^2 > 0$,
    \begin{equation}
    \label{eq:cstr:ap.0.2}
    \R\big(\max\{p, \gamma^*\}\big) + \Psi\big(p, \{F_i^*\}_{i = 1}^n\big) = \delta^* + \Psi\big(p, \{F_i^*\}_{i = 1}^n\big) \leq \R(p) - \delta^* \leq \R(p).
    \end{equation}
\end{description}

\paragraph{Near-optimality.}
We have $\opt\big(\cont(\gamma^*) \cup \{F_i^*\}_{i = 1}^n\big) \geq \opt\big(\{F_i^*\}_{i = 1}^n\big)$. To measure the loss of objective value, it suffices to verify that $\opt\big(\{F_i^*\}_{i = 1}^n\big) \geq \opt\big(\{F_i\}_{i = 1}^n\big) - 5\epsilon$, that is,
\[
2 + \displaystyle{\int_1^{u^*}} \left(1 - \prod\limits_{i = 1}^n D_i^*(x)\right) dx \geq 2 + \displaystyle{\int_1^{\infty}} \left(1 - \prod\limits_{i = 1}^n D_i(x)\right) dx - 5\epsilon.
\]
Under the definition of $D_i^*$'s in line~(\ref{alg1-preprocess}) of algorithm \textsf{MAIN}, this inequality comes from combining Lemma~\ref{lem:obj_approx1} and Lemma~\ref{lem:obj_approx2} (whose proofs are deferred to Appendix~\ref{subapp:upper2-cstr-opt}) in the below.

\begin{lemma}
\label{lem:obj_approx1}
$\displaystyle{\int_1^{\infty}} \left(1 - \prod\limits_{i = 1}^n D_i(x)\right) dx - \displaystyle{\int_1^{\infty}} \left(1 - \prod\limits_{i = 1}^n D_i\big((1 + \epsilon) \cdot x\big)\right) dx \leq 3\epsilon$.
\end{lemma}

\begin{lemma}
\label{lem:obj_approx2}
$\displaystyle{\int_{u^*}^{\infty}} \left(1 - \prod\limits_{i = 1}^n D_i\big((1 + \epsilon) \cdot x\big)\right) dx \leq \displaystyle{\int_{u^*}^{\infty}} \left(1 - \prod\limits_{i = 1}^n D_i(x)\right) dx \leq 2\epsilon$.
\end{lemma}

\subsection{\textsf{SUBROUTINE} and Running Time}
\label{subsec:upper3-subroutine}

As mentioned before, algorithm \textsf{SUBROUTINE} iteratively calls algorithm \textsf{DIMINISH} to construct new instances that are regular, feasible to Program~(\ref{prog:1}), and generate greater objective values (compared with previous instances). In each iteration, for brevity we respectively re-denote by $\cont(\gamma) \cup \{F_i\}_{i = 1}^n$ and $\cont(\ggamma) \cup \{\FF_i\}_{i = 1}^n$ the current and new instances.

When potential $\Psi\big(\{F_i\}_{i = 1}^n\big)$ decreases to $0$, algorithm \textsf{SUBROUTINE} terminates, and outputs a continuous instance $\cont(\hgamma)$. We present the algorithms as follows.

\begin{algorithm}[H]
\caption{$\textsf{SUBROUTINE}(\Delta^*, u^*, \cont(\gamma^*) \cup \{F_i^*\}_{i = 1}^n)$}
\begin{algorithmic}[1]
\Require step-size $\Delta^*$; support-supremum $u^*$; feasible instance $\cont(\gamma^*) \cup \{F_i^*\}_{i = 1}^n$.
\Ensure desired continuous instance $\cont(\hgamma)$.
\State Let $u \eqdef u^*$ and $\cont(\gamma) \cup \{F_i\}_{i = 1}^n \eqdef \cont(\gamma^*) \cup \{F_i^*\}_{i = 1}^n$.
\Comment{\emph{Initialization.}}
\While{$\Psi\big(\{F_i\}_{i = 1}^n\big) > 0$}
    \State Let $v \eqdef \max \big\{v_i: \,\, \big(i \in [n]\big) \bigwedge \big(\Psi\big(\{F_i\}\big) > 0\big)\big\}$.
    \label{alg2:v}
    \Comment{\emph{Referred to the next $F_i$ to ``vanish''.}}
    \State Let $\Delta \eqdef \min\big\{\Delta^*, \Psi\big(v, \{F_i\}_{i = 1}^n\big)\big\}$.
    \label{alg2:Delta}
    \Comment{\emph{Potential-decrease in this iteration.}}
    \State Let $\uu \eqdef \mathop{\arg\max} \big\{p \in [v, u]: \,\, \Psi\big(p, \{F_i\}_{i = 1}^n\big) \geq \Delta\big\}$.
    \label{alg2:uu}
    \Comment{\emph{Support-supremum of $\{\FF_i\}_{i = 1}^n$.}}
    \Statex \quad\,\,\,$------------------------------------------$
    \State \fcolorbox{white}{lightgray}{Define $\{\FF_i\}_{i = 1}^n \eqdef \textsf{DIMINISH}(\Delta, \uu, \{F_i\}_{i = 1}^n)$.}
    \label{alg2:discrete}
    \State \fcolorbox{white}{lightgray}{Define $\cont(\ggamma)$ by letting $\R(\ggamma) = \R(\gamma) + \Delta$.}
    \label{alg2:continuous}
    \State Update: $u \leftarrow \uu$ and $\cont(\gamma) \cup \{F_i\}_{i = 1}^n \leftarrow \cont(\ggamma) \cup \{\FF_i\}_{i = 1}^n$.
    \label{alg2:update}
\EndWhile
\State \Return $\cont(\gamma)$.
\end{algorithmic}
\label{alg2}
\end{algorithm}

\begin{algorithm}[H]
\caption{$\textsf{DIMINISH}(\Delta, \uu, \{F_i\}_{i = 1}^n)$}
\begin{algorithmic}[1]
\Require potential-decrease $\Delta$; next support-supremum $\uu$; current instance $\{F_i\}_{i = 1}^n$.
\Ensure next instance $\{\FF_i\}_{i = 1}^n$.
\State Let $\Delta^{\#} \eqdef \Delta$ and $W \eqdef \{i \in [n]: \,\, F_i \text{ has probability-mass at } \uu\}$.
\Comment{\emph{Initialization.}}
\ForAll{$i \in [n] \setminus W$}
    \State\label{alg3:Delta_no_mass} Let $\Delta_i \eqdef \Psi\big(\uu, \{F_i\}\big)$.
    Update: $\Delta^{\#} \leftarrow \big(\Delta^{\#} - \Delta_i\big)$.
    \State\label{alg3:FF1} \fcolorbox{white}{lightgray}{Define $\FF_i$ as $\Psi\big(p, \{\FF_i\}\big) \equiv \Psi\big(p, \{F_i\}\big) - \Delta_i$, for all $p \in (0, \uu]$.}
\EndFor
\Comment{\emph{$\Delta^{\#} \geq 0$, by definition of $\uu$.}}
\ForAll{$i \in W$}
    \Comment{\emph{Orders of $i$'s do not matter.}}
    \State Let $\Delta_i \eqdef \min\left\{\Delta^{\#}, \Psi\big(\uu, \{F_i\}\big)\right\}$.
    Update: $\Delta^{\#} \leftarrow \big(\Delta^{\#} - \Delta_i\big)$.
    \State\label{alg3:FF2} \fcolorbox{white}{lightgray}{Define $\FF_i$ as $\Psi\big(p, \{\FF_i\}\big) \equiv \Psi\big(p, \{F_i\}\big) - \Delta_i$, for all $p \in (0, \uu]$.}
\EndFor
\Comment{\emph{$\Delta^{\#} = 0$, by definition of $\uu$.}}
\State \Return $\{\FF_i\}_{i = 1}^n$.
\end{algorithmic}
\label{alg3}
\end{algorithm}

For algorithm \textsf{SUBROUTINE}, we first list several useful facts (namely Lemma~\ref{lem:alg_facts}, whose proofs are deferred to Appendix~\ref{app:upper3-lem:alg_facts}), and then measure its running time  in Lemma~\ref{lem:terminate}. Recall Lemma~\ref{lem:constant} for fixed-parameter $v^* = \min \big\{v_i^*: i \in [n]\big\} > 1$.

\begin{lemma}
\label{lem:alg_facts}
In each iteration of \textsf{SUBROUTINE}, (omit line~(\ref{alg2:update}) which is for updating parameters),
\begin{enumerate}
\item $\gamma^* \geq \gamma \geq \ggamma \geq u \geq \uu$, and $\uu_i = \uu$ for each $i \in [n]$ that $\Delta_i > 0$;
\item $\qq_i \leq q_i$ and $\uu_i \geq \vv_i = v_i = v_i^* \geq v^* > 1$, for each $i \in [n]$;
\item $\Psi\big(p, \{\FF_i\}_{i = 1}^n\big) = \Psi\big(p, \{F_i\}_{i = 1}^n\big) - \Delta$ for all $p \in (0, \uu]$.
\end{enumerate}
\end{lemma}

\begin{lemma}
\label{lem:terminate}
\textsf{SUBROUTINE} terminates in at most $\left\lceil\frac{1}{\Delta^*} \cdot \Psi\big(\{F_i^*\}_{i = 1}^n\big) + n\right\rceil$ iterations, and then output a continuous instance $\cont(\hgamma)$ that $\R(\hgamma) = \R(\gamma^*) + \Psi\big(\{F_i^*\}_{i = 1}^n\big)$, where $\hgamma \in (1, \infty)$ is well-defined.
\end{lemma}

\begin{proof}
Emphasize again, when algorithm \textsf{SUBROUTINE} terminates, potential $\Psi(\{F_i\}_{i = 1}^n) = 0$. Each iteration falls into either of the following two cases:
\begin{description}
\item[Case I (when $\Delta = \Delta^*$):] In this case, we know from Lemma~\ref{lem:alg_facts}.3, after the update in line~(\ref{alg2:update}) of algorithm \textsf{SUBROUTINE}, that potential $\Psi\big(\{F_i\}_{i = 1}^n\big)$ decreases by $\Delta = \Delta^*$. \\
    With fixed step-size $\Delta^* > 0$ and input instance $\cont(\gamma^*) \cup \{F_i^*\}_{i = 1}^n$, this case happens for at most $\left\lceil\frac{1}{\Delta^*} \cdot \Psi\big(\{F_i^*\}_{i = 1}^n\big)\right\rceil$ times;
\item[Case II (when $\Delta < \Delta^*$):] In this case, due to line~(\ref{alg2:Delta}) and line~(\ref{alg2:uu}) of algorithm \textsf{SUBROUTINE}, $\uu = v$ and $\Delta = \Psi\big(v, \{F_i\}_{i = 1}^n\big) = \Psi\big(\uu, \{F_i\}_{i = 1}^n\big)$. Hence, there exists $k \in [n]$ such that
    \begin{itemize}
    \item Since $\uu = v$, it follows from line~(\ref{alg2:v}) of algorithm \textsf{SUBROUTINE} that
        \[
        \uu = v = v_k \xlongequal{Lemma~\ref{lem:alg_facts}.2} \vv_k \quad\quad\quad\quad \Psi\big(\{F_k\}\big) > 0;
        \]
    \item Since $\Delta = \Psi\big(\uu, \{F_i\}_{i = 1}^n\big)$, we know from Lemma~\ref{lem:alg_facts}.3 that $\Psi\big(\uu, \{\FF_i\}_{i = 1}^n\big) = 0$ and thus,
        \[
        \Psi\big(\{\FF_k\}\big) = \Psi\big(\vv_k, \{\FF_k\}\big) \xlongequal{\vv_k = \uu} \Psi\big(\uu, \{\FF_k\}\big) \leq \Psi\big(\uu, \{\FF_i\}_{i = 1}^n\big) = 0.
        \]
    \end{itemize}
    Thus, the $k$-th distribution ``vanishes'' in this iteration, in that $\Psi\big(\{F_k\}\big) > 0$ and $\Psi\big(\FF_k\big) = 0$. Clearly, with input instance $\cont(\gamma^*) \cup \{F_i^*\}_{i = 1}^n$, this case happens for at most $n$ times.
\end{description}
In total, \textbf{Case I} and \textbf{Case II} can happen for at most $\left\lceil\frac{1}{\Delta^*} \cdot \Psi\big(\{F_i^*\}_{i = 1}^n\big) + n\right\rceil$ times. Afterwards, potential $\Psi\big(\{F_i\}_{i = 1}^n\big)$ decreases to $0$, and algorithm \textsf{SUBROUTINE} terminates.

In each iteration of algorithm \textsf{SUBROUTINE} (omit the update in line~(\ref{alg2:update})), combining Lemma~\ref{lem:alg_facts}.3 and line~(\ref{alg2:continuous}) together implies that $\R(\ggamma) + \Psi\big(p, \{F_i\}_{i = 1}^n\big) = \R(\gamma) + \Psi\big(p, \{\FF_i\}_{i = 1}^n\big)$, for all $p \in (0, \uu]$. Particularly, we must have
\[
\R(\ggamma) + \Psi\big(\{F_i\}_{i = 1}^n\big) = \R(\gamma) + \Psi\big(\{\FF_i\}_{i = 1}^n\big).
\]
By induction, output instance $\cont(\hgamma)$ satisfies that $\R(\hgamma) = \R(\gamma^*) + \Psi\big(\{F_i^*\}_{i = 1}^n\big)$. Furthermore, parameter $\hgamma$ must be well-defined in interval $(1, \infty)$, due to Lemma~\ref{lem:RQ} that $R(p)$ is decreasing on $p \in (1, \infty)$, $\lim\limits_{p \rightarrow \infty} \R(p) = 0$ and $\lim\limits_{p \rightarrow 1^+} \R(p) = \infty$.
\end{proof}

\subsection{Proof Plan about \textsf{SUBROUTINE}}
\label{subsec:upper3-reg-cstr}

The rest proofs about algorithm \textsf{SUBROUTINE} (refer to the regularity, feasibility and optimality of involved instances) all relies on induction. Recall Section~\ref{subsec:upper2-cstr-opt} and inequality~(\ref{eq:cstr:ap.0.2}) for the \emph{base-case},
\begin{itemize}
\item $\cont(\gamma^*) \cup \{F_i^*\}_{i = 1}^n$ is regular, and feasible to Program~(\ref{prog:1});
\item $\R(\gamma^*) + \Psi\big(p, \{F_i^*\}_{i = 1}^n\big) \leq \R(p) - \delta^*$ for all $p \in (1, u^*]$.
\end{itemize}
From now on, we would focus on a specific iteration of algorithm \textsf{SUBROUTINE}, safely omitting the update in line~(\ref{alg2:update}). After verifying Main Lemma~\ref{lem:subroutine}, by induction output instance $\cont(\hgamma)$ of algorithm \textsf{SUBROUTINE} satisfies all properties as desired.
\vspace{5.5pt} \\
\fbox{\begin{minipage}{\textwidth}
\begin{mainlemma}
\label{lem:subroutine}
In each iteration of \textsf{SUBROUTINE}, suppose that $\cont(\gamma) \cup \{F_i\}_{i = 1}^n$ is regular, and feasible to Program~(\ref{prog:1}), and that $\R(\gamma) + \Psi\big(p, \{F_i\}_{i = 1}^n\big) \leq \R(p) - \delta^*$ for all $p \in (1, u]$, then
\begin{enumerate}
\item $\cont(\ggamma) \cup \{\FF_i\}_{i = 1}^n$ is regular, and feasible to Program~(\ref{prog:1}). See Sections~\ref{subsec:upper3-reg} and~\ref{subsec:upper3-cstr};
\item $\R(\ggamma) + \Psi\big(p, \{\FF_i\}_{i = 1}^n\big) \leq \R(p) - \delta^*$ for all $p \in (1, \uu]$. See \textbf{Case II} in Section~\ref{subsec:upper3-cstr};
\item $\opt\big(\cont(\ggamma) \cup \{\FF_i\}_{i = 1}^n\big) \geq \opt\big(\cont(\gamma) \cup \{F_i\}_{i = 1}^n\big)$. See Section~\ref{subsec:upper3-opt}.
\end{enumerate}
\end{mainlemma}
\end{minipage}}

\subsubsection{Regularity of $\{\FF_i\}_{i = 1}^n$}
\label{subsec:upper3-reg}

For each $k \in [n]$, the regularity premise of Main Lemma~\ref{lem:subroutine} indicates that $r_k(q)$ is continuous and concave on $[0, 1]$ (see Section~\ref{subsec:prelim-reg-vv}). Under our potential-based construction of $\FF_k$,
\begin{itemize}
\item $\rr_k(q)$ is a left-differentiable and right-differentiable function on $[0, 1]$;
\item $\Psi\big(p, \{\FF_k\}\big)$ remains $\ln\left(1 + \frac{\vv_k \qq_k}{1 - \qq_k}\right) = \ln\left(1 + \frac{v_k q_k}{1 - q_k}\right) - \Delta_k$ when $p \in (0, \vv_k] = (0, v_k]$, due to Lemma~\ref{lem:potential}. In terms of revenue-quantile curve, we know $\rr_k(q) = \frac{\vv_k \qq_k}{1 - \qq_k} \cdot (1 - q)$ when $q \in [\qq_k, 1]$. Clearly, $\rr_k(q)$ is concave on $q \in [\qq_k, 1]$.
\end{itemize}
It remains to prove that $\rr_k(q)$ is concave on $[0, \qq_k]$. For our potential-based construction, we obtain the following structural lemma in Appendix~\ref{subapp:upper3-cstr:regular}.

\begin{lemma}
\label{lem:regular}
$\rr_k(q) \leq \partial_+ \rr_k(\overline{x}) \cdot (q - \overline{x}) + \rr_k(\overline{x})$ for all $\overline{x} \in (0, \qq_k)$ and $q \in [\overline{x}, \qq_k)$.
\end{lemma}

Equipped with Lemma~\ref{lem:regular}, the concavity of $\rr_k(q)$ (and thus the regularity of $\FF_k$) immediately follows from the equivalent condition for concave functions.
\vspace{5.5pt} \\
\fbox{\begin{minipage}{\textwidth}
\begin{fact}[Equivalent Condition for Concavity]
\label{fact:concave}
Suppose $r(q)$ is a left- and right-differentiable function on $[a, b]$, then $r(q)$ is concave on $[a, b]$ iff
\[
r(q) \leq \partial_+ r(x) \cdot (q - x) + r(x) \quad\quad \forall x \in (a, b), q \in [x, b).
\]
\end{fact}
\end{minipage}}

\subsubsection{Feasibility of $\cont(\ggamma) \cup \{\FF_i\}_{i = 1}^n$}
\label{subsec:upper3-cstr}

We turn to check the feasibility to constraint~(\ref{cstr:value}). For each $k \in [n]$, the feasibility premise of Main Lemma~\ref{lem:subroutine} ensures that $v_k > 1$ and $\Phi_k(v_k^+) \geq 1$. It follows from Lemma~\ref{lem:alg_facts}.2 that $\vv_k = v_k > 1$. For virtual value functions, we further acquire Lemma~\ref{lem:upper3-cstr2} in Appendix~\ref{subapp:upper3-cstr:value}, which implies that
\[
\PPhi_k(\vv_k^+) \geq \Phi_k(\vv_k^+) \xlongequal{\vv_k = v_k} \Phi_k(v_k^+) \geq 1 \quad\quad \forall k \in [n].
\]
Later, Lemma~\ref{lem:upper3-cstr2} also serves as a cornerstone of our optimality-analyses (specifically, the proof of Lemma~\ref{lem:abel_ineq1}) in Section~\ref{subsec:upper3-opt}.

\begin{lemma}
\label{lem:upper3-cstr2}
$\PPhi_k(p) \geq \Phi_k(p)$ for all $p \in (\vv_k, \infty) = (v_k, \infty)$ and $k \in [n]$.
\end{lemma}

As for the feasibility to constraint~(\ref{cstr:ap2}), for all $p \in (1, \infty)$, we shall check that
\begin{equation}
\label{eq:cstr:ap.t}
\R\big(\max\{p, \ggamma\}\big) + \Psi\big(p, \{\FF_i\}_{i = 1}^n\big) \leq \R(p).
\end{equation}
\begin{description}
\item[Case I (when $\uu < p < \infty$):] In this case, $\FF_i(p) = 1$ for each $i \in [n]$. Recall Lemma~\ref{lem:RQ} that $\R(p)$ is decreasing on $p \in (1, \infty)$, clearly $\text{LHS of~(\ref{eq:cstr:ap.t})} = \R\big(\max\{p, \ggamma\}\big) \leq \R(p)$.
\item[Case II (when $1 < p \leq \uu$):] In this case, we know $\R\big(\max\{p, \ggamma\}\big) = \R(\ggamma)$, by Lemma~\ref{lem:alg_facts}.1 that $\ggamma \geq u \geq \uu \geq p$. Under premise in Main Lemma~\ref{lem:subroutine} that $\R(\gamma) + \Psi\big(p, \{F_i\}_{i = 1}^n\big) \leq \R(p) - \delta^*$,
    \[
    \begin{aligned}
    \text{LHS of~(\ref{eq:cstr:ap.t})}
    = & \R(\ggamma) + \Psi\big(p, \{\FF_i\}_{i = 1}^n\big)
    \overset{(\ast)}{=} \R(\ggamma) + \Psi\big(p, \{F_i\}_{i = 1}^n\big) - \Delta \\
    \overset{(\diamond)}{=} & \R(\gamma^*) + \Psi\big(p, \{F_i^*\}_{i = 1}^n\big) \leq \R(p) - \delta^*,
    \end{aligned}
    \]
    where $(\ast)$ comes from Lemma~\ref{lem:alg_facts}.3, and $(\diamond)$ comes from line~(\ref{alg2:continuous}) of algorithm \textsf{SUBROUTINE}.
\end{description}

\subsection{Optimality: Monotonicity of Objective Value}
\label{subsec:upper3-opt}

Each iteration of algorithm \textsf{SUBROUTINE} can be divided into $n$ processes. In the $k$-th process, assume\footnote{Otherwise, i.e., when $\Delta_k = 0$, we know $\ggamma_k = \gamma_k$ and $\FF_k \equiv F_k$, and therefore the optimality trivially holds.} w.l.o.g. $\Delta_k > 0$, then Lemma~\ref{lem:alg_facts}.1 ensures that $\uu_k = \uu$. The following box reviews the changes during the $k$-th process:
\vspace{5.5pt} \\
\fcolorbox{white}{lightgray}{\begin{minipage}{\textwidth}
\begin{itemize}
\item $\{\FF_i\}_{i = 1}^{k - 1} \cup \{F_i\}_{i = k}^n$ becomes $\{\FF_i\}_{i = 1}^k \cup \{F_i\}_{i = k + 1}^n$, where $\FF_k$ is defined as
    \begin{equation}
    \label{eq:upper3-F}
    \Psi\big(p, \{\FF_k\}\big) \equiv \Psi\big(p, \{F_k\}\big) - \Delta_k \quad\quad \forall p \in (0, \uu] = (0, \uu_k];
    \end{equation}
\item $\cont(\gamma_k)$ becomes $\cont(\ggamma_k)$, where $\gamma_k$ and $\ggamma_k$ satisfies $\R(\gamma_k) = \R(\gamma) + \sum\limits_{i = 1}^{k - 1} \Delta_i$ and
    \begin{equation}
    \label{eq:upper3-cont}
    \R(\ggamma_k) = \R(\gamma_k) + \Delta_k;
    \end{equation}
\item Clearly, $\Delta_k > 0$ means $\Psi\big(\{F_k\}\big) > 0$ and $\gamma_k > \ggamma_k$ (see Lemma~\ref{lem:RQ}.1). Based on Lemma~\ref{lem:alg_facts}, and since $\vv_k \leq \uu_k \leq u_k \leq u = \max \big\{u_i: \,\, \big(i \in [n]\big) \bigwedge \big(\Psi\big(\{F_i\}\big) > 0\big)\big\}$, we can conclude that
    \begin{equation}
    \label{eq:upper3-p}
    \gamma^* \geq \gamma_k > \ggamma_k \geq u_k \geq \uu_k \geq v_k = \vv_k \geq v^* > 1.
    \end{equation}
\end{itemize}
\end{minipage}}
\vspace{5.5pt}

The objective value can never decrease after the above process, which is formulated as Lemma~\ref{lem:contopt} in the below. Clearly, applying Lemma~\ref{lem:contopt} over all $k \in [n]$ indicates Main Lemma~\ref{lem:subroutine}.3.
\begin{lemma}
\label{lem:contopt}
$\opt\big(\cont(\ggamma_k) \cup \{\FF_i\}_{i = 1}^k \cup \{F_i\}_{i = k + 1}^n\big) \geq \opt\big(\cont(\gamma_k) \cup \{\FF_i\}_{i = 1}^{k - 1} \cup \{F_i\}_{i = k}^n\big)$.
\end{lemma}

\begin{proof}
Let $d_k(p) \eqdef \prod\limits_{i = 1}^{k - 1} \DD_i(p) \cdot \prod\limits_{i = k + 1}^n D_i(p)$ for notational convenience. Recall Lemma~\ref{lem:opt_cont} and the objective of Program~(\ref{prog:1}), the statement of the lemma is given by
\[
2 + \displaystyle{\int_1^{\infty}} \left(1 - d_k(x) \cdot \DD_k(x) \cdot e^{-\Q\big(\max\{x, \ggamma_k\}\big)}\right) dx
\geq 2 + \displaystyle{\int_1^{\infty}} \left(1 - d_k(x) \cdot D_k(x) \cdot e^{-\Q\big(\max\{x, \gamma_k\}\big)}\right) dx.
\]
Recall Lemma~\ref{lem:alg_facts}.1 that $\gamma \geq \ggamma \geq u = \max \big\{u_i: \,\, \Psi\big(\{F_i\}\big) > 0\big\} \geq \uu = \max \big\{\uu_i: \,\, \Psi\big(\{\FF_i\}\big) > 0\big\}$. Certainly, $D_k(p) = \DD_k(p) = d_k(p) = 1$ for all $p \in (\ggamma, \infty)$. As per this, after rearranging the above inequality, we are left to justify that
\begin{equation}
\label{eq:main1}\tag{E1}
\displaystyle{\int_1^{\gamma_k}} d_k(x) \cdot \left(D_k(x) \cdot e^{-\Q(\gamma_k)} - \DD_k(x) \cdot e^{-\Q(\ggamma_k)}\right) dx
\geq \displaystyle{\int_{\ggamma_k}^{\gamma_k}} \left(e^{-\Q(x)} - e^{-\Q(\ggamma_k)}\right) dx.
\end{equation}
In spirit, the following observations are crucial for testifying inequality~(\ref{eq:main1}). Independent of anything else, suppose that step-size $\Delta_k$ is a \emph{first-order} infinitesimal. Since then,
\begin{itemize}
\item $\big(\DD_k(p) - D_k(p)\big)$ is a positive first-order infinitesimal, for all $p \in (1, \infty)$;
\item $\big(\gamma_k - \ggamma_k\big)$ and $\left(e^{-\Q(\gamma_k)} - e^{-\Q(\ggamma_k)}\right)$ is positive first-order infinitesimals;
\item The $\text{LHS of inequality~(\ref{eq:main1})}$ is a \emph{first-order} infinitesimal: the interval of integration, $[1, \gamma_k]$, is of constant length; the involved integrand is (pointwisely) of first-order infinitesimal;
\item The $\text{RHS of inequality~(\ref{eq:main1})}$ is a \emph{positive second-order} infinitesimal: the interval of integration, $[\ggamma_k, \gamma_k]$, has a length of first-order infinitesimal; $\left(e^{-\Q(p)} - e^{-\Q(\ggamma_k)}\right)$ is a first-order infinitesimal, for all $p \in [\ggamma_k, \gamma_k]$.
\end{itemize}
Conceivably, a first-order infinitesimal (namely the LHS) should be no less than a second-order infinitesimal (namely the RHS). However, because of the following two issues, converting the above intuition into the ultimate proof is far from obvious.
\begin{itemize}
\item Whether the $\text{LHS of inequality~(\ref{eq:main1})}$ is \emph{positive} is quite non-trivial: in fact, we can construct a sophisticated example such that, on interval $p \in [1, \gamma_k]$, the values of the involved integrand (of the LHS) change signs.

    The main tool for overcoming this issue is Abel's inequality. Showing that Abel's inequality is applicable requires enormous calculations. Nevertheless, the subsequent benefits are twofold: (see \textbf{Analysis I})~we do catch a new definite integration, which always takes a positive value, to bound the $\text{LHS of inequality~(\ref{eq:main1})}$ from below; (see \textbf{Analysis II})~the involved integrand is greatly simplified, and thus the new definite integration admits an \emph{explicit formula}.
\item Whether the chosen step-size $\Delta_k$ is small enough to guarantee inequality~(\ref{eq:main1})?

    We shall rewrite/relax both hand side of inequality~(\ref{eq:main1}) to be formulas containing $\Delta_k$, and measure the resulting coefficients (of $\Delta_k$ and $\Delta_k^2$). For the RHS, this task is relatively easy. To treat the LHS in a similar manner, the following two observations are crucial.
    \begin{enumerate}
    \item Recall step~(\ref{alg1:gamma}) of \textsf{MAIN} that $\Delta^* = \frac{v^* - 1}{3\gamma^{*2}} \cdot \kappa^* \cdot \delta^*$. Adopting the idea from Section~\ref{subsec:upper3-cstr}, we can show a slack of constraint~(\ref{cstr:ap2}): for all $p \in (1, \uu]$,
        \[
        \Psi\big(\cont(\ggamma_k) \cup \{\FF_i\}_{i = 1}^k \cup \{F_i\}_{i = k + 1}^n\big) = \Psi\big(\cont(\gamma_k) \cup \{\FF_i\}_{i = 1}^{k - 1} \cup \{F_i\}_{i = k}^n\big) \leq \R(p) - \delta^*.
        \]
        In fact, this constraint-slack will reward us with a lower bound (containing $\Delta^*$) of the LHS of inequality~(\ref{eq:main1}).
    \item Hitherto, it suffices to show that the relax LHS (i.e., a formula of order $\Delta^* \geq \Delta = \sum\limits_{i = 1}^n \Delta_i \geq \Delta_k$) is no less than the relaxed RHS (i.e., a formula of order $\Delta_k^2$): this task is relatively easy. Notably, the value $\Delta^* = \frac{v^* - 1}{3\gamma^{*2}} \cdot \kappa^* \cdot \delta^*$ is carefully chosen, to make it fit in the coefficients on both hand sides of the relaxed inequality~(\ref{eq:main1}).
    \end{enumerate}
    We will elaborate on the above ideas in \textbf{Analysis III}.
\end{itemize}
In the rest part of Section~\ref{subsec:upper3-opt}, we will implement the above proof plan.
\vspace{5.5pt} \\
\fbox{\begin{minipage}{\textwidth}
\begin{fact}[Abel's Inequality]
\label{fact:abel_ineq}
Suppose (1)~$\lambda(x)$ is a non-negative and increasing function on $[a, b]$; and (2)~$\mu(x)$ is an integrable function such that $\displaystyle{\int_y^b} \mu(x) dx \geq 0$ for all $y \in [a, b]$, then
\[
\displaystyle{\int_a^b} \lambda(x) \cdot \mu(x) dx \geq \lambda(a^+) \cdot \displaystyle{\int_a^b} \mu(x) dx.
\]
\end{fact}
\end{minipage}}

\paragraph{Analysis I: Exploiting Abel's Inequality.}
To see Fact~\ref{fact:abel_ineq} can be applied to the LHS of inequality~(\ref{eq:main1}), we shall check the two involved conditions. Specifically,
\begin{enumerate}
\item $d_k(p) = \prod\limits_{i = 1}^{k - 1} \DD_i(p) \cdot \prod\limits_{i = k + 1}^n D_i(p)$ is a non-negative and increasing function on $p \in (0, \infty)$;
\item $\displaystyle{\int_p^{\gamma_k}} \left(D_k(x) \cdot e^{-\Q(\gamma_k)} - \DD_k(x) \cdot e^{-\Q(\ggamma_k)}\right) dx \geq 0$ for all $p \in [1, \gamma_k]$ .
\end{enumerate}
The correctness of the first condition is trivial. The second condition is more technical. Based on Lemma~\ref{lem:upper3-cstr2}, Lemma~\ref{lem:main} (to appear soon after, whose proof is deferred to Appendix~\ref{app:upper3-lem}) and extra mathematical facts, we further confirm Lemma~\ref{lem:abel_ineq1} in Appendix~\ref{app:upper3-opt}. Hence, the second condition for applying Fact~\ref{fact:abel_ineq} (to the LHS of inequality~(\ref{eq:main1})) is also guaranteed.

\begin{lemma}
\label{lem:abel_ineq1}
$\displaystyle{\int_p^{\gamma_k}} D_k(x) \cdot e^{-\Q(\gamma_k)} dx \geq \displaystyle{\int_p^{\gamma_k}} \DD_k(x) \cdot e^{-\Q(\ggamma_k)} dx$ for all $p \in [1, \gamma_k]$.
\end{lemma}

To generate useful bound for the LHS of inequality~(\ref{eq:main1}), we need to measure $d_k(1^+)$. For each $i \in [n]$, we have shown in Section~\ref{subsec:upper3-cstr} that $\PPhi_i(\vv_i^+) \geq \Phi_i(v_i^+) \geq 1$, which means
\begin{equation}
\label{eq:abel_ineq5}
\DD_i(p) = \DD_i(1^+) = 1 - \qq_i \quad\quad D_i(p) = D_i(1^+) = 1 - q_i \quad\quad \forall p\in (0,1].
\end{equation}
Recall Lemma~\ref{lem:constant} for fixed-parameter $\kappa^* = \prod\limits_{i = 1}^n (1 - q_i^*) \in (0, 1]$.
\begin{equation}
\label{eq:abel_ineq4}
d_k(1^+) \overset{(\ref{eq:abel_ineq5})}{=} \prod\limits_{i = 1}^{k - 1} (1 - \qq_i) \cdot \prod\limits_{i = k + 1}^n (1 - q_i) \overset{(\diamond)}{=} \prod\limits_{i \in [n] \setminus \{k\}} (1 - q_i^*) \geq \kappa^*,
\end{equation}
where in $(\diamond)$ we apply Lemma~\ref{lem:alg_facts}.2 inductively. Combine Fact~\ref{fact:abel_ineq} and inequality~(\ref{eq:abel_ineq4}) together,
\[
\text{LHS of~(\ref{eq:main1})}
\geq \kappa^* \cdot \displaystyle{\int_1^{\gamma_k}} \left(D_k(x) \cdot e^{-\Q(\gamma_k)} - \DD_k(x) \cdot e^{-\Q(\ggamma_k)}\right) dx.
\]

\paragraph{Analysis II: Getting Explicit Formulas.}
The relaxed formula mentioned above (for the LHS of inequality~(\ref{eq:main1})) admits an explicit formula. Recall that $D_k(p) = 1$ for all $p \in (\ggamma, \infty)$.
\[
\begin{aligned}
\displaystyle{\int_1^{\gamma_k}} D_k(x) dx
= & \gamma_k - 1 + \displaystyle{\int_0^1} \big(1 - D_k(x)\big) dx - \displaystyle{\int_0^{\infty}} \big(1 - D_k(x)\big) dx \\
\overset{(\ref{eq:abel_ineq5})}{=} & \gamma_k - 1 + q_k - \displaystyle{\int_0^{\infty}} \big(1 - D_k(x)\big) dx \\
\overset{(\diamond)}{=} & \gamma_k - 1 + q_k - v_k q_k,
\end{aligned}
\]
where $(\diamond)$ comes from Main Lemma~\ref{lem:p/(p+1)}.2 (that $r_k(0) = 0$) and Main Lemma~\ref{lem:virtual_value} (see Appendix~\ref{app:prelim}). Similarly, $\displaystyle{\int_1^{\gamma_k}} \DD_k(x) dx = \gamma_k - 1 + \qq_k - \vv_k \qq_k \xlongequal{Lemma~\ref{lem:alg_facts}.2} \gamma_k - 1 + \qq_k - v_k \qq_k$. To sum up,
\[
\text{LHS of~(\ref{eq:main1})}
\geq \kappa^* \cdot \left[\left(\gamma_k - 1 - v_k q_k + q_k\right) \cdot e^{-\Q(\gamma_k)} - \left(\gamma_k - 1 - v_k \qq_k + \qq_k\right) \cdot e^{-\Q(\ggamma_k)}\right]
\]
Moreover, we can relax the RHS of inequality~(\ref{eq:main1}) as follows:
\[
\text{RHS of~(\ref{eq:main1})} = \displaystyle{\int_{\ggamma_k}^{\gamma_k}} \left(e^{-\Q(x)} - e^{-\Q(\ggamma_k)}\right) dx \leq (\gamma_k - \ggamma_k) \cdot \left(e^{-\Q(\gamma_k)} - e^{-\Q(\ggamma_k)}\right),
\]
which follows from Lemma~\ref{lem:RQ}.1 that $\Q(p)$ is decreasing on $p \in (1, \infty)$. Adopt both tricks, and then rearrange the intermediate inequality. Afterwards, we are left to prove the following.
\begin{equation}
\label{eq:main2}\tag{E2}
\gamma_k - (1 - q_k + v_k q_k) - (v_k - 1) \cdot \frac{q_k - \qq_k}{e^{\Q(\ggamma_k) - \Q(\gamma_k)} - 1}
\geq \frac{1}{\kappa^*} \cdot (\gamma_k - \ggamma_k).
\end{equation}

\paragraph{Analysis III: Exploiting Auxiliary Lemmas.}
To conquer inequality~(\ref{eq:main2}), we shall introduce the following technical lemma (whose proof is deferred to Appendix~\ref{app:upper3-lem}).

\begin{lemma}
\label{lem:main}
For all $p \in [\vv_k, \uu_k] \subset [v_k, u_k]$,
\begin{equation}
\label{eq:lem:main}
\gamma_k - \Big[F_k(p) + p \cdot \big(1 - F_k(p)\big)\Big] - (p - 1) \cdot \frac{\FF_k(p) - F_k(p)}{e^{\Q(\ggamma_k) - \Q(\gamma_k)} - 1} \geq \frac{1}{\kappa^*} \cdot \gamma^{*2} \cdot \Delta^*.
\end{equation}
\end{lemma}

We continue to relax both hand sides of inequality~(\ref{eq:main2}). Assigning $p \leftarrow v_k$ in Lemma~\ref{lem:main} gives
\[
\text{LHS of~(\ref{eq:main2})} = \text{LHS of~(\ref{eq:lem:main})} \geq \frac{1}{\kappa^*} \cdot \gamma^{*2} \cdot \Delta^*.
\]
Moreover, recall Lemma~\ref{lem:RQ}.1 that $\R(p)$ is a decreasing and convex function on $p \in (1, \infty)$,
\[
\text{RHS of~(\ref{eq:main2})} \leq \frac{1}{\kappa^*} \cdot \frac{\R(\ggamma_k) - \R(\gamma_k)}{\big|\R'(\gamma_k)\big|} \overset{(\ref{eq:upper3-cont})}{=} \frac{1}{\kappa^*} \cdot \frac{\Delta_k}{\big|\R'(\gamma_k)\big|} \overset{(\star)}{\leq} \frac{1}{\kappa^*} \cdot \gamma_k^2 \cdot \Delta_k,
\]
where in $(\star)$ we apply Lemma~\ref{lem:RQ}.3 that $\big|\R'(p)\big| \geq \frac{1}{p^2}$ for all $p \in (1, \infty)$. Gather everything together, it remains to show that
\[
\frac{1}{\kappa^*} \cdot \gamma^{*2} \cdot \Delta^* \geq \frac{1}{\kappa^*} \cdot \gamma_k^2 \cdot \Delta_k.
\]
This follows as $\gamma^* \geq \gamma_k$ and $\Delta^* \geq \Delta \geq \Delta_k$, according to inequality~(\ref{eq:upper3-p}) and line~(\ref{alg2:Delta}) of algorithm \textsf{SUBROUTINE}, respectively. This completes the proof of Lemma~\ref{lem:contopt}.
\end{proof}

%

\section{Bayesian Multi-Dimensional Unit-Demand Mechanism Design}
\label{sec:extension}
In this section, we move to the multi-dimensional unit-demand environment (with a single buyer), aiming to certify Theorem~\ref{thm:bupp_upm}.
\vspace{5.5pt} \\
\fbox{\begin{minipage}{\textwidth}
\paragraph{[Theorem~\ref{thm:bupp_upm}].}
\emph{To sell heterogeneous items to a unit-demand buyer with independent regular distributions for different items, the supremum of the ratio of optimal item pricing to uniform pricing equals to $\C \approx 2.6202$.}
\end{minipage}}
\vspace{5.5pt}

Intuitively, the upper-bound part comes from combining Theorem~\ref{thm:opt_ap} and previous results (i.e., the \emph{copying} technique developed in~\cite{CHK07,CHMS10,CMS15}). As for the lower-bound part, it turns out that applying a natural extension of Example~\ref{exp:opt-ap-lower} (i.e., the lower-bound instance for Theorem~\ref{thm:opt_ap}) would work. We will settle everything in the rest of this section.

\subsection{Recap of Multi-Dimensional Unit-Demand Setting}

The problem of selling \emph{heterogeneous} items among unit-demand buyers, is termed the \emph{Bayesian multi-dimensional unit-demand mechanism design} ($\bmumd$) problem~\cite{CHMS10,KW12,CMS15}. In the restricted setting with a \emph{single} buyer, basic concepts for the problem are given as follows. We re-denote by $n \in \mathbb{N}_+$ the number of items, and re-denote by $\{F_i\}_{i = 1}^n$ the buyer's value distributions for the items.

\paragraph{$\bmumd$ problem.}
In the single-buyer setting, Chawla et al.~\cite{CMS15} proved that any a mechanism (possibly randomized) can be mapped to a \emph{lottery-pricing} mechanism. We denote a lottery $\ell = \left(p^{\ell}, x_1^{\ell}, x_2^{\ell}, \cdots, x_n^{\ell}\right)$ by $(n + 1)$ elements: the embedded price $p^{\ell} \in \mathbb{R}_+$; and each $i$-th item's allocation probability $x_i^{\ell}$. With values $(w_1, w_2, \cdots, w_n)$ drawn from distributions $\{F_i\}_{i = 1}^n$, the buyer will gain utility $u^l \eqdef \sum\limits_{i = 1}^n w_i \cdot x_i^{\ell} - p^{\ell}$ after accepting this lottery. W.l.o.g. we have $\sum\limits_{i = 1}^n x_i^{\ell} \leq 1$, since the buyer is unit-demand.

A \emph{lottery-pricing} mechanism involves an assemblage of lotteries $\mathcal{L} = \{\ell_0, \ell_1, \ell_2, \cdots\}$, where the special lottery $\ell_0 = \{0\}^{n + 1}$ respects the buyer's rationality. Facing lottery set $\mathcal{L}$, the buyer will take a utility-optimal lottery $\ell^* \in \mathop{\arg\max} \big\{u^l: \,\, \ell \in \mathcal{L}\big\}$ (possibly the special lottery $l_0$, which results in $u^{l_0} = 0$). For brevity, denote by $\bmumd\big(\{F_i\}_{i = 1}^n\big)$ the optimal lottery-pricing and revenue.

Restricted to deterministic mechanisms, the $\bmumd$ problem is also known as the \emph{Bayesian unit-demand item-pricing} ($\bupp$) problem in~\cite{CHK07,CD15,CMS15,CDOPSY15}. In an \emph{item pricing} scheme: the seller posts prices $(p_1, p_2, \cdots, p_n) \in \mathbb{R}_+^n$ for the items; with values $(w_1, w_2, \cdots, w_n)$, the buyer selects a utility-optimal item $i^* \in \mathop{\arg\max} \big\{w_i - p_i: \,\, i \in [n]\big\}$ (or nothing when $w_{i^*} < p_{i^*}$). Again, we abuse notation $\bupp\big(\{F_i\}_{i = 1}^n\big)$ to denote the optimal item pricing and revenue.

A \emph{uniform-pricing} ($\upm$) mechanism~\cite{CD15,HH15} is a restricted item pricing, where $p_1 = p_2 = \cdots = p_n = p$. For a specific price $p \in \mathbb{R}_+$, let $\upm\big(p, \{F_i\}_{i = 1}^n\big)$ be the revenue from such uniform-pricing. Since then, $\upm\big(\{F_i\}_{i = 1}^n\big) \eqdef \max\big\{\upm\big(p, \{F_i\}_{i = 1}^n\big): \,\, p \in \mathbb{R}_+\big\}$ presents the optimum.

\paragraph{Copying Technique.}
Chawla et al.~\cite{CHK07,CMS15} developed the so-called \emph{copying} technique. Derived from an \emph{original} instance $\{F_i\}_{i = 1}^n$ for the single-buyer $\bmumd$ problem, consider a new scenario: the seller aims to sell a single item among $n$ buyers; for the item, each $i$-th buyer's value follows distribution $F_i$. Chawla et al. named this instance the \emph{copied} instance. For brevity, we use $\{F_i\}_{i = 1}^n$ to denote both instances. Which scenario we refer to will be clear from the context.

For the involved two instances, Chawla et al.~\cite{CHK07,CMS15} acquired the following two lemmas, hence a connection to Myerson Auction ($\opt$) in the single-dimensional setting.

\begin{lemma}[\cite{CHK07}]
\label{lem:copy_bupp}
$\bupp\big(\{F_i\}_{i = 1}^n\big) \leq \opt\big(\{F_i\}_{i = 1}^n\big)$.
\end{lemma}

\begin{lemma}[\cite{CMS15}]
\label{lem:copy_bmumd}
$\bmumd\big(\{F_i\}_{i = 1}^n\big) \leq 2\opt\big(\{F_i\}_{i = 1}^n\big)$.
\end{lemma}

\subsection{Upper-Bound Analysis of Theorem~\ref{thm:bupp_upm}}
\label{subsec:extension:bupp_upm_upper}

As Lemma~\ref{lem:copy_upm_ap} suggests, in terms of revenue, the multi-dimensional uniform-pricing ($\upm$) and the single-dimensional anonymous pricing ($\ap$) are equivalent.

\begin{lemma}
\label{lem:copy_upm_ap}
$\upm\big(p, \{F_i\}_{i = 1}^n\big) = \ap\big(p, \{F_i\}_{i = 1}^n\big)$ for all $p \in (0, \infty)$. Thus, $\upm\big(\{F_i\}_{i = 1}^n\big) = \ap\big(\{F_i\}_{i = 1}^n\big)$.
\end{lemma}

\begin{proof}
Given $p \in (0, \infty)$, in both of the original and copied instances, the seller allocates an/the item with probability $1 - \prod\limits_{i = 1}^n F_i(p)$. Clearly, this fact indicates the lemma.
\end{proof}

Recall Theorem~\ref{thm:opt_ap} that $\C \cdot \ap\big(\{F_i\}_{i = 1}^n\big) \geq \opt\big(\{F_i\}_{i = 1}^n\big)$ for all $\{F_i\}_{i = 1}^n \in \reg^n$. Combining this with Lemma~\ref{lem:copy_bupp} and Lemma~\ref{lem:copy_upm_ap}, we settle the upper-bound part of Theorem~\ref{thm:bupp_upm} as follows:
\[
\C \cdot \upm\big(\{F_i\}_{i = 1}^n\big) = \C \cdot \ap\big(\{F_i\}_{i = 1}^n\big) \geq \opt\big(\{F_i\}_{i = 1}^n\big) \geq \bupp\big(\{F_i\}_{i = 1}^n\big).
\]
Similarly, gathering Theorem~\ref{thm:opt_ap}, Lemma~\ref{lem:copy_bmumd} and Lemma~\ref{lem:copy_upm_ap} together results in Corollary~\ref{crl:bmumd_upm}.
\vspace{5.5pt} \\
\fbox{\begin{minipage}{\textwidth}
\begin{corollary}
\label{crl:bmumd_upm}
To sell heterogeneous items, to a unit-demand buyer with regular distributions, the ratio of $\bmumd$ to $\upm$ is upper-bounded by $2\C \approx 5.2404$.
\end{corollary}
\end{minipage}}
\vspace{5.5pt}

Note that the ratio involved in Corollary~\ref{crl:bmumd_upm} is lower-bound by $\C \approx 2.6202$ (see Section~\ref{subsec:extension:bupp_upm_lower}). As stressed by Alaei et al.~\cite{AHNPY15}, closing this gap remains an important open problem, either by proving improved upper bound, or constructing sharper examples.

\subsection{Lower-Bound Analysis of Theorem~\ref{thm:bupp_upm}}
\label{subsec:extension:bupp_upm_lower}

We are left with the lower-bound part of Theorem~\ref{thm:bupp_upm}, which comes from gathering Main Lemma~\ref{lem:bupp_upm_lower} (in the below) with Theorem~\ref{thm:opt_ap_lower} and Lemma~\ref{lem:copy_upm_ap}: viewed as a single-buyer multi-item instance\footnote{In terms of notations defined in Example~\ref{exp:opt-ap-lower}, the seller brings $(n + 2)$ heterogeneous items, for which the buyer's values are drawn from distributions $\{\tri(\infty)\} \cup \{\tri(v_i, q_i)\}_{i = 1}^{n + 1}$.}, Example~\ref{exp:opt-ap-lower} turns out to be the desired lower-bound instance for Theorem~\ref{thm:bupp_upm}.

\paragraph{[Example~\ref{exp:opt-ap-lower}~(\cite{JLTX18})].}
\emph{Given $\epsilon \in (0, 1)$, let $a \eqdef \Q^{-1}\left(\ln\frac{8}{\epsilon}\right) > 1$, $b \eqdef \frac{8}{\epsilon}$ and $\delta \eqdef \frac{b - a}{n}$. Define triangular instance $\{\tri(\infty)\} \cup \{\tri(v_i, q_i)\}_{i = 1}^{n + 1}$ as, for each $i \in [n + 1]$,}
\[
v_i = b - (i - 1) \cdot \delta \quad\quad\quad\quad q_i = \frac{\R(v_i) - \R(v_{i - 1})}{v_i + \R(v_i) - \R(v_{i - 1})}.
\]

\paragraph{[Theorem~\ref{thm:opt_ap_lower}~(\cite{JLTX18})].}
\emph{Consider the triangular instance in Example~\ref{exp:opt-ap-lower}. When $n \in \mathbb{N}_+$ is sufficiently large, $\ap\big(\{\tri(\infty)\} \cup \{\tri(v_i, q_i)\}_{i = 1}^{n + 1}\big) \leq 1$ and further,}
\[
\C \geq \opt\big(\{\tri(\infty)\} \cup \{\tri(v_i, q_i)\}_{i = 1}^{n + 1}\big) = 1 + \sum\limits_{i = 1}^{n + 1} v_i q_i \cdot \prod\limits_{j = 1}^{i - 1} (1 - q_i) \geq \C - \epsilon.
\]
\fbox{\begin{minipage}{\textwidth}
\begin{mainlemma}
\label{lem:bupp_upm_lower}
Consider the triangular instance in Example~\ref{exp:opt-ap-lower}. For any fixed $\epsilon' \in (0, 1)$,
\[
\bupp\big(\{\tri(\infty)\} \cup \{\tri(v_i, q_i)\}_{i = 1}^{n + 1}\big) \geq  -\epsilon' + \opt\big(\{\tri(\infty)\} \cup \{\tri(v_i, q_i)\}_{i = 1}^{n + 1}\big).
\]
\end{mainlemma}
\end{minipage}}

\begin{proof}
Recall Lemma~\ref{lem:copy_upm_ap}, the revenue from the multi-dimensional $\bupp$ is no more than the revenue from the single-dimensional $\opt$: in the former case, the unit-demand buyer would choose an item in favor of himself (rather than in favor of the seller). However, with a triangular instance like Example~\ref{exp:opt-ap-lower}, the seller can manipulate the items' prices, so as to incentive the buyer to choose higher-price items, while keeping revenue-loss small enough (compared with $\opt$).

Emphasizing again, here we regard Example~\ref{exp:opt-ap-lower} as a single-buyer multi-item instance. In particular, with parameters $h \in (0, 1)$ and $H > v_1 = b$, take the item pricing $(H, p_1, p_2, \cdots, p_{n + 1})$ below into account:
\begin{itemize}
\item Post price $H$ for item $\tri(\infty)$;
\item Post price $p_i \eqdef (1 - h) \cdot v_i$ for the $i$-th item $\tri(v_i, q_i)$, for each $i \in [n + 1]$. Since then, by choosing the $i$-th item, the buyer's utility is at most $v_i - p_i = h \cdot v_i$.
\end{itemize}
Under such item pricing, consider the following $(n + 2)$ disjoint events: ($A_{\infty}$)~that the buyer gets item $\tri(\infty)$; and ($A_i$ for each $i \in [n + 1]$)~that the buyer gets the $i$-th item $\tri(v_i, q_i)$.

Recall Example~\ref{exp:opt-ap-lower} that $b = v_1 > v_2 > \cdots > v_{n + 1} = a$, event $\{A_{\infty}\}$ happens when the buyer values item $\tri(\infty)$ at $w_{\infty} > H + h \cdot b$, with utility exceeding $h \cdot b$ (thus exceeding $h \cdot v_i$, the greatest possible utility from the $i$-th item, for each $i \in [n + 1]$). Besides, event $\{A_{\infty}\}$ happens only if the buyer values item $\tri(\infty)$ at $w_{\infty} \geq H$. Formally,
\[
\Pr\big\{w_{\infty} > H + h \cdot b\big\} \leq \Pr\big\{A_{\infty}\big\} \leq \Pr\big\{w_{\infty} \geq H\big\}.
\]
Independent of anything else, we can impel $H$ to approach to infinity. Afterwards, the buyer gets item $\tri(\infty)$ with \emph{negligible} probability and expected payment of $1$, in that
\begin{align*}
& \lim\limits_{H \rightarrow \infty} \Pr\big\{A_{\infty}\big\} \leq \lim\limits_{H \rightarrow \infty} \Pr\big\{w_{\infty} \geq H\big\} = \lim\limits_{H \rightarrow \infty} \frac{1}{H + 1} = 0 \\
& \lim\limits_{H \rightarrow \infty} H \cdot \Pr\big\{A_{\infty}\big\} \leq \lim\limits_{H \rightarrow \infty} H \cdot \Pr\big\{w_{\infty} \geq H\big\} = \lim\limits_{H \rightarrow \infty} \frac{H}{H + 1} = 1, \\
& \lim\limits_{H \rightarrow \infty} H \cdot \Pr\big\{A_{\infty}\big\} \geq \lim\limits_{H \rightarrow \infty} H \cdot \Pr\big\{w_{\infty} > H + h \cdot b\big\} = \lim\limits_{H \rightarrow \infty} \frac{H}{H + h \cdot b + 1} = 1.
\end{align*}
For each $i \in [n + 1]$, even $\{A_i\}$ happens when the following two conditions hold:
\begin{itemize}
\item The buyer values the $j$-th item at $w_j < p_j = (1 - h) \cdot v_j$, for each $j \in [i - 1]$. Respecting his own rationality, the buyer will discard all these items;
\item The buyer values the $i$-th item at $w_i = v_i$, gaining utility $w_i - p_i = h \cdot v_i$. This utility exceeds $h \cdot v_k$ (namely the greatest possible utility from the $k$-th item), for each $k \in [n + 1] \setminus [i]$.
\end{itemize}
Conditioned on that event $\{A_i\}$ happens, we can bound the expected revenue from below:
\[
\begin{aligned}
p_i \cdot \Pr\{A_i\}
\geq & (1 - h) \cdot v_i \cdot \Pr\big\{w_i = v_i\big\} \cdot \prod\limits_{j = 1}^{i - 1} \Pr\big\{w_j < (1 - h) \cdot v_j\big\} \\
= & (1 - h)^i \cdot v_i q_i \cdot \prod\limits_{j = 1}^{i - 1} \left[\frac{1 - q_j}{1 - h \cdot (1 - q_j)}\right] \\
\geq & \big[1 - (n + 1) \cdot h\big] \cdot v_i q_i \cdot \prod\limits_{j = 1}^{i - 1} (1 - q_j),
\end{aligned}
\]
for each $i \in [n + 1]$. To sum up, the revenue from such item pricing $(H, p_1, p_2, \cdots, p_{n + 1})$ is
\begin{flalign*}
1 + \sum\limits_{i = 1}^{n + 1} p_i \cdot \Pr\{A_i\}
\geq & 1 + \big[1 - (n + 1) \cdot h\big] \cdot \sum\limits_{i = 1}^{n + 1} v_i q_i \cdot \prod\limits_{j = 1}^{i - 1} (1 - q_j) \\
\geq & \big[1 - (n + 1) \cdot h\big] \cdot \opt\big(\{\tri(\infty)\} \cup \{\tri(v_i, q_i)\}_{i = 1}^{n + 1}\big). & \text{(By Fact~\ref{fact:revenue_spm_opt})}
\end{flalign*}
Recall Theorem~\ref{thm:opt_ap_lower} that $\opt\big(\{\tri(\infty)\} \cup \{\tri(v_i, q_i)\}_{i = 1}^{n + 1}\big) \leq \C < 3$. After assigning $h \leftarrow \frac{\epsilon'}{3 \cdot (n + 1)}$, where $\epsilon' > 0$ is given in the lemma, we conclude that
\[
1 + \sum\limits_{i = 1}^{n + 1} p_i \cdot \Pr\{A_i\} \geq  -\epsilon' + \opt\big(\{\tri(\infty)\} \cup \{\tri(v_i, q_i)\}_{i = 1}^{n + 1}\big)
\]
In terms of revenue, optimal deterministic mechanism $\bupp\big(\{\tri(\infty)\} \cup \{\tri(v_i, q_i)\}_{i = 1}^{n + 1}\big)$ must surpass the above item pricing $(H, p_1, p_2, \cdots, p_{n + 1})$. This completes the proof of Main lemma~\ref{lem:bupp_upm_lower}.
\end{proof}

%

\bibliographystyle{plain}
\bibliography{main}

\newpage

\appendix
\renewcommand{\appendixname}{Appendix~\Alph{section}}

\section{Extended Preliminaries}
\label{app:prelim}
In this appendix, we formally prove a few facts relating revenue-quantile curve and virtual value CDF, which will be useful in our proof. 
 Fact~\ref{fact1} captures reductions from revenue-quantile curve $r_i(q)$, and Fact~\ref{fact2} formalizes reductions from virtual value CDF $D_i$. Notably, inverse function $D_i^{-1}$ (respect virtual value CDF) is defined as follows: for all $x \in [1 - q_i, 1]$,
\[
D_i^{-1}(x) \eqdef \arg\max \big\{p \in (0, \infty]:\,\, D_i(p) \leq x\big\}.
\]
\fbox{\begin{minipage}{\textwidth}
\begin{fact}
\label{fact1}
Given a continuous and concave revenue-quantile curve $r_i(q)$ on $q \in [0, 1]$,
\begin{enumerate}
\item Value CDF $F_i$ is given by $F_i(p) \eqdef 0$ for all $p \in \big(0, r_i(1)\big]$ and
    \[
    F_i(p) \eqdef 1 - \max \big\{q \in [0, 1]: \,\, \big.r_i(q)\big/q \leq p\big\} \quad\quad\quad\quad \forall p \in \big(r_i(1), \infty\big];
    \]
\item Virtual Value CDF $D_i$ is given by $D_i(p) \eqdef 1 - q_i$ for all $q \in \big(0, \Phi_i(v_i^+)\big]$ and
    \[
    D_i(p) \eqdef 1 - \max \big\{q \in [0, 1]: \,\, \partial_+ r_i(q) \leq p\big\} \quad\quad\quad\quad \forall p \in \big(\Phi_i(v_i^+), \infty\big];
    \]
\end{enumerate}
\end{fact}
\end{minipage}}
\vspace{5.5pt} \\
\fbox{\begin{minipage}{\textwidth}
\begin{fact}
\label{fact2}
Given a well-defined virtual value CDF $D_i$ and $r_i(0) = \lim\limits_{p \rightarrow \infty} p \cdot \big(1 - F_i(p)\big)$,
\begin{enumerate}
\item In the range of $q \in [0, q_i]$, revenue-quantile curve $r_i(q)$ is given by
    \[
    r_i(q) \eqdef r_i(0) + \displaystyle{\int_0^q} D_i^{-1}(1 - x) \cdot dx;
    \]
\item In the range of $p \in [v_i, \infty]$, parameterized by $t \in [0, q_i]$, value CDF $F_i$ is given by
    \[
    \left\{
    \begin{aligned}
    & F_i = 1 - t \\
    & p = \big.r_i(t)\big/t
    \end{aligned}.
    \right.
    \]
\end{enumerate}
\end{fact}
\end{minipage}}
\vspace{5.5pt}

As a corollary of Fact~\ref{fact1} and Fact~\ref{fact2}, we can prove the following structural lemma.

\begin{mainlemma}[Virtual Value]
\label{lem:virtual_value}
Given a regular distribution $F_i \in \reg$,
\[
r_i(0) + \displaystyle{\int_p^{\infty}} \big(1 - D_i(x)\big) dx = \left(\Phi_i^{-1}(p^+) - p\right) \cdot \big(1 - D_i(p^+)\big),
\]
for all $p \in [0, \infty)$. Particularly, $r_i(0) + \displaystyle{\int_0^{\infty}} \big(1 - D_i(x)\big) dx = r_i(q_i) = v_i q_i$.
\end{mainlemma}

\begin{proof}
Due to integration by substitution,
\begin{flalign*}
p \cdot \big(1 - D_i(p^+)\big) + \displaystyle{\int_p^{\infty}} \big(1 - D_i(x)\big) dx
= & \displaystyle{\int_0^{1 - D_i(p^+)}} D_i^{-1}(1 - y) dy & \text{(Let $y = 1 - D_i(x)$)} \\
= & r_i\big(1 - D_i(p^+)\big) - r_i(0) & \text{(By Fact~\ref{fact2}.1)} \\
= & \big(1 - D_i(p^+)\big) \cdot F_i^{-1}\big(D_i(p^+)\big) - r_i(0) & \text{(By Fact~\ref{fact1}.1)} \\
= & \Phi_ i^{-1}(p^+) \cdot \big(1 - D_i(p^+)\big) - r_i(0).
\end{flalign*}
After being rearranged, this equation is exactly the formula in the lemma. Moreover, for the special case that $p = 0$, we shall notice $D_i(0^+) = 1 - q_i$ and $\Phi_i^{-1}(0^+) = v_i$ (see our definition of inverse function $\Phi_i^{-1}$ in Section~\ref{subsec:prelim-reg-vv}). This completes the proof of the lemma.
\end{proof}

\subsection{Proof of Fact~\ref{fact:opt_rev}}
\label{subapp:prelim:opt_rev}

\paragraph{[Fact~\ref{fact:opt_rev}~{\normalfont(Characterization)}].}
\emph{Given a regular instance $\{F_i\}_{i = 1}^n \in \reg^n$,}
\[
\opt\big(\{F_i\}_{i = 1}^n\big) = \sum\limits_{i = 1}^n r_i(0) + \displaystyle{\int_0^{\infty}} \left(1 - \prod\limits_{i = 1}^n D_i(x)\right) dx.
\]

\begin{proof}
In Myerson Auction, the item is always allocated to the buyer with highest virtual value (required to be no less than $0$). Let $\bm{x}(\cdot)$ be the corresponding (monotone) allocation rule. Clearly, the $k$-th buyer always wins, when his bid goes to $\infty$, while other buyers' bids are finite:
\begin{equation}
\label{eq:lem:opt_rev1}
\lim\limits_{b_k \rightarrow \infty} x_k(b_k, \bm{b}_{-k}) = 1 \quad\quad \forall k \in [n] \text{, } \forall \bm{b}_{-k} \in \mathbb{R}_+^{n - 1}.
\end{equation}
Given $\bm{b} \in \mathbb{R}_+^n$, Lemma~\ref{lem:myerson} (Myerson Lemma) explicitly formulates the payment of the $k$-th buyer:
\begin{equation}
\label{eq:lem:opt_rev2}
\pi_k(\bm{b}) = \displaystyle{\int}_0^{b_k} z \cdot dx_k(z, \bm{b}_{-k}).
\end{equation}
Hence, with fixed $\bm{b}_{-k} \in \mathbb{R}_+^{n - 1}$ and $b_k$ drawn from $F_k$,
\[
\begin{aligned}
\E_{\left.b_k\right|\bm{b}_{-k}}\big[\pi_k(\bm{b})\big]
= & \displaystyle{\int}_0^{\infty} \pi_k(b_k, \bm{b}_{-k}) \cdot dF_k(b_k)
\overset{(\ref{eq:lem:opt_rev2})}{=} \displaystyle{\int}_0^{\infty} \displaystyle{\int}_0^{b_k} z \cdot dx_k(z, \bm{b}_{-k}) \cdot dF_k(b_k) \\
& \text{(By interchange of the order of integration)} \\
= & \displaystyle{\int}_0^{\infty} \left(\displaystyle{\int}_z^{\infty} dF_k(b_k)\right) \cdot z \cdot dx_k(z, \bm{b}_{-k})
= \displaystyle{\int}_0^{\infty} \big(1 - F_k(z)\big) \cdot z \cdot dx_k(z, \bm{b}_{-k}) \\
& \text{(By integration by parts)} \\
= & \left.\Big[\big(1 - F_k(z)\big) \cdot z \cdot x_k(z, \bm{b}_{-k})\Big]\right|_0^{\infty} - \displaystyle{\int}_0^{\infty} \left(1 - F_k(z) - z \cdot \frac{dF_k(z)}{dz}\right) \cdot x_k(z, \bm{b}_{-k}) \cdot dz \\
= & \lim\limits_{z \rightarrow \infty} \Big[\big(1 - F_k(z)\big) \cdot z \cdot x_k(z, \bm{b}_{-k})\Big] + \displaystyle{\int}_0^{\infty} \Phi_k(z) \cdot x_k(z, \bm{b}_{-k}) \cdot dF_k(z) \\
\overset{(\ref{eq:lem:opt_rev1})}{=} & r_k(0) + \E_{\left.b_k\right|\bm{b}_{-k}} \big[\Phi_k(b_k) \cdot x_k(b_k, \bm{b}_{-k})\big].
\end{aligned}
\]
As per this, with $\bm{b}$ drawn from $\{F_i\}_{i = 1}^n$, the expected payment of the $k$-th buyer equals to
\[
\E_{\bm{b}}\big[\pi_k(\bm{b})\big] = r_k(0) + \E_{\bm{b}}\big[\Phi_k(b_k) \cdot x_k(b_k, \bm{b}_{-k})\big] \quad\quad \forall k \in [n].
\]
Sum this equation over all $k \in [n]$, the expected revenue from Myerson Auction equals to
\[
\E_{\bm{b}}\left[\sum\limits_{i = 1}^n \pi_i(\bm{b})\right] = \sum\limits_{i = 1}^n r_i(0) + \E_{\bm{b}}\left[\sum\limits_{i = 1}^n \Phi_i(b_i) \cdot x_i(b_i, \bm{b}_{-i})\right].
\]
Again, the buyer with highest virtual value (required to be no less than $0$) always wins, the above formula can be rewritten as
\[
\E_{\bm{b}}\left[\sum\limits_{i = 1}^n \pi_i(\bm{b})\right]
= \sum\limits_{i = 1}^n r_i(0) + \E_{\bm{b}}\left[\left(\max \big\{\Phi_i(b_i): \,\, i \in [n]\big\}\right)_+\right]
\overset{(\ast)}{=} \sum\limits_{i = 1}^n r_i(0) + \displaystyle{\int_0^{\infty}} \left(1 - \prod\limits_{i = 1}^n D_i(x)\right) dx,
\]
where $\{D_i\}_{i = 1}^n$ denotes the virtual value CDF's, and $(\ast)$ follows from order statistics. This completes the proof of the lemma.
\end{proof}



\section{Mathematical Facts}
\label{app:math_facts}
In this appendix, we verify a few mathematical facts (mainly inequalities) which are omitted in the main text. 

\subsection{Missing Proof of Lemma~\ref{lem:ineq:p/(p+1)}}
\label{subapp:math_facts:p/(p+1)}
\begin{lemma}
\label{lem:ineq:p/(p+1)}
$H(s, p) < 0$ for all $s \in \left[0, \frac{1}{\sqrt{3}}\right]$ and $p \in (1, \infty)$, where
\[
H(s, p) \eqdef \ln\left(1 - \frac{1}{p}\right) + \ln\left(1 + \frac{s}{p}\right) + \big[p - (1 + s)\big] \cdot \left(\frac{1}{p - 1} + \frac{1}{p + s} - \frac{2}{p}\right).
\]
\end{lemma}

\begin{proof}
To conquer the inequality in the lemma, we first investigate the monotonicity of $H(s, p)$, and thus acquire Lemma~\ref{lem:ineq:p/(p+1)-1} (whose proof is deferred to the following).
\vspace{5.5pt} \\
\fbox{\begin{minipage}{\textwidth}
\begin{lemma}
\label{lem:ineq:p/(p+1)-1}
Let $M_p \eqdef \frac{5}{6} + \frac{1}{3\sqrt{3}} + \sqrt{\left(\frac{5}{6} + \frac{1}{3\sqrt{3}}\right)^2 + \frac{2}{3\sqrt{3}}} \approx 2.2246$,
\begin{enumerate}
\item Given $p \in (1, M_p]$, $\frac{\partial H}{\partial s} \leq 0$ for all $s \in \left[0, \frac{1}{\sqrt{3}}\right]$;
\item Given $p \in (M_p, \infty)$, there exists a unique $M_s \in \left(0, \frac{1}{\sqrt{3}}\right)$ such that $\frac{\partial H}{\partial s} \leq 0$ for all $s \in [0, M_s]$, and $\frac{\partial H}{\partial s} > 0$ for all $s \in \left(M_s, \frac{1}{\sqrt{3}}\right]$.
\end{enumerate}
\end{lemma}
\end{minipage}}
\vspace{5.5pt} \\
Lemma~\ref{lem:ineq:p/(p+1)-1} implies, to settle Lemma~\ref{lem:ineq:p/(p+1)} (no matter whether $p > M_p$ or not), that it suffices to show
\[
H_1(p) \eqdef H(0, p) < 0 \quad\quad\quad\quad H_2(p) \eqdef H\left(\frac{1}{\sqrt{3}}, p\right) < 0,
\]
for all $p \in (1, \infty)$. We would justify these two inequalities as follows.

Firstly, $H_1(p) = \ln\left(1 - \frac{1}{p}\right) + \frac{1}{p} \overset{(\diamond)}{=} -\sum\limits_{k = 2}^{\infty} \frac{1}{k} \cdot \frac{1}{p^k} < 0$ for all $p \in (1, \infty)$, where $(\diamond)$ follows from Taylor series. As for $H_2(p) = \ln\left(1 - \frac{1}{p}\right) + \ln\left(1 + \frac{1}{\sqrt{3}p}\right) + \left[p - \left(1 + \frac{1}{\sqrt{3}}\right)\right] \cdot \left(\frac{1}{p - 1} + \frac{\sqrt{3}}{\sqrt{3}p + 1} - \frac{2}{p}\right)$, it can be checked that $H_2'(p) = h_2(p) \cdot p^{-2} \cdot (p - 1)^{-2} \cdot \left(p + \frac{1}{\sqrt{3}}\right)^{-2}$, where
\[
h_2(p) \eqdef \left(2 + \frac{14}{3\sqrt{3}}\right) \cdot p^2 - \left(\frac{2}{3} + \frac{8}{3\sqrt{3}}\right) \cdot p - \left(\frac{2}{3} + \frac{2}{3\sqrt{3}}\right).
\]
As a convex quadratic function, $h_2(p) > 0$ for all $p \in (1, \infty)$, since (1)~$h_2(1) = \frac{2}{3} + \frac{4}{3\sqrt{3}} > 0$; and (2)~$h_2(p)$ has axis of symmetry $\frac{1}{2} \cdot \left.\left(\frac{2}{3} + \frac{8}{3\sqrt{3}}\right)\right/\left(2 + \frac{14}{3\sqrt{3}}\right) = \frac{19 - 5\sqrt{3}}{44} \approx 0.2350 < 1$. Thus, $H_2(p)$ is strictly increasing on $(1, \infty)$. Noting that $\lim\limits_{p \rightarrow \infty} H_2(p) = 0$, we complete the proof of the lemma.
\end{proof}

\paragraph{[Lemma~\ref{lem:ineq:p/(p+1)-1}].}
\emph{Let $M_p \eqdef \frac{5}{6} + \frac{1}{3\sqrt{3}} + \sqrt{\left(\frac{5}{6} + \frac{1}{3\sqrt{3}}\right)^2 + \frac{2}{3\sqrt{3}}} \approx 2.2246$,
\begin{enumerate}
\item Given $p \in (1, M_p]$, $\frac{\partial H}{\partial s} \leq 0$ for all $s \in \left[0, \frac{1}{\sqrt{3}}\right]$;
\item Given $p \in (M_p, \infty)$, there exists a unique $M_s \in \left(0, \frac{1}{\sqrt{3}}\right)$ such that $\frac{\partial H}{\partial s} \leq 0$ for all $s \in [0, M_s]$, and $\frac{\partial H}{\partial s} > 0$ for all $s \in \left(M_s, \frac{1}{\sqrt{3}}\right]$.
\end{enumerate}}

\begin{proof}
Given $p \in (1, \infty)$, it can be checked that $\frac{\partial H}{\partial s} = \frac{h(s, p)}{p \cdot (p - 1) \cdot (p + s)^2}$, where
\[
h(s, p) \eqdef (p - 2) \cdot s^2 + (3p^2 - 5p) \cdot s - p.
\]
We would explore the real roots of this function as follows.
\begin{description}
\item [Case I (when $p = 2$):] In this case, $h(s, 2) = 2s - 2 < 0$ for all $s \in \left(0, \frac{1}{\sqrt{3}}\right)$;
\end{description}
Otherwise, given $p \in (1, 2) \cup (2, \infty)$, $h(s, p)$ is a quadratic function of $s$, with axis of symmetry
\[
\rho_h(p) \eqdef -\frac{3p^2 - 5p}{2 \cdot (p - 2)} = -\frac{3}{2} \cdot (p - 2) - \frac{1}{p - 2} - \frac{7}{2}.
\]
Clearly, $h(s, p)$ has real roots in the range of $s \in \left(0, \frac{1}{\sqrt{3}}\right)$ only if
\[
\Delta_h(p) \eqdef (3p^2 - 5p)^2 + 4 \cdot (p - 2) \cdot p = p \cdot (p - 1) \cdot \left(p - \frac{7 - \sqrt{17}}{6}\right) \cdot \left(p - \frac{7 + \sqrt{17}}{6}\right) \geq 0.
\]
Since $\frac{7 - \sqrt{17}}{6} \approx 0.4795 < 1 < \frac{7 + \sqrt{17}}{6} \approx 1.8539$, we safely assume $p \in \left[\frac{7 + \sqrt{17}}{6}, 2\right) \cup (2, \infty)$ henceforth.
\begin{description}
\item [Case II (when $\frac{7 + \sqrt{17}}{6} \leq p < 2$):] In this case, $h(s, p)$ is a concave function of $s$, with axis of symmetry $\rho_h(p) \overset{(\dagger)}{\geq} \rho_h\left(\frac{7 + \sqrt{17}}{6}\right) = \frac{3 + \sqrt{17}}{2} \approx 3.5616 > \frac{1}{\sqrt{3}} \approx 0.5774$, where $(\dagger)$ follows as $\rho_h(p)$ is increasing on $p \in \left(2 - \sqrt{\frac{2}{3}}, 2\right)$, and $2 - \sqrt{\frac{2}{3}} \approx 1.1835 < \frac{7 + \sqrt{17}}{6} \approx 1.8539$;
\item [Case III (when $2 < p < \infty$):] In this case, $h(s, p)$ is a convex function of $s$, with axis of symmetry $\rho_h(p) = -\frac{3}{2} \cdot (p - 2) - \frac{1}{p - 2} - \frac{7}{2} < -\frac{7}{2} < 0$;
\end{description}
To sum up, given $p \in \left[\frac{7 + \sqrt{17}}{6}, 2\right) \cup (2, \infty)$, we know $h(s, p)$ is strictly increasing on $s \in \left(0, \frac{1}{\sqrt{3}}\right)$. Together with the fact $h(0, p) = -p < 0$, such monotonicity implies that $h(s, p)$ has (a unique) real root in the range of $s \in \left(0, \frac{1}{\sqrt{3}}\right)$, iff $h\left(\frac{1}{\sqrt{3}}, p\right) = \sqrt{3}p^2 - \left(\frac{2}{3} + \frac{5}{\sqrt{3}}\right) \cdot p - \frac{2}{3} > 0$, that is, iff
\[
p > M_p = \frac{5}{6} + \frac{1}{3\sqrt{3}} + \sqrt{\left(\frac{5}{6} + \frac{1}{3\sqrt{3}}\right)^2 + \frac{2}{3\sqrt{3}}} \approx 2.2246.
\]
Gathering everything together, we can summarize that
\begin{itemize}
\item Given $p \in (1, M_p]$, $h(s, p) \leq 0$ for all $s \in \left[0, \frac{1}{\sqrt{3}}\right]$, in that $h(0, p) = -p < 0$ and $h(s, p)$ has no roots in the interior of this interval;
\item Given $p \in (M_p, \infty)$, $h(s, p)$ is strictly increasing when $s \in \left(0, \frac{1}{\sqrt{3}}\right)$, and has a unique root, namely $M_s$, in this range.
\end{itemize}
Since then, Lemma~\ref{lem:ineq:p/(p+1)-1} follows from above immediately.
\end{proof}

\subsection{Missing Proof of Lemma~\ref{lem:RQ}}
\label{subapp:math_facts:RQ}

Lemma~\ref{lem:RQ}.1 and Lemma~\ref{lem:RQ}.2 were first obtained in~\cite{JLTX18}. For the sake of completeness, here we offer a proof of the whole lemma.

\paragraph{[Lemma~\ref{lem:RQ}~(\cite{JLTX18})].}
\emph{For $\R(p) = p\ln\left(\frac{p^2}{p^2 - 1}\right)$ and $\Q(p) = \ln\left(\frac{p^2}{p^2 - 1}\right) - \frac{1}{2}Li_2\left(\frac{1}{p^2}\right)$,
\begin{enumerate}
\item $\R'(p) = p \cdot \Q'(p)$, $\R'(p) < \Q'(p) < 0$ and $\R''(p) > 0$, for all $p \in (1, \infty)$;
\item $\R(\infty) = \Q(\infty) = 0$ and $\R(1^+) = \Q(1^+) = \infty$;
\item $\frac{1}{p} \leq \R(p) \leq \frac{1}{p - 1}$ and $\frac{1}{p^2} \leq \left|\R'(p)\right| \leq \frac{1}{p^2 - p}$, for all $p \in (1, \infty)$.
\end{enumerate}}

\begin{proof}[Proof of Lemma~\ref{lem:RQ}.1]
Given $p \in (1, \infty)$, since $\ln(1 + x) \leq x$ for all $x \geq 0$, we have
\begin{equation}
\label{eq:math_fact1}
\R'(p) = -\frac{2}{p^2 - 1} + \ln\left(1 + \frac{1}{p^2 - 1}\right) \overset{(\ast)}{\leq} -\frac{1}{p^2 - 1} < 0,
\end{equation}
Besides, $\Q'(p) = -\frac{2}{p(p^2 - 1)} + \sum\limits_{k = 1}^{\infty} \frac{1}{k} \cdot \frac{1}{p^{2k + 1}} = -\frac{2}{p(p^2 - 1)} - \frac{1}{p} \cdot \ln\left(1 - \frac{1}{p^2}\right) \overset{(\ref{eq:math_fact1})}{=} \frac{\R'(p)}{p}$, for all $p \in (1, \infty)$. Combining the above two facts together directly implies $\R'(p) < \Q'(p) < 0$ for all $p \in (1, \infty)$. Finally, observing that $\R''(p) \overset{(\ref{eq:math_fact1})}{=} \frac{2(p^2 + 1)}{p(p^2 - 1)^2} > 0$, we complete the proof of Lemma~\ref{lem:RQ}.1.
\end{proof}

\begin{proof}[Proof of Lemma~\ref{lem:RQ}.2]
For the first limit, since $\ln(1 + x) \leq x$ for all $x \geq 0$,
\[
\begin{aligned}
& 0 \leq \R(\infty) = \lim\limits_{p \rightarrow \infty} p \ln\left(1 + \frac{1}{p^2 - 1}\right) \leq \lim\limits_{p \rightarrow \infty} \frac{p}{p^2 - 1} = 0, \\
& 0 \leq \Q(\infty) = \lim\limits_{p \rightarrow \infty} \displaystyle{\int_p^{\infty}} \left(\frac{-\R'(x)}{x}\right) dx \leq \lim\limits_{p \rightarrow \infty} \displaystyle{\int_p^{\infty}} \left(-\R'(x)\right) dx = \lim\limits_{p \rightarrow \infty} \R(p) = 0.
\end{aligned}
\]
For the second limit,
\[
\begin{aligned}
& \R(1^+) \geq \ln\left(\lim\limits_{p \rightarrow 1^+}\frac{p^2}{p^2 - 1}\right) = \ln(\infty) = \infty, \\
& \Q(1^+) = \lim\limits_{p \rightarrow 1^+} \displaystyle{\int_p^{\infty}} \frac{-\R'(x)}{x} dx \geq \lim\limits_{p \rightarrow 1^+} \displaystyle{\int_p^2} \frac{-\R'(x)}{2} dx = \lim\limits_{p \rightarrow 1^+} \frac{\R(p) - \R(2)}{2} = \infty.
\end{aligned}
\]
This completes the proof of Lemma~\ref{lem:RQ}.2.
\end{proof}

\begin{proof}[Proof of Lemma~\ref{lem:RQ}.3]
Since $\ln(1 - x) \leq -x$ for all $x \in [0, 1)$,
\[
\R(p) = -p\ln\left(1 - \frac{1}{p^2}\right) \geq p \cdot \frac{1}{p^2} = \frac{1}{p} \quad\quad \forall p \in (1, \infty).
\]
Since $\ln(1 + x) \leq x$ for all $x \geq 0$,
\[
\R(p) = p\ln\left(1 + \frac{1}{p^2 - 1}\right) \leq \frac{p}{p^2 - 1} \leq \frac{p + 1}{p^2 - 1} = \frac{1}{p - 1} \quad\quad \forall p \in (1, \infty).
\]
Furthermore, $\left|\R'(p)\right| \overset{(\ref{eq:math_fact1})}{\geq} \frac{1}{p^2 - 1} > \frac{1}{p^2}$ and
\[
\left|\R'(p)\right| \overset{(\ref{eq:math_fact1})}{=} \frac{2}{p^2 - 1} + \ln\left(1 - \frac{1}{p^2}\right) \overset{(\star)}{\leq} \frac{2}{p^2 - 1} - \frac{1}{p^2} = \frac{1}{p^2 - p} - \frac{1}{p^3 + p^2} < \frac{1}{p^2 - p},
\]
where $(\star)$ follows as $\ln(1 - x) \leq -x$ for all $x \in [0, 1)$. This completes the proof of Lemma~\ref{lem:RQ}.3.
\end{proof}

\subsection{Useful Inequalities for Lemma~\ref{lem:main} and Lemma~\ref{lem:abel_ineq1}}
\label{subapp:math_facts:ineq}

\begin{lemma}
\label{lem:ineq}
For all $y \geq x > 1$,
\begin{enumerate}
\item $1 - e^{-\big(\R(x) - \R(y)\big)} \leq \left(1 - \frac{x}{y}\right) \cdot \frac{x + e^{\R(x) - \R(y)} - 1}{(x - 1)^2}$;
\item $\frac{1}{y} \cdot \left[1 - e^{-\big(\R(x) - \R(y)\big)}\right] \leq 1 - e^{-\big(\Q(x) - \Q(y)\big)} \leq \frac{1}{x} \cdot \left[1 - e^{-\big(\R(x) - \R(y)\big)}\right]$;
\item $\frac{e^{\R(x) - \R(y)}}{x + e^{\R(x) - \R(y)} - 1} \geq \frac{1}{y}$;
\item $1 - \frac{1}{y} \geq (x - 1) \cdot \frac{e^{\R(x) - \R(y)}}{x + e^{\R(x) - \R(y)} - 1}$.
\end{enumerate}
\end{lemma}

\begin{proof}[Proof of Lemma~\ref{lem:ineq}.1]
Since $\R(x) \geq \R(y)$ when $y \geq x > 1$ (see Lemma~\ref{lem:RQ}.1), and $1 - e^{-z} \leq z$ for all $z \geq 0$, we can relax the left hand side (of the inequality in the lemma) as follows:
\[
1 - e^{-\big(\R(x) - \R(y)\big)} \leq \R(x) - \R(y).
\]
Furthermore, the right hand side (of the inequality in the lemma) can be relaxed as follows:
\[
\left(1 - \frac{x}{y}\right) \cdot \frac{x + e^{\R(x) - \R(y)} - 1}{(x - 1)^2} \geq \left(1 - \frac{x}{y}\right) \cdot \frac{x}{(x - 1)^2}.
\]
As per these, we only need to show $G_1(x, y) \leq 0$ when $y \geq x > 1$, where
\[
G_1(x, y) \eqdef \R(x) - \R(y) - \left(1 - \frac{x}{y}\right) \cdot \frac{x}{(x - 1)^2}.
\]
For this, we shall notice $\frac{\partial G_1}{\partial y} = -\R'(y) - \frac{1}{y^2} \cdot \left(1 - \frac{1}{x}\right)^{-2} \leq g_1(y)$, where
\begin{equation}
\label{eq:math_fact2}
g_1(y) \eqdef -\R'(y) - \frac{1}{y^2} \cdot \left(1 - \frac{1}{y}\right)^{-2} = \frac{2}{y^2 - 1} - \ln\left(\frac{y^2}{y^2 - 1}\right) - \frac{1}{(y - 1)^2}.
\end{equation}
Since $g_1(\infty) = 0$ and $g_1'(y) = \frac{6y^2 + 2}{y \cdot (y - 1) \cdot (y^2 - 1)^2} \geq 0$, when $y \geq x > 1$, we surely have
\begin{equation}
\label{eq:math_fact3}
\frac{\partial G_1}{\partial y} \leq g_1(y) \leq 0.
\end{equation}
Together with the fact $G_1(x, y) = 0$ when $y = x > 1$, we complete the proof of Lemma~\ref{lem:ineq}.1.
\end{proof}

\begin{proof}[Proof of Lemma~\ref{lem:ineq}.2]
Define $G_2(x, y) \eqdef \frac{1}{y} \cdot \left[1 - e^{-\big(\R(x) - \R(y)\big)}\right] - 1 + e^{-\big(\Q(x) - \Q(y)\big)}$. Noticing that $G_2(x, y) = 0$ when $y = x > 1$, we safely turn to certify
\begin{flalign*}
\frac{\partial G_2}{\partial x}
= & -\frac{1}{y} \cdot e^{-\big(\R(x) - \R(y)\big)} \cdot \big(-\R'(x)\big) + e^{-\big(\Q(x) - \Q(y)\big)} \cdot \big(-\Q'(x)\big) \\
= & \big(-\R'(x)\big) \cdot e^{-\R(x) + \Q(y)} \cdot \left(\frac{1}{x} \cdot e^{\R(x) - \Q(x)} - \frac{1}{y} \cdot e^{\R(y) - \Q(y)}\right) \geq 0, & \text{(Since $\Q'(x) = \frac{1}{x} \cdot \R'(x)$)}
\end{flalign*}
for all $y \geq x > 1$. Due to Lemma~\ref{lem:RQ}.1 that $\R'(x) \leq 0$, we only need to check $g_2(x) \geq g_2(y)$, where $g_2(x) \eqdef \R(x) - \Q(x) - \ln x$. This follows from the monotonicity of $g_2$:
\[
g_1'(x) = -(x - 1) \cdot \big(-\Q'(x)\big) - \frac{1}{x} \leq -\frac{1}{x} \leq 0 \quad\quad \forall x \in (1, \infty).
\]

We continue to attest the second inequality in Lemma~\ref{lem:ineq}.2, that is,
\[
G_3(x, y) \eqdef \frac{1}{x} \cdot \left[1 - e^{-\big(\R(x) - \R(y)\big)}\right] - 1 + e^{-\big(\Q(x) - \Q(y)\big)} \geq 0,
\]
for all $y \geq x > 1$. Since $G_3(x, y) = 0$ when $y = x > 1$, we only need to verify
\begin{flalign*}
\frac{\partial G_3}{\partial y}
= & \frac{1}{x} \cdot e^{-\left(\R(x) - \R(y)\right)} \cdot \left(-\R'(y)\right) - e^{-\left(\Q(x) - \Q(y)\right)} \cdot \left(-\Q'(y)\right) \\
= & \frac{1}{x} \cdot \left(-\Q'(y)\right) \cdot e^{-\R(x) + \Q(y)} \cdot \left(y \cdot e^{\R(y) - \Q(y)} - x \cdot e^{\R(x) - \Q(x)}\right) \geq 0. & \text{(Since $\R'(y) = y \cdot \Q'(y)$)}
\end{flalign*}
Since $\Q'(y) \leq 0$ for all $y > 1$ (see Lemma~\ref{lem:RQ}.1), it remains to reveal $g_3(y) \geq g_3(x)$ for all $y \geq x > 1$, where $g_3(x) \eqdef \ln x + \R(x) - \Q(x)$. The monotonicity of $g_3(x)$ is proved as follows:
\[
g_3'(x) \overset{(\ref{eq:math_fact1})}{=} \frac{1}{x} + \left(1 - \frac{1}{x}\right) \cdot \left[-\frac{2}{x^2 - 1} + \ln\left(\frac{x^2}{x^2 - 1}\right)\right] \geq \frac{1}{x} - \left(1 - \frac{1}{x}\right) \cdot \frac{2}{x^2 - 1} = \frac{x - 1}{x(x + 1)} \geq 0.
\]
This completes the proof of the Lemma~\ref{lem:ineq}.2.
\end{proof}

\begin{proof}[Proof of Lemma~\ref{lem:ineq}.3]
After being rearranged, the inequality in Lemma~\ref{lem:ineq}.3 becomes
\[
e^{\R(x) - \R(y)} \geq \frac{x - 1}{y - 1}.
\]
Since $y \geq x > 1$, the lemma directly follows from Lemma~\ref{lem:RQ}.1 as $e^{\R(x) - \R(y)} \geq 1 \geq \frac{x - 1}{y - 1}$.
\end{proof}

\begin{proof}[Proof of Lemma~\ref{lem:ineq}.4]
Since $0 \leq \R(x) \leq \frac{1}{x - 1}$ (see Lemma~\ref{lem:RQ}.3), and $1 - e^{-z} \leq z$ for all $z \geq 0$, we know $1 - e^{-\R(x)} \leq \frac{1}{x - 1}$ for all $x \in (1, \infty)$. After being rearranged, this inequality becomes
\[
1 \geq (x - 1) \cdot \frac{e^{\R(x)}}{x + e^{\R(x)} - 1} \overset{(\ddagger)}{\geq} (x - 1) \cdot \frac{e^{\R(x) - \R(y)}}{x + e^{\R(x) - \R(y)} - 1},
\]
for all $y \geq x > 1$, where $(\ddagger)$ follows as $g_4(z) \eqdef \frac{z}{x + z - 1}$ is increasing on $z \in (0, \infty)$. As per these, we can safely rearrange the inequality in Lemma~\ref{lem:ineq}.4 as follows:
\[
y - 1 \geq \frac{1}{1 - (x - 1) \cdot \frac{e^{\R(x) - \R(y)}}{x + e^{\R(x) - \R(y)} - 1}} - 1 = \frac{(x - 1) \cdot e^{\R(x) - \R(y)}}{x + e^{\R(x) - \R(y)} - 1 - (x - 1) \cdot e^{\R(x) - \R(y)}},
\]
Further rearranging results in $G_4(x, y) \eqdef \left(\frac{1}{x - 1} - \frac{1}{y - 1}\right) - \left(1 - e^{-(\R(x) - \R(y))}\right) \geq 0$, for all $y \geq x > 1$. This inequality can be confirmed by observing that $G_4(x, y) = 0$ when $y = x > 1$, and that
\[
\frac{\partial G_4}{\partial y}
= \frac{1}{(y - 1)^2} + e^{-\big(\R(x) - \R(y)\big)} \cdot \R'(y)
\overset{(\ddagger)}{\geq} \frac{1}{(y - 1)^2} + \R'(y)
\overset{(\ref{eq:math_fact2})}{=} -g_1(y) \overset{(\ref{eq:math_fact3})}{\geq} 0,
\]
where $(\ddagger)$ follows as $\R'(p) \leq 0$ for all $p \in (1, \infty)$. This completes the proof of Lemma~\ref{lem:ineq}.4.
\end{proof}

\section{Missing Proofs in Section~\ref{sec:upper1}}
\label{app:upper1}
\subsection{Proof of Lemma~\ref{lem:p/(p+1)-3}}
\label{subapp:upper1-p/(p+1)-3}

\fcolorbox{white}{lightgray}{\begin{minipage}{\textwidth}
In a worst-case instance $\{F_i\}_{i = 0}^n$, w.l.o.g.
\begin{itemize}
\item With some pre-fixed $r_0(0) \geq 0$, $r_0(q) = r_0(0) \cdot (1 - q)$ for all $q \in [0, 1]$;
\item For each $i \in [n]$, $q_i > 0$, $\partial_+ r_i(q) > 1$ for all $q \in [0, q_i)$, $v_i = \frac{r_i(q_i)}{q_i} > 1$.
\end{itemize}
\end{minipage}}

\paragraph{[Lemma~\ref{lem:p/(p+1)-3}].}
\emph{There exists a worst-case instance $\{F_i\}_{i = 0}^n$ of Program~(\ref{prog:0}), such that for all $i \in [n]$ with $r_i(0) > 0$,}
\[
\partial_+ r_i(0) > 1 + r_0(0).
\]

\begin{proof}
Assume to the contrary: There exists $k \in [n]$ such that $r_k(0) > 0$ and $\partial_+ r_k(0) \leq 1 + r_0(0)$. In terms of revenue-quantile curve, let us replace $F_0$ by $\FF_0$ (resp. replace $F_k$ by $\FF_k$).
\begin{itemize}
\item Define $\rr_0(q) \eqdef \big(r_0(0) + r_k(0)\big) \cdot (1 - q)$ for all $q \in [0, 1]$;
\item Define $\qq_k \eqdef q_k$ and $\rr_k(q) \eqdef
    \begin{cases}
    r_k(q) - r_k(0) & \forall q \in [0, \qq_k] \\
    \frac{r_k(q_k) - r_k(0)}{1 - q_k} \cdot (1 - q) & \forall q \in (\qq_k, 1]
    \end{cases}$.
\end{itemize}
We have shown in Section~\ref{subsec:upper1-lem:p/(p+1)-1} that $\opt\left(\{\FF_0\} \bigcup \{\FF_k\} \bigcup \{F_i\}_{i \in [n] \setminus \{k\}}\right) = \opt\big(\{F_i\}_{i = 0}^n\big)$. In the rest proof, we focus on the feasibility of instance $\{\FF_0\} \bigcup \{\FF_k\} \bigcup \{F_i\}_{i \in [n] \setminus \{k\}}$.


\paragraph{Feasibility Analysis.}
Under our construction of $\FF_k$, its feasibility to constraint~(\ref{cstr:value}) trivially follows from Lemma~\ref{lem:shape}.1: Since $\partial_+ \rr_k(q) = \partial_+ r_k(q) > 1$ for all $q \in [0, \qq_k) = [0, q_k)$,
\[
\vv_k = \frac{\vv_k(\qq_k)}{\qq_k} > 1 \quad\quad\quad\quad\quad\quad \PPhi_k(\vv_k^+) = \lim\limits_{q \rightarrow \qq_k^-} \partial_+ \rr_k(q) = \lim\limits_{q \rightarrow q_k^-} \partial_+ r_k(q) = \Phi_k(v_k^+) \geq 1.
\]
It remains to preserve constraint~(\ref{cstr:ap0}) that $\ap\big(p, \{\FF_0\} \bigcup \{\FF_k\} \bigcup \{F_i\}_{i \in [n] \setminus \{k\}}\big) \leq \ap\big(p, \{F_i\}_{i = 0}^n\big)$, or equivalently,
\[
\FF_0(p) \cdot \FF_k(p) \geq F_0(p) \cdot F_k(p),
\]
for all $p \in (1, \infty)$. Depending on whether $p > \vv_k \eqdef \frac{r_k(q_k) - r_k(0)}{q_k}$ or not, our claim is clarified in the following two cases. Noticeably, premise $\partial_+ r_k(0) \leq 1 + r_0(0)$ is only used in \textbf{Case II}.

\begin{figure}[H]
\centering
\subfigure[Original $r_k(q)$]{
\begin{tikzpicture}[thick, smooth, domain = 0: 0.5, scale = 2.55]
\draw[->] (0, 0) -- (1.1, 0);
\draw[->] (0, 0) -- (0, 1.1);
\node[above] at (0, 1.1) {\scriptsize $r_k(q)$};
\node[right] at (1.1, 0) {\scriptsize $q$};
\node[left] at (0, 0) {\scriptsize $0$};
\draw[color = blue] plot (\x, {0.2 + 1.6 * sqrt(\x / 2)});
\draw[thin, dashed] (0, 1) -- (0.5, 1);
\draw[very thick] (0, 1) -- (1.5pt, 1);
\node[anchor = 0] at (0, 1) {\scriptsize $r_k(q_k)$};

\draw[very thick] (0, 0.2) -- (1.5pt, 0.2);
\node[left] at (0, 0.2) {\scriptsize $r_k(0)$};

\draw[thin, dashed] (0.2222, 0) -- (0.2222, 0.7333);
\draw[very thick] (0.2222, 0) -- (0.2222, 1.5pt);
\draw[thin, dashed] (0, 0.7333) -- (0.2222, 0.7333);
\draw[very thick] (0, 0.7333) -- (1.5pt, 0.7333);
\node[below] at (0.2222, 0) {\scriptsize $y$};
\node[left] at (0, 0.7333) {\scriptsize $r_k(y)$};

\draw[color = blue] (0.5, 1) -- (1, 0);
\draw[color = Sepia] (0, 0) -- (0.5, 1.2);
\draw[color = Sepia] (0.2222, 0.7333) -- (1, 0);
\node[anchor = 60, color = Sepia] at (0.7879, 0.2) {\scriptsize $\mathcal{L}(q)$};
\draw[very thick] (1, 0) -- (1, 1.5pt);
\node[below] at (1, 0) {\scriptsize $1$};
\node[anchor = west, color = Sepia] at (0.5, 1.2) {\scriptsize $r(q) = p \cdot q$};
\draw[thin, color = blue, fill = green] (0.2222, 0.7333) circle(0.5pt);
\draw[thin, color = blue, fill = green] (0.3701, 0.8883) circle(0.5pt);
\draw[thin, color = blue, fill = green] (0.2821, 0.6769) circle(0.5pt);
\draw[thin, color = blue, fill = red] (1, 0) circle(0.5pt);
\draw[thin, color = blue, fill = red] (0.5, 1) circle(0.5pt);
\draw[thin, color = blue, fill = red] (0, 0.2) circle(0.5pt);
\end{tikzpicture}
\label{fig:lem:p/(p+1)-3.1}
}
\quad\quad\quad\quad\quad\quad
\subfigure[$\rr_k(q)$ after reduction]{
\begin{tikzpicture}[thick, smooth, domain = 0: 0.5, scale = 2.55]
\draw[->] (0, 0) -- (1.1, 0);
\draw[->] (0, 0) -- (0, 1.1);
\node[above] at (0, 1.1) {\scriptsize $\rr_k(q)$};
\node[right] at (1.1, 0) {\scriptsize $q$};
\node[left] at (0, 0) {\scriptsize $\rr_k(0) = 0$};
\draw[color = blue] plot (\x, {1.6 * sqrt(\x / 2)});
\draw[thin, dashed] (0, 0.8) -- (0.5, 0.8);
\draw[very thick] (0, 0.8) -- (1.5pt, 0.8);
\node[anchor = 0] at (0, 0.8) {\scriptsize $\rr_k(\qq_k)$};

\draw[thin, dashed] (0.2222, 0) -- (0.2222, 0.5333);
\draw[very thick] (0.2222, 0) -- (0.2222, 1.5pt);
\draw[thin, dashed] (0, 0.5333) -- (0.2222, 0.5333);
\draw[very thick] (0, 0.5333) -- (1.5pt, 0.5333);
\node[below] at (0.2222, 0) {\scriptsize $y$};
\node[left] at (0, 0.5333) {\scriptsize $\rr_k(y)$};

\draw[color = blue] (0.5, 0.8) -- (1, 0);
\draw[color = Sepia] (0, 0) -- (0.45, 1.08);
\draw[very thick] (1, 0) -- (1, 1.5pt);
\node[below] at (1, 0) {\scriptsize $1$};
\node[anchor = north west, color = Sepia] at (0.45, 1.08) {\scriptsize $r(q) = p \cdot q$};
\draw[thin, color = blue, fill = green] (0.2222, 0.5333) circle(0.5pt);
\draw[thin, color = blue, fill = red] (1, 0) circle(0.5pt);
\draw[thin, color = blue, fill = red] (0.5, 0.8) circle(0.5pt);
\draw[thin, color = blue, fill = red] (0, 0) circle(0.5pt);
\end{tikzpicture}
\label{fig:lem:p/(p+1)-3.2}
}
\caption{Demonstrations for the reduction in Lemma~\ref{lem:p/(p+1)-3}}
\label{fig:lem:p/(p+1)-3}
\end{figure}
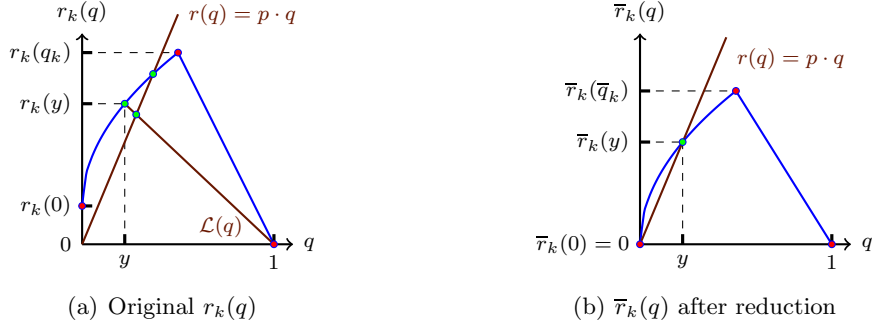

\paragraph{Case I (when $\vv_k < p < \infty$).}
In this case, there exists some $y \in [0, \qq_k) = [0, q_k)$ such that
\begin{equation}
\label{eq:lem:p/(p+1)-3.1}
p \cdot y = \rr_k(y) = r_k(y) - r_k(0),
\end{equation}
which is shown in Figure~\ref{fig:lem:p/(p+1)-3.2}. In Figure~\ref{fig:lem:p/(p+1)-3.1}, line segment $\mathcal{L}(q)$ has endpoints $\big(y, r_k(y)\big)$ and $(1, 0)$. Clearly, $\mathcal{L}(q) \eqdef \frac{r_k(y)}{1 - y} \cdot (1 - q)$ for all $q \in [y, 1]$, corresponding to CDF
\[
F_{\mathcal{L}}(p') \eqdef \frac{(1 - y) \cdot p'}{(1 - y) \cdot p' + r_k(y)} \quad\quad \forall p' \in \left[0, \frac{r_k(y)}{y}\right].
\]
We know\footnote{Particularly, when $y = 0$, premise $r_k(0) > 0$ imposes $\frac{r_k(y)}{y} = \infty$, and thus $p \leq \frac{r_k(y)}{y}$ still holds.} $p \in \left[0, \frac{r_k(y)}{y}\right]$ and $r_k(q) \geq \mathcal{L}(q)$ for all $q \in [y, 1]$, since $r_k(q)$ is concave on $[y, 1] \subset [0, 1]$. As a consequence, $F_k(p') \leq F_{\mathcal{L}}(p')$ for all $p' \in \left[0, \frac{r_k(y)}{y}\right]$, and specifically,
\begin{equation}
\label{eq:lem:p/(p+1)-3.2}
F_k(p) \leq F_{\mathcal{L}}(p) = \frac{(1 - y) \cdot p}{(1 - y) \cdot p + r_k(y)}.
\end{equation}
Since $\FF_0(p) = \frac{p}{p + \rr_0(0)} = \frac{p}{p + r_0(0) + r_k(0)}$, $\FF_k(p) = 1 - y$ and $F_0(p) = \frac{p}{p + r_0(0)}$, our claim follows as
\[
\begin{aligned}
\left(\FF_0(p) \cdot \FF_k(p)\right)^{-1}
= & \left(1 + \frac{r_0(0) + r_k(0)}{p}\right) \cdot \frac{1}{1 - y} \\
\leq & \left(1 + \frac{r_0(0)}{p}\right) \cdot \left(1 + \frac{r_k(0)}{p}\right) \cdot \frac{1}{1 - y} \\
\overset{(\ref{eq:lem:p/(p+1)-3.1})}{=} & \left(1 + \frac{r_0(0)}{p}\right) \cdot \left(1 + \frac{r_k(y)}{1 - y} \cdot \frac{1}{p}\right)
\overset{(\ref{eq:lem:p/(p+1)-3.2})}{\leq} \big(F_0(p) \cdot F_k(p)\big)^{-1}.
\end{aligned}
\]

\paragraph{Case II (when $1 < p \leq \vv_k$).}
In this case, observe that $p \leq \vv_k = \frac{r_k(q_k) - r_k(0)}{q_k} \leq v_k = \frac{r_k(q_k)}{q_k} < \infty$, all involved CDF's admit explicit formulas:
\begin{equation}
\label{eq:lem:p/(p+1)-3.3}
\FF_0(p) = \frac{p}{p + \rr_0(0)} \quad\quad \FF_k(p) = \frac{p}{p + \frac{\rr_k(\qq_k)}{1 - \qq_k}} \quad\quad F_0(p) = \frac{p}{p + r_0(0)} \quad\quad F_k(p) = \frac{p}{p + \frac{r_k(q_k)}{1 - q_k}}.
\end{equation}
Due to premise $\partial_+ r_k(0) \leq 1 + r_0(0)$, and since $r_k(q)$ is concave,
\begin{equation}
\label{eq:lem:p/(p+1)-3.4}
\frac{r_k(q_k) - r_k(0)}{q_k} \leq \partial_+ r_k(0) \leq 1 + r_0(0).
\end{equation}
Recall that $\rr_0(0) = r_0(0) + r_k(0)$, $\rr_k(\qq_k) = r_k(q_k) - r_k(0)$ and $\qq_k = q_k$. Put everything together,
\[
\begin{aligned}
& \left(\FF_0(p) \cdot \FF_k(p)\right)^{-1} - \big(F_0(p) \cdot F_k(p)\big)^{-1} \\
\overset{(\ref{eq:lem:p/(p+1)-3.3})}{=} & \left(1 + \frac{r_0(0) + r_k(0)}{p}\right) \cdot \left(1 + \frac{r_k(q_k) - r_k(0)}{1 - q_k} \cdot \frac{1}{p}\right) - \left(1 + \frac{r_0(0)}{p}\right) \cdot \left(1 + \frac{r_k(q_k)}{1 - q_k} \cdot \frac{1}{p}\right) \\
= & -\left(p + \frac{r_0(0)}{q_k} - \frac{r_k(q_k) - r_k(0)}{q_k}\right) \cdot \frac{r_k(0)}{1 - q_k} \cdot \frac{q_k}{p^2} \\
\overset{(\ref{eq:lem:p/(p+1)-3.4})}{\leq} & -\left(p + \frac{r_0(0)}{q_k} - 1 - r_0(0)\right) \cdot \frac{r_k(0)}{1 - q_k} \cdot \frac{q_k}{p^2} \leq 0,
\end{aligned}
\]
where the last inequality follows as $p > 1$ and $q_k \leq 1$. This completes the proof of the lemma.
\end{proof}

\subsection{Proofs of Lemma~\ref{lem:ap:increase} and Lemma~\ref{lem:ap_infinity}}
\label{subapp:upper1-p/(p+1)-1}

\fcolorbox{white}{lightgray}{\begin{minipage}{\textwidth}
In a worst-case instance $\{F_i\}_{i = 0}^n$, w.l.o.g.
\begin{itemize}
\item With some pre-fixed $r_0(0) \geq 0$, $r_0(q) = r_0(0) \cdot (1 - q)$ for all $q \in [0, 1]$;
\item For each $i \in [n]$, $q_i > 0$, $\partial_+ r_i(q) > 1$ for all $q \in [0, q_i)$, $v_i = \frac{r_i(q_i)}{q_i} > 1$.
\end{itemize}
\end{minipage}}

\paragraph{[Lemma~\ref{lem:ap:increase}].}
\emph{$\ap\big(p, \{F_i\}_{i = 0}^n\big)$ is increasing, when $1 < p < \min \big\{v_i: \,\, i \in [n]\big\}$.}

\begin{proof}
When $1 < p < \min \big\{v_i: \,\, i \in [n]\big\}$, $\ap\big(p, \{F_i\}_{i = 0}^n\big) = p \cdot \left\{1 - \prod\limits_{i = 0}^n \left[\frac{(1 - q_i) \cdot p}{(1 - q_i) \cdot p + r_i(q_i)}\right]\right\}$. Therefore, in this range, we can settle the lemma by observing that
\[
\begin{aligned}
\ap'\big(p, \{F_i\}_{i = 0}^n\big)
= & 1 - \prod\limits_{i = 0}^n \left[\frac{(1 - q_i) \cdot p}{(1 - q_i) \cdot p + r_i(q_i)}\right] \cdot \left[1 + \sum\limits_{i = 0}^n \frac{r_i(q_i)}{(1 - q_i) \cdot p + r_i(q_i)}\right] \\
\geq & 1 - \prod\limits_{i = 0}^n \left[\frac{(1 - q_i) \cdot p}{(1 - q_i) \cdot p + r_i(q_i)}\right] \cdot \left[1 + \sum\limits_{i = 0}^n \frac{r_i(q_i)}{(1 - q_i) \cdot p}\right] \\
\overset{(\uplus)}{\geq} & 1 - \prod\limits_{i = 0}^n \left[\frac{(1 - q_i) \cdot p}{(1 - q_i) \cdot p + r_i(q_i)}\right] \cdot \prod\limits_{i = 0}^n \left[1 + \frac{r_i(q_i)}{(1 - q_i) \cdot p}\right] = 0, \\
\end{aligned}
\]
where $(\uplus)$ follows as $1 + \sum\limits_{i = 0}^n z_i \leq \prod\limits_{i = 0}^n (1 + z_i)$ when $z_i \geq 0$ for each $i \in \{0\} \bigcup [n]$.
\end{proof}

\paragraph{[Lemma~\ref{lem:ap_infinity}].}
\emph{$\ap\big(\infty, \{F_i\}_{i = 0}^n\big) = \sum\limits_{i = 0}^n r_i(0)$.}

\begin{proof}
Recall Section~\ref{subsec:prelim-reg-vv}, distribution $F_i$ is regular iff its revenue-quantile curve $r_i(q)$ is continuous and concave on $q \in [0, 1]$. Such continuity and concavity promise that $\lim\limits_{q \rightarrow 0^+} r_i(q) = r_0(0) \in [0, \infty)$. Recall the definition of revenue-quantile curve in Section~\ref{subsec:prelim-rq-curve}, for each $i \in \{0\} \cup [n]$,
\[
\lim\limits_{p \rightarrow \infty} p \cdot \big(1 - F_i(p)\big)
= \lim\limits_{p \rightarrow \infty} r_i\big(1 - F_i(p)\big)
= \lim\limits_{q \rightarrow 0^+} r_i(q) = r_i(0),
\]
which means $F_i(p) = 1 - \frac{r_i(0)}{p} + o\left(\frac{1}{p}\right)$ for sufficiently large $p$. Therefore, with pre-fixed $n \in \mathbb{N}$ and sufficiently large $p$, we must have
\[
\ap\big(p, \{F_i\}_{i = 0}^n\big) = p \cdot \left[1 - \prod\limits_{i = 0}^n \left(1 - \frac{r_i(0)}{p} + o\left(\frac{1}{p}\right)\right)\right] = \sum\limits_{i = 0}^n r_i(0) + o\left(\frac{1}{p}\right).
\]
Letting $p$ approach to $\infty$, we complete the proof of the lemma.
\end{proof}

\subsection{Proof of Lemma~\ref{lem:p/(p+1)-4}}
\label{subapp:upper1-p/(p+1)-4}

\fcolorbox{white}{lightgray}{\begin{minipage}{\textwidth}
In a worst-case instance $\{F_i\}_{i = 0}^n$, w.l.o.g.
\begin{itemize}
\item With some pre-fixed $r_0(0) \geq 0$, $r_0(q) = r_0(0) \cdot (1 - q)$ for all $q \in [0, 1]$;
\item For each $i \in [n]$, $q_i > 0$, $\partial_+ r_i(q) > 1$ for all $q \in [0, q_i)$, $v_i = \frac{r_i(q_i)}{q_i} > 1$.
\end{itemize}
\end{minipage}}

\paragraph{[Lemma~\ref{lem:p/(p+1)-4}].}
\emph{Given a feasible instance $\{F_i\}_{i = 0}^n$ of Program~(\ref{prog:0}) that $r_0(0) \geq \frac{1}{2}$, there exists another feasible instance $\{\FF_i\}_{i = 0}^n$ such that $\opt\big(\{\FF_i\}_{i = 0}^n\big) \geq \opt\big(\{F_i\}_{i = 0}^n\big)$, and}
\[
\rr_0(q) = 1 - q \quad\quad \forall q \in [0, 1] \quad\quad\quad\quad\quad\quad\quad\quad \rr_i(0) = 0 \quad\quad \forall i \in [n].
\]

\begin{proof}
Given a feasible instance $\{F_i\}_{i = 0}^n$ that $r_0(0) \geq \frac{1}{2}$, a desired instance $\{\FF_i\}_{i = 0}^n$ is constructed in Algorithm~\ref{alg0.3}. Notably, premise $r_0(0) \geq \frac{1}{2}$ is only used in \textbf{Case II.1} of \textbf{Analysis III}.

\paragraph{Analysis I: Reduction.}
It was explicitly shown in~\cite{AHNPY15} that $\sum\limits_{i = 0}^n r_i(q_i) \geq \opt\big(\{F_i\}_{i = 0}^n\big)$. Assuming w.l.o.g. $\opt\big(\{F_i\}_{i = 0}^n\big) > 1$, we have $r_0(0) + \sum\limits_{i = 1}^n r_i(q_i) > 1$. Due to premise $v_i > 1$ for each $i \in [n]$, definitely $\min\{v_i: \,\, i \in [n]\} > 1$. Since then,
\begin{itemize}
\item $\mathsf{ExAnte}(p) \eqdef r_0(0) + \sum\limits_{i \in [n]: v_i \geq p} r_i(q_i) + \sum\limits_{i \in [n]: v_i < p} r_i\big(1 - F_i(p)\big)$ is decreasing on $(0, \infty)$;
\item $\mathsf{ExAnte}\big(\min\{v_i: \,\, i \in [n]\}\big) = r_0(0) + \sum\limits_{i = 1}^n r_i(q_i) > 1$, ;
\item $\mathsf{ExAnte}(\infty) = \sum\limits_{i = 0}^n r_i(0) \overset{(\ast)}{=} \ap\big(\infty, \{F_i\}_{i = 0}^n\big) \leq 1$, where $(\ast)$ follows from Lemma~\ref{lem:ap_infinity}.
\item $\mathsf{ExAnte}(p)$ has at most $n$ points of discontinuity: For each $i \in [n]$, $F_i$ has at most one point of discontinuity (which must be $u_i$ if exists, see Section~\ref{subsec:prelim-reg-vv}).
\end{itemize}
As per these, we know (1)~parameter $\theta \eqdef \max \big\{p \in (1, \infty]: \,\, \mathsf{ExAnte}(p) \geq 1\big\}$ is well-defined; and (2)~let $S \eqdef \{i \in [n]: \,\, F_i \text{ has probability-mass at } \theta\}$, then
\begin{equation}
\label{eq:upper1.1}
\mathsf{ExAnte}(\theta) - \sum\limits_{i \in S: v_i \geq \theta} r_i(q_i) - \sum\limits_{i \in S: v_i < \theta} r_i\big(1 - F_i(\theta)\big) \leq 1.
\end{equation}
Afterwards, we would (1)~determine parameter $z_i \in [0, q_i]$ for each $i \in [n]$; then (2)~construct a desired instance $\{\FF_i\}_{i = 0}^n$ in term of revenue-quantile curves (in a fashion suggested by Figure~\ref{fig:lem:p/(p+1)-4.1}).

\begin{figure}[H]
\centering
\subfigure[Original $r_i(q)$]{
\begin{tikzpicture}[thick, smooth, domain = 0: 0.5, scale = 2.5]
\draw[->] (0, 0) -- (1.1, 0);
\draw[->] (0, 0) -- (0, 1.1);
\node[above] at (0, 1.1) {\footnotesize $r_i(q)$};
\node[right] at (1.1, 0) {\footnotesize $q$};
\node[left] at (0, 0) {\footnotesize $0$};
\draw[color = blue] plot (\x, {0.2 + 1.6 * sqrt(\x / 2)});

\draw[color = Sepia] (0, 0) -- (0.25, 1.2);
\node[right, color = Sepia] at (0.25, 1.2) {$r(q) = \theta \cdot q$};

\draw[very thick] (0.5, 0) -- (0.5, 1.25pt);
\node[below] at (0.5, 0) {\footnotesize $q_i$};
\draw[very thick] (0, 1) -- (1.25pt, 1);
\node[left] at (0, 1) {\footnotesize $r_i(q_i)$};

\draw[very thick] (0.125, 0) -- (0.125, 1.25pt);
\node[below] at (0.175, 0.01) {\footnotesize $z_i$};
\draw[very thick] (0, 0.6) -- (1.25pt, 0.6);
\node[left] at (0, 0.6) {\footnotesize $r_i(z_i)$};

\draw[very thick] (0, 0.2) -- (1.25pt, 0.2);
\node[left] at (0, 0.2) {\footnotesize $r_i(0)$};

\draw[very thick, color = blue] (0.5, 1) -- (1, 0);
\draw[thin, dashed] (0.5, 0) -- (0.5, 1);
\draw[thin, dashed] (0, 1) -- (0.5, 1);
\draw[thin, dashed] (0.125, 0) -- (0.125, 0.6);
\draw[thin, dashed] (0, 0.6) -- (0.125, 0.6);
\draw[thin, color = blue, fill = green] (0.5, 1) circle(0.5pt);
\draw[thin, color = blue, fill = green] (0.125, 0.6) circle(0.5pt);
\draw[very thick] (1, 0) -- (1, 1.25pt);
\node[below] at (1, 0) {\footnotesize $1$};
\end{tikzpicture}
\label{fig:lem:p/(p+1)-4.1.1}
}
\quad\quad\quad\quad\quad\quad
\subfigure[$\rr_i(q)$ after reduction]{
\begin{tikzpicture}[thick, smooth, domain = 0: 0.375, scale = 2.5]
\draw[->] (0, 0) -- (1.1, 0);
\draw[->] (0, 0) -- (0, 0.7);
\node[above] at (0, 0.7) {\footnotesize $\rr_i(q)$};
\node[right] at (1.1, 0) {\footnotesize $q$};
\node[left] at (0, 0) {\footnotesize $0$};
\draw[color = blue] plot (\x, {1.6 * sqrt(\x / 2 + 0.0625) - 0.4});

\draw[very thick] (0.375, 0) -- (0.375, 1.25pt);
\node[below] at (0.375, 0) {\footnotesize $\qq_i$};
\draw[very thick] (0, 0.4) -- (1.25pt, 0.4);
\node[left] at (0, 0.4) {\footnotesize $\rr(\qq_i)$};

\draw[very thick, color = blue] (0.375, 0.4) -- (1, 0);
\draw[thin, dashed] (0.375, 0) -- (0.375, 0.4);
\draw[thin, dashed] (0, 0.4) -- (0.375, 0.4);
\draw[thin, color = blue, fill = green] (0, 0) circle(0.5pt);
\draw[thin, color = blue, fill = green] (0.375, 0.4) circle(0.5pt);
\draw[very thick] (1, 0) -- (1, 1.25pt);
\node[below] at (1, 0) {\footnotesize $1$};
\end{tikzpicture}
\label{fig:lem:p/(p+1)-4.1.2}
}
\caption{Demonstration for Algorithm~\ref{alg0.3}, where $\qq_i = q_i - z_i$ and $\rr(\qq_i) = r_i(q_i) - r_i(z_i)$.}
\label{fig:lem:p/(p+1)-4.1}
\end{figure}
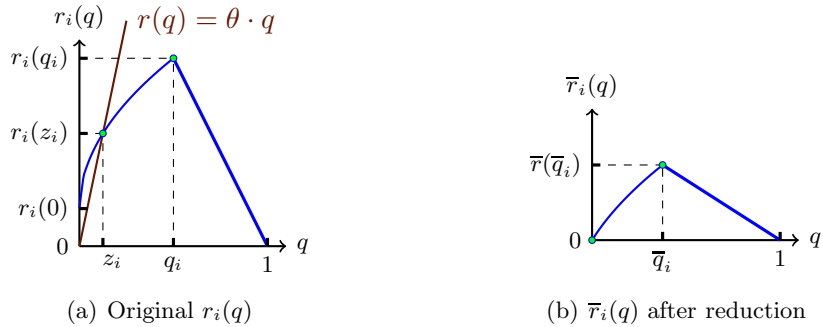

\begin{algorithm}[H]
\caption{(Reduction for Lemma~\ref{lem:p/(p+1)-4})}
\begin{algorithmic}[1]
\Require parameter $\theta$, indices $S$, feasible instance $\{r_i\}_{i = 0}^n$ that $r_0(0) \geq \frac{1}{2}$ and $\sum\limits_{i = 0}^n r_i(q_i) > 1$.
\Ensure parameters $\{z_i\}_{i = 1}^n$, desired instance $\{\rr_i\}_{i = 0}^n$, then $\{\FF_i\}_{i = 0}^n$ follows.
\Statex $------------------------------------------$
\State Let $\Delta_{\#} \eqdef 1 - r_0(0)$.
\Comment{\emph{Firstly, compute parameters $\{z_i\}_{i = 0}^n$.}}
\ForAll{$i \in [n] \setminus S$}
    \If{$v_i \geq \theta$} Define $z_i \eqdef q_i$.
    \label{alg0.3-z}
    \Else{} Define $z_i \eqdef \big(1 - F_i(\theta)\big) \leq q_i$.
    \EndIf
    \State Update: $\Delta_{\#} \leftarrow \big(\Delta_{\#} - r_i(z_i)\big)$
\EndFor
\Comment{\emph{$\Delta_{\#} \geq 0$ at the end of this loop, by inequality~(\ref{eq:upper1.1}).}}
\ForAll{$i \in S$}
    \Comment{\emph{Different orders of $i$'s in $S$ do not matter.}}
    \State Define $z_i \eqdef \min\left\{1 - F_i(\theta), \frac{\Delta_{\#}}{\theta}\right\} \leq q_i$.
    \Comment{\emph{Such $F_i$ has probability-mass at $u_i = \theta$.}}
    \State Update: $\Delta_{\#} \leftarrow \big(\Delta_{\#} - r_i(z_i)\big)$.
\EndFor
\Comment{\emph{$\Delta_{\#} = 0$ at the end of this loop, since $\mathsf{ExAnte}(\theta) \geq 1$.}}
\Statex $------------------------------------------$
\State Define $\rr_0(q) \eqdef 1 - q$ for all $q \in [0, 1]$.
\Comment{\emph{Secondly, construct desired $\{\rr_i\}_{i = 0}^n$.}}
\ForAll{$i \in [n]$}
    \State\label{alg0.3-reduct} Define $\qq_i \eqdef (q_i - z_i) \in [0, q_i]$ and $\rr_i(q) \eqdef
    \begin{cases}
    r_i(q + z_i) - r_i(z_i) & \forall q \in [0, \qq_i] \\
    \frac{r_i(q_i) - r_i(z_i)}{1 - \qq_i} \cdot (1 - q) & \forall q \in (\qq_i, 1]
    \end{cases}.$
\EndFor
\State \Return $\{z_i\}_{i = 1}^n$ and $\{\rr_i\}_{i = 0}^n$.
\end{algorithmic}
\label{alg0.3}
\end{algorithm}

\paragraph{Analysis II: Optimality.}
The facts below are central to justifying the optimality of $\{\FF_i\}_{i = 0}^n$. Here, the concavity of $r_i(q)$ ensures $\partial_+ r_i(q) \geq \partial_+ r_i(q + z_i) = \partial_+ \rr_i(q)$ for all $q \in (0, \qq_i)$. Since then, fact~\textbf{(c)} follows directly from Fact~\ref{fact1}.2 (namely how to convert a revenue-quantile curve to the corresponding virtual value CDF).
\vspace{5.5pt} \\
\fcolorbox{black}{white}{\begin{minipage}{\textwidth}
\begin{description}
\item [(a):] $r_0(0) + \sum\limits_{i = 1}^n r_i(z_i) = \rr_0(0) = 1$. This follows from definitions of $\theta$ and $z_i$'s;
\item [(b):] $\rr_i(0) = 0$ and $\rr_i(\qq_i) = r_i(q_i) - r_i(z_i)$ for all $i \in [n]$, due to line~(\ref{alg0.3-reduct}) of Algorithm~\ref{alg0.3};
\item [(c):] $1 \geq \DD_i(p) \geq D_i(p)$ for all $p \in (0, \infty)$ and $i \in [n]$.
\end{description}
\end{minipage}}
\vspace{5.5pt}

After our reduction, the objective value of Program~(\ref{prog:0}) can never decrease: $\opt\big(\{\FF_i\}_{i = 0}^n\big) = 2 + \displaystyle{\int_1^{\infty}} \left(1 - \prod\limits_{i = 1}^n \DD_i(x)\right) dx \geq 1 + \sum\limits_{i = 0}^{n} r_i(0) + \displaystyle{\int_1^{\infty}} \left(1 - \prod\limits_{i = 1}^n D_i(x)\right) dx = \opt\big(\{F_i\}_{i = 0}^n\big)$. After being rearranged, this inequality is equivalent to
\begin{equation}
\label{eq:upper1.3}
\begin{aligned}
1 - \sum\limits_{i = 0}^{n} r_i(0)
\geq & \displaystyle{\int_1^{\infty}} \left(\prod\limits_{i = 1}^n \DD_i(x) - \prod\limits_{i = 1}^n D_i(x)\right) dx \\
= & \sum\limits_{k = 1}^n \displaystyle{\int_1^{\infty}} \prod\limits_{i = 1}^{k - 1} \DD_i(x) \cdot \prod\limits_{i = k + 1}^n D_i(x) \cdot \big(\DD_k(x) - D_k(x)\big) dx.
\end{aligned}
\end{equation}
Based on the above facts, we can certify inequality~(\ref{eq:upper1.3}) as follows:
\begin{flalign*}
\text{RHS of~(\ref{eq:upper1.3})}
\overset{(\textbf{c})}{\leq} & \sum\limits_{k = 1}^n \displaystyle{\int_1^{\infty}} \big(\DD_k(x) - D_k(x)\big) dx
\overset{(\textbf{c})}{\leq} \sum\limits_{k = 1}^n \displaystyle{\int_0^{\infty}} \big(\DD_k(x) - D_k(x)\big) dx \\
= & \sum\limits_{k = 1}^n \displaystyle{\int_0^{\infty}} \big(1 - D_k(x)\big) dx - \sum\limits_{k = 1}^n \displaystyle{\int_0^{\infty}} \big(1 - \DD_k(x)\big) dx \\
= & \sum\limits_{k = 1}^n \big(r_k(q_k) - r_k(0)\big) - \sum\limits_{k = 1}^n \big(\rr_k(\qq_k) - \rr_k(0)\big) & \text{(By Main Lemma~\ref{lem:virtual_value})} \\
\overset{(\textbf{b})}{=} & \sum\limits_{k = 1}^n \big(r_k(z_k) - r_k(0)\big)
\overset{(\textbf{a})}{=} 1 - \sum\limits_{k = 0}^n r_k(0) = \text{LHS of~(\ref{eq:upper1.3})}.
\end{flalign*}

\paragraph{Analysis III: Feasibility.}
For output $\{\FF_i\}_{i = 0}^n$ of Algorithm~\ref{alg0.3}, the feasibility to constraint~(\ref{cstr:value}) trivially follows from Lemma~\ref{lem:shape}.1: For each $i \in [n]$, under premise $\partial_+ r_i(q) > 1$ for all $q \in [0, q_i)$, it follows from line~(\ref{alg0.3-reduct}) that $\partial_+ \rr_i(q) = \partial_+ r_i(q + z_i) > 1$ for all $q \in [0, \qq_i) = [0, q_i - z_i)$. Clearly\footnote{Particularly, $\FF_i$ ``vanishes'' when $\qq_i = 0$, i.e., when $z_i = q_i$. If so, $\vv_i \equiv \PPhi_i(\vv_i^+)$ can be set as any finite value.},
\[
\vv_i = \frac{\rr_i(\qq_i)}{\qq_i} > 1 \quad\quad\quad\quad\quad\quad \PPhi_i(\vv_i^+) = \lim\limits_{q \rightarrow \qq_i^-} \partial_+ \rr_i(q) \geq 1.
\]
On the other hand, the feasibility to constraint~(\ref{cstr:ap0}) is specified in the following two cases. For this, the facts mentioned in \textbf{Analysis II} are also useful.

\paragraph{Case I (when $p > \theta$):}
In this case, for each $i \in [n]$, it is easy to see $\FF_i(p) = 1$. Consequently, $\ap\big(p, \{\FF_i\}_{i = 0}^n\big) = \ap\big(p, \{\FF_0\}\big) = p \cdot \left(1 - \frac{p}{p + 1}\right) = \frac{p}{p + 1} \leq 1$.

\paragraph{Case II (when $1 < p \leq \theta$):}
In this case, for each $j \in \{0\} \bigcup [n]$, let $s_j \eqdef r_0(0) + \sum\limits_{i = 1}^j r_i(z_i)$ for notational simplicity. For all $j \in [n]$, we claim
\begin{equation}
\label{eq:lem:p/(p+1)-4.0}
\frac{p}{p + s_j} \cdot \prod\limits_{i = 1}^j \FF_i(p) \cdot \prod\limits_{i = j + 1}^n F_i(p) \geq \frac{p}{p + s_{j - 1}} \cdot \prod\limits_{i = 1}^{j - 1} \FF_i(p) \cdot \prod\limits_{i = j}^n F_i(p) \quad\quad \forall p \in (1, \theta].
\end{equation}
If so, since $s_0 = r_0(0)$ and $s_n = r_0(0) + \sum\limits_{i = 1}^n r_i(z_i) \overset{(\textbf{a})}{=} 1 = \rr_0(0)$, this chain of inequalities implies
\[
\prod\limits_{i = 0}^n \FF_i(p) \xlongequal{s_n = \rr_0(0)} \frac{p}{p + s_n} \cdot \prod\limits_{i = 1}^n \FF_i(p) \overset{(\ref{eq:lem:p/(p+1)-4.0})}{\geq} \frac{p}{p + s_0} \cdot \prod\limits_{i = 1}^n F_i(p) \xlongequal{s_0 = r_0(0)} \prod\limits_{i = 0}^n F_i(p) \quad\quad \forall p \in (1, \theta],
\]
which guarantees $\ap\big(p, \{\FF_i\}_{i = 0}^n\big) \leq \ap\big(p, \{F_i\}_{i = 0}^n\big) \leq 1$ for all $p \in (1, \theta]$.

Given $k \in [n]$, under induction hypothesis that inequality~(\ref{eq:lem:p/(p+1)-4.0}) holds for all $j \in [k - 1]$, we shall attest that so it is when $j = k$. Rearrange inequality~(\ref{eq:lem:p/(p+1)-4.0}), it suffices to show
\begin{equation}
\label{eq:lem:p/(p+1)-4.1}
\left(1 + \frac{s_{k - 1} + r_k(z_k)}{p}\right) \cdot \big(\FF_k(p)\big)^{-1} \leq \left(1 + \frac{s_{k - 1}}{p}\right) \cdot \big(F_k(p)\big)^{-1} \quad\quad \forall p \in (1, \theta].
\end{equation}
For each $i \in [n]$, let $\vv_i \eqdef \frac{\rr_i(\qq_i)}{\qq_i} > 1$. Define $U \eqdef \big\{i \in [n]: \,\, (v_i \leq \theta) \bigwedge (\vv_i \geq p)\big\}$. Depending on whether $k \in U$ or not, we would deal with inequality~(\ref{eq:lem:p/(p+1)-4.1}) in different manners.

\paragraph{Case II.1 (when $1 < p \leq \theta$ and $k \in U$):}
In this case, since $r_k(q)$ is concave on $[0, 1]$ and $k \in U$,
\[
v_k = \frac{r_k(q_k)}{q_k} \geq \frac{r_k(q_k) - r_k(z_k)}{q_k - z_k} \overset{(\uplus)}{=} \vv_k \geq p,
\]
where $(\uplus)$ follows from line~(\ref{alg0.3-reduct}) of Algorithm~\ref{alg0.3}. Clearly, $\FF_k(p)$ and $F_k(p)$ admit explicit formulas:
\begin{equation}
\label{eq:lem:p/(p+1)-4.2}
\FF_k(p) = \frac{(1 - \qq_k) \cdot p}{(1 - \qq_k) \cdot p + \rr_k(\qq_k)} \quad\quad\quad\quad F_k(p) = \frac{(1 - q_k) \cdot p}{(1 - q_k) \cdot p + r_k(q_k)}.
\end{equation}
Since then, we can check inequality~(\ref{eq:lem:p/(p+1)-4.1}) as follows:
\begin{flalign*}
& \text{LHS of~(\ref{eq:lem:p/(p+1)-4.1})} - \text{RHS of~(\ref{eq:lem:p/(p+1)-4.1})} \\
\overset{(\ref{eq:lem:p/(p+1)-4.2})}{=} & \left(1 + \frac{s_{k - 1} + r_k(z_k)}{p}\right) \cdot \left(1 + \frac{\rr_k(\qq_k)}{1 - \qq_k} \cdot \frac{1}{p}\right) - \left(1 + \frac{s_{k - 1}}{p}\right) \cdot \left(1 + \frac{r_k(q_k)}{1 - q_k} \cdot \frac{1}{p}\right) \\
& \text{(By line~(\ref{alg0.3-reduct}) of Algorithm~\ref{alg0.3}, $\rr_k(\qq_k) = r_k(q_k) - r_k(z_k)$ and $\qq_k = q_k - z_k \leq q_k$)} \\
\leq & \left(1 + \frac{s_{k - 1} + r_k(z_k)}{p}\right) \cdot \left(1 + \frac{r_k(q_k) - r_k(z_k)}{1 - q_k} \cdot \frac{1}{p}\right) - \left(1 + \frac{s_{k - 1}}{p}\right) \cdot \left(1 + \frac{r_k(q_k)}{1 - q_k} \cdot \frac{1}{p}\right) \\
= & -\left(p - \frac{r_k(q_k) - s_{k - 1}}{q_k} + \frac{r_k(z_k)}{q_k}\right) \cdot \frac{q_k}{1 - q_k} \cdot \frac{r_k(z_k)}{p^2} \\
\overset{(\diamond)}{\leq} & -\left(p - 1 + \frac{r_k(z_k)}{q_k}\right) \cdot \frac{q_k}{1 - q_k} \cdot \frac{r_k(z_k)}{p^2} \leq 0, & \text{(As $p > 1$)}
\end{flalign*}
where in $(\diamond)$ we apply Lemma~\ref{lem:p/(p+1)-4.1} (whose premise $v_k \leq \theta$ follows as $k \in U$).
\vspace{5.5pt} \\
\fbox{\begin{minipage}{\textwidth}
\begin{lemma}
\label{lem:p/(p+1)-4.1}
Suppose $v_k \leq \theta$, then $\frac{r_k(q_k) - s_{k - 1}}{q_k} = v_k - \frac{s_{k - 1}}{q_k} \leq 1$.
\end{lemma}
\end{minipage}}

\begin{proof}[Proof of Lemma~\ref{lem:p/(p+1)-4.1}]
Recall premise $r_0(0) \geq \frac{1}{2}$ in Lemma~\ref{lem:p/(p+1)-4}, and facts mentioned in \textbf{Analysis II}, $\frac{1}{2} \leq r_0(0) = s_0 \leq \cdots \leq s_{k - 1} \leq \cdots \leq s_n = r_0(0) + \sum\limits_{i = 1}^n r_i(z_i) \overset{(\textbf{a})}{=} 1$, which means
\begin{equation}
\label{eq:lem:p/(p+1)-4.3}
0 \leq \frac{1 - s_{k - 1}}{s_{k - 1}} \leq 1.
\end{equation}
Given $p \in (1, \theta]$, under induction hypothesis that inequality~(\ref{eq:lem:p/(p+1)-4.0}) holds for all $j \in [k - 1]$, certainly $\frac{p}{p + s_{k - 1}} \cdot F_{k - 1}(p) \geq \frac{p}{p + s_{k - 1}} \cdot \prod\limits_{i = 1}^{k - 1} \FF_i(p) \cdot \prod\limits_{i = k}^n F_i(p) \overset{(\ref{eq:lem:p/(p+1)-4.0})}{\geq} \frac{p}{p + r_0(0)} \cdot \prod\limits_{i = 1}^n F_i(p) \xlongequal{s_0 = r_0(0)} \prod\limits_{i = 0}^n F_i(p)$. As per this, and under our premise that $\{F_i\}_{i = 0}^n$ is feasible to constraint~(\ref{cstr:ap0}), definitely
\[
p \cdot \left(1 - \frac{p}{p + s_{k - 1}} \cdot F_{k - 1}(p)\right) \leq \ap\big(p, \{F_i\}_{i = 0}^n\big) \leq 1,
\]
for all $p \in (1, \theta]$. Under premise $\theta \geq v_k > 1$, assigning $p \leftarrow v_k$ gives $v_k \cdot \left[1 - \frac{v_k}{v_k + s_{k - 1}} \cdot (1 - q_k)\right] \leq 1$, or equivalently,
\begin{equation}
\label{eq:lem:p/(p+1)-4.4}
q_k \leq \frac{1 - s_{k - 1}}{v_k} + \frac{s_{k - 1}}{v_k^2}.
\end{equation}
Accordingly, $v_k - \frac{s_{k - 1}}{q_k} \overset{(\ref{eq:lem:p/(p+1)-4.4})}{\leq} v_k - v_k^2\left/\left(\frac{1 - s_{k - 1}}{s_{k - 1}} \cdot v_k + 1\right)\right. \overset{(\ref{eq:lem:p/(p+1)-4.3})}{\leq} v_k - \frac{v_k^2}{v_k + 1} = \frac{v_k}{v_k + 1} \leq 1$. This completes the proof of Lemma~\ref{lem:p/(p+1)-4.1}.
\end{proof}

\paragraph{Case II.2 (when $1 < p \leq \theta$ and $k \in [n] \setminus U$):}
Due to line~(\ref{alg0.3-reduct}) of Algorithm~\ref{alg0.3}, and as Figure~\ref{fig:lem:p/(p+1)-4.2.2} shows, there exists $y \in [0, \qq_k) = [0, q_k - z_k)$ such that
\begin{equation}
\label{eq:lem:p/(p+1)-4.5}
p \cdot y = \rr_k(y) = r_k(y + z_k) - r_k(z_k).
\end{equation}
No matter whether $y = 0$ or not, we always have $p \in \left[0, \frac{r_k(y + z_k)}{y + z_k}\right]$:
\begin{itemize}
\item Suppose $y = 0$, then $p \leq \theta \overset{(\vee)}{\leq} \frac{r_k(z_k)}{z_k} = \frac{r_k(y + z_k)}{y + z_k}$, where $(\vee)$ follows from the definition of $z_k$;
\item Otherwise, $p \overset{(\ref{eq:lem:p/(p+1)-4.5})}{=} \frac{r_k(y + z_k) - r_k(z_k)}{y + z_k - z_k} \leq \frac{r_k(y + z_k)}{y + z_k}$, where the inequality follows as $r_k(q)$ is concave.
\end{itemize}

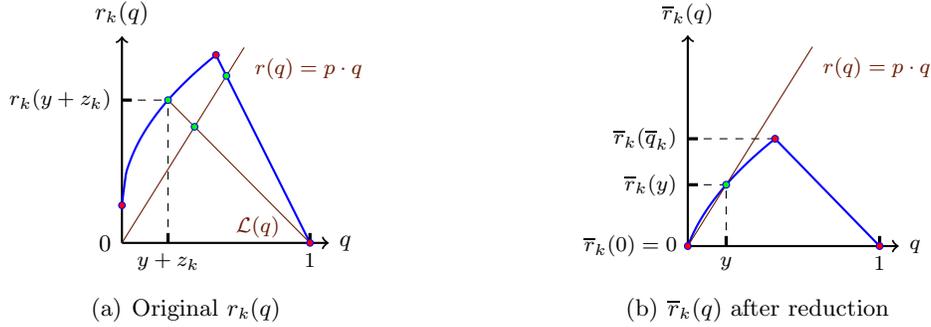
\begin{figure}[H]
\centering
\subfigure[Original $r_k(q)$]{
\begin{tikzpicture}[thick, smooth, domain = 0: 0.5, scale = 2.5]
\draw[->] (0, 0) -- (1.1, 0);
\draw[->] (0, 0) -- (0, 1.1);

\node[above] at (0, 1.1) {\footnotesize $r_k(q)$};
\node[right] at (1.1, 0) {\footnotesize $q$};
\node[left] at (0, 0) {\footnotesize $0$};
\draw[color = blue] plot (\x, {0.2 + 1.6 * sqrt(\x / 2)});

\draw[thin, dashed] (0.245, 0) -- (0.245, 0.76);
\draw[very thick] (0.245, 0) -- (0.245, 1.5pt);
\draw[thin, dashed] (0, 0.76) -- (0.245, 0.76);
\draw[very thick] (0, 0.76) -- (1.5pt, 0.76);
\node[below] at (0.245, 0) {\scriptsize $y + z_k$};
\node[left] at (0, 0.76) {\scriptsize $r_k(y + z_k)$};

\draw[color = blue] (0.5, 1) -- (1, 0);
\draw[thin, color = Sepia] (0, 0) -- (0.65, 1.04);
\draw[thin, color = Sepia] (0.245, 0.76) -- (1, 0);
\node[anchor = 60, color = Sepia] at (0.7879, 0.2) {\scriptsize $\mathcal{L}(q)$};
\draw[very thick] (1, 0) -- (1, 1.5pt);
\node[below] at (1, 0) {\scriptsize $1$};
\node[anchor = north west, color = Sepia] at (0.65, 1.04) {\scriptsize $r(q) = p \cdot q$};

\draw[thin, color = blue, fill = green] (0.245, 0.76) circle(0.5pt);
\draw[thin, color = blue, fill = green] (0.5556, 0.8889) circle(0.5pt);
\draw[thin, color = blue, fill = green] (0.3862, 0.6179) circle(0.5pt);
\draw[thin, color = blue, fill = red] (1, 0) circle(0.5pt);
\draw[thin, color = blue, fill = red] (0.5, 1) circle(0.5pt);
\draw[thin, color = blue, fill = red] (0, 0.2) circle(0.5pt);
\end{tikzpicture}
\label{fig:lem:p/(p+1)-4.2.1}
}
\quad\quad\quad\quad\quad\quad
\subfigure[$\rr_k(q)$ after reduction]{
\begin{tikzpicture}[thick, smooth, domain = 0: 0.455, scale = 2.55]
\draw[->] (0, 0) -- (1.1, 0);
\draw[->] (0, 0) -- (0, 1.1);
\node[above] at (0, 1.1) {\scriptsize $\rr_k(q)$};
\node[right] at (1.1, 0) {\scriptsize $q$};
\node[left] at (0, 0) {\scriptsize $\rr_k(0) = 0$};
\draw[color = blue] plot (\x, {1.6 * sqrt(\x / 2 + 0.0225) - 0.24});
\draw[thin, dashed] (0, 0.56) -- (0.455, 0.56);
\draw[very thick] (0, 0.56) -- (1.5pt, 0.56);
\node[anchor = 0] at (0, 0.56) {\scriptsize $\rr_k(\qq_k)$};

\draw[thin, dashed] (0.2, 0) -- (0.2, 0.32);
\draw[very thick] (0.2, 0) -- (0.2, 1.5pt);
\draw[thin, dashed] (0, 0.32) -- (0.2, 0.32);
\draw[very thick] (0, 0.32) -- (1.5pt, 0.32);
\node[below] at (0.2, 0) {\scriptsize $y$};
\node[left] at (0, 0.32) {\scriptsize $\rr_k(y)$};

\draw[color = blue] (0.455, 0.56) -- (1, 0);
\draw[thin, color = Sepia] (0, 0) -- (0.65, 1.04);
\draw[very thick] (1, 0) -- (1, 1.5pt);
\node[below] at (1, 0) {\scriptsize $1$};
\node[anchor = north west, color = Sepia] at (0.65, 1.04) {\scriptsize $r(q) = p \cdot q$};
\draw[thin, color = blue, fill = green] (0.2, 0.32) circle(0.5pt);
\draw[thin, color = blue, fill = red] (1, 0) circle(0.5pt);
\draw[thin, color = blue, fill = red] (0.455, 0.56) circle(0.5pt);
\draw[thin, color = blue, fill = red] (0, 0) circle(0.5pt);
\end{tikzpicture}
\label{fig:lem:p/(p+1)-4.2.2}
}
\caption{Demonstrations for the reduction in Algorithm~\ref{alg0.3}}
\label{fig:lem:p/(p+1)-4.2}
\end{figure}

In Figure~\ref{fig:lem:p/(p+1)-4.2.1}, line segment $\mathcal{L}(q)$ has endpoints $\big(y + z_k, r_k(y + z_k)\big)$ and $(1, 0)$. Transparently, $\mathcal{L}(q) \eqdef \frac{r_k(y + z_k)}{1 - y - z_k} \cdot (1 - q)$ for all $q \in [y + z_k, 1]$, which corresponds to CDF
\[
F_{\mathcal{L}}(p') \eqdef \frac{(1 - y - z_k) \cdot p'}{(1 - y - z_k) \cdot p' + r_k(y + z_k)} \quad\quad \forall p' \in \left[0, \frac{r_k(y + z_k)}{y + z_k}\right].
\]
Since $r_k(q)$ is concave on $[y + z_k, 1] \subset [0, 1]$, we know $r_k(q) \geq \mathcal{L}(q)$ for all $q \in [y + z_k, 1]$ and further, $F_k(p') \leq F_{\mathcal{L}}(p')$ for all $p' \in \left[0, \frac{r_k(y + z_k)}{y + z_k}\right]$. Specifically, recall that $p \in \left[0, \frac{r_k(y + z_k)}{y + z_k}\right]$,
\begin{equation}
\label{eq:lem:p/(p+1)-4.6}
F_k(p) \leq F_{\mathcal{L}}(p) = \frac{(1 - y - z_k) \cdot p}{(1 - y - z_k) \cdot p + r_k(y + z_k)} \leq \frac{(1 - y) \cdot p}{(1 - y) \cdot p + r_k(y + z_k)}.
\end{equation}
Together with the fact $\FF_k(p) = 1 - y$, we settle inequality~(\ref{eq:lem:p/(p+1)-4.1}) as follows:
\[
\begin{aligned}
\text{LHS of~(\ref{eq:lem:p/(p+1)-4.1})}
= & \left(1 + \frac{s_{k - 1} + r_k(z_k)}{p}\right) \cdot \frac{1}{1 - y}
\leq \left(1 + \frac{s_{k - 1}}{p}\right) \cdot \left(1 + \frac{r_k(z_k)}{p}\right) \cdot \frac{1}{1 - y} \\
\overset{(\ref{eq:lem:p/(p+1)-4.5})}{=} & \left(1 + \frac{s_{k - 1}}{p}\right) \cdot \left(1 + \frac{r_k(y + z_k)}{1 - y} \cdot \frac{1}{p}\right)
\overset{(\ref{eq:lem:p/(p+1)-4.6})}{\leq} \text{RHS of~(\ref{eq:lem:p/(p+1)-4.1})}.
\end{aligned}
\]
To sum up, by induction inequality~(\ref{eq:lem:p/(p+1)-4.1}) holds for each $j \in [n]$. As mentioned before, this indicates $\ap\big(p, \{\FF_i\}_{i = 0}^n\big) \leq \ap\big(p, \{F_i\}_{i = 0}^n\big) \leq 1$ for all $p \in (1, \theta]$.

Henceforth, we settle the optimality and feasibility of $\{\FF_i\}_{i = 0}^n$, hence Lemma~\ref{lem:p/(p+1)-4}.
\end{proof}

\subsection{Proof of Lemma~\ref{lem:p/(p+1)-1.1}}
\label{subapp:upper1-p/(p+1)-1.1}

The following mathematical fact is got in Appendix~\ref{subapp:math_facts:p/(p+1)}, and is central to the proof of Lemma~\ref{lem:p/(p+1)-1.1}.

\paragraph{[Lemma~\ref{lem:ineq:p/(p+1)}].}
\emph{$H(s, p) < 0$ for all $s \in \left[0, \frac{1}{\sqrt{3}}\right]$ and $p \in (1, \infty)$, where}
\[
H(s, p) \eqdef \ln\left(1 - \frac{1}{p}\right) + \ln\left(1 + \frac{s}{p}\right) + \big[p - (1 + s)\big] \cdot \left(\frac{1}{p - 1} + \frac{1}{p + s} - \frac{2}{p}\right).
\]

Recall the definition of critical instance:
\vspace{5.5pt} \\
\fcolorbox{white}{lightgray}{\begin{minipage}{\textwidth}
\paragraph{[Definition~\ref{def:critical}].}
\emph{A critical instance $\{F_i\}_{i = 0}^n$ is feasible, has well-defined and finite $\beta$, and satisfies:
\begin{enumerate}
\item $\forall \big(i \in [n] \text{: } F_i(\beta) < 1\big)$, $\exists a_i > 1$, $\exists b_i \geq 0$: $F_i(p) = 1 - \frac{b_i}{p - a_i}$ for all $p \in (\beta, \infty)$;
\item $\partial_+ \ap\big(\beta\big) = 0$.
\end{enumerate}}
\end{minipage}}

\paragraph{[Lemma~\ref{lem:p/(p+1)-1.1}].}
\emph{Given a critical instance $\{F_i\}_{i = 0}^n$ of Program~(\ref{prog:0}) that satisfies the property from Lemma~\ref{lem:p/(p+1)-3} (i.e., $\partial_+ r_i(0) > 1 + r_0(0)$ for all $i \in [n]$ with $r_i(0) > 0$), then $r_0(0) > \frac{1}{\sqrt{3}}$.}

\begin{proof}
We would settle Lemma~\ref{lem:p/(p+1)-1.1} in two steps, both of which are based on proof by contradiction.

\paragraph{Analysis I.}
We first claim \emph{there exists at least one $k \in [n]$ such that $r_k(0) > 0$}.

When $\{F_i\}_{i = 0}^n$ satisfies all premises, clearly\footnote{If $\Gamma = \emptyset$, we obtain a contradiction $1 = \ap\big(\beta, \{F_i\}_{i = 0}^n\big) = \ap\big(\beta, \{F_0\}\big) = \beta \cdot \left(1 - \frac{\beta}{\beta + r_0(0)}\right) \leq r_0(0) \leq \frac{1}{\sqrt{3}}$.} $\Gamma \eqdef \big\{i \in [n]: \,\, F_i(\beta) < 1\big\} \neq \emptyset$. By the \emph{first} condition for criticality, for each $i \in \Gamma$, there exists $a_i > 1$ and $b_i \geq 0$ such that $F_i(p) = 1 - \frac{b_i}{p - a_i}$ for all $p \in (\beta, \infty)$. Indeed, $\Gamma_0 \eqdef \big\{i \in \Gamma: b_i = 0\big\} = \emptyset$. For this, we shall observe that
\begin{description}
\item [\textbf{(a):}] $F_i(\beta^+) = F_i(\beta) < 1$ for each $i \in (\Gamma \setminus \Gamma_0)$. Clearly, such distribution $F_i$ has support-supremum $u_i = \infty$ (in that $r_i(0) = b_i > 0$), and thus must be continuous at $\beta < u_i = \infty$ (see Section~\ref{subsec:prelim-reg-vv});
\item [\textbf{(b):}] $F_i(\beta^+) = 1$ yet $F_i(\beta) < 1$, for each $i \in \Gamma_0$.
\end{description}
Assume to the contrary $\Gamma_0 \neq \emptyset$, then it follows from above that
\[
\begin{aligned}
\ap\big(\beta^+, \{F_i\}_{i = 0}^n\big)
= & \ap\big(\beta^+, \{F_0\} \bigcup \{F_i\}_{i \in \Gamma}\big)
\overset{(\textbf{b})}{=} \ap\big(\beta^+, \{F_0\} \bigcup \{F_i\}_{i \in \Gamma \setminus \Gamma_0}\big) \\
\overset{(\textbf{a})}{=} & \ap\big(\beta, \{F_0\} \bigcup \{F_i\}_{i \in \Gamma \setminus \Gamma_0}\big)
\overset{(\textbf{b})}{<} \ap\big(\beta^+, \{F_0\} \bigcup \{F_i\}_{i \in \Gamma}\big)
= \ap\big(\beta, \{F_i\}_{i = 0}^n\big),
\end{aligned}
\]
which violates to premise $\partial_+ \ap\big(\beta, \{F_i\}_{i = 0}^n\big) = 0$. Hitherto, we know $(\Gamma \setminus \Gamma_0) = \Gamma \neq \emptyset$. To sum up, for each $i \in \Gamma \neq \emptyset$, there exists $a_i > 1$ and $b_i > 0$ such that
\begin{itemize}
\item $F_i(p) = 1 - \frac{b_i}{p - a_i}$ for all $p \in (\beta, \infty)$, and $F_i(\beta^+) = F_i(\beta) < 1$;
\item $f_i(p) = F_i'(p) = \frac{b_i}{(p - a_i)^2} > 0$ when $p \in (\beta, \infty)$;
\item $r_i(q) = a_i \cdot q + b_i$ when $q \in \big[0, 1 - F_i(\beta)\big]$, due to the reduction of revenue-quantile curve in Section~\ref{subsec:prelim-rq-curve}.
\end{itemize}
Now, we achieve our claim by observing $r_i(0) = b_i > 0$ and $\partial_+ r_i(0) = a_i > 1$, for each $i \in \Gamma \neq \emptyset$.

\paragraph{Analysis II.}
To see the lemma, assume to the contrary that $r_0(0) \leq \frac{1}{\sqrt{3}}$. Based on the conclusions of \textbf{Analysis I},
\[
p - \frac{1 - F_i(p)}{f_i(p)} = a_i > 1 + r_0(0),
\]
for all $p \in (\beta, \infty)$ and $i \in \Gamma \neq \emptyset$. Particularly, when $p$ goes to $\beta^+$, rearranging this inequality gives
\begin{equation}
\label{eq:lem:p/(p+1)-1.1.2}
\Big[\beta - \big(1 + r_0(0)\big)\Big] \cdot f_i(\beta^+) \geq 1 - F_i(\beta^+) = 1 - F_i(\beta),
\end{equation}
for each $i \in \Gamma \neq \emptyset$. Recall premise that $\beta = \max \big\{p \in (1, \infty]: \,\, \ap\big(p, \{F_i\}_{i = 0}^n\big) = 1\big\}$ is well-defined and finite, and the \emph{second} condition for criticality that $\partial_+ \ap\big(\beta, \{F_i\}_{i = 0}^n\big) = 0$. We have
\begin{description}
\item [\textbf{(c):}] $\ln\left(1 - \frac{1}{\beta}\right) + \ln\left(1 + \frac{r_0(0)}{\beta}\right) = \sum\limits_{i \in \Gamma} \ln F_i(\beta)$. This follows from rearranging $1 = \ap\big(\beta, \{F_i\}_{i = 0}^n\big) = \ap\big(\beta, \{F_0\} \bigcup \{F_i\}_{i \in \Gamma}\big) = \beta \cdot \left(1 - \frac{\beta}{\beta + r_0(0)} \cdot \prod\limits_{i \in \Gamma} F_i(\beta)\right)$;
\item [\textbf{(d):}] $\frac{1}{\beta - 1} + \frac{1}{\beta + r_0(0)} - \frac{2}{\beta} = \sum\limits_{i \in \Gamma} \frac{f_i(\beta^+)}{F_i(\beta)}$. This follows from taking right-derivatives on both hand sides of the equation in~\textbf{(c)}, under premise $\partial_+ \ap\big(\beta, \{F_0\} \bigcup \{F_i\}_{i \in \Gamma}\big) = \partial_+ \ap\big(\beta, \{F_i\}_{i = 0}^n\big) = 0$.
\end{description}
Gathering everything together results in
\[
\begin{aligned}
\ln\left(1 - \frac{1}{\beta}\right) + \ln\left(1 + \frac{r_0(0)}{\beta}\right)
\overset{(\textbf{c})}{=} & -\sum\limits_{i \in \Gamma} \ln\left(1 + \frac{1 - F_i(\beta)}{F_i(\beta)}\right)
\overset{(\star)}{\geq} -\sum\limits_{i \in \Gamma} \frac{1 - F_i(\beta)}{F_i(\beta)} \\
\overset{(\ref{eq:lem:p/(p+1)-1.1.2})}{\geq} & -\Big[\beta - \big(1 + r_0(0)\big)\Big] \cdot \sum\limits_{i \in \Gamma} \frac{f_i(\beta^+)}{F_i(\beta)} \\
\overset{(\textbf{d})}{=} & -\Big[\beta - \big(1 + r_0(0)\big)\Big] \cdot \left(\frac{1}{\beta - 1} + \frac{1}{\beta + r_0(0)} - \frac{2}{\beta}\right),
\end{aligned}
\]
where $(\star)$ follows as $x \geq \ln(1 + x)$ for all $x \geq 0$.

In terms of $H(s, p)$ defined in Lemma~\ref{lem:ineq:p/(p+1)}, rearranging this inequality gives $H\big(r_0(0), \beta\big) \geq 0$. Under premises $1 < \beta < \infty$ and $0 \leq r_0(0) \leq \frac{1}{\sqrt{3}}$, this contradicts Lemma~\ref{lem:ineq:p/(p+1)} that $H(s, p) < 0$ for all $s \in \left[0, \frac{1}{\sqrt{3}}\right]$ and $p \in (1, \infty)$. Refuting our previous assumption, we settle the proof of Lemma~\ref{lem:p/(p+1)-1.1}.
\end{proof}

\subsection{Proof of Lemma~\ref{lem:p/(p+1)-2}}
\label{subapp:upper1-p/(p+1)-2}

\fcolorbox{white}{lightgray}{\begin{minipage}{\textwidth}
In a worst-case instance $\{F_i\}_{i = 0}^n$, w.l.o.g.
\begin{itemize}
\item With some pre-fixed $r_0(0) \geq 0$, $r_0(q) = r_0(0) \cdot (1 - q)$ for all $q \in [0, 1]$;
\item For each $i \in [n]$, $q_i > 0$, $\partial_+ r_i(q) > 1$ for all $q \in [0, q_i)$, $v_i = \frac{r_i(q_i)}{q_i} > 1$.
\end{itemize}
\end{minipage}}

\paragraph{[Lemma~\ref{lem:p/(p+1)-2}].}
\emph{Given a feasible instance $\{F_i\}_{i = 0}^n$ of Program~(\ref{prog:0}) that $\sum\limits_{i = 0}^n r_i(0) = 1$ and satisfies the property from Lemma~\ref{lem:p/(p+1)-3} (i.e., $\partial_+ r_i(0) > 1 + r_0(0)$ for all $i \in [n]$ with $r_i(0) > 0$), then $r_0(0) > \frac{1}{\sqrt{3}}$.}

\begin{proof}
For ease of notation, let $S \eqdef \big\{i \in [n]: \,\, r_i(0) > 0\big\} \subset [n]$. Assuming to the contrary that $r_0(0) \leq \frac{1}{\sqrt{3}}$, we claim that $\ap\big(p, \{F_i\}_{i = 0}^n\big) > 1$ when $p$ is sufficiently large. This would violate the feasibility to constraint~(\ref{cstr:ap0}).

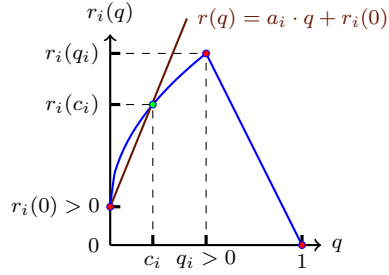
\begin{figure}[H]
\centering
\begin{tikzpicture}[thick, smooth, domain = 0: 0.5, scale = 2.55]
\draw[->] (0, 0) -- (1.1, 0);
\draw[->] (0, 0) -- (0, 1.1);
\node[above] at (0, 1.1) {\scriptsize $r_i(q)$};
\node[right] at (1.1, 0) {\scriptsize $q$};
\node[left] at (0, 0) {\scriptsize $0$};
\draw[color = blue] plot (\x, {0.2 + 1.6 * sqrt(\x / 2)});
\draw[thin, dashed] (0, 1) -- (0.5, 1);
\draw[very thick] (0, 1) -- (1.5pt, 1);
\draw[thin, dashed] (0.5, 0) -- (0.5, 1);
\draw[very thick] (0.5, 0) -- (0.5, 1.5pt);
\node[left] at (0, 1) {\scriptsize $r_i(q_i)$};
\node[below] at (0.5, 0.75pt) {\scriptsize $q_i > 0$};

\draw[very thick] (0, 0.2) -- (1.5pt, 0.2);
\node[left] at (0, 0.2) {\scriptsize $r_i(0) > 0$};

\draw[thin, dashed] (0.2222, 0) -- (0.2222, 0.7333);
\draw[very thick] (0.2222, 0) -- (0.2222, 1.5pt);
\draw[thin, dashed] (0, 0.7333) -- (0.2222, 0.7333);
\draw[very thick] (0, 0.7333) -- (1.5pt, 0.7333);
\node[below] at (0.2222, 0) {\scriptsize $c_i$};
\node[left] at (0, 0.7333) {\scriptsize $r_i(c_i)$};

\draw[color = blue] (0.5, 1) -- (1, 0);
\draw[color = Sepia] (0, 0.2) -- (0.4, 1.1822);
\draw[very thick] (1, 0) -- (1, 1.5pt);
\node[below] at (1, 0.25pt) {\scriptsize $1$};
\node[anchor = west, color = Sepia] at (0.4, 1.1822) {\scriptsize $r(q) = a_i \cdot q + r_i(0)$};
\draw[thin, color = blue, fill = green] (0.2222, 0.7333) circle(0.5pt);
\draw[thin, color = blue, fill = red] (1, 0) circle(0.5pt);
\draw[thin, color = blue, fill = red] (0.5, 1) circle(0.5pt);
\draw[thin, color = blue, fill = red] (0, 0.2) circle(0.5pt);
\end{tikzpicture}
\caption{Demonstrations for Lemma~\ref{lem:p/(p+1)-2}}
\label{fig:lem:p/(p+1)-6}
\end{figure}

For each $i \in S$, given $a_i \in \big(1 + r_0(0), \partial_+ r_i(0)\big)$, let $c_i > 0$ be that $r_i(c_i) = a_i c_i + r_i(0)$ (which refers to Figure~\ref{fig:lem:p/(p+1)-6} in the above). Recall that $q_i > 0$ for each $i \in S$. As Figure~\ref{fig:lem:p/(p+1)-6} suggests, for each $i \in S$, we shall notice that
\begin{itemize}
\item The concavity of $r_i(q)$ indicates that $r_i(q) \geq a_i \cdot q + r_i(0)$ when $q \in [0, c_i]$;
\item $F_i(p) \geq 1 - \frac{r_i(0)}{p - a_i}$ for all $p \in \left(\frac{r_i(c_i)}{c_i}, \infty\right)$, due to Fact~\ref{fact1} (i.e., the reduction from $r_i(q)$ to $F_i$).
\end{itemize}
Define $\theta \eqdef \max \left\{\frac{r_i(c_i)}{c_i}: \,\, i \in S\right\}$. For all $p \in [\theta, \infty]$, we know from the above that
\[
p \cdot \left[1 - \frac{p}{p + r_0(0)} \cdot \prod\limits_{i \in S} \left(1 - \frac{r_i(0)}{p - a_i}\right)\right] \leq \ap\big(p, \{F_0\} \bigcup \{F_i\}_{i \in S}\big) \leq \ap\big(p, \{F_i\}_{i = 0}^n\big).
\]
To obtain the desired contradiction, it suffices to prove that $p \cdot \left[1 - \frac{p}{p + r_0(0)} \cdot \prod\limits_{i \in S} \left(1 - \frac{r_i(0)}{p - a_i}\right)\right] > 1$, when $p$ is sufficiently large. Let $t \eqdef \frac{1}{p} \in \left[0, \theta^{-1}\right]$ for notational brevity. After rearranging, we turn to show that $G_{\ap}(t) > 0$ when $t > 0$ is sufficiently small, where
\[
G_{\ap}(t) \eqdef \ln(1 - t) + \ln\big(1 + r_0(0) \cdot t\big) + \sum\limits_{i \in S} \ln(1 - a_i \cdot t) - \sum\limits_{i \in S} \ln\left[1 - \big(a_i + r_i(0)\big) \cdot t\right].
\]
Equipped with the new notations, we shall observe $G_{\ap}(0) = 0$ and further,
\begin{itemize}
\item \emph{First Derivative}: $G_{\ap}'(t) = -\frac{1}{1 - t} + \frac{r_0(0)}{1 + r_0(0) \cdot t} - \sum\limits_{i \in S} \frac{a_i}{1 - a_i \cdot t} + \sum\limits_{i \in S} \frac{a_i + r_i(0)}{1 - \big(a_i + r_i(0)\big) \cdot t}$. In that $\sum\limits_{i = 0}^n r_i(0) = 1$, we must have $G_{\ap}'(0) = -1 + r_0(0) + \sum\limits_{i \in S} r_i(0) = -1 + \sum\limits_{i = 0}^n r_i(0) = 0$;
\item \emph{Second Derivative}: $G_{\ap}''(t) = -\frac{1}{(1 - t)^2} - \frac{r_0^2(0)}{\big(1 + r_0(0) \cdot t\big)^2} - \sum\limits_{i \in S} \frac{a_i^2}{(1 - a_i \cdot t)^2} + \sum\limits_{i \in S} \frac{\big(a_i + r_i(0)\big)^2}{\left[1 - \big(a_i + r_i(0)\big) \cdot t\right]^2}$. Hence,
    \[
    \begin{aligned}
    G_{\ap}''(0)
    = & -1 - r_0^2(0) - \sum\limits_{i \in S} a_i^2 + \sum\limits_{i \in S} \big(a_i + r_i(0)\big)^2
    > -1 - r_0^2(0) + 2 \cdot \sum_{i \in S} a_i \cdot r_i(0) \\
    > & -1 - r_0^2(0) + 2 \cdot \sum_{i \in S} \big(1 + r_0(0)\big) \cdot r_i(0)
    = -1 - r_0^2(0) + 2 \cdot \big(1 - r_0^2(0)\big) \geq 0,
    \end{aligned}
    \]
    where the last inequality follows from our assumption that $r_0(0) \leq \frac{1}{\sqrt{3}}$.
\end{itemize}
Combining everything together, we conclude that $G_{\ap}(t) > 0$ when $t > 0$ is sufficiently small, which contradicts the feasibility to constraint~(\ref{cstr:ap0}). This completes the proof of the lemma.
\end{proof}

\section{Proof of Lemma~\ref{lem:opt_cont}}
\label{subapp:upper2-cont}
\fcolorbox{white}{lightgray}{\begin{minipage}{\textwidth}
\paragraph{[Definition~\ref{def:cont}].}
\emph{Given $\gamma \in [1, \infty)$ and $n \in \mathbb{N}_+$, triangular instance $\{\tri(v_i, q_i)\}_{i = 1}^{n^2}$ satisfies that
\begin{itemize}
\item $v_i = \gamma + n - \frac{i - 1}{n}$, for all $i \in [n^2]$. For notational simplicity, let $v_0 \eqdef \infty$;
\item $\sum\limits_{j = 1}^i \ln\left(1 + \frac{v_i q_i}{1 - q_i}\right) = \R(v_i)$, that is, $q_i = \frac{e^{\R(v_i) - \R(v_{i - 1})} - 1}{v_i + e^{\R(v_i) - \R(v_{i - 1})} - 1}$, for all $i \in [n^2]$.
\end{itemize}
When $n$ approaches to infinity, $\{\tri(v_i, q_i)\}_{i = 1}^{n^2}$ converges to $\cont(\gamma)$.}
\end{minipage}}

\paragraph{[Lemma~\ref{lem:opt_cont}].}
\emph{$\opt\big(\cont(\gamma)\big) = 2 + \displaystyle{\int_1^{\infty}} \left(1 - e^{-\Q\big(\max\{x, \gamma\}\big)}\right) dx$ is decreasing for all $\gamma \in [1, \infty)$.}

\begin{figure}[H]
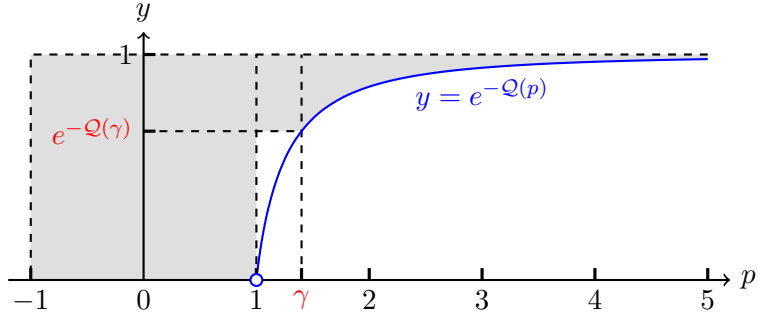

\centering

\caption{Demonstration for $y = e^{-Q(p)}$, where $\opt\big(\cont(\gamma)\big)$ equals to the area of shadow}
\label{fig:opt_cont}
\end{figure}

\begin{proof}
Figure~\ref{fig:opt_cont} is offered to make things mimic, from which the monotonicity of $\opt\big(\cont(\gamma)\big)$ follows immediately. Given $\gamma \in [1, \infty)$ and $n \in \mathbb{N}_+$ in Definition~\ref{def:cont}, we would use a formula obtained in~\cite{JLTX18} (see Lemma~1 therein) to rewrite the objective value of Program~(\ref{prog:1}):
\begin{equation}
\label{eq:lem:opt_cont1}
\opt\left(\{\tri(v_i, q_i)\}_{i = 1}^{n^2}\right) = 2 + \displaystyle{\int_1^{\infty}} \left[1 - \prod\limits_{i \in [n^2]: v_i > x} (1 - q_i)\right] dx.
\end{equation}
When $n$ approaches to infinity, combining Definition~\ref{def:cont} and Lemma~\ref{lem:RQ} gives
\begin{flalign*}
& \lim\limits_{n \rightarrow \infty} \prod\limits_{i \in [n^2]: v_i > p} (1 - q_i)
= \lim\limits_{n \rightarrow \infty} \prod\limits_{i \in [n^2]: v_i > p} \left(1 + \frac{e^{\R(v_i) - \R(v_{i - 1})} - 1}{v_i}\right)^{-1} \\
= & \lim\limits_{n \rightarrow \infty} \prod\limits_{i \in [n^2]: v_i > p} \left(1 + \frac{\R(v_i) - \R(v_{i - 1})}{v_i}\right)^{-1}
= \lim\limits_{n \rightarrow \infty} \prod\limits_{i \in [n^2]: v_i > p} \left[1 + \frac{\left|\R'(v_i)\right|}{v_i} \cdot (v_{i - 1} - v_i)\right]^{-1} \\
= & \exp\left(-\lim\limits_{n \rightarrow \infty} \sum\limits_{i \in [n^2]: v_i > p} \frac{\left|\R'(v_i)\right|}{v_i} \cdot (v_{i - 1} - v_i)\right)
= \exp\left(-\int_{p}^{\infty} \frac{\left|\R'(x)\right|}{x} dx\right) = e^{-\Q(p)},
\end{flalign*}
for all $p \in (\gamma, \infty)$. Similarly, we also have $\lim\limits_{n \rightarrow \infty} \prod\limits_{i \in [n^2]: v_i > p} (1 - q_i) = e^{-\Q(\gamma)}$, for all $p \in (1, \gamma]$. Applying these to Eq~(\ref{eq:lem:opt_cont1}) completes the proof of the lemma.
\end{proof}

\section{Missing Proofs in Section~\ref{subsec:upper2-main} and Section~\ref{subsec:upper2-cstr-opt}}
\label{app:upper2}

\subsection{Proof of Lemma~\ref{lem:constant}}
\label{subapp:upper2-alg}

\paragraph{[Lemma~\ref{lem:constant}].}
\emph{$\frac{1}{32}\epsilon^2 \leq \delta^* \leq \frac{1}{2}\epsilon$, and $1 < v^* \leq u^* \leq \gamma^* \leq \frac{33}{\epsilon^2}$, and $\kappa^* \in (0, 1]$, and $\Delta^* \in (0, \delta^*)$.}

\begin{proof}
Given $\epsilon \in (0, 1)$, for $u^* = \R^{-1}(\epsilon)$, it follows from Lemma~\ref{lem:RQ}.1 and Lemma~\ref{lem:RQ}.3 that
\begin{align}
\label{eq:lem:constant2}
& \frac{1}{\epsilon} = \frac{1}{\R(u^*)} \leq u^* \leq \frac{1}{\R(u^*)} + 1 \leq \frac{2}{\epsilon}, \\
\label{eq:lem:constant3}
& \left|\R'\big((1 + \epsilon) \cdot u^*\big)\right| \geq \left|\R'\big(2u^*\big)\right| \geq \frac{1}{4u^{*2}} \overset{(\ref{eq:lem:constant2})}{\geq} \frac{1}{16}\epsilon^2.
\end{align}
For each $i \in [n]$, clearly $D_i\big((1 + \epsilon) \cdot p\big) \leq D_i^*(p)$ when $p \in (0, \infty)$, due to the definition of $\{D_i^*\}_{i = 1}^n$ in line~(\ref{alg1-preprocess}) of Algorithm~\ref{alg1}. By Fact~\ref{fact2}.2, we also have $F_i\big((1 + \epsilon) \cdot p\big) \leq F_i^*(p)$ when $p \in (0, \infty)$. Thus,
\begin{equation}
\label{eq:lem:constant4}
\Psi\big(p, \{F_i^*\}_{i = 1}^n\big)
\leq \Psi\big((1 + \epsilon) \cdot p, \{F_i\}_{i = 1}^n\big)
\overset{(\diamond)}{\leq} \R\big((1 + \epsilon) \cdot p\big),
\end{equation}
where $(\diamond)$ follows as $\{F_i\}_{i = 1}^n$ is feasible to constraint~(\ref{cstr:ap2}). As a consequence,
\begin{flalign*}
\delta^*
= & \frac{1}{2} \cdot \min \Big\{\R(p) - \Psi\big(p, \{F_i^*\}_{i = 1}^n\big): \,\, p \in (1, u^*]\Big\} \\
\overset{(\ref{eq:lem:constant4})}{\geq} & \frac{1}{2} \cdot \min \Big\{\R(p) - \R\big((1 + \epsilon) \cdot p\big): \,\, p \in (1, u^*]\Big\} \\
= & \frac{1}{2} \cdot \Big(\R(u^*) - \R\big((1 + \epsilon) \cdot u^*\big)\Big) & \text{(By Lemma~\ref{lem:RQ}.1, $\R'(p) \leq \R'\big((1 + \epsilon) \cdot p\big)$)} \\
\geq & \frac{1}{2} \cdot \left|\R'\big((1 + \epsilon) \cdot u^*\big)\right| \cdot \epsilon \cdot u^*
\overset{(\ref{eq:lem:constant2},\ref{eq:lem:constant3})}{\geq} \frac{1}{32}\epsilon^2. & \text{(Since $\R(p)$ is convex and decreasing)}
\end{flalign*}
The monotonicity of $\R(p)$ also implies $\delta^* \leq \frac{1}{2} \cdot \min \big\{\R(p): \,\, p \in (1, u^*]\big\} = \frac{1}{2} \cdot \R(u^*) = \frac{1}{2}\epsilon$.

As for the second claim, it is easy to see $u^* \geq v^* > 1$. Moreover, noticing that $\R(p)$ is strictly decreasing on $p \in (1, \infty)$, and $\R(p) \leq \frac{1}{p - 1}$, we can infer $\frac{33}{\epsilon^2} \geq 1 + \frac{32}{\epsilon^2} \geq \gamma^* \geq u^*$ by observing
\[
\frac{1}{\gamma^* - 1} \geq \R(\gamma^*) = \delta^* \geq \frac{1}{32}\epsilon^2 \quad\quad\quad\quad\quad\quad \R(\gamma^*) = \delta^* \leq \frac{1}{2}\epsilon < \epsilon = \R(u^*).
\]
Besides, $1 \geq \prod\limits_{i = 1}^n (1 - q_i^*) = \kappa^* \geq e^{-\R(v^*)} > 0$, in that
\begin{flalign*}
\R(v^*)
\geq & \R\big((1 + \epsilon) \cdot v^*\big)
\overset{(\ref{eq:lem:constant4})}{\geq} \Psi\big(v^*, \{F_i^*\}_{i = 1}^n\big) & \text{(Since $\R(p)$ is decreasing)} \\
= & \sum\limits_{i = 1}^n \ln\left(1 + \frac{v_i^* q_i^*}{1 - q_i^*}\right)
\geq \sum\limits_{i = 1}^n \ln\left(1 + \frac{q_i^*}{1 - q_i^*}\right)
= -\ln\kappa^*. & \text{(Since $v_i^* \geq v^* > 1$ for all $i \in [n]$)}
\end{flalign*}
Eventually, $0 < \Delta^* = \frac{v^* - 1}{3\gamma^{*2}} \cdot \kappa^* \cdot \delta^* < \delta^*$, in that $\gamma^* \geq v^* > 1$ and $\kappa \in (0, 1]$. This completes the proof of the lemma.
\end{proof}

\subsection{Proofs of Lemma~\ref{lem:obj_approx1} and Lemma~\ref{lem:obj_approx2}}
\label{subapp:upper2-cstr-opt}

\paragraph{[Lemma~\ref{lem:obj_approx1}].}
\emph{$\displaystyle{\int_1^{\infty}} \left(1 - \prod\limits_{i = 1}^n D_i(x)\right) dx - \displaystyle{\int_1^{\infty}} \left(1 - \prod\limits_{i = 1}^n D_i\big((1 + \epsilon) \cdot x\big)\right) dx \leq 3\epsilon$.}

\begin{proof}
By line~(\ref{alg1-preprocess}) of Algorithm~\ref{alg1}, $D_i(p) \leq D_i\big((1 + \epsilon) \cdot p\big)$ for all $p \in (0, \infty)$ and $i \in [n]$. Hence,
\[
\begin{aligned}
\text{LHS}
\leq & \left(\displaystyle{\int_0^1} + \displaystyle{\int_1^{\infty}}\right) \left(1 - \prod\limits_{i = 1}^n D_i(x)\right) dx - \left(\displaystyle{\int_0^1} + \displaystyle{\int_1^{\infty}}\right) \left(1 - \prod\limits_{i = 1}^n D_i\big((1 + \epsilon) \cdot x\big)\right) dx \\
\overset{(\ast)}{=} & \left(1 - \frac{1}{1 + \epsilon}\right) \cdot \displaystyle{\int_0^{\infty}} \left(1 - \prod\limits_{i = 1}^n D_i(x)\right) dx \leq \epsilon \cdot \opt\big(\{F_i\}_{i = 1}^n\big),
\end{aligned}
\]
where $(\ast)$ follows from integration by substitution. Alaei et al.~\cite{AHNPY15} proved $\opt\big(\{F_i\}_{i = 1}^n\big) \leq e \leq 3$. Applying this to the above inequality completes the proof of the lemma.
\end{proof}

\paragraph{[Lemma~\ref{lem:obj_approx2}].}
\emph{$\displaystyle{\int_{u^*}^{\infty}} \left(1 - \prod\limits_{i = 1}^n D_i\big((1 + \epsilon) \cdot x\big)\right) dx \leq \displaystyle{\int_{u^*}^{\infty}} \left(1 - \prod\limits_{i = 1}^n D_i(x)\right) dx \leq 2\epsilon$.}

\begin{proof}
The first inequality follows from the monotonicity of $D_i$'s. As for the second one, note that $1 - \prod\limits_{i = 1}^n z_i \leq \sum\limits_{i = 1}^n (1 - z_i)$ for all $(z_1, z_2, \cdots, z_n) \in [0, 1]^n$ (which can be easily attested by induction),
\begin{equation}
\label{eq:lem:obj_approx2.1}
\displaystyle{\int_{u^*}^{\infty}} \left(1 - \prod\limits_{i = 1}^n D_i(x)\right) dx
\leq \displaystyle{\int_{u^*}^{\infty}} \sum\limits_{i = 1}^n \big(1 - D_i(x)\big) dx
\overset{(\star)}{\leq} \sum\limits_{i = 1}^n u^* \cdot \frac{1 - F_i(u^*)}{F_i(u^*)}
\end{equation}
where in $(\star)$ we apply Lemma~\ref{lem:obj_approx3} in the below. For each $i \in [n]$, we have
\[
\ln\left(1 + u^* \cdot \frac{1 - F_i(u^*)}{F_i(u^*)}\right) = \Psi\big(u^*, \{F_i^*\}\big) \leq \Psi\big(u^*, \{F_i^*\}_{i = 1}^n\big) \overset{(\vee)}{\leq} \R(u^*) = \epsilon \leq 1.
\]
where $(\vee)$ follows as $\{F_i^*\}_{i = 1}^n$ is feasible to constraint~(\ref{cstr:ap2}). Accordingly, $u^* \cdot \frac{1 - F_i(u^*)}{F_i(u^*)} \leq e - 1 \leq 2$ for each $i \in [n]$. Note that $z \leq 2\ln(1 + z)$ for all $z \in [0, 2]$,
\[
\text{LHS of~(\ref{eq:lem:obj_approx2.1})}
\leq 2 \cdot \sum\limits_{i = 1}^n \ln\left(1 + u^* \cdot \frac{1 - F_i(u^*)}{F_i(u^*)}\right)
= 2 \Psi\big(u^*, \{F_i^*\}_{i = 1}^n\big)
\overset{(\circ)}{\leq} 2 \R(u^*) = 2\epsilon,
\]
where $(\circ)$ also follows from the feasibility of $\{F_i^*\}_{i = 1}^n$. This completes the proof of the lemma.
\end{proof}

\begin{lemma}
\label{lem:obj_approx3}
$\displaystyle{\int_p^{\infty}} \big(1 - D_i(x)\big) dx \leq p \cdot \frac{1 - F_i(p)}{F_i(p)}$ for all $p \in (0, \infty)$ and $i \in [n]$.
\end{lemma}

\begin{proof}
Given $p \in (0, \infty)$, since $D_i(p^+) = F_i(\Phi_i^{-1}(p^+))$, we settle the lemma by observing
\begin{flalign*}
\displaystyle{\int_p^{\infty}} \big(1 - D_i(x)\big) dx
= & \left.\Big[(x - p) \cdot \big(1 - F_i(x)\big)\Big]\right|_{x = \Phi_i^{-1}(p^+)} & \text{(By Main Lemma~\ref{lem:virtual_value})} \\
\leq & \left.\left(x \cdot \frac{1 - F_i(x)}{F_i(x)}\right)\right|_{x = \Phi_i^{-1}(p^+)}
\overset{(\wedge)}{\leq} p \cdot \frac{1 - F_i(p)}{F_i(p)},
\end{flalign*}
where $(\wedge)$ follows as $\Phi_i^{-1}(p) \geq p$, and $y = x \cdot \frac{1 - F_i(x)}{F_i(x)}$ is decreasing on $(0, \infty)$ (see Lemma~\ref{lem:potential}).
\end{proof}

\section{Proof of Lemma~\ref{lem:alg_facts}}
\label{app:upper3-lem:alg_facts}
\paragraph{[Lemma~\ref{lem:alg_facts}].}
\emph{In a non-terminal iteration of Algorithm~\ref{alg2}, omit the update in line~(\ref{alg2:update}),
\begin{enumerate}
\item $\gamma^* \geq \gamma \geq \ggamma \geq u \geq \uu$, and $\uu_i = \uu$ for each $i \in [n]$ that $\Delta_i > 0$;
\item $\qq_i \leq q_i$ and $\uu_i \geq \vv_i = v_i = v_i^* \geq v^* > 1$, for each $i \in [n]$;
\item $\Psi\big(p, \{\FF_i\}_{i = 1}^n\big) = \Psi\big(p, \{F_i\}_{i = 1}^n\big) - \Delta$ for all $p \in (0, \uu]$.
\end{enumerate}}

\begin{proof}
By line~(\ref{alg2:continuous}) of Algorithm~\ref{alg2}, $\R(\ggamma) = \R(\gamma) + \Delta \geq \R(\gamma)$. As per this, we can infer $\ggamma \leq \gamma$ from Lemma~\ref{lem:RQ}.1. By induction, we also have $\gamma \leq \gamma^*$. Essentially, $u$ (resp. $\uu$) is the maximum support-supremum, among all non-vanishing $F_i$'s (resp. $\FF_i$'s). Since then, we know
\begin{itemize}
\item $\uu \leq u$ (even though some $F_i$'s vanishes in the current iteration), since $\uu_i \leq u_i$ for all $i \in [n]$;
\item $\uu_i = \uu$ for each $i \in [n]$ that $\Delta_i > 0$. Assuming to the contrary that $\uu_i \neq \uu$, we would see a contradiction that $\Delta_i > 0$, in both of the following two cases:
    \begin{enumerate}
    \item When $\uu_i < \uu$, \emph{this distribution's aggregated potential-decrease is still $0$}, since $\uu$ has never met this distribution's support, in the current and all previous iterations;
    \item Otherwise, i.e., when $\uu_i > \uu$, \emph{this distribution's potential already decreases to $0$, in some previous iteration}, which can be inferred from line~(\ref{alg2:uu}) of Algorithm~\ref{alg2}.
    \end{enumerate}
\end{itemize}
It remains to expose $\ggamma \geq \uu$. For convenience, we would justify this after confirming the rest claims.

For the second claim, it is easy to see $\qq_i \geq q_i$ and $\uu_i \geq \vv_i$, for each $i \in [n]$. By Lemma~\ref{lem:constant} that $v^* = \min\big\{v_j^*: \,\, j \in [n]\big\} > 1$, clearly $v_i^* \geq v^* > 1$ for each $i \in [n]$. Besides, (1)~when $\Delta_i = 0$, we surely have $\FF_i \equiv F_i$ and thus, $\vv_i = v_i$; otherwise, i.e. (2)~when $\Delta_i > 0$, it follows from Lemma~\ref{lem:potential} and line~(\ref{alg3:FF1})/line~(\ref{alg3:FF2}) of Algorithm~\ref{alg3} that
\begin{itemize}
\item $\Psi\big(p, \{\FF_i\}\big)$ is strictly decreasing on $p \in [v_i, \uu] = [v_i, \uu_i] \subset [v_i, u_i]$, which means that $\vv_i \leq v_i$;
\item $\Psi\big(p, \{\FF_i\}\big)$ remains $\ln\left(1 + \frac{v_i q_i}{1 - q_i}\right) - \Delta_i$ on $p \in (0, v_i]$, which means $\vv_i \geq v_i$.
\end{itemize}
Therefore, $\vv_i = v_i$ always holds, whether $\Delta_i = 0$ or not. By induction, $\vv_i = v_i = v_i^*$ for all $i \in [n]$.

Furthermore, the third claim can be easily inferred from line~(\ref{alg3:FF1})/line~(\ref{alg3:FF2}) of Algorithm~\ref{alg3}. For this, we shall emphasize again that $\uu_i = \uu$, for all $i \in [n]$ that $\Delta_i > 0$.

Eventually, to see $\ggamma \geq u$, we shall recall inequality~(\ref{eq:cstr:ap.0.2}) that $\R(\gamma^*) + \Psi\big(p, \{F_i^*\}_{i = 1}^n\big) \leq \R(p) - \delta^*$ for all $p \in (1, u^*]$. In the current iteration, assume w.l.o.g.
\begin{equation}
\label{eq:lem:alg_facts1}
\R(\gamma) + \Psi\big(p, \{F_i\}_{i = 1}^n\big) \leq \R(p) - \delta^* \quad\quad \forall p \in (1, u].
\end{equation}
Particularly, assigning $p \leftarrow u$ in inequality~(\ref{eq:lem:alg_facts1}) results in
\[
\R(u) \overset{(\ref{eq:lem:alg_facts1})}{\geq} \R(\gamma) + \Psi\big(u, \{F_i\}_{i = 1}^n\big) + \delta^* \geq \R(\gamma) + \delta^* \overset{(\vee)}{\geq} \R(\gamma) + \Delta \overset{(\wedge)}{=} \R(\ggamma),
\]
where $(\vee)$ follows as $\delta^* \geq \Delta^* \geq \Delta$ (see Lemma~\ref{lem:constant} and line~(\ref{alg2:Delta}) of Algorithm~\ref{alg2}), and $(\wedge)$ follows from line~(\ref{alg2:continuous}) of Algorithm~\ref{alg2}. As per this, we can infer $\ggamma \geq u$ from Lemma~\ref{lem:RQ}.1. Additionally,
\begin{flalign*}
\R(\ggamma) + \Psi\big(p, \{\FF_i\}_{i = 1}^n\big)
= & \R(\ggamma) + \Psi\big(p, \{F_i\}_{i = 1}^n\big) - \Delta & \text{(By the third claim)} \\
= & \R(\gamma) + \Psi\big(p, \{F_i\}_{i = 1}^n\big)
\overset{(\ref{eq:lem:alg_facts1})}{\leq} \R(p) - \delta^*, & \text{(By line~(\ref{alg2:continuous}) of Algorithm~\ref{alg2})}
\end{flalign*}
for all $p \in (1, \uu] \subset (1, u]$.  This means inequality~(\ref{eq:lem:alg_facts1}) holds in the next iteration, which completes the proof of the lemma.
\end{proof}

\section{Missing Proofs in Section~\ref{subsec:upper3-reg-cstr}}
\label{app:upper3-cstr}
\subsection{Auxiliary Lemma~\ref{lem:tangent}}
\label{subapp:upper3-cstr:lem}

\begin{lemma}
\label{lem:tangent}
Given $a \geq 1$, $b \geq 0$ and $F_{\tan}(p) =
\begin{cases}
1 - \frac{b}{p - a} & p \in (a + b, \infty) \\
0 & p \in (0, a + b]
\end{cases}$. For all $\theta \in (a + b, \infty)$ and $\Delta_k \in \Big(0, \Psi\big(\theta, \{F_{\tan}\}\big)\Big]$, the $\FF_{\tan}$ defined below is regular in the range of $p \in (0, \theta]$:
\[
\Psi\big(p, \{\FF_{\tan}\}\big) \equiv \Psi\big(p, \{F_{\tan}\}\big) - \Delta_k \quad\quad \forall p \in (a + b, \theta].
\]
\end{lemma}

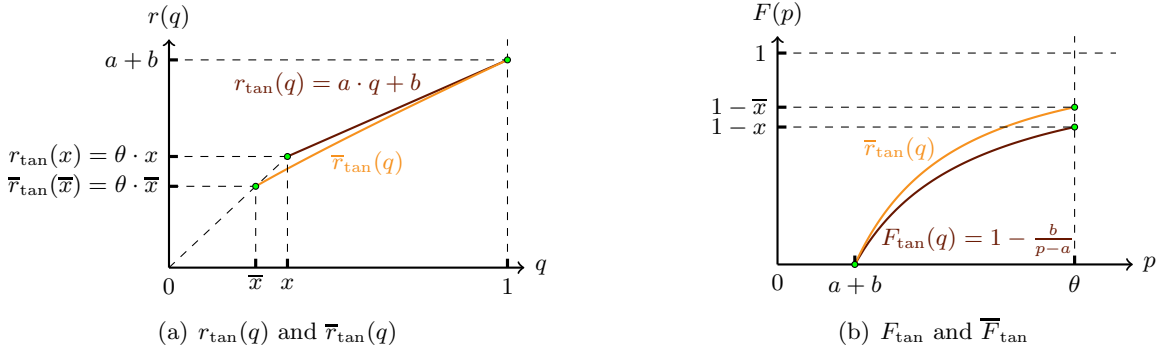
\begin{figure}[H]
\centering
\subfigure[$r_{\tan}(q)$ and $\rr_{\tan}(q)$]{
\begin{tikzpicture}[thick, smooth, scale = 2.25]
\draw[->] (0, 0) -- (2.1, 0);
\draw[->] (0, 0) -- (0, 1.35);
\draw[color = Sepia] (0.7, 0.6572) -- (2, 1.2284);
\draw[color = BurntOrange] (0.5123, 0.4809) -- (0.5142, 0.4821) -- (0.5162, 0.4832) -- (0.5182, 0.4843) -- (0.5202, 0.4855) -- (0.5222, 0.4866) -- (0.5242, 0.4877) -- (0.5262, 0.4888) -- (0.5282, 0.4900) -- (0.5302, 0.4911) -- (0.5322, 0.4922) -- (0.5342, 0.4934) -- (0.5362, 0.4945) -- (0.5382, 0.4956) -- (0.5402, 0.4968) -- (0.5422, 0.4979) -- (0.5442, 0.4990) -- (0.5462, 0.5002) -- (0.5482, 0.5013) -- (0.5502, 0.5024) -- (0.5522, 0.5036) -- (0.5543, 0.5047) -- (0.5563, 0.5058) -- (0.5583, 0.5070) -- (0.5603, 0.5081) -- (0.5623, 0.5092) -- (0.5644, 0.5104) -- (0.5664, 0.5115) -- (0.5684, 0.5126) -- (0.5705, 0.5138) -- (0.5725, 0.5149) -- (0.5745, 0.5160) -- (0.5765, 0.5172) -- (0.5786, 0.5183) -- (0.5806, 0.5194) -- (0.5827, 0.5206) -- (0.5847, 0.5217) -- (0.5867, 0.5228) -- (0.5888, 0.5240) -- (0.5908, 0.5251) -- (0.5929, 0.5263) -- (0.5949, 0.5274) -- (0.5970, 0.5285) -- (0.5990, 0.5297) -- (0.6011, 0.5308) -- (0.6031, 0.5319) -- (0.6052, 0.5331) -- (0.6072, 0.5342) -- (0.6093, 0.5353) -- (0.6114, 0.5365) -- (0.6134, 0.5376) -- (0.6155, 0.5388) -- (0.6176, 0.5399) -- (0.6196, 0.5410) -- (0.6217, 0.5422) -- (0.6238, 0.5433) -- (0.6258, 0.5444) -- (0.6279, 0.5456) -- (0.6300, 0.5467) -- (0.6320, 0.5479) -- (0.6341, 0.5490) -- (0.6362, 0.5501) -- (0.6383, 0.5513) -- (0.6404, 0.5524) -- (0.6424, 0.5536) -- (0.6445, 0.5547) -- (0.6466, 0.5558) -- (0.6487, 0.5570) -- (0.6508, 0.5581) -- (0.6529, 0.5592) -- (0.6550, 0.5604) -- (0.6571, 0.5615) -- (0.6591, 0.5627) -- (0.6612, 0.5638) -- (0.6633, 0.5649) -- (0.6654, 0.5661) -- (0.6675, 0.5672) -- (0.6696, 0.5684) -- (0.6717, 0.5695) -- (0.6738, 0.5707) -- (0.6759, 0.5718) -- (0.6780, 0.5729) -- (0.6801, 0.5741) -- (0.6823, 0.5752) -- (0.6844, 0.5764) -- (0.6865, 0.5775) -- (0.6886, 0.5786) -- (0.6907, 0.5798) -- (0.6928, 0.5809) -- (0.6949, 0.5821) -- (0.6970, 0.5832) -- (0.6992, 0.5843) -- (0.7013, 0.5855) -- (0.7034, 0.5866) -- (0.7055, 0.5878) -- (0.7076, 0.5889) -- (0.7098, 0.5901) -- (0.7119, 0.5912) -- (0.7140, 0.5923) -- (0.7161, 0.5935) -- (0.7183, 0.5946) -- (0.7204, 0.5958) -- (0.7225, 0.5969) -- (0.7247, 0.5981) -- (0.7268, 0.5992) -- (0.7289, 0.6003) -- (0.7311, 0.6015) -- (0.7332, 0.6026) -- (0.7353, 0.6038) -- (0.7375, 0.6049) -- (0.7396, 0.6061) -- (0.7418, 0.6072) -- (0.7439, 0.6083) -- (0.7461, 0.6095) -- (0.7482, 0.6106) -- (0.7503, 0.6118) -- (0.7525, 0.6129) -- (0.7546, 0.6141) -- (0.7568, 0.6152) -- (0.7589, 0.6164) -- (0.7611, 0.6175) -- (0.7633, 0.6186) -- (0.7654, 0.6198) -- (0.7676, 0.6209) -- (0.7697, 0.6221) -- (0.7719, 0.6232) -- (0.7740, 0.6244) -- (0.7762, 0.6255) -- (0.7784, 0.6267) -- (0.7805, 0.6278) -- (0.7827, 0.6290) -- (0.7849, 0.6301) -- (0.7870, 0.6312) -- (0.7892, 0.6324) -- (0.7914, 0.6335) -- (0.7935, 0.6347) -- (0.7957, 0.6358) -- (0.7979, 0.6370) -- (0.8000, 0.6381) -- (0.8022, 0.6393) -- (0.8044, 0.6404) -- (0.8066, 0.6416) -- (0.8087, 0.6427) -- (0.8109, 0.6439) -- (0.8131, 0.6450) -- (0.8153, 0.6461) -- (0.8175, 0.6473) -- (0.8196, 0.6484) -- (0.8218, 0.6496) -- (0.8240, 0.6507) -- (0.8262, 0.6519) -- (0.8284, 0.6530) -- (0.8306, 0.6542) -- (0.8328, 0.6553) -- (0.8349, 0.6565) -- (0.8371, 0.6576) -- (0.8393, 0.6588) -- (0.8415, 0.6599) -- (0.8437, 0.6611) -- (0.8459, 0.6622) -- (0.8481, 0.6634) -- (0.8503, 0.6645) -- (0.8525, 0.6657) -- (0.8547, 0.6668) -- (0.8569, 0.6679) -- (0.8591, 0.6691) -- (0.8613, 0.6702) -- (0.8635, 0.6714) -- (0.8657, 0.6725) -- (0.8679, 0.6737) -- (0.8701, 0.6748) -- (0.8723, 0.6760) -- (0.8745, 0.6771) -- (0.8767, 0.6783) -- (0.8789, 0.6794) -- (0.8812, 0.6806) -- (0.8834, 0.6817) -- (0.8856, 0.6829) -- (0.8878, 0.6840) -- (0.8900, 0.6852) -- (0.8922, 0.6863) -- (0.8944, 0.6875) -- (0.8967, 0.6886) -- (0.8989, 0.6898) -- (0.9011, 0.6909) -- (0.9033, 0.6921) -- (0.9055, 0.6932) -- (0.9078, 0.6944) -- (0.9100, 0.6955) -- (0.9122, 0.6967) -- (0.9144, 0.6978) -- (0.9167, 0.6990) -- (0.9189, 0.7001) -- (0.9211, 0.7013) -- (0.9233, 0.7024) -- (0.9256, 0.7036) -- (0.9278, 0.7047) -- (0.9300, 0.7059) -- (0.9323, 0.7070) -- (0.9345, 0.7082) -- (0.9367, 0.7093) -- (0.9390, 0.7105) -- (0.9412, 0.7116) -- (0.9434, 0.7128) -- (0.9457, 0.7139) -- (0.9479, 0.7151) -- (0.9501, 0.7162) -- (0.9524, 0.7174) -- (0.9546, 0.7185) -- (0.9569, 0.7197) -- (0.9591, 0.7208) -- (0.9613, 0.7220) -- (0.9636, 0.7231) -- (0.9658, 0.7243) -- (0.9681, 0.7254) -- (0.9703, 0.7266) -- (0.9726, 0.7277) -- (0.9748, 0.7289) -- (0.9771, 0.7300) -- (0.9793, 0.7312) -- (0.9816, 0.7323) -- (0.9838, 0.7335) -- (0.9861, 0.7346) -- (0.9883, 0.7358) -- (0.9906, 0.7369) -- (0.9928, 0.7381) -- (0.9951, 0.7392) -- (0.9973, 0.7404) -- (0.9996, 0.7415) -- (1.0019, 0.7427) -- (1.0041, 0.7438) -- (1.0064, 0.7450) -- (1.0086, 0.7461) -- (1.0109, 0.7473) -- (1.0132, 0.7485) -- (1.0154, 0.7496) -- (1.0177, 0.7508) -- (1.0200, 0.7519) -- (1.0222, 0.7531) -- (1.0245, 0.7542) -- (1.0268, 0.7554) -- (1.0290, 0.7565) -- (1.0313, 0.7577) -- (1.0336, 0.7588) -- (1.0358, 0.7600) -- (1.0381, 0.7611) -- (1.0404, 0.7623) -- (1.0426, 0.7634) -- (1.0449, 0.7646) -- (1.0472, 0.7657) -- (1.0495, 0.7669) -- (1.0517, 0.7680) -- (1.0540, 0.7692) -- (1.0563, 0.7703) -- (1.0586, 0.7715) -- (1.0608, 0.7727) -- (1.0631, 0.7738) -- (1.0654, 0.7750) -- (1.0677, 0.7761) -- (1.0700, 0.7773) -- (1.0722, 0.7784) -- (1.0745, 0.7796) -- (1.0768, 0.7807) -- (1.0791, 0.7819) -- (1.0814, 0.7830) -- (1.0837, 0.7842) -- (1.0859, 0.7853) -- (1.0882, 0.7865) -- (1.0905, 0.7876) -- (1.0928, 0.7888) -- (1.0951, 0.7899) -- (1.0974, 0.7911) -- (1.0997, 0.7923) -- (1.1020, 0.7934) -- (1.1043, 0.7946) -- (1.1065, 0.7957) -- (1.1088, 0.7969) -- (1.1111, 0.7980) -- (1.1134, 0.7992) -- (1.1157, 0.8003) -- (1.1180, 0.8015) -- (1.1203, 0.8026) -- (1.1226, 0.8038) -- (1.1249, 0.8049) -- (1.1272, 0.8061) -- (1.1295, 0.8072) -- (1.1318, 0.8084) -- (1.1341, 0.8096) -- (1.1364, 0.8107) -- (1.1387, 0.8119) -- (1.1410, 0.8130) -- (1.1433, 0.8142) -- (1.1456, 0.8153) -- (1.1479, 0.8165) -- (1.1502, 0.8176) -- (1.1525, 0.8188) -- (1.1548, 0.8199) -- (1.1571, 0.8211) -- (1.1595, 0.8222) -- (1.1618, 0.8234) -- (1.1641, 0.8246) -- (1.1664, 0.8257) -- (1.1687, 0.8269) -- (1.1710, 0.8280) -- (1.1733, 0.8292) -- (1.1756, 0.8303) -- (1.1779, 0.8315) -- (1.1802, 0.8326) -- (1.1826, 0.8338) -- (1.1849, 0.8349) -- (1.1872, 0.8361) -- (1.1895, 0.8372) -- (1.1918, 0.8384) -- (1.1941, 0.8396) -- (1.1965, 0.8407) -- (1.1988, 0.8419) -- (1.2011, 0.8430) -- (1.2034, 0.8442) -- (1.2057, 0.8453) -- (1.2081, 0.8465) -- (1.2104, 0.8476) -- (1.2127, 0.8488) -- (1.2150, 0.8499) -- (1.2173, 0.8511) -- (1.2197, 0.8523) -- (1.2220, 0.8534) -- (1.2243, 0.8546) -- (1.2266, 0.8557) -- (1.2290, 0.8569) -- (1.2313, 0.8580) -- (1.2336, 0.8592) -- (1.2359, 0.8603) -- (1.2383, 0.8615) -- (1.2406, 0.8626) -- (1.2429, 0.8638) -- (1.2453, 0.8650) -- (1.2476, 0.8661) -- (1.2499, 0.8673) -- (1.2523, 0.8684) -- (1.2546, 0.8696) -- (1.2569, 0.8707) -- (1.2593, 0.8719) -- (1.2616, 0.8730) -- (1.2639, 0.8742) -- (1.2663, 0.8753) -- (1.2686, 0.8765) -- (1.2709, 0.8777) -- (1.2733, 0.8788) -- (1.2756, 0.8800) -- (1.2779, 0.8811) -- (1.2803, 0.8823) -- (1.2826, 0.8834) -- (1.2850, 0.8846) -- (1.2873, 0.8857) -- (1.2896, 0.8869) -- (1.2920, 0.8880) -- (1.2943, 0.8892) -- (1.2967, 0.8904) -- (1.2990, 0.8915) -- (1.3013, 0.8927) -- (1.3037, 0.8938) -- (1.3060, 0.8950) -- (1.3084, 0.8961) -- (1.3107, 0.8973) -- (1.3131, 0.8984) -- (1.3154, 0.8996) -- (1.3178, 0.9007) -- (1.3201, 0.9019) -- (1.3225, 0.9031) -- (1.3248, 0.9042) -- (1.3272, 0.9054) -- (1.3295, 0.9065) -- (1.3319, 0.9077) -- (1.3342, 0.9088) -- (1.3366, 0.9100) -- (1.3389, 0.9111) -- (1.3413, 0.9123) -- (1.3436, 0.9135) -- (1.3460, 0.9146) -- (1.3483, 0.9158) -- (1.3507, 0.9169) -- (1.3530, 0.9181) -- (1.3554, 0.9192) -- (1.3577, 0.9204) -- (1.3601, 0.9215) -- (1.3624, 0.9227) -- (1.3648, 0.9238) -- (1.3672, 0.9250) -- (1.3695, 0.9262) -- (1.3719, 0.9273) -- (1.3742, 0.9285) -- (1.3766, 0.9296) -- (1.3790, 0.9308) -- (1.3813, 0.9319) -- (1.3837, 0.9331) -- (1.3860, 0.9342) -- (1.3884, 0.9354) -- (1.3908, 0.9366) -- (1.3931, 0.9377) -- (1.3955, 0.9389) -- (1.3978, 0.9400) -- (1.4002, 0.9412) -- (1.4026, 0.9423) -- (1.4049, 0.9435) -- (1.4073, 0.9446) -- (1.4097, 0.9458) -- (1.4120, 0.9469) -- (1.4144, 0.9481) -- (1.4168, 0.9493) -- (1.4191, 0.9504) -- (1.4215, 0.9516) -- (1.4239, 0.9527) -- (1.4262, 0.9539) -- (1.4286, 0.9550) -- (1.4310, 0.9562) -- (1.4334, 0.9573) -- (1.4357, 0.9585) -- (1.4381, 0.9597) -- (1.4405, 0.9608) -- (1.4428, 0.9620) -- (1.4452, 0.9631) -- (1.4476, 0.9643) -- (1.4500, 0.9654) -- (1.4523, 0.9666) -- (1.4547, 0.9677) -- (1.4571, 0.9689) -- (1.4595, 0.9701) -- (1.4618, 0.9712) -- (1.4642, 0.9724) -- (1.4666, 0.9735) -- (1.4690, 0.9747) -- (1.4713, 0.9758) -- (1.4737, 0.9770) -- (1.4761, 0.9781) -- (1.4785, 0.9793) -- (1.4808, 0.9804) -- (1.4832, 0.9816) -- (1.4856, 0.9828) -- (1.4880, 0.9839) -- (1.4904, 0.9851) -- (1.4927, 0.9862) -- (1.4951, 0.9874) -- (1.4975, 0.9885) -- (1.4999, 0.9897) -- (1.5023, 0.9908) -- (1.5047, 0.9920) -- (1.5070, 0.9932) -- (1.5094, 0.9943) -- (1.5118, 0.9955) -- (1.5142, 0.9966) -- (1.5166, 0.9978) -- (1.5190, 0.9989) -- (1.5213, 1.0001) -- (1.5237, 1.0012) -- (1.5261, 1.0024) -- (1.5285, 1.0035) -- (1.5309, 1.0047) -- (1.5333, 1.0059) -- (1.5357, 1.0070) -- (1.5381, 1.0082) -- (1.5404, 1.0093) -- (1.5428, 1.0105) -- (1.5452, 1.0116) -- (1.5476, 1.0128) -- (1.5500, 1.0139) -- (1.5524, 1.0151) -- (1.5548, 1.0162) -- (1.5572, 1.0174) -- (1.5596, 1.0186) -- (1.5620, 1.0197) -- (1.5644, 1.0209) -- (1.5667, 1.0220) -- (1.5691, 1.0232) -- (1.5715, 1.0243) -- (1.5739, 1.0255) -- (1.5763, 1.0266) -- (1.5787, 1.0278) -- (1.5811, 1.0289) -- (1.5835, 1.0301) -- (1.5859, 1.0313) -- (1.5883, 1.0324) -- (1.5907, 1.0336) -- (1.5931, 1.0347) -- (1.5955, 1.0359) -- (1.5979, 1.0370) -- (1.6003, 1.0382) -- (1.6027, 1.0393) -- (1.6051, 1.0405) -- (1.6075, 1.0416) -- (1.6099, 1.0428) -- (1.6123, 1.0440) -- (1.6147, 1.0451) -- (1.6171, 1.0463) -- (1.6195, 1.0474) -- (1.6219, 1.0486) -- (1.6243, 1.0497) -- (1.6267, 1.0509) -- (1.6291, 1.0520) -- (1.6315, 1.0532) -- (1.6339, 1.0543) -- (1.6363, 1.0555) -- (1.6387, 1.0567) -- (1.6411, 1.0578) -- (1.6435, 1.0590) -- (1.6459, 1.0601) -- (1.6483, 1.0613) -- (1.6507, 1.0624) -- (1.6531, 1.0636) -- (1.6555, 1.0647) -- (1.6579, 1.0659) -- (1.6603, 1.0670) -- (1.6628, 1.0682) -- (1.6652, 1.0694) -- (1.6676, 1.0705) -- (1.6700, 1.0717) -- (1.6724, 1.0728) -- (1.6748, 1.0740) -- (1.6772, 1.0751) -- (1.6796, 1.0763) -- (1.6820, 1.0774) -- (1.6844, 1.0786) -- (1.6868, 1.0797) -- (1.6892, 1.0809) -- (1.6917, 1.0820) -- (1.6941, 1.0832) -- (1.6965, 1.0844) -- (1.6989, 1.0855) -- (1.7013, 1.0867) -- (1.7037, 1.0878) -- (1.7061, 1.0890) -- (1.7085, 1.0901) -- (1.7110, 1.0913) -- (1.7134, 1.0924) -- (1.7158, 1.0936) -- (1.7182, 1.0947) -- (1.7206, 1.0959) -- (1.7230, 1.0970) -- (1.7254, 1.0982) -- (1.7279, 1.0994) -- (1.7303, 1.1005) -- (1.7327, 1.1017) -- (1.7351, 1.1028) -- (1.7375, 1.1040) -- (1.7399, 1.1051) -- (1.7423, 1.1063) -- (1.7448, 1.1074) -- (1.7472, 1.1086) -- (1.7496, 1.1097) -- (1.7520, 1.1109) -- (1.7544, 1.1120) -- (1.7569, 1.1132) -- (1.7593, 1.1144) -- (1.7617, 1.1155) -- (1.7641, 1.1167) -- (1.7665, 1.1178) -- (1.7690, 1.1190) -- (1.7714, 1.1201) -- (1.7738, 1.1213) -- (1.7762, 1.1224) -- (1.7786, 1.1236) -- (1.7811, 1.1247) -- (1.7835, 1.1259) -- (1.7859, 1.1270) -- (1.7883, 1.1282) -- (1.7907, 1.1293) -- (1.7932, 1.1305) -- (1.7956, 1.1316) -- (1.7980, 1.1328) -- (1.8004, 1.1340) -- (1.8029, 1.1351) -- (1.8053, 1.1363) -- (1.8077, 1.1374) -- (1.8101, 1.1386) -- (1.8126, 1.1397) -- (1.8150, 1.1409) -- (1.8174, 1.1420) -- (1.8198, 1.1432) -- (1.8223, 1.1443) -- (1.8247, 1.1455) -- (1.8271, 1.1466) -- (1.8295, 1.1478) -- (1.8320, 1.1489) -- (1.8344, 1.1501) -- (1.8368, 1.1512) -- (1.8392, 1.1524) -- (1.8417, 1.1536) -- (1.8441, 1.1547) -- (1.8465, 1.1559) -- (1.8490, 1.1570) -- (1.8514, 1.1582) -- (1.8538, 1.1593) -- (1.8563, 1.1605) -- (1.8587, 1.1616) -- (1.8611, 1.1628) -- (1.8635, 1.1639) -- (1.8660, 1.1651) -- (1.8684, 1.1662) -- (1.8708, 1.1674) -- (1.8733, 1.1685) -- (1.8757, 1.1697) -- (1.8781, 1.1708) -- (1.8806, 1.1720) -- (1.8830, 1.1731) -- (1.8854, 1.1743) -- (1.8879, 1.1754) -- (1.8903, 1.1766) -- (1.8927, 1.1778) -- (1.8952, 1.1789) -- (1.8976, 1.1801) -- (1.9000, 1.1812) -- (1.9025, 1.1824) -- (1.9049, 1.1835) -- (1.9073, 1.1847) -- (1.9098, 1.1858) -- (1.9122, 1.1870) -- (1.9146, 1.1881) -- (1.9171, 1.1893) -- (1.9195, 1.1904) -- (1.9219, 1.1916) -- (1.9244, 1.1927) -- (1.9268, 1.1939) -- (1.9292, 1.1950) -- (1.9317, 1.1962) -- (1.9341, 1.1973) -- (1.9366, 1.1985) -- (1.9390, 1.1996) -- (1.9414, 1.2008) -- (1.9439, 1.2019) -- (1.9463, 1.2031) -- (1.9487, 1.2042) -- (1.9512, 1.2054) -- (1.9536, 1.2065) -- (1.9561, 1.2077) -- (1.9585, 1.2088) -- (1.9609, 1.2100) -- (1.9634, 1.2111) -- (1.9658, 1.2123) -- (1.9683, 1.2135) -- (1.9707, 1.2146) -- (1.9731, 1.2158) -- (1.9756, 1.2169) -- (1.9780, 1.2181) -- (1.9805, 1.2192) -- (1.9829, 1.2204) -- (1.9853, 1.2215) -- (1.9878, 1.2227) -- (1.9902, 1.2238) -- (1.9927, 1.2250) -- (1.9951, 1.2261) -- (1.9976, 1.2273) -- (2.0000, 1.2284);

\node[above] at (0, 1.35) {\footnotesize $r(q)$};
\node[right] at (2.1, 0) {\footnotesize $q$};
\node[below] at (0, 0) {\footnotesize $0$};
\draw[very thick] (2, 0) -- (2, 1.5pt);
\node[below] at (2, 0) {\footnotesize $1$};
\draw[thin, dashed] (0, 0) -- (0.7, 0.6572);

\draw[very thick] (0.7, 0) -- (0.7, 1.5pt);
\node[below] at (0.7, 0) {\footnotesize $x$};
\node[left] at (0, 0.6572) {\footnotesize $r_{\tan}(x) = \theta \cdot x$};
\draw[thin, dashed] (0.7, 0) -- (0.7, 0.6572) -- (0, 0.6572);

\draw[very thick] (0.5123, 0) -- (0.5123, 1.5pt);
\node[below] at (0.5123, 0.025) {\footnotesize $\overline{x}$};
\node[left] at (0, 0.4809) {\footnotesize $\rr_{\tan}(\overline{x}) = \theta \cdot \overline{x}$};
\draw[thin, dashed] (0.5123, 0) -- (0.5123, 0.4809) -- (0, 0.4809);

\draw[very thick] (0, 1.2284) -- (1.5pt, 1.2284);
\node[left] at (0, 1.2284) {\footnotesize $a + b$};
\draw[thin, dashed] (0, 1.2284) -- (2, 1.2284);

\draw[thin, dashed] (2, 0) -- (2, 1.35);
\node[anchor = south east, color = Sepia] at (1.55, 0.95) {\footnotesize $r_{\tan}(q) = a \cdot q + b$};
\node[anchor = north west, color = BurntOrange] at (0.9, 0.75) {\footnotesize $\rr_{\tan}(q)$};

\draw[thin, fill = green] (0.7, 0.6572) circle(0.5pt);
\draw[thin, fill = green] (0.5123, 0.4809) circle(0.5pt);
\draw[thin, fill = green] (2, 1.2284) circle(0.5pt);

\draw[very thick] (0, 0.4809) -- (1.5pt, 0.4809);
\draw[very thick] (0, 0.6572) -- (1.5pt, 0.6572);
\end{tikzpicture}
\label{fig:lem:tangent1}
}
\quad\quad\quad\quad
\subfigure[$F_{\tan}$ and $\FF_{\tan}$]{
\begin{tikzpicture}[thick, smooth, scale = 2.25]
\draw[->] (0, 0) -- (2.1, 0);
\draw[->] (0, 0) -- (0, 1.35);
\draw[color = Sepia, domain = 0.4568: 1.7554] plot (\x, {1.25 * (1 - 0.69925 / (\x + 0.24245))});
\draw[color = BurntOrange] (1.7554, 0.9298) -- (1.7497, 0.9286) -- (1.7440, 0.9274) -- (1.7384, 0.9261) -- (1.7328, 0.9249) -- (1.7273, 0.9236) -- (1.7217, 0.9224) -- (1.7162, 0.9211) -- (1.7108, 0.9199) -- (1.7053, 0.9186) -- (1.6999, 0.9174) -- (1.6945, 0.9161) -- (1.6892, 0.9149) -- (1.6839, 0.9136) -- (1.6786, 0.9124) -- (1.6733, 0.9111) -- (1.6681, 0.9099) -- (1.6629, 0.9086) -- (1.6577, 0.9074) -- (1.6525, 0.9061) -- (1.6474, 0.9048) -- (1.6423, 0.9036) -- (1.6373, 0.9023) -- (1.6322, 0.9011) -- (1.6272, 0.8998) -- (1.6222, 0.8985) -- (1.6173, 0.8973) -- (1.6123, 0.8960) -- (1.6074, 0.8947) -- (1.6025, 0.8935) -- (1.5977, 0.8922) -- (1.5928, 0.8909) -- (1.5880, 0.8897) -- (1.5833, 0.8884) -- (1.5785, 0.8871) -- (1.5738, 0.8858) -- (1.5691, 0.8846) -- (1.5644, 0.8833) -- (1.5597, 0.8820) -- (1.5551, 0.8807) -- (1.5505, 0.8795) -- (1.5459, 0.8782) -- (1.5413, 0.8769) -- (1.5368, 0.8756) -- (1.5323, 0.8743) -- (1.5278, 0.8730) -- (1.5233, 0.8718) -- (1.5189, 0.8705) -- (1.5145, 0.8692) -- (1.5101, 0.8679) -- (1.5057, 0.8666) -- (1.5013, 0.8653) -- (1.4970, 0.8640) -- (1.4927, 0.8627) -- (1.4884, 0.8614) -- (1.4841, 0.8602) -- (1.4798, 0.8589) -- (1.4756, 0.8576) -- (1.4714, 0.8563) -- (1.4672, 0.8550) -- (1.4630, 0.8537) -- (1.4589, 0.8524) -- (1.4548, 0.8511) -- (1.4507, 0.8498) -- (1.4466, 0.8485) -- (1.4425, 0.8472) -- (1.4384, 0.8459) -- (1.4344, 0.8446) -- (1.4304, 0.8433) -- (1.4264, 0.8420) -- (1.4224, 0.8407) -- (1.4185, 0.8393) -- (1.4145, 0.8380) -- (1.4106, 0.8367) -- (1.4067, 0.8354) -- (1.4028, 0.8341) -- (1.3990, 0.8328) -- (1.3951, 0.8315) -- (1.3913, 0.8302) -- (1.3875, 0.8289) -- (1.3837, 0.8275) -- (1.3799, 0.8262) -- (1.3762, 0.8249) -- (1.3724, 0.8236) -- (1.3687, 0.8223) -- (1.3650, 0.8210) -- (1.3613, 0.8196) -- (1.3577, 0.8183) -- (1.3540, 0.8170) -- (1.3504, 0.8157) -- (1.3468, 0.8144) -- (1.3432, 0.8130) -- (1.3396, 0.8117) -- (1.3360, 0.8104) -- (1.3324, 0.8091) -- (1.3289, 0.8077) -- (1.3254, 0.8064) -- (1.3219, 0.8051) -- (1.3184, 0.8037) -- (1.3149, 0.8024) -- (1.3114, 0.8011) -- (1.3080, 0.7998) -- (1.3046, 0.7984) -- (1.3012, 0.7971) -- (1.2978, 0.7958) -- (1.2944, 0.7944) -- (1.2910, 0.7931) -- (1.2876, 0.7917) -- (1.2843, 0.7904) -- (1.2810, 0.7891) -- (1.2777, 0.7877) -- (1.2744, 0.7864) -- (1.2711, 0.7851) -- (1.2678, 0.7837) -- (1.2646, 0.7824) -- (1.2613, 0.7810) -- (1.2581, 0.7797) -- (1.2549, 0.7783) -- (1.2517, 0.7770) -- (1.2485, 0.7757) -- (1.2453, 0.7743) -- (1.2422, 0.7730) -- (1.2390, 0.7716) -- (1.2359, 0.7703) -- (1.2328, 0.7689) -- (1.2297, 0.7676) -- (1.2266, 0.7662) -- (1.2235, 0.7649) -- (1.2204, 0.7635) -- (1.2174, 0.7622) -- (1.2143, 0.7608) -- (1.2113, 0.7595) -- (1.2083, 0.7581) -- (1.2053, 0.7568) -- (1.2023, 0.7554) -- (1.1993, 0.7540) -- (1.1963, 0.7527) -- (1.1934, 0.7513) -- (1.1904, 0.7500) -- (1.1875, 0.7486) -- (1.1846, 0.7473) -- (1.1817, 0.7459) -- (1.1788, 0.7445) -- (1.1759, 0.7432) -- (1.1730, 0.7418) -- (1.1702, 0.7405) -- (1.1673, 0.7391) -- (1.1645, 0.7377) -- (1.1617, 0.7364) -- (1.1589, 0.7350) -- (1.1561, 0.7336) -- (1.1533, 0.7323) -- (1.1505, 0.7309) -- (1.1477, 0.7295) -- (1.1450, 0.7282) -- (1.1422, 0.7268) -- (1.1395, 0.7254) -- (1.1368, 0.7240) -- (1.1340, 0.7227) -- (1.1313, 0.7213) -- (1.1286, 0.7199) -- (1.1260, 0.7186) -- (1.1233, 0.7172) -- (1.1206, 0.7158) -- (1.1180, 0.7144) -- (1.1153, 0.7131) -- (1.1127, 0.7117) -- (1.1101, 0.7103) -- (1.1075, 0.7089) -- (1.1049, 0.7076) -- (1.1023, 0.7062) -- (1.0997, 0.7048) -- (1.0971, 0.7034) -- (1.0946, 0.7020) -- (1.0920, 0.7007) -- (1.0895, 0.6993) -- (1.0869, 0.6979) -- (1.0844, 0.6965) -- (1.0819, 0.6951) -- (1.0794, 0.6937) -- (1.0769, 0.6924) -- (1.0744, 0.6910) -- (1.0719, 0.6896) -- (1.0695, 0.6882) -- (1.0670, 0.6868) -- (1.0646, 0.6854) -- (1.0621, 0.6840) -- (1.0597, 0.6827) -- (1.0573, 0.6813) -- (1.0549, 0.6799) -- (1.0525, 0.6785) -- (1.0501, 0.6771) -- (1.0477, 0.6757) -- (1.0453, 0.6743) -- (1.0429, 0.6729) -- (1.0406, 0.6715) -- (1.0382, 0.6701) -- (1.0359, 0.6687) -- (1.0336, 0.6673) -- (1.0312, 0.6659) -- (1.0289, 0.6645) -- (1.0266, 0.6632) -- (1.0243, 0.6618) -- (1.0220, 0.6604) -- (1.0197, 0.6590) -- (1.0175, 0.6576) -- (1.0152, 0.6562) -- (1.0130, 0.6548) -- (1.0107, 0.6534) -- (1.0085, 0.6520) -- (1.0062, 0.6506) -- (1.0040, 0.6492) -- (1.0018, 0.6478) -- (0.9996, 0.6464) -- (0.9974, 0.6449) -- (0.9952, 0.6435) -- (0.9930, 0.6421) -- (0.9908, 0.6407) -- (0.9886, 0.6393) -- (0.9865, 0.6379) -- (0.9843, 0.6365) -- (0.9822, 0.6351) -- (0.9800, 0.6337) -- (0.9779, 0.6323) -- (0.9758, 0.6309) -- (0.9737, 0.6295) -- (0.9715, 0.6281) -- (0.9694, 0.6267) -- (0.9673, 0.6252) -- (0.9653, 0.6238) -- (0.9632, 0.6224) -- (0.9611, 0.6210) -- (0.9590, 0.6196) -- (0.9570, 0.6182) -- (0.9549, 0.6168) -- (0.9529, 0.6154) -- (0.9508, 0.6139) -- (0.9488, 0.6125) -- (0.9468, 0.6111) -- (0.9447, 0.6097) -- (0.9427, 0.6083) -- (0.9407, 0.6069) -- (0.9387, 0.6054) -- (0.9367, 0.6040) -- (0.9348, 0.6026) -- (0.9328, 0.6012) -- (0.9308, 0.5998) -- (0.9288, 0.5984) -- (0.9269, 0.5969) -- (0.9249, 0.5955) -- (0.9230, 0.5941) -- (0.9210, 0.5927) -- (0.9191, 0.5912) -- (0.9172, 0.5898) -- (0.9153, 0.5884) -- (0.9134, 0.5870) -- (0.9114, 0.5856) -- (0.9095, 0.5841) -- (0.9076, 0.5827) -- (0.9058, 0.5813) -- (0.9039, 0.5798) -- (0.9020, 0.5784) -- (0.9001, 0.5770) -- (0.8983, 0.5756) -- (0.8964, 0.5741) -- (0.8946, 0.5727) -- (0.8927, 0.5713) -- (0.8909, 0.5699) -- (0.8890, 0.5684) -- (0.8872, 0.5670) -- (0.8854, 0.5656) -- (0.8836, 0.5641) -- (0.8818, 0.5627) -- (0.8800, 0.5613) -- (0.8782, 0.5598) -- (0.8764, 0.5584) -- (0.8746, 0.5570) -- (0.8728, 0.5555) -- (0.8710, 0.5541) -- (0.8693, 0.5527) -- (0.8675, 0.5512) -- (0.8657, 0.5498) -- (0.8640, 0.5484) -- (0.8622, 0.5469) -- (0.8605, 0.5455) -- (0.8587, 0.5441) -- (0.8570, 0.5426) -- (0.8553, 0.5412) -- (0.8536, 0.5397) -- (0.8519, 0.5383) -- (0.8501, 0.5369) -- (0.8484, 0.5354) -- (0.8467, 0.5340) -- (0.8450, 0.5325) -- (0.8434, 0.5311) -- (0.8417, 0.5297) -- (0.8400, 0.5282) -- (0.8383, 0.5268) -- (0.8367, 0.5253) -- (0.8350, 0.5239) -- (0.8333, 0.5225) -- (0.8317, 0.5210) -- (0.8300, 0.5196) -- (0.8284, 0.5181) -- (0.8268, 0.5167) -- (0.8251, 0.5152) -- (0.8235, 0.5138) -- (0.8219, 0.5123) -- (0.8203, 0.5109) -- (0.8186, 0.5095) -- (0.8170, 0.5080) -- (0.8154, 0.5066) -- (0.8138, 0.5051) -- (0.8122, 0.5037) -- (0.8107, 0.5022) -- (0.8091, 0.5008) -- (0.8075, 0.4993) -- (0.8059, 0.4979) -- (0.8044, 0.4964) -- (0.8028, 0.4950) -- (0.8012, 0.4935) -- (0.7997, 0.4921) -- (0.7981, 0.4906) -- (0.7966, 0.4892) -- (0.7950, 0.4877) -- (0.7935, 0.4863) -- (0.7920, 0.4848) -- (0.7904, 0.4834) -- (0.7889, 0.4819) -- (0.7874, 0.4804) -- (0.7859, 0.4790) -- (0.7844, 0.4775) -- (0.7829, 0.4761) -- (0.7814, 0.4746) -- (0.7799, 0.4732) -- (0.7784, 0.4717) -- (0.7769, 0.4703) -- (0.7754, 0.4688) -- (0.7739, 0.4673) -- (0.7724, 0.4659) -- (0.7710, 0.4644) -- (0.7695, 0.4630) -- (0.7680, 0.4615) -- (0.7666, 0.4600) -- (0.7651, 0.4586) -- (0.7637, 0.4571) -- (0.7622, 0.4557) -- (0.7608, 0.4542) -- (0.7594, 0.4527) -- (0.7579, 0.4513) -- (0.7565, 0.4498) -- (0.7551, 0.4484) -- (0.7537, 0.4469) -- (0.7522, 0.4454) -- (0.7508, 0.4440) -- (0.7494, 0.4425) -- (0.7480, 0.4411) -- (0.7466, 0.4396) -- (0.7452, 0.4381) -- (0.7438, 0.4367) -- (0.7424, 0.4352) -- (0.7410, 0.4337) -- (0.7397, 0.4323) -- (0.7383, 0.4308) -- (0.7369, 0.4293) -- (0.7355, 0.4279) -- (0.7342, 0.4264) -- (0.7328, 0.4249) -- (0.7315, 0.4235) -- (0.7301, 0.4220) -- (0.7288, 0.4205) -- (0.7274, 0.4191) -- (0.7261, 0.4176) -- (0.7247, 0.4161) -- (0.7234, 0.4147) -- (0.7221, 0.4132) -- (0.7207, 0.4117) -- (0.7194, 0.4102) -- (0.7181, 0.4088) -- (0.7168, 0.4073) -- (0.7154, 0.4058) -- (0.7141, 0.4044) -- (0.7128, 0.4029) -- (0.7115, 0.4014) -- (0.7102, 0.3999) -- (0.7089, 0.3985) -- (0.7076, 0.3970) -- (0.7064, 0.3955) -- (0.7051, 0.3941) -- (0.7038, 0.3926) -- (0.7025, 0.3911) -- (0.7012, 0.3896) -- (0.7000, 0.3882) -- (0.6987, 0.3867) -- (0.6974, 0.3852) -- (0.6962, 0.3837) -- (0.6949, 0.3823) -- (0.6937, 0.3808) -- (0.6924, 0.3793) -- (0.6912, 0.3778) -- (0.6899, 0.3763) -- (0.6887, 0.3749) -- (0.6874, 0.3734) -- (0.6862, 0.3719) -- (0.6850, 0.3704) -- (0.6837, 0.3690) -- (0.6825, 0.3675) -- (0.6813, 0.3660) -- (0.6801, 0.3645) -- (0.6789, 0.3630) -- (0.6776, 0.3616) -- (0.6764, 0.3601) -- (0.6752, 0.3586) -- (0.6740, 0.3571) -- (0.6728, 0.3556) -- (0.6716, 0.3542) -- (0.6704, 0.3527) -- (0.6692, 0.3512) -- (0.6681, 0.3497) -- (0.6669, 0.3482) -- (0.6657, 0.3467) -- (0.6645, 0.3453) -- (0.6633, 0.3438) -- (0.6622, 0.3423) -- (0.6610, 0.3408) -- (0.6598, 0.3393) -- (0.6587, 0.3378) -- (0.6575, 0.3364) -- (0.6564, 0.3349) -- (0.6552, 0.3334) -- (0.6540, 0.3319) -- (0.6529, 0.3304) -- (0.6518, 0.3289) -- (0.6506, 0.3274) -- (0.6495, 0.3260) -- (0.6483, 0.3245) -- (0.6472, 0.3230) -- (0.6461, 0.3215) -- (0.6449, 0.3200) -- (0.6438, 0.3185) -- (0.6427, 0.3170) -- (0.6416, 0.3155) -- (0.6405, 0.3141) -- (0.6394, 0.3126) -- (0.6382, 0.3111) -- (0.6371, 0.3096) -- (0.6360, 0.3081) -- (0.6349, 0.3066) -- (0.6338, 0.3051) -- (0.6327, 0.3036) -- (0.6316, 0.3021) -- (0.6305, 0.3006) -- (0.6295, 0.2992) -- (0.6284, 0.2977) -- (0.6273, 0.2962) -- (0.6262, 0.2947) -- (0.6251, 0.2932) -- (0.6241, 0.2917) -- (0.6230, 0.2902) -- (0.6219, 0.2887) -- (0.6208, 0.2872) -- (0.6198, 0.2857) -- (0.6187, 0.2842) -- (0.6177, 0.2827) -- (0.6166, 0.2812) -- (0.6156, 0.2798) -- (0.6145, 0.2783) -- (0.6135, 0.2768) -- (0.6124, 0.2753) -- (0.6114, 0.2738) -- (0.6103, 0.2723) -- (0.6093, 0.2708) -- (0.6082, 0.2693) -- (0.6072, 0.2678) -- (0.6062, 0.2663) -- (0.6052, 0.2648) -- (0.6041, 0.2633) -- (0.6031, 0.2618) -- (0.6021, 0.2603) -- (0.6011, 0.2588) -- (0.6000, 0.2573) -- (0.5990, 0.2558) -- (0.5980, 0.2543) -- (0.5970, 0.2528) -- (0.5960, 0.2513) -- (0.5950, 0.2498) -- (0.5940, 0.2483) -- (0.5930, 0.2468) -- (0.5920, 0.2453) -- (0.5910, 0.2438) -- (0.5900, 0.2423) -- (0.5890, 0.2408) -- (0.5880, 0.2393) -- (0.5871, 0.2378) -- (0.5861, 0.2363) -- (0.5851, 0.2348) -- (0.5841, 0.2333) -- (0.5831, 0.2318) -- (0.5822, 0.2303) -- (0.5812, 0.2288) -- (0.5802, 0.2273) -- (0.5793, 0.2258) -- (0.5783, 0.2243) -- (0.5773, 0.2228) -- (0.5764, 0.2213) -- (0.5754, 0.2198) -- (0.5745, 0.2183) -- (0.5735, 0.2168) -- (0.5726, 0.2153) -- (0.5716, 0.2138) -- (0.5707, 0.2123) -- (0.5697, 0.2108) -- (0.5688, 0.2093) -- (0.5678, 0.2078) -- (0.5669, 0.2063) -- (0.5660, 0.2048) -- (0.5650, 0.2033) -- (0.5641, 0.2018) -- (0.5632, 0.2002) -- (0.5622, 0.1987) -- (0.5613, 0.1972) -- (0.5604, 0.1957) -- (0.5595, 0.1942) -- (0.5586, 0.1927) -- (0.5576, 0.1912) -- (0.5567, 0.1897) -- (0.5558, 0.1882) -- (0.5549, 0.1867) -- (0.5540, 0.1852) -- (0.5531, 0.1837) -- (0.5522, 0.1822) -- (0.5513, 0.1807) -- (0.5504, 0.1791) -- (0.5495, 0.1776) -- (0.5486, 0.1761) -- (0.5477, 0.1746) -- (0.5468, 0.1731) -- (0.5459, 0.1716) -- (0.5450, 0.1701) -- (0.5441, 0.1686) -- (0.5433, 0.1671) -- (0.5424, 0.1656) -- (0.5415, 0.1641) -- (0.5406, 0.1625) -- (0.5397, 0.1610) -- (0.5389, 0.1595) -- (0.5380, 0.1580) -- (0.5371, 0.1565) -- (0.5363, 0.1550) -- (0.5354, 0.1535) -- (0.5345, 0.1520) -- (0.5337, 0.1505) -- (0.5328, 0.1489) -- (0.5319, 0.1474) -- (0.5311, 0.1459) -- (0.5302, 0.1444) -- (0.5294, 0.1429) -- (0.5285, 0.1414) -- (0.5277, 0.1399) -- (0.5268, 0.1384) -- (0.5260, 0.1368) -- (0.5251, 0.1353) -- (0.5243, 0.1338) -- (0.5235, 0.1323) -- (0.5226, 0.1308) -- (0.5218, 0.1293) -- (0.5210, 0.1278) -- (0.5201, 0.1262) -- (0.5193, 0.1247) -- (0.5185, 0.1232) -- (0.5176, 0.1217) -- (0.5168, 0.1202) -- (0.5160, 0.1187) -- (0.5152, 0.1172) -- (0.5143, 0.1156) -- (0.5135, 0.1141) -- (0.5127, 0.1126) -- (0.5119, 0.1111) -- (0.5111, 0.1096) -- (0.5103, 0.1081) -- (0.5095, 0.1065) -- (0.5087, 0.1050) -- (0.5078, 0.1035) -- (0.5070, 0.1020) -- (0.5062, 0.1005) -- (0.5054, 0.0990) -- (0.5046, 0.0974) -- (0.5038, 0.0959) -- (0.5030, 0.0944) -- (0.5023, 0.0929) -- (0.5015, 0.0914) -- (0.5007, 0.0898) -- (0.4999, 0.0883) -- (0.4991, 0.0868) -- (0.4983, 0.0853) -- (0.4975, 0.0838) -- (0.4967, 0.0822) -- (0.4960, 0.0807) -- (0.4952, 0.0792) -- (0.4944, 0.0777) -- (0.4936, 0.0762) -- (0.4929, 0.0747) -- (0.4921, 0.0731) -- (0.4913, 0.0716) -- (0.4905, 0.0701) -- (0.4898, 0.0686) -- (0.4890, 0.0670) -- (0.4882, 0.0655) -- (0.4875, 0.0640) -- (0.4867, 0.0625) -- (0.4860, 0.0610) -- (0.4852, 0.0594) -- (0.4845, 0.0579) -- (0.4837, 0.0564) -- (0.4829, 0.0549) -- (0.4822, 0.0534) -- (0.4814, 0.0518) -- (0.4807, 0.0503) -- (0.4799, 0.0488) -- (0.4792, 0.0473) -- (0.4785, 0.0457) -- (0.4777, 0.0442) -- (0.4770, 0.0427) -- (0.4762, 0.0412) -- (0.4755, 0.0397) -- (0.4748, 0.0381) -- (0.4740, 0.0366) -- (0.4733, 0.0351) -- (0.4726, 0.0336) -- (0.4718, 0.0320) -- (0.4711, 0.0305) -- (0.4704, 0.0290) -- (0.4696, 0.0275) -- (0.4689, 0.0259) -- (0.4682, 0.0244) -- (0.4675, 0.0229) -- (0.4668, 0.0214) -- (0.4660, 0.0198) -- (0.4653, 0.0183) -- (0.4646, 0.0168) -- (0.4639, 0.0153) -- (0.4632, 0.0137) -- (0.4625, 0.0122) -- (0.4618, 0.0107) -- (0.4611, 0.0092) -- (0.4603, 0.0076) -- (0.4596, 0.0061) -- (0.4589, 0.0046) -- (0.4582, 0.0031) -- (0.4575, 0.0015) -- (0.4568, 0.0000);

\node[above] at (0, 1.35) {\footnotesize $F(p)$};
\node[right] at (2.1, 0) {\footnotesize $p$};
\node[below] at (0, 0) {\footnotesize $0$};
\draw[very thick] (1.7554, 0) -- (1.7554, 1.5pt);
\node[below] at (1.7554, 0) {\footnotesize $\theta$};
\draw[very thick] (0.4568, 0) -- (0.4568, 1.5pt);
\node[below] at (0.4568, 0) {\footnotesize $a + b$};

\node[left] at (0, 1.25) {\footnotesize $1$};
\draw[thin, dashed] (0, 1.25) -- (2, 1.25);

\node[left] at (0, 0.8125) {\footnotesize $1 - x$};
\draw[thin, dashed] (0, 0.8125) -- (1.7554, 0.8125);
\node[anchor = north west, color = Sepia] at (0.55, 0.3) {\footnotesize $F_{\tan}(q) = 1 - \frac{b}{p - a}$};

\node[left] at (0, 0.9298) {\footnotesize $1 - \overline{x}$};
\draw[thin, dashed] (0, 0.9298) -- (1.7554, 0.9298);

\draw[thin, dashed] (1.7554, 0) -- (1.7554, 1.35);
\node[anchor = south east, color = BurntOrange] at (1, 0.55) {\footnotesize $\rr_{\tan}(q)$};

\draw[thin, fill = green] (0.4568, 0) circle(0.5pt);
\draw[thin, fill = green] (1.7554, 0.9298) circle(0.5pt);
\draw[thin, fill = green] (1.7554, 0.8125) circle(0.5pt);

\draw[very thick] (0, 0.8125) -- (1.5pt, 0.8125);
\draw[very thick] (0, 0.9298) -- (1.5pt, 0.9298);
\draw[very thick] (0, 1.25) -- (1.5pt, 1.25);
\end{tikzpicture}
\label{fig:lem:tangent2}
}
\caption{Demonstrations for $r_{\tan}(q)$, $\rr_{\tan}(q)$, $F_{\tan}$ and $\FF_{\tan}$.}
\label{fig:lem:tangent}
\end{figure}

\begin{proof}
To make things mimic, we provide Figure~\ref{fig:lem:tangent} to exhibit involved revenue-quantile curves and CDF's. When $p \in (0, a + b]$, clearly $\Psi\big(p, \{\FF_{\tan}\}\big) = \Psi\big(p, \{F_{\tan}\}\big) = \infty$ and $\FF_{\tan}(p) = F_{\tan}(p) = 0$, which indicates the regularity in this range. It remains to deal with the case when $p \in (a + b, \theta]$. Recall Section~\ref{subsec:prelim-reg-vv} for the definition of regularity, the following two claims are equivalent:
\begin{itemize}
\item $\FF_{\tan}$ is regular, in the range of $p \in (a + b, \theta]$;
\item $\PPhi_{\tan}(p) \eqdef p - \frac{1 - \FF_{\tan}(p)}{\ff_{\tan}(p)}$ is increasing when $p \in (a + b, \theta]$, where $\ff_{\tan}(p) \eqdef \FF_{\tan}'(p)$.
\end{itemize}
To settle the lemma, we would verify the second claim instead.

Given $p \in (a + b, \theta]$, let $y \eqdef \left(1 - \frac{a + b}{p}\right) \in \left(0, 1 - \frac{a + b}{\theta}\right) \subset (0, 1)$ for notational simplicity. Clearly, $y$ can be viewed as an increasing function of $p$. Under the construction in the lemma,
\begin{align}
\notag
& \FF_{\tan}(p) = \FF_{\tan}\left(\frac{a + b}{1 - y}\right) = \frac{(a + b) \cdot y}{ay + b - (1 - e^{-\Delta_k}) \cdot (1 - y) \cdot (b + y)} \\
\label{eq:lem:tangent2}
& \ff_{\tan}(p) = \ff_{\tan}\left(\frac{a + b}{1 - y}\right) = \frac{\left[b - (1 - e^{-\Delta_k}) \cdot (b + y^2)\right] \cdot (1 - y)^2}{\left[ay + b - (1 - e^{-\Delta_k}) \cdot (1 - y) \cdot (b + y)\right]^2}.
\end{align}
As a consequence, $\PPhi_{\tan}(p) = p - \frac{1 - \FF_{\tan}(p)}{\ff_{\tan}(p)} = a + \Upsilon(y) \cdot \left(1 - e^{-\Delta_k}\right)$, where
\[
\Upsilon(y) \eqdef \frac{(a - 1) \cdot y^2 + e^{-\Delta_k} \cdot (b + y)^2}{b - \left(1 - e^{-\Delta_k}\right) \cdot (b + y^2)}.
\]
To see the lemma, it suffices to show $\Upsilon(y)$ is increasing on $y \in (0, 1)$. Since $a \geq 1$ and $b \geq 0$,
\begin{itemize}
\item $\Upsilon_1(y) \eqdef (a - 1) \cdot y^2 + e^{-\Delta_k} \cdot (b + y)^2$ is a positive and increasing function on $y \in (0, 1)$;
\item $\Upsilon_2(y) \eqdef b - \left(1 - e^{-\Delta_k}\right) \cdot (b + y^2)$ is a decreasing function on $y \in (0, 1)$;
\item $\Upsilon_2(y) \overset{(\ref{eq:lem:tangent2})}{=} \left[\frac{ay + b}{1 - y} - \left(1 - e^{-\Delta_k}\right) \cdot (b + y)\right]^2 \cdot \ff_{\tan}\left(\frac{a + b}{1 - y}\right) \geq 0$ for all $y \in (0, 1)$.
\end{itemize}
Putting these together indicates the monotonicity of $\Upsilon(y)$, and thus settles the lemma.
\end{proof}

\subsection{Proof of Lemma~\ref{lem:regular}}
\label{subapp:upper3-cstr:regular}

\paragraph{[Lemma~\ref{lem:regular}].}
\emph{$\rr_k(q) \leq \partial_+ \rr_k(\overline{x}) \cdot (q - \overline{x}) + \rr_k(\overline{x})$ for all $\overline{x} \in (0, \qq_k)$ and $q \in [\overline{x}, \qq_k)$.}

\begin{proof}
W.l.o.g., we would assume $\Delta_k > 0$, and thus $\uu = \uu_k$ (by Lemma~\ref{lem:alg_facts}.1). To make things mimic, all involved revenue-quantile curves are shown in Figure~\ref{fig:lem:regular}. Recall Fact~\ref{fact1}, these revenue-quantile curves can be converted into the corresponding CDF's, and vice versa. Under our induction hypothesis that $F_k$ is regular, $r_k(q)$ is continuous and concave on $q \in [0, 1]$ (see Section~\ref{subsec:prelim-reg-vv}).

\begin{figure}[H]
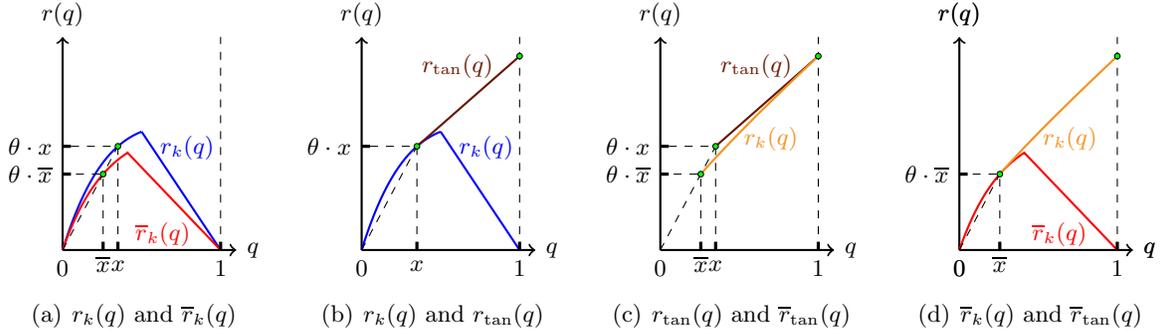

\centering
\subfigure[$r_k(q)$ and $\rr_k(q)$]{

\label{fig:lem:regular4}
}
\caption{Demonstrations for $r_k(q)$, $\rr_k(q)$, $r_{\tan}(q)$ and $\rr_{\tan}(q)$.}
\label{fig:lem:regular}
\end{figure}

As mentioned before, $r_k(0) = 0$. Since $r_k(q)$ is continuous and concave, and as Figure~\ref{fig:lem:regular1} suggests, $y = \frac{r_k(x)}{1 - x}$ is continuous and strictly increasing on $x \in [0, q_k]$. As per these, given $\Delta_k \in \left(0, \ln\left(1 + \frac{v_k q_k}{1 - q_k}\right)\right]$, there exists a unique $x^* \in (0, q_k]$ such that
\[
\ln\left(1 + \frac{r_k(x^*)}{1 - x^*}\right) = \Delta_k \quad\quad\quad\quad\quad\quad \frac{r_k(x^*)}{x^*} = \uu_k \xlongequal{Lemma~\ref{lem:alg_facts}.1} \uu.
\]
\begin{description}
\item [(a):] Given $x \in (x^*, q_k)$, let $\theta \eqdef \frac{r_k(x)}{x} \in [\vv_k, \uu_k] \subset [v_k, u_k]$, $a \eqdef \partial_+ r_k(x)$ and $b \eqdef r_k(x) - \partial_+ r_k(x) \cdot x$ for notational brevity. Consider the tangent to $r_k(q)$ at $\big(x, r_k(x)\big)$, that is,
    \[
    r_{\tan}(q) \eqdef a \cdot q + b \quad\quad \forall q \in [x, 1].
    \]
    Clearly, $r_{\tan}(q)$ corresponds to CDF $F_{\tan}(p) \eqdef
    \begin{cases}
    1 - \frac{b}{p - a} & \forall p \in (a + b, \theta] \\
    0 & \forall p \in (0, a + b]
    \end{cases}$. As Figure~\ref{fig:lem:regular2} illustrates, the concavity of $r_k(q)$ and tangency imply that
    \begin{description}
    \item [(1):] $r_k(q) \leq r_{\tan}(q)$ for all $q \in [x, 1]$, and thus $F_k(p) \geq F_{\tan}(p)$ for all $p \in (0, \theta]$;
    \item [(2):] $r_k(x) = r_{\tan}(x)$ and $\partial_+ r_k(x) = \partial_+ r_{\tan}(x)$. Thus, $F_k(\theta) = F_{\tan}(\theta)$ and $f_k(\theta) = f_{\tan}(\theta)$;
    \item [(3):] $a = \partial_+ r_k(x) = \Phi_k(\theta) \geq \Phi_k(v_k^+) \geq 1$ and $b \geq 0$, where the two inequalities respectively follows from our induction hypotheses that $F_k$ is regular (that is, $\Psi_k(p)$ is increasing when $p \geq \vv_k = v_k$), and that $F_k$ satisfies constraint~(\ref{cstr:value}).
    \end{description}
\item [(b):]  Parameter $\overline{x}$ can be determined by $\theta$ and $x$ as $\ln\left(1 + \frac{\theta \cdot \overline{x}}{1 - \overline{x}}\right) = \ln\left(1 + \frac{\theta \cdot x}{1 - x}\right) - \Delta_k \geq 0$.
    \begin{description}
    \item [(1):] As Figure~\ref{fig:lem:regular1} suggests, $\overline{x}$ ranges from $0$ to $\qq_k$, when $x$ increases from $x^*$ to $q_k$.
    \end{description}
    We further define $\rr_{\tan}(q)$ from $r_{\tan}(q)$ and $\Delta_k$: In terms of CDF's,
    \[
    \Psi\big(p, \{\FF_{\tan}\}\big) \equiv \Psi\big(p, \{F_{\tan}\}\big) - \Delta_k \quad\quad \forall p \in (0, \theta].
    \]
    In the range of $p \in (0, \theta] \subset (0, \uu_k] = (0, \uu]$, recall line~(\ref{alg3:FF1})/line~(\ref{alg3:FF2}) of Algorithm~\ref{alg3},
    \[
    \Psi\big(p, \{\FF_k\}\big) \equiv \Psi\big(p, \{F_k\}\big) - \Delta_k.
    \]
    Respectively, $\FF_k$ and $\FF_{\tan}$ are constructed from $F_k$ and $F_{\tan}$ in same fashions. It follows from fact~\textbf{(a.1)} that $\Psi\big(p, \{F_k\}\big) \leq \Psi\big(p, \{\FF_k\}\big)$ for all $p \in (0, \theta]$, which means
    \[
    \Psi\big(p, \{\FF_k\}\big) = \Psi\big(p, \{F_k\}\big) - \Delta_k \leq \Psi\big(p, \{F_{\tan}\}\big) - \Delta_k = \Psi\big(p, \{\FF_{\tan}\}\big),
    \]
    for all $p \in (0, \theta]$. In other words,
    \begin{description}
    \item [(2):] $\FF_k(p) \geq \FF_{\tan}(p)$ for all $p \in (0, \theta]$. Hence, $\rr_k(q) \leq \rr_{\tan}(q)$ for all $q \in [\overline{x}, 1]$.
    \end{description}
    Similarly, as Figure~\ref{fig:lem:regular3} suggests, we can also infer from fact~\textbf{(a.2)} that
    \begin{description}
    \item [(3):] $\FF_k(\theta) = \FF_{\tan}(\theta)$ and $\ff_k(\theta) = \ff_{\tan}(\theta)$, which means that $\rr_k(\overline{x}) = \rr_{\tan}(\overline{x})$ and $\partial_+ \rr_k(\overline{x}) = \partial_+ \rr_{\tan}(\overline{x})$.
    \end{description}
\item [(c):] $\FF_{\tan}$ turns out to be regular on $p \in (0, \theta]$, which is formalized as Lemma~\ref{lem:tangent} in the below (for this, fact~\textbf{(a.3)} ensures the lemma is applicable). That is, $\rr_{\tan}(q)$ is concave on $q \in [\overline{x}, 1]$. By equivalent condition for concave functions (see Fact~\ref{fact:concave}),
    \begin{description}
    \item [(1):] $\rr_{\tan}(q) \leq \partial_+ \rr_{\tan}(\overline{x}) \cdot (q - \overline{x}) + \rr_{\tan}(\overline{x})$ for all $q \in [\overline{x}, 1]$.
    \end{description}
\end{description}
To sum up, and as Figure~\ref{fig:lem:regular4} conveys,
\[
\rr_k(q) \overset{(\textbf{b.2})}{\leq} \rr_{\tan}(q) \overset{(\textbf{c.1})}{\leq} \partial_+ \rr_{\tan}(\overline{x}) \cdot (q - \overline{x}) + \rr_{\tan}(\overline{x}) \overset{(\textbf{b.3})}{=} \partial_+ \rr_k(\overline{x}) \cdot (q - \overline{x}) + \rr_k(\overline{x}),
\]
for all $q \in [\overline{x}, 1]$. Combining this with fact~\textbf{(b.1)}, we complete the proof of the lemma.
\end{proof}

\subsection{Proof of Lemma~\ref{lem:upper3-cstr2}}
\label{subapp:upper3-cstr:value}

\paragraph{[Lemma~\ref{lem:upper3-cstr2}].}
\emph{$\PPhi_k(p) \geq \Phi_k(p)$ for all $p  \in (\vv_k, \infty) = (v_k, \infty)$ and $k \in [n]$.}

\begin{proof}
First of all, when $p \in (\uu_k, \infty)$, the lemma certainly holds in that $\PPhi_k(p) = p \geq \Phi_k(p)$. In the range of $p \in (\vv_k, \uu_k]$, under our construction in line~(\ref{alg3:FF1})/line~(\ref{alg3:FF2}) of Algorithm~\ref{alg3},
\begin{align*}
& \FF_k(p) = \frac{F_k(p)}{1 - \left(1 - e^{-\Delta_k}\right) \cdot \left[1 - \left(1 - \frac{1}{p}\right) \cdot F_k(p)\right]}, \\
& \ff_k(p) \eqdef \FF_k'(p) = \frac{f_k(p) - \left(1 - e^{-\Delta_k}\right) \cdot \left(f_k(p) + \frac{F_k(p)^2}{p^2}\right)}{\left\{1 - \left(1 - e^{-\Delta_k}\right) \cdot \left[1 - \left(1 - \frac{1}{p}\right) \cdot F_k(p)\right]\right\}^2}.
\end{align*}
$\PPhi_k(p) \geq \Phi_k(p)$ is equivalent to $p - \frac{1 - \FF_k(p)}{\ff_k(p)} \geq p - \frac{1 - F_k(p)}{f_k(p)}$. To reveal this, we first plug the above two formulas into the left hand side of this inequality, and then rearrange the intermediate inequality. Afterwards, it remains to deal with $f_k(p) \geq \frac{1 - F_k(p)}{p - 1 + e^{-\Delta_k} \cdot \left(1 + p \cdot \frac{1 - F_k(p)}{F_k(p)}\right)^2}$, or equivalently,
\begin{equation}
\label{eq:upper3-cstr3}
\Phi_k(p) = p - \frac{1 - F_k(p)}{f_k(p)} \geq 1 - e^{-\Delta_k} \cdot \left(1 + p \cdot \frac{1 - F_k(p)}{F_k(p)}\right)^2.
\end{equation}
Recall our induction hypotheses, $F_k$ is regular (which indicates the monotonicity of $\Phi_k(p)$), and satisfies constraint~(\ref{cstr:value}). As per these, inequality~(\ref{eq:upper3-cstr3}) surely holds in that
\[
\text{LHS of~(\ref{eq:upper3-cstr3})} \geq \Phi_k(v_k^+) \geq 1 \geq \text{RHS of~(\ref{eq:upper3-cstr3})},
\]
for all $p \in (\vv_k, \uu_k] \subset (v_k, u_k]$. This completes the proof of the lemma.
\end{proof}

\section{Proof of Lemma~\ref{lem:main}}
\label{app:upper3-lem}
\fcolorbox{white}{lightgray}{\begin{minipage}{\textwidth}
Assume w.l.o.g. $\Delta_k > 0$, then Lemma~\ref{lem:alg_facts}.1 ensures $\uu_k = \uu$: While $\{\FF_i\}_{i = 1}^{k - 1} \cup \{F_i\}_{i = k + 1}^n$ remains,
\begin{itemize}
\item $F_k$ becomes $\FF_k$, where $\FF_k$ is defined as $\Psi\big(p, \{\FF_k\}\big) \equiv \Psi\big(p, \{F_k\}\big) - \Delta_k$ for all $p \in (0, \uu_k]$;
\item $\cont(\gamma_k)$ becomes $\cont(\ggamma_k)$, where $\ggamma_k$ is defined by letting $\R(\ggamma_k) = \R(\gamma_k) + \Delta_k$;
\item $\gamma^* \geq \gamma_k > \ggamma_k \geq u_k \geq \uu_k \geq v_k = \vv_k \geq v^* > 1$.
\end{itemize}
\end{minipage}}

\subsection{Auxiliary Lemma~\ref{lem:main1}}
\label{subapp:upper3-lem1}

\begin{lemma}
\label{lem:main1}
Given $p \in [\vv_k, \uu_k]$,
\[
(p - 1) \cdot \frac{\FF_k(p) - F_k(p)}{e^{\Q(\ggamma_k) - \Q(\gamma_k)} - 1}
\leq 2\gamma^{*2} \cdot \Delta^* + \gamma_k \cdot \Big[F_k(p) + p \cdot \big(1 - F_k(p)\big)\Big] \cdot \left(1 - \frac{1}{p}\right) \cdot F_k(p).
\]
\end{lemma}

\begin{proof}
Under premise $\ggamma_k \leq \gamma_k$, clearly $\Q(\ggamma_k) \geq \Q(\gamma_k)$. Since $e^x - 1 \geq x$ for all $x \geq 0$,
\[
e^{\Q(\ggamma_k) - \Q(\gamma_k)} - 1
\geq \Q(\ggamma_k) - \Q(\gamma_k)
\xlongequal{Lemma~\ref{lem:RQ}.1} \displaystyle{\int_{\ggamma_k}^{\gamma_k}} \frac{|\R'(x)|}{x} dx
\geq \displaystyle{\int_{\ggamma_k}^{\gamma_k}} \frac{|\R'(x)|}{\gamma_k} dx
\overset{(\ref{eq:upper3-cont})}{=} \frac{\Delta_k}{\gamma_k}.
\]
Together with Lemma~\ref{lem:main1.1} in the below, we can infer from above
\[
(p - 1) \cdot \frac{\FF_k(p) - F_k(p)}{e^{\Q(\ggamma_k) - \Q(\gamma_k)} - 1}
\leq 2(p - 1) \cdot \gamma_k \cdot \Delta_k + \gamma_k \cdot \Big[F_k(p) + p \cdot \big(1 - F_k(p)\big)\Big] \cdot \left(1 - \frac{1}{p}\right) \cdot F_k(p).
\]
Due to Algorithm~\ref{alg3} and line~(\ref{alg2:Delta}) of Algorithm~\ref{alg2}, we know $\Delta_k \leq \sum\limits_{i = 1}^n \Delta_i = \Delta \leq \Delta^*$. Under premise $\gamma^* \geq \gamma_k \geq \uu_k \geq p \geq \vv_k > 1$, we settle the lemma by observing $2(p - 1) \cdot \gamma_k \cdot \Delta_k \leq 2\gamma^{*2} \cdot \Delta^*$.
\end{proof}

\begin{lemma}
\label{lem:main1.1}
Given $p \in [\vv_k, \uu_k]$, $\FF_k(p) - F_k(p) \leq 2\Delta_k^2 + \Big[F_k(p) + p \cdot \big(1 - F_k(p)\big)\Big] \cdot \frac{\Delta_k}{p} \cdot F_k(p)$.
\end{lemma}

\begin{proof}
Let $t \eqdef \Big[F_k(p) + p \cdot \big(1 - F_k(p)\big)\Big] \cdot \frac{1}{p}$ for notational simplicity. We shall prove
\begin{equation}
\label{eq:lem:main2}
\FF_k(p) - F_k(p) \leq 2\Delta_k^2 + t \cdot \Delta_k \cdot F_k(p),
\end{equation}
for all $p \in [\vv_k, \uu_k]$. Recall $\Psi\big(p, \{F_k\}\big) = \ln\left(1 + p \cdot \frac{1 - F_k(p)}{F_k(p)}\right)$, rearranging Eq.~(\ref{eq:upper3-F}) results in
\[
\FF_k(p) \overset{(\ref{eq:upper3-F})}{=} \frac{F_k(p)}{1 - t \cdot \left(1 - e^{-\Delta_k}\right)} \overset{(\dagger)}{\leq} \frac{F_k(p)}{1 - t \cdot \Delta_k},
\]
where $(\dagger)$ follows as $0 \leq t \cdot \Delta_k \leq \frac{1}{2} < 1$ (see Lemma~\ref{lem:main1.2}), and $1 - e^{-x} \leq x$ for all $x \geq 0$. Since then,
\[
\text{LHS of~(\ref{eq:lem:main2})}
\leq \frac{t \cdot \Delta_k}{1 - t \cdot \Delta_k} \cdot F_k(p)
\overset{(\ddagger)}{\leq} \big(2t^2 \cdot \Delta_k^2 + t \cdot \Delta_k\big) \cdot F_k(p)
\overset{(\star)}{\leq} \text{RHS of~(\ref{eq:lem:main2})},
\]
where $(\ddagger)$ follows as $\frac{x}{1 - x} \leq 2x^2 + x$ for all $x \in \left[0, \frac{1}{2}\right]$, and $(\star)$ follows as $F_k(p) \leq 1$ and $t \leq 1$ (see Lemma~\ref{lem:main1.2}). This completes the proof of the lemma.
\end{proof}

\begin{lemma}
\label{lem:main1.2}
$0 \leq \Delta_k \leq \frac{1}{2}$, and $0 \leq \Big[F_k(p) + p \cdot \big(1 - F_k(p)\big)\Big] \cdot \frac{1}{p} \leq 1$ for all $p \in [\vv_k, \uu_k]$.
\end{lemma}

\begin{proof}
For the first claim, recall line~(\ref{alg2:Delta}) of Algorithm~\ref{alg2}, Algorithm~\ref{alg3} and Lemma~\ref{lem:constant},
\[
0 \leq \Delta_k \leq \sum\limits_{i = 1}^n \Delta_i = \Delta \leq \Delta^* \leq \delta^* \leq \frac{1}{2}\epsilon \leq \frac{1}{2}.
\]
Furthermore, under premise $p \geq \vv_k > 1$ (see inequality~(\ref{eq:upper3-p})), the second claim follows as
\[
0 \leq \Big[F_k(p) + p \cdot \big(1 - F_k(p)\big)\Big] \cdot \frac{1}{p} = 1 - \left(1 - \frac{1}{p}\right) \cdot F_k(p) \leq 1.
\]
This completes the proof of the lemma.
\end{proof}

\subsection{Auxiliary Lemma~\ref{lem:main2}}
\label{subapp:upper3-lem2}

The following mathematical fact is obtained in Appendix~\ref{subapp:math_facts:ineq}.
\vspace{5.5pt} \\
\fbox{\begin{minipage}{\textwidth}
\paragraph{[Lemma~\ref{lem:ineq}.1].}
\emph{For all $y \geq x > 1$,
\begin{enumerate}
\item $1 - e^{-\big(\R(x) - \R(y)\big)} \leq \left(1 - \frac{x}{y}\right) \cdot \frac{x + e^{\R(x) - \R(y)} - 1}{(x - 1)^2}$.
\end{enumerate}}
\end{minipage}}

\begin{lemma}
\label{lem:main2}
Given $p \in [\vv_k, \uu_k]$, define $\hp$ by letting $\R(\hp) = \R(p) - \Psi\big(p, \{F_k\}\big)$, then
\[
\frac{1}{\hp} \leq \frac{1}{F_k(p) + p \cdot (1 - F_k(p))} - \left(1 - \frac{1}{p}\right) \cdot F_k(p).
\]
\end{lemma}

\begin{proof}
Since $\R(p) - \R(\hp) = \Psi\big(p, \{F_k\}\big) = \ln\left(1 + p \cdot \frac{1 - F_k(p)}{F_k(p)}\right)$, we surely have $F_k(p) = \frac{p}{p + e^{\R(p) - \R(\hp)} - 1}$. Plug this into the inequality in the lemma, it remains to show
\[
\frac{1}{\hp} \leq \frac{p + e^{\R(p) - \R(\hp)} - 1}{p} \cdot \frac{1}{e^{\R(p) - \R(\hp)}} - \frac{p - 1}{p + e^{\R(p) - \R(\hp)} - 1}.
\]
After rearranging, we are left with
\[
1 - e^{-\big(\R(p) - \R(\hp)\big)} \leq \left(1 - \frac{p}{\hp}\right) \cdot \frac{p + e^{\R(p) - \R(\hp)} - 1}{(p - 1)^2}.
\]
Clearly, $\R(\hp) \leq \R(p)$ means $\hp \geq p \geq \vv_k > 1$. Since then, the above inequality can be inferred from Lemma~\ref{lem:ineq}.1. This completes the proof of the lemma.
\end{proof}

\subsection{Auxiliary Lemma~\ref{lem:main3}}
\label{subapp:upper3-lem3}

\begin{lemma}
\label{lem:main3}
Given $p \in [\vv_k, \uu_k]$, define $\hp$ by letting $\R(\hp) = \R(p) - \Psi\big(p, \{F_k\}\big)$, then
\[
\gamma_k - \hp \geq \frac{1}{\left|\R'(\hp)\right|} \cdot \delta^* \geq (\hp^2 - \hp) \cdot \delta^*.
\]
\end{lemma}

\begin{proof}
The second inequality follows from Lemma~\ref{lem:RQ}.3 that $|\R'(\hp)| \leq (\hp^2 - \hp)$. As for the first one, recall the definition of $\gamma_k$ that $\R(\gamma_k) = \R(\gamma) + \sum\limits_{i = 1}^{k - 1} \Delta_i$. In the range of $p \in [\vv_k, \uu_k] \subset (1, \uu] \subset (1, u]$,
\[
\begin{aligned}
\R(\gamma_k) - \R(\hp)
= & \R(\gamma) + \sum\limits_{i = 1}^{k - 1} \Delta_i - \R(p) + \Psi\big(p, \{F_k\}\big) \\
\leq & \R(\gamma) + \sum\limits_{i = 1}^{k - 1} \Delta_i - \R(p) + \sum\limits_{i = 1}^{k - 1} \Psi\big(p, \{\FF_i\}\big) + \sum\limits_{i = k}^n \Psi\big(p, \{F_i\}\big) \\
\overset{(\ast)}{=} & \R(\gamma) - \R(p) + \sum\limits_{i = 1}^n \Psi\big(p, \{F_i\}\big)
= \R(\gamma) - \R(p) + \Psi\big(p, \{F_i\}_{i = 1}^n\big)
\overset{(\diamond)}{\leq} -\delta^*,
\end{aligned}
\]
where $(\ast)$ follows from line~(\ref{alg3:FF1})/line~(\ref{alg3:FF2}) of Algorithm~\ref{alg3}, and $(\diamond)$ follows from induction hypothesis $\R(\gamma) + \Psi\big(p, \{F_i\}_{i = 1}^n\big) \leq \R(p) - \delta^*$ for all $p \in (1, u]$. Noting that $\R(p)$ is decreasing and convex on $p \in (1, \infty)$ (see Lemma~\ref{lem:RQ}.1), we can infer from above $\gamma_k - \hp \geq \frac{1}{\left|\R'(\hp)\right|} \cdot \delta^*$, hence settling the lemma.
\end{proof}

\subsection{Proof of Lemma~\ref{lem:main}}

\paragraph{[Lemma~\ref{lem:main}].}
\emph{For all $p \in [\vv_k, \uu_k] \subset [v_k, u_k]$,
\begin{equation}
\tag{\ref{eq:lem:main}}
\gamma_k - \Big[F_k(p) + p \cdot \big(1 - F_k(p)\big)\Big] - (p - 1) \cdot \frac{\FF_k(p) - F_k(p)}{e^{\Q(\ggamma_k) - \Q(\gamma_k)} - 1}
\geq \frac{1}{\kappa^*} \cdot \gamma^{*2} \cdot \Delta^*.
\end{equation}}

\begin{proof}
The proof of Lemma~\ref{lem:main} consists of the following three parts.

\paragraph{Part I.}
Based on definitions of $\FF_k$ and $\ggamma_k$, we obtain Lemma~\ref{lem:main1} in Appendix~\ref{subapp:upper3-lem1}, which indicates
\begin{flalign*}
\bullet\quad \text{LHS of~(\ref{eq:lem:main})}
\geq & -2\gamma^{*2} \cdot \Delta^* + \gamma_k \cdot \Big[F_k(p) + p \cdot \big(1 - F_k(p)\big)\Big] & \\
& \phantom{-2\gamma^{*2} \cdot \Delta^* + \gamma_k\,\,} \cdot \left[-\frac{1}{\gamma_k} + \frac{1}{F_k(p) + p \cdot \big(1 - F_k(p)\big)} - \left(1 - \frac{1}{p}\right) \cdot F_k(p)\right]; &
\end{flalign*}
\fbox{\begin{minipage}{\textwidth}
\paragraph{[Lemma~\ref{lem:main1}].}
\emph{Given $p \in [\vv_k, \uu_k]$,}
\[
(p - 1) \cdot \frac{\FF_k(p) - F_k(p)}{e^{\Q(\ggamma_k) - \Q(\gamma_k)} - 1}
\leq 2\gamma^{*2} \cdot \Delta^* + \gamma_k \cdot \Big[F_k(p) + p \cdot \big(1 - F_k(p)\big)\Big] \cdot \left(1 - \frac{1}{p}\right) \cdot F_k(p).
\]
\end{minipage}}

\paragraph{Part II.}
Define $\hp$ by letting $\R(\hp) = \R(p) - \Psi\big(p, \{F_k\}\big)$. In Lemma~\ref{lem:main2} (justified in Appendix~\ref{subapp:upper3-lem2}), we get a measurement of $\hp$. Applying the lemma to the conclusion of~\textbf{Part I}, we know
\begin{flalign*}
\bullet\quad \text{LHS of~(\ref{eq:lem:main})}
\geq & -2\gamma^{*2} \cdot \Delta^* + \gamma_k \cdot \Big[F_k(p) + p \cdot \big(1 - F_k(p)\big)\Big] \cdot \left(\frac{1}{\hp} - \frac{1}{\gamma_k}\right) \\
\overset{(\ast)}{\geq} & -2\gamma^{*2} \cdot \Delta^* + \gamma_k \cdot \left(\frac{1}{\hp} - \frac{1}{\gamma_k}\right)
= -2\gamma^{*2} \cdot \Delta^* + \frac{1}{\hp} \cdot \left(\gamma_k - \hp\right); &
\end{flalign*}
where $(\ast)$ follows as $\hp \leq \gamma_k$ (see Lemma~\ref{lem:main3}) and $F_k(p) + p \cdot \big(1 - F_k(p)\big) = 1 + (p - 1) \cdot \big(1 - F_k(p)\big) \geq 1$.
\vspace{5.5pt} \\
\fbox{\begin{minipage}{\textwidth}
\paragraph{[Lemma~\ref{lem:main2}].}
\emph{Given $p \in [\vv_k, \uu_k]$, define $\hp$ by letting $\R(\hp) = \R(p) - \Psi\big(p, \{F_k\}\big)$, then}
\[
\frac{1}{\hp} \leq \frac{1}{F_k(p) + p \cdot \big(1 - F_k(p)\big)} - \left(1 - \frac{1}{p}\right) \cdot F_k(p).
\]
\end{minipage}}

\paragraph{Part III.}
We know $\hp \geq p$, in that $\R(\hp) = \R(p) - \Psi\big(p, \{F_k\}\big) \leq \R(p)$. Under premise $p \geq \vv_k \geq v^* > 1$, clearly $\hp \geq v^* > 1$. Since then, and recall line~(\ref{alg1:Delta}) of Algorithm~\ref{alg1} that $\Delta^* = \frac{v^* - 1}{3\gamma^{*2}} \cdot \kappa^* \cdot \delta^*$,
\begin{equation}
\label{eq:lem:main1}
\delta^* = \frac{3}{\kappa^*} \cdot \frac{\gamma^{*2} \cdot \Delta^*}{v^* - 1} \geq \frac{3}{\kappa^*} \cdot \frac{\gamma^{*2} \cdot \Delta^*}{\hp - 1}
\end{equation}
Applying Lemma~\ref{lem:main3} (to be acquired in Appendix~\ref{subapp:upper3-lem3}) to the conclusion of \textbf{Part II} results in
\begin{flalign*}
\bullet\quad \text{LHS of~(\ref{eq:lem:main})}
\geq & -2\gamma^{*2} \cdot \Delta^* + (\hp - 1) \cdot \delta^* \\
\overset{(\ref{eq:lem:main1})}{\geq} & \left(\frac{3}{\kappa^*} - 2\right) \cdot \gamma^{*2} \cdot \Delta^*
\geq \text{RHS of~(\ref{eq:lem:main})}. & \text{(By Lemma~\ref{lem:constant} that $\kappa^* \in (0, 1]$)}
\end{flalign*}
\fbox{\begin{minipage}{\textwidth}
\paragraph{[Lemma~\ref{lem:main3}].}
\emph{Given $p \in [\vv_k, \uu_k]$, define $\hp$ by letting $\R(\hp) = \R(p) - \Psi\big(p, \{F_k\}\big)$, then}
\[
\gamma_k - \hp \geq \frac{1}{\left|\R'(\hp)\right|} \cdot \delta^* \geq (\hp^2 - \hp) \cdot \delta^*.
\]
\end{minipage}}
\vspace{5.5pt}

This completes the proof of Lemma~\ref{lem:main}.
\end{proof}

\section{Proof of Lemma~\ref{lem:abel_ineq1}}
\label{app:upper3-opt}
\fcolorbox{white}{lightgray}{\begin{minipage}{\textwidth}
Assume w.l.o.g. $\Delta_k > 0$, then Lemma~\ref{lem:alg_facts}.1 ensures $\uu_k = \uu$: While $\{\FF_i\}_{i = 1}^{k - 1} \cup \{F_i\}_{i = k + 1}^n$ remains,
\begin{itemize}
\item $F_k$ becomes $\FF_k$, where $\FF_k$ is defined as $\Psi\big(p, \{\FF_k\}\big) \equiv \Psi\big(p, \{F_k\}\big) - \Delta_k$ for all $p \in (0, \uu_k]$;
\item $\cont(\gamma_k)$ becomes $\cont(\ggamma_k)$, where $\ggamma_k$ is defined by letting $\R(\ggamma_k) = \R(\gamma_k) + \Delta_k$;
\item $\gamma^* \geq \gamma_k > \ggamma_k \geq u_k \geq \uu_k \geq v_k = \vv_k \geq v^* > 1$;
\end{itemize}
\end{minipage}}

\paragraph{Analysis 0: Auxiliary Facts.}
The following mathematical facts are certified in Appendix~\ref{subapp:math_facts:ineq}.
\vspace{5.5pt} \\
\fbox{\begin{minipage}{\textwidth}
\paragraph{[Lemma~\ref{lem:ineq}.2,3,4].}
\emph{For all $y \geq x > 1$,
\begin{enumerate}
\setcounter{enumi}{1}
\item $1 - e^{-\big(\R(x) - \R(y)\big)} \leq \left(1 - \frac{x}{y}\right) \cdot \frac{x + e^{\R(x) - \R(y)} - 1}{(x - 1)^2}$;
\item $\frac{e^{\R(x) - \R(y)}}{x + e^{\R(x) - \R(y)} - 1} \geq \frac{1}{y}$;
\item $1 - \frac{1}{y} \geq (x - 1) \cdot \frac{e^{\R(x) - \R(y)}}{x + e^{\R(x) - \R(y)} - 1}$.
\end{enumerate}}
\end{minipage}}
\vspace{5.5pt}

To justify Lemma~\ref{lem:abel_ineq1}, the following facts are also required:
\vspace{5.5pt} \\
\fbox{\begin{minipage}{\textwidth}
\begin{description}
\item [(a):] $\gamma_k \geq \Phi_k(\uu_k^+) \geq 1$. This is confirmed as follows:
    \begin{itemize}
    \item We know $\Phi_k(\uu_k^+) \leq \uu_k \leq \gamma_k$, according to the definition of $\Phi_k$;
    \item $\Phi_k(\uu_k^+) \geq \Phi_k(\vv_k^+) = \Phi_k(v_k^+) \geq 1$, which follows from induction hypotheses that $F_k$ is regular (that is, $\Phi_k$ is increasing when $p \geq \vv_k = v_k$), and that $\Phi_k(v_k^+) \geq 1$;
    \end{itemize}
\item [(b):] $F_k(\uu_k^+) \geq \frac{\uu_k}{\uu_k + e^{\R(\ggamma_k) - \R(\gamma_k)} - 1}$. This can be inferred from Lemma~\ref{lem:abel_ineq4} in the below;
\item [(c):] $\FF_k(p) \geq e^{\Q(\ggamma_k) - \Q(\gamma_k)} \cdot F_k(p)$ for all $p \in [\vv_k, \uu_k]$. This will be attested soon in Lemma~\ref{lem:abel_ineq3}.
\end{description}
\end{minipage}}

\begin{lemma}
\label{lem:abel_ineq4}
Suppose $\Delta_k > 0$, then $\ln\left(1 + \uu_k \cdot \frac{1 - F_k(\uu_k^+)}{F_k(\uu_k^+)}\right) = \Psi\big(\uu_k^+, \{F_k\}\big) \leq \Delta_k = \R(\ggamma_k) - \R(\gamma_k)$.
\end{lemma}

\begin{proof}
For notational convenience, we would reuse $W = \{i \in [n]: \,\, \text{$F_i$ has probability-mass at $\uu$}\}$ defined in Algorithm~\ref{alg3}. By Lemma~\ref{lem:alg_facts}.1, $\Delta_k > 0$ means $\uu_k = \uu$. Suppose $k \in W$, under induction hypothesis that $F_k$ is regular, $\uu_k = \uu$ is $F_k$'s support-supremum (see Section~\ref{subsec:prelim-reg-vv}). Hence,
\[
F_k(\uu_k^+) = 1 \quad\quad\quad\quad\quad\quad \Psi\big(\uu_k^+, \{F_k\}\big) = 0 \leq \Delta_k.
\]
Otherwise, i.e., when $k \in [n] \setminus W$, it follows from line~(\ref{alg3:Delta_no_mass}) of Algorithm~\ref{alg3} that
\[
\Psi\big(\uu_k^+, \{F_k\}\big) \xlongequal{\uu_k = \uu} \Psi\big(\uu^+, \{F_k\}\big) \xlongequal{continuity} \Psi\big(\uu, \{F_k\}\big) \xlongequal{line~(\ref{alg3:Delta_no_mass})} \Delta_k.
\]
This completes the proof of the lemma.
\end{proof}

\begin{lemma}
\label{lem:abel_ineq3}
$\FF_k(p) \geq e^{\Q(\ggamma_k) - \Q(\gamma_k)} \cdot F_k(p)$ for all $p \in [\vv_k, \uu_k]$.
\end{lemma}

\begin{proof}
Given $p \in [\vv_k, \uu_k]$, under our construction that $\Psi\big(p, \{\FF_k\}\big) = \Psi\big(p, \{F_k\}\big) - \Delta_k$,
\begin{flalign*}
\FF_k(p)
\overset{(\ref{eq:upper3-F})}{=} & \left\{1 - \left(1 - e^{-\Delta_k}\right) \cdot \left[1 - \left(1 - \frac{1}{p}\right) \cdot F_k(p)\right]\right\}^{-1} \cdot F_k(p) \\
\geq & \left[1 - \left(1 - e^{-\Delta_k}\right) \cdot \frac{1}{p}\right]^{-1} \cdot F_k(p) & \text{(Since $F_k(p) \leq 1$ and $p \geq \vv_k > 1$)} \\
\geq & \left[1 - \left(1 - e^{-\Delta_k}\right) \cdot \frac{1}{\ggamma_k}\right]^{-1} \cdot F_k(p). & \text{(Since $p \leq \uu_k \leq \ggamma_k$)}
\end{flalign*}
To settle the lemma, we shall prove $\left[1 - \left(1 - e^{-\Delta_k}\right) \cdot \frac{1}{\ggamma_k}\right]^{-1} \geq e^{\Q(\ggamma_k) - \Q(\gamma_k)}$, or equivalently,
\[
1 - e^{-\big(\Q(\ggamma_k) - \Q(\gamma_k)\big)} \leq \frac{1}{\ggamma_k} \cdot \left(1 - e^{-\Delta_k}\right) \overset{(\ref{eq:upper3-cont})}{=} \frac{1}{\ggamma_k} \cdot \left[1 - e^{-\big(\R(\ggamma_k) - \R(\gamma_k)\big)}\right],
\]
which follows from Lemma~\ref{lem:ineq}.2 immediately. This completes the proof of the lemma.
\end{proof}

Depending on whether $p \geq \Phi_k(\uu_k^+)$ or not, we would handle Lemma~\ref{lem:abel_ineq1} in different manners.

\paragraph{Analysis I: When $p \in \big[\Phi_k(\uu_k^+), \gamma_k\big]$.}
In this case, Lemma~\ref{lem:abel_ineq1} is established on fact~\textbf{(a)}, fact~\textbf{(b)}, Main Lemma~\ref{lem:virtual_value} and the above mathematical facts.

\paragraph{[Lemma~\ref{lem:abel_ineq1}.1].}
\emph{$\displaystyle{\int_p^{\gamma_k}} D_k(x) \cdot e^{-\Q(\gamma_k)} dx \geq \displaystyle{\int_p^{\gamma_k}} \DD_k(x) \cdot e^{-\Q(\ggamma_k)} dx$ for all $p \in \big[\Phi_k(\uu_k^+), \gamma_k\big]$.}

\begin{proof}
Since $\DD_k(p) \leq 1$ for all $p \in (0, \infty)$, it suffices to show $\displaystyle{\int_p^{\gamma_k}} D_k(x) \cdot e^{-\Q(\gamma_k)} dx \geq \displaystyle{\int_p^{\gamma_k}} e^{-\Q(\ggamma_k)} dx$, or equivalently,
\[
\displaystyle{\int_p^{\gamma_k}} \left[D_k(x) - e^{-\big(\Q(\ggamma_k) - \Q(\gamma_k)\big)}\right] dx \geq 0.
\]
Note that $D_k(p)$ is increasing on $p \in (0, \infty)$, the left hand side of this inequality can be viewed as a concave function of $p$. Thus, we only need to deal with the case that $p = \Phi_k(\uu_k^+)$, in which the above inequality is equivalent to $\displaystyle{\int_{\Phi_k(\uu_k^+)}^{\gamma_k}} \left[1 - e^{-\big(\Q(\ggamma_k) - \Q(\gamma_k)\big)}\right] \geq \displaystyle{\int_{\Phi_k(\uu_k^+)}^{\gamma_k}} \big(1 - D_k(x)\big) dx$. For this,
\begin{itemize}
\item The left hand side surely equals to $\big(\gamma_k - \Phi_k(\uu_k^+)\big) \cdot \left[1 - e^{-\big(\Q(\ggamma_k) - \Q(\gamma_k)\big)}\right]$;
\item The right hand side equals to $\displaystyle{\int_{\Phi_k(\uu_k^+)}^{\infty}} \big(1 - D_k(x)\big) dx$, since $D_k(p) = 1$ when $p > \gamma_k \geq u_k$.
\end{itemize}
Further applying Main Lemma~\ref{lem:virtual_value} to the right hand side, we are left with
\begin{equation}
\label{eq:abel_ineq1}
\big(\gamma_k - \Phi_k(\uu_k^+)\big) \cdot \left[1 - e^{-\big(\Q(\ggamma_k) - \Q(\gamma_k)\big)}\right]
\geq \big(\uu_k - \Phi_k(\uu_k^+)\big) \cdot \big(1 - F_k(\uu_k^+)\big).
\end{equation}
We continue to relax both hand sides of inequality~(\ref{eq:abel_ineq1}). Firstly, due to Lemma~\ref{lem:ineq}.2,
\[
\text{LHS of~(\ref{eq:abel_ineq1})} \geq \big(\gamma_k - \Phi_k(\uu_k^+)\big) \cdot \frac{1}{\gamma_k} \cdot \left[1 - e^{-\big(\R(\ggamma_k) - \R(\gamma_k)\big)}\right];
\]
On the other hand, applying aforementioned fact~\textbf{(b)} results in
\[
\text{RHS of~(\ref{eq:abel_ineq1})} \leq \big(\uu_k - \Phi_k(\uu_k^+)\big) \cdot \frac{e^{\R(\ggamma_k) - \R(\gamma_k)} - 1}{\uu_k + e^{\R(\ggamma_k) - \R(\gamma_k)} - 1}.
\]
After rearranging, instead of inequality~(\ref{eq:abel_ineq1}), we safely turn to deal with
\begin{equation}
\label{eq:abel_ineq2}
\left(1 - \frac{\Phi_k(\uu_k^+)}{\gamma_k}\right) - \big(\uu_k - \Phi_k(\uu_k^+)\big) \cdot \frac{e^{\R(\ggamma_k) - \R(\gamma_k)}}{\uu_k + e^{\R(\ggamma_k) - \R(\gamma_k)} - 1} \geq 0.
\end{equation}
Suppose $\uu_k$ and $\Phi_k(\uu_k^+)$ are independent variables, then the left hand side of inequality~(\ref{eq:abel_ineq1}) can be viewed as an increasing function of $\Phi_k(\uu_k^+)$, in that
\[
\frac{1}{\gamma_k} \overset{(\ast)}{\leq} \frac{e^{\R(\ggamma_k) - \R(\gamma_k)}}{\ggamma_k + e^{\R(\ggamma_k) - \R(\gamma_k)} - 1} \overset{(\star)}{\leq} \frac{e^{\R(\ggamma_k) - \R(\gamma_k)}}{\uu_k + e^{\R(\ggamma_k) - \R(\gamma_k)} - 1},
\]
where $(\ast)$ follows from Lemma~\ref{lem:ineq}.3, and $(\star)$ follows as $\ggamma_k \geq \uu_k$ (see inequality~(\ref{eq:upper3-p})). Moreover, recall fact~\textbf{(a)} that $\Phi_k(\uu_k^+) \geq 1$. After assigning $\Phi_k(\uu_k^+) \leftarrow 1$ in inequality~(\ref{eq:abel_ineq2}), we are left with
\[
\left(1 - \frac{1}{\gamma_k}\right) - (\uu_k - 1) \cdot \frac{e^{\R(\ggamma_k) - \R(\gamma_k)}}{\uu_k + e^{\R(\ggamma_k) - \R(\gamma_k)} - 1} \geq 0.
\]
Apparently, here the left hand side can be viewed as decreasing function of $\uu_k$. Recall inequality~(\ref{eq:upper3-p}) that $\uu_k \leq \ggamma_k$, assigning $\uu_k \leftarrow \ggamma_k$ results in
\[
\left(1 - \frac{1}{\gamma_k}\right) - (\ggamma_k - 1) \cdot \frac{e^{\R(\ggamma_k) - \R(\gamma_k)}}{\ggamma_k + e^{\R(\ggamma_k) - \R(\gamma_k)} - 1} \geq 0,
\]
which follows from Lemma~\ref{lem:ineq}.4 immediately. This settles Lemma~\ref{lem:abel_ineq1} when $p \in \big[\Phi_k(\uu_k^+), \gamma_k\big]$.
\end{proof}

\paragraph{Analysis II: When $p \in \big[1, \Phi_k(\uu_k^+)\big)$.}
In this case, Lemma~\ref{lem:abel_ineq1} is established on aforementioned fact~\textbf{(c)}, Main Lemma~\ref{lem:virtual_value} and Lemma~\ref{lem:main}.

\paragraph{[Lemma~\ref{lem:abel_ineq1}.2].}
\emph{$\displaystyle{\int_p^{\gamma_k}} D_k(x) \cdot e^{-\Q(\gamma_k)} dx \geq \displaystyle{\int_p^{\gamma_k}} \DD_k(x) \cdot e^{-\Q(\ggamma_k)} dx$ for all $p \in \big[1, \Phi_k(\uu_k^+)\big)$.}

\begin{proof}
Given $p \in \big[1, \Phi_k(\uu_k^+)\big)$, for notational simplicity, let
\begin{equation}
\label{eq:abel_ineq2.1}
\hp \eqdef \Phi_k^{-1}(p^+) \in [\vv_k, \uu_k) \quad\quad\quad\quad\quad\quad F_k(\hp) = D_k\big(\Phi_k^{-1}(\hp)\big) = D_k(p^+).
\end{equation}
It is noteworthy that $\hp$ cannot reach $\uu_k$, since $p$ cannot reach $\Phi_k(\uu_k^+)$. We claim
\begin{itemize}
\item $\displaystyle{\int_p^{\gamma_k}} D_k(x) dx = (\gamma_k - p) - \big(\hp - p\big) \cdot \big(1 - F_k(\hp)\big)$;
\item $\displaystyle{\int_p^{\gamma_k}} \DD_k(x) dx \leq \gamma_k - \FF_k(\hp^+) - \hp \cdot \big(1 - \FF_k(\hp^+)\big) - (p - 1) \cdot F_k(\hp^+) \cdot e^{-\big(\Q(\gamma_k) - \Q(\ggamma_k)\big)}$.
\end{itemize}
On the one hand, the first claim follows as
\begin{flalign*}
\displaystyle{\int_p^{\gamma_k}} D_k(x) dx
= & (\gamma_k - p) - \displaystyle{\int_p^{\gamma_k}} \big(1 - D_k(x)\big) dx \\
= & (\gamma_k - p) - \displaystyle{\int_p^{\infty}} \big(1 - D_k(x)\big) dx & \text{(Since $D_k(x) = 1$ when $x \geq \gamma_k \geq u_k$)} \\
= & (\gamma_k - p) - \big(\Phi_k^{-1}(p^+) - p\big) \cdot \big(1 - D_k(p^+)\big) & \text{(By Main Lemma~\ref{lem:virtual_value})} \\
\overset{(\ref{eq:abel_ineq2.1})}{=} & (\gamma_k - p) - \big(\hp - p\big) \cdot \big(1 - F_k(\hp)\big). &
\end{flalign*}
Since $\PPhi_k(\hp^+) \leq \hp \overset{(\ref{eq:abel_ineq2.1})}{<} \uu_k \overset{(\ref{eq:upper3-p})}{\leq} \gamma_k$, and $\DD_k(x) = 1$ when $x \geq \gamma_k \geq \uu_k$, following similar steps results in
\begin{equation}
\label{eq:abel_ineq2.3}
\displaystyle{\int_{\PPhi_k(\hp^+)}^{\gamma_k}} \DD_k(x) dx = \big(\gamma_k - \PPhi_k(\hp^+)\big) - \big(\hp - \PPhi_k(\hp^+)\big) \cdot \big(1 - \FF_k(\hp^+)\big).
\end{equation}
By Lemma~\ref{lem:upper3-cstr2} and the monotonicity of $\Phi_k$, we know $\PPhi_k(\hp^+) \geq \Phi_k(\hp^+) \geq \Phi_k(\hp) \overset{(\ref{eq:abel_ineq2.1})}{=} p$ and thus,
\[
\begin{aligned}
\displaystyle{\int_p^{\gamma_k}} \DD_k(x) dx
\overset{(\uplus)}{\leq} & \displaystyle{\int_p^{\PPhi_k(\hp^+)}} \DD_k\big(\PPhi_k(\hp^+)\big) dx + \displaystyle{\int_{\PPhi_k(\hp^+)}^{\gamma_k}} \DD_k(x) dx \\
\overset{(\ref{eq:abel_ineq2.1})}{=} & \big(\PPhi_k(\hp^+) - p\big) \cdot \FF_k(\hp^+) + \displaystyle{\int_{\PPhi_k(\hp^+)}^{\gamma_k}} \DD_k(x) dx \\
\overset{(\ref{eq:abel_ineq2.3})}{=} & \gamma_k - \FF_k(\hp^+) - \hp \cdot \big(1 - \FF_k(\hp^+)\big) - (p - 1) \cdot \FF_k(\hp^+),
\end{aligned}
\]
where $(\uplus)$ follows as $\DD_k(x)$ is increasing when $x > 0$. Since $\hp$ belongs to half-open interval $[\vv_k, \uu_k)$, apply fact~\textbf{(c)} to the above inequality, our second claim can be inferred easily.

Plug the two claims to the inequality in the lemma. After rearranging, it remains to certify $e^{-\Q(\gamma_k)} \cdot \Big[\gamma_k - F_k(\hp^+) - \hp \cdot \big(1 - F_k(\hp^+)\big)\Big]
\geq e^{-\Q(\ggamma_k)} \cdot \Big[\gamma_k - \FF_k(\hp^+) - \hp \cdot \big(1 - \FF_k(\hp^+)\big)\Big]$, or equivalently,
\begin{equation}
\label{eq:abel_ineq3}
\gamma_k - \Big[F_k(\hp^+) + \hp \cdot \big(1 - F_k(\hp^+)\big)\Big] - (\hp - 1) \cdot \frac{\FF_k(\hp^+) - F_k(\hp^+)}{e^{\Q(\ggamma_k) - \Q(\gamma_k)} - 1} \geq 0.
\end{equation}
Again, notice that $\hp$ belongs to half-open interval $[\vv_k, \uu_k)$, Lemma~\ref{lem:main} is applicable here. After assigning $p \leftarrow \hp$ in Lemma~\ref{lem:main}, we have
\[
\text{LHS of~(\ref{eq:abel_ineq3})} = \text{LHS of~(\ref{eq:lem:main})} \geq \frac{1}{\kappa^*} \cdot \gamma^{*2} \cdot \Delta^* \geq 0.
\]
This completes the proof of Lemma~\ref{lem:abel_ineq1} when $p \in \big[1, \Phi_k(\uu_k^+)\big)$.
\end{proof}

\end{document}